\documentclass[11pt]{article} 
\usepackage[margin=1in]{geometry}

\usepackage{amsmath, amssymb, amsthm}
\usepackage{proof}
\usepackage{stmaryrd}
\usepackage{url} 

\newcommand{\logic}{$\mathcal{L}_{LF}$}
\newcommand{\dom}[1]{dom\left({#1}\right)}

\newcommand{\ie}{{\rm i.e.}}
\newcommand{\emptybb}{\cdot}
\newcommand{\emptycs}{\cdot}
\newcommand{\declinst}[4]{#1 \vdash #2 \leadsto_{\mbox{\sl \scriptsize dec}} #3 \bowtie #4}
\newcommand{\bsinst}[4]{#1 ; #2 \vdash #3 \leadsto_{\mbox{\sl \scriptsize bs}} #4}
\newcommand{\csinst}[4]{#1 ; #2 \vdash #3 \leadsto_{\mbox{\sl \scriptsize cs}} #4}
\newcommand{\csinstone}[4]{#1 ; #2 \vdash #3 \leadsto^1_{\mbox{\sl \scriptsize cs}} #4}

\newcommand{\emptycb}{\cdot}

\newcommand{\fvalid}[2]{#1 \vDash #2\ \mbox{\tt valid}}
\newcommand{\fatmann}[3]{{\fatm{#1}{#2}}^{#3}}
\newcount\counter
\newcommand{\repeatchar}[2]{
  \ifnum #2>-1
    \counter=0
    \loop
      \ifnum \counter<#2
        #1
        \advance \counter+1
    \repeat
  \fi
}
\newcommand{\ltann}[2][1]{{#2}^{\repeatchar{*}{#1}}}
\newcommand{\ltannaux}[2]{{#2}^{*^{#1}}}
\newcommand{\eqann}[2][1]{{#2}^{\repeatchar{@}{#1}}}
\newcommand{\eqannaux}[2]{{#2}^{@^{#1}}}
\newcommand{\assn}[3]{#1\lbrack{#2}\mapsfrom{#3}\rbrack}

\newcommand{\typedpi}[3]{\Pi {#1}{:}{#2}.{\mkern 3mu} #3}
\newcommand{\type}{\mbox{Type}}
\newcommand{\lflam}[2]{\lambda {#1}.{\mkern 3mu} #2}
\newcommand{\imp}[2]{#1\rightarrow #2}
\newcommand{\app}{\ }
\newcommand{\emptysig}{\cdot}
\newcommand{\emptyctx}{\cdot}

\newcommand{\lfprove}[3][\Sigma]{#2 \vdash_{#1} #3}
\newcommand{\lfsynthkind}[4][\Sigma]{\lfprove[#1]{#2}{#3\Rightarrow #4}}
\newcommand{\lfsynthtype}[4][\Sigma]{\lfprove[#1]{#2}{#3\Rightarrow #4}}
\newcommand{\lfchecktype}[4][\Sigma]{\lfprove[#1]{#2}{#3\Leftarrow #4}}
\newcommand{\lfkind}[3][\Sigma]{#2 \vdash_{#1} #3\ \mbox{\tt kind}}
\newcommand{\lftype}[3][\Sigma]{#2 \vdash_{#1} #3\ \mbox{\tt type}}
\newcommand{\lfctx}[2][\Sigma]{\vdash_{#1} #2\ \mbox{\tt ctx}}
\newcommand{\lfsig}[1]{\vdash {#1}\ \mbox{\tt sig}}

\newcommand{\domain}[1]{\mbox{\sl dom}(#1)}
\newcommand{\range}[1]{\mbox{\sl rng}(#1)}
\newcommand{\context}[1]{\mbox{\sl ctx}(#1)}
\newcommand{\supportof}[1]{\mbox{\sl supp}(#1)}
 
\newcommand{\subst}[2]{#2[#1]}
\newcommand{\hsubst}[2]{#2\llbracket {#1} \rrbracket} 
\newcommand{\hsubstseq}[3]{#3\llbracket {#2} \rrbracket_{#1}}
\newcommand{\hsubr}[4][]{#3\llbracket {#2} \rrbracket^r_{#1} = #4}
\newcommand{\hsub}[4][]{#3\llbracket {#2} \rrbracket_{#1} = #4}

\newcommand{\comp}[2]{#1\circ #2}

\newcommand{\seqsub}[2]{\langle{#1},{#2}\rangle}
\newcommand{\ctxvarminus}[2]{#1_{#2}}
\newcommand{\restrict}[2]{#1|_{#2}}

\newcommand{\STLCGamma}{\Theta}
\newcommand{\aritysum}[2]{#1 \uplus #2}
\newcommand{\stlcder}{\vdash_{\mbox{\sl \scriptsize at}}}
\newcommand{\stlctyjudg}[3]{#1 \stlcder #2 : #3}
\newcommand{\stlcderr}{\vdash^r_{\mbox{\sl\scriptsize at}}}
\newcommand{\stlctyjudgr}[3]{#1 \stlcderr #2 : #3}
\newcommand{\kindingder}{\vdash_{\mbox{\sl \scriptsize ak}}}
\newcommand{\wftype}[2]{#1 \kindingder #2\ \mbox{\sl type}}
\newcommand{\kindingderp}{\vdash^p_{\mbox{\sl \scriptsize ak}}}
\newcommand{\akindingp}[3]{#1 \kindingderp #2 : #3}
\newcommand{\declder}{\vdash_{\mbox{\sl \scriptsize dec}}}
\newcommand{\wfdecls}[3]{#1 \declder #2 \Rightarrow #3}

\newcommand{\abstyping}[1]{\vdash #1\ \mbox{\sl blk schema}}

\newcommand{\acstyping}[1]{\vdash #1\ \mbox{\sl ctx schema}}
\newcommand{\wfctx}[3]{#1;#2 \vdash #3\ \mbox{\sl context}}

\newcommand{\noms}{\mathcal{N}} 
\newcommand{\atyarr}{\rightarrow}
\newcommand{\arr}[2]{#1\rightarrow #2}
\newcommand{\erase}[1]{{(#1)}^{-}}
\newcommand{\oty}{o}
\newcommand{\ctxty}[2]{#1\lbrack {#2}\rbrack}
\newcommand{\ctxvarty}[3]{#1\!\uparrow\!#2 : #3}

\newcommand{\supp}[1]{\text{supp}\left({#1}\right)}

\newcommand{\ctxsanstype}[1]{#1^-}

\newcommand{\of}[2]{#1:#2}
\newcommand{\emptyce}{\cdot}
\newcommand{\fatm}[2]{\left\{#1 \vdash #2\right\}}
\newcommand{\fand}[2]{#1\wedge #2}
\newcommand{\for}[2]{#1\vee #2}
\newcommand{\fimp}[2]{#1\supset #2}
\newcommand{\fall}[2]{\forall #1. #2}
\newcommand{\fexists}[2]{\exists #1. #2}
\newcommand{\genericq}{\mathcal{Q}}
\newcommand{\fgeneric}[2]{\genericq #1. #2}
\newcommand{\fctx}[3]{\Pi\,#1 : #2.#3}

\newcommand{\ftrue}{\top}
\newcommand{\ffalse}{\bot}

\newcommand{\typdecomp}[4]{\mbox{\sl Decompose}\left(#1,#2, #3, #4\right)}
\newcommand{\typdecompsans}{\mbox{\sl Decompose}} 

\newcommand{\seq}[5][\mathbb{N}]{#1;#2;#3;#4\longrightarrow #5}
\newcommand{\seqsans}[3][\mathbb{N}]{#1;\emptyset;\emptyset;#2\longrightarrow#3}
\newcommand{\seqsansfmlas}[3]{#1;#2;#3;}
\newcommand{\seqsansctxts}[2]{#1\longrightarrow#2}

\newcommand{\setand}[2]{#1, #2}

\newcommand{\formeq}[4]{#1 \vdash #3 \equiv_{#2} #4}
\newcommand{\ctxeq}[4]{#1 \vdash #3 \equiv_{#2} #4}

\newcommand{\streq}[4]{#1 \vdash #3 \succeq_{#2} #4}

\newcommand{\ind}{\mbox{\sl ind}}
\newcommand{\cut}{\mbox{\sl cut}}

\newcommand{\id}{\mbox{\sl id}}
\newcommand{\absR}{\mbox{\sl atm-abs-R}}
\newcommand{\annabsR}{\mbox{\sl atm-abs-R}^*}
\newcommand{\absL}{\mbox{\sl atm-abs-L}}
\newcommand{\annabsL}{\mbox{\sl atm-abs-L}^*}
\newcommand{\appR}{\mbox{\sl atm-app-R}}
\newcommand{\annappR}{\mbox{\sl atm-app-R}^*}
\newcommand{\appL}{\mbox{\sl atm-app-L}}
\newcommand{\annappL}{\mbox{\sl atm-app-L}^*}
\newcommand{\topR}{\mbox{\sl $\top$-R}}
\newcommand{\botL}{\mbox{\sl $\bot$-L}}
\newcommand{\ctxR}{\mbox{\sl $\Pi$-R}}
\newcommand{\ctxL}{\mbox{\sl $\Pi$-L}}
\newcommand{\allR}{\mbox{\sl $\forall$-R}}
\newcommand{\allL}{\mbox{\sl $\forall$-L}}
\newcommand{\existsR}{\mbox{\sl $\exists$-R}}
\newcommand{\existsL}{\mbox{\sl $\exists$-L}}
\newcommand{\impR}{\mbox{\sl $\supset$-R}}
\newcommand{\impL}{\mbox{\sl $\supset$-L}}
\newcommand{\andR}{\mbox{\sl $\wedge$-R}}
\newcommand{\andL}{\mbox{\sl $\wedge$-L$_i$}}
\newcommand{\orR}{\mbox{\sl $\vee$-R$_i$}}
\newcommand{\orL}{\mbox{\sl $\vee$-L}}
\newcommand{\lfwk}{\mbox{\sl LF-wk}}
\newcommand{\lfstr}{\mbox{\sl LF-str}}
\newcommand{\lfperm}{\mbox{\sl LF-perm}}
\newcommand{\lfinst}{\mbox{\sl LF-inst}}
\newcommand{\sweak}{\mbox{\sl ctx-wk}}
\newcommand{\sstr}{\mbox{\sl ctx-str}}
\newcommand{\weakening}{\mbox{\sl wk}}
\newcommand{\contraction}{\mbox{\sl cont}}

\newcommand{\emptyseq}{\mbox{\sl nil}}
\newcommand{\consseq}[2]{ #1 :: #2}

\newcommand{\addblksans}{\mbox{\sl AddBlock}}
\newcommand{\addblk}[7]{\addblksans(#1,#2,#3,#4,#5,#6,#7)}
\newcommand{\allblks}[3]{\mbox{\sl AllBlocks}(#1,#2,#3)}
\newcommand{\implheadssans}{\mbox{\sl ImplicitHeads}}
\newcommand{\implheads}[2]{\implheadssans(#1,#2)}
\newcommand{\namessans}{\mbox{\sl NamesLsts}}
\newcommand{\names}[3]{\namessans(#1,#2,#3)}
\newcommand{\headssans}{\mbox{\sl Heads}}
\newcommand{\hds}[2]{\headssans(#1,#2)}
\newcommand{\decompseqsans}{\mbox{\sl ReduceSeq}}
\newcommand{\decompseq}[2]{\decompseqsans\left({#2}, #1\right)}

\newcommand{\makecasessans}{\mbox{\sl Cases}}
\newcommand{\makecases}[3][\fatm{G}{R:P}]{\makecasessans\left({#2}, #1, #3\right)}
\newcommand{\casessans}{\mbox{\sl AllCases}}
\newcommand{\casesfn}[2][\fatm{G}{R:P}]{\casessans \left(#2, #1\right)}

\newcommand{\unif}[3]{\left\langle{#1}, {#2}, {#3}\right\rangle}
\newcommand{\eqn}[2]{{#1}={#2}}
\newcommand{\solun}[2]{\langle{#1},{#2}\rangle}
\newcommand{\permute}[2]{{#1}.{#2}}
\newcommand{\inv}[1]{{#1}^{-1}}

\newcommand{\wfform}[3]{#1;#2 \vdash #3\ \mbox{\sl fmla}}

\newcommand{\wfctxvarty}[3]{#1 ; #2 \vdash #3 \mbox{ ctx-ty}}

\newcommand{\ctxtyinst}[5]{#1;#2; #3 \vdash #4 \leadsto_{\mbox{\sl \scriptsize csty}} #5}

\newcommand{\sigempty}{\mbox{\sl\small SIG\_EMPTY }}
\newcommand{\sigterm}{\mbox{\sl\small SIG\_TERM}}
\newcommand{\sigfam}{\mbox{\sl\small SIG\_FAM}}
\newcommand{\ctxempty}{\mbox{\sl\small CTX\_EMPTY}}
\newcommand{\ctxterm}{\mbox{\sl\small CTX\_TERM}}
\newcommand{\canonkindtype}{\mbox{\sl\small CANON\_KIND\_TYPE}}
\newcommand{\canonkindpi}{\mbox{\sl\small CANON\_KIND\_PI}}
\newcommand{\canonfamatom}{\mbox{\sl\small CANON\_FAM\_ATOM}}
\newcommand{\canonfampi}{\mbox{\sl\small CANON\_FAM\_PI}}
\newcommand{\atomfamconst}{\mbox{\sl\small ATOM\_FAM\_CONST}}
\newcommand{\atomfamapp}{\mbox{\sl\small ATOM\_FAM\_APP}}
\newcommand{\canontermatom}{\mbox{\sl\small CANON\_TERM\_ATOM}}
\newcommand{\canontermlam}{\mbox{\sl\small CANON\_TERM\_LAM}}
\newcommand{\atomtermvar}{\mbox{\sl\small ATOM\_TERM\_VAR}}
\newcommand{\atomtermconst}{\mbox{\sl\small ATOM\_TERM\_CONST}}
\newcommand{\atomtermapp}{\mbox{\sl\small ATOM\_TERM\_APP}}

\newcommand{\tpty}{\mbox{\sl tp}}
\newcommand{\unittm}{\mbox{\sl unit}}
\newcommand{\arrtm}{\mbox{\sl arr}}
\newcommand{\tmty}{\mbox{\sl tm}}
\newcommand{\emptytm}{\mbox{\sl empty}}
\newcommand{\lamtm}{\mbox{\sl lam}}
\newcommand{\apptm}{\mbox{\sl app}}
\newcommand{\ofty}{\mbox{\sl of}}
\newcommand{\ofemptytm}{\mbox{\sl of\_empty}}
\newcommand{\ofapptm}{\mbox{\sl of\_app}}
\newcommand{\oflamtm}{\mbox{\sl of\_lam}}
\newcommand{\eqty}{\mbox{\sl eq}}
\newcommand{\refltm}{\mbox{\sl refl}} 

\numberwithin{figure}{section}
\newtheorem{theorem}{Theorem}[section]
\newtheorem{lemma}[theorem]{Lemma}
\newtheorem{definition}[theorem]{Definition}

\title{A Logic for Reasoning About LF Specifications}
\author{Gopalan Nadathur \and Mary Southern}

\begin{document}
\maketitle

\begin{abstract}
\noindent We present a logic named \logic\ whose intended use is to
formalize properties of  specifications developed in the dependently
typed lambda calculus LF.  
The logic is parameterized by the LF signature that constitutes the
specification.
Atomic formulas correspond to typing derivations relative to this
signature. 
The logic includes a collection of propositional connectives and quantifiers.
Quantification ranges over expressions that denote LF terms and LF
contexts.
Quantifiers of the first variety are qualified by simple types that
describe the functional structure associated with the variables they
bind; deeper, dependency related properties are expressed by the body
of the formula.
Context-level quantifiers are qualified by \emph{context schemas} that
identify patterns of declarations out of which actual contexts may be 
constructed.
The semantics of variable-free atomic formulas is articulated via the
derivability in LF of the judgements they encode.
Propositional constants and connectives are understood in the usual
manner and the meaning of quantifiers is explicated through
substitutions of expressions that adhere to the type qualifications.
The logic is complemented by a proof system that enables reasoning
that is sound with respect to the described semantics.
The main novelties of the proof system are the provision for
case-analysis style reasoning about LF judgements, support for
inductive reasoning over the heights of LF derivations and the
encoding of LF meta-theorems.
The logic is motivated by the paradigmatic example of type assignment
in the simply-typed lambda calculus and the proof system is
illustrated through the formalization of a proof of type uniqueness
for this calculus. 

\end{abstract}

\section{Introduction}\label{sec:intro}

The Edinburgh Logical Framework of LF has been proposed as a vehicle
for formalizing rule-based presentation of object systems and has also
been successfully used in many such formalizations.
Two aspects of LF are critical to its use in such applications: its
basis in a regime of dependent types and its ability to support a
meta-level treatment of binding.
The typical deployment of LF in this context involves the
presentation of a signature that represents the object system.
Exploiting their parameterization by terms, types are then used to
encode relations between objects. 
Inhabitation of such a type, demonstrated through the
derivation of a typing judgement, provides witness to the validity of
a relation of interest.
The meta-level encoding of binding is exhibited in the form of typing
judgements in which typing contexts that scope over types and terms
play an important role.

The objective of this paper is to develop a framework for formalizing
and proving properties of specifications that are constructed in the
manner described above.
This framework is intended to provide a formal counterpart to the
informal style of reasoning about object systems based on their
rule-based description as is seen, for example, in \cite{pierce02book}.
To fit this role, the framework would need to be based on a logic that
allows for the expression of complex properties that are built on top
of derivability judgements whose validity is determined by the rules
describing the system. 
Moreover, the framework would need to accommodate case analysis and
inductive forms of reasoning, thereby enabling the least-fixed point
treatment of the rules that is typically assumed in the informal
arguments. 

Towards delivering on the mentioned objective, we describe here a logic
in which properties of LF specifications can be formalized.
The logic is parameterized by an LF signature that describes the
particular object system that is of interest.
The atomic formulas in the logic then comprise typing judgements that
assert the inhabitation of types relative to LF contexts.
Complex formulas can be constructed from these judgements using a
collection of propositional constants and connectives and
quantifiers.
The quantification that is permitted ranges over expressions that
denote both LF terms and LF contexts.
A critical issue that must be addressed in this context is the
qualification of quantifiers to limit their scope to meaningful
subclasses of expressions.
For term quantifiers, we address this issue by using a special class
of simple types that are referred to as \emph{arity types}.
Such a type limits attention to terms that satisfy a particular functional
structure, leaving deeper, dependency related properties to the domain
of the formula that the quantifier ranges over.
For context quantifiers, we introduce a new kind of types that we call
\emph{context schemas}.
These types, which are motivated by the \emph{regular worlds} used
with Twelf developments~\cite{Pfenning02guide,schurmann00phd},
identify patterns of declarations out of which the actual contexts
that instantiate the quantifiers must be constructed.
The validity of atomic formulas that are devoid of free variables is
determined, as might be expected, by the derivability of the LF
judgements that they represent.
Propositional symbols are understood in the usual manner.
An especially interesting part of the semantics is the treatment of
quantifiers: these are interpreted via the substitution of expressions
that adhere to the type qualifications that adorn them.

Perhaps the most significant contribution of this paper is a proof
system that can be used to establish the validity of formulas in 
the logic that is described.
This proof system is oriented around sequents that intuitively
encapsulate states in the development of a proof.
A basic part of this system is a set of rules that encode the
meanings of the propositional symbols and quantifiers.
The truly innovative part of the system is its treatment of atomic
formulas that embodies their interpretation as typing judgements in LF.
Included in this part is a mechanism for analyzing
an atomic assumption formula via the parameterizing signature and the
declarations in the LF context that is a part of the formula.
The rules also build in the capability to reason by induction on the
heights of LF derivations and to utilize meta-theorems about LF
derivability.
We show in this paper that the proof system is both coherent and
sound.
While we illustrate the logic and the proof system, we leave the
demonstration of their effectiveness to other work, e.g. see
\cite{adelfa.website,southern21lfmtp}.

The rest of the paper is organized as follows.
In the next section, we present the notions related to LF that
we need in other parts of this work.
A major part of this presentation is the treatment of \emph{simultaneous
hereditary substitution}, a concept that is important to much of the
technical development in this paper and that has also not been treated
in prior work.
Section~\ref{sec:logic} then presents the logic \logic\ and illustrates
its use in formalizing the type uniqueness property for the
simply-typed lambda calculus, a paradigmatic example for such work.
Section~\ref{sec:proof-system} describes the notion of a sequent and
develops a collection of proof rules for deriving sequents.
The main burden of the technical work in this section is the
demonstration of the coherence and soundness of the rules that are
presented. 
Section~\ref{sec:examples} illustrates the proof system by showing how
it can be used to encode the informal argument outlined in
Section~\ref{sec:logic}.
We conclude the paper in Section~\ref{sec:conclusion} by placing the
ideas in this paper in the context of other work related to reasoning
about LF specifications. 

\section{LF and the Formalization of Object Systems}
\label{sec:lf}

The methodology for modeling object systems in LF relies on adequacy
theorems that, in turn, depend on canonical $\beta\eta$-forms for
terms.
The original presentation of LF~\cite{harper93jacm} includes terms
in both canonical and non-canonical form.
Such a presentation simplifies the treatment of substitution but at
the price of complicating arguments concerning adequacy and LF
derivability.  
In light of this an alternative treatment of LF has been proposed that 
admits only terms that are in $\beta$-normal form and that are
well-typed only if they are additionally in $\eta$-long
form~\cite{harper07jfp,watkins03tr}.
We use this presentation of LF, called \emph{canonical LF}, as the
basis for our work. 
The first subsection below recalls this presentation and develops
notions related to it that will be used in the later parts of this
paper.  
Towards motivating the development of the logic that is the main
content of this paper, we then discuss the use of LF in representing
object systems and in reasoning about them at an informal level.
The section concludes with an identification of meta-theorems related
to derivability in LF that will be used in articulating proof rules in
our logic.

\subsection{Canonical LF}\label{ssec:lf-syntax}

Our presentation of canonical LF, henceforth referred to simply as LF,
differs from that in ~\cite{harper07jfp} in two respects.
First, we elide the subordination relation in typing judgements since
it is orthogonal to the thrust of this paper.
Second, we treat substitution independently of LF typing judgements
and we also extend the notion to include the simultaneous
replacement of multiple variables. 
The elaboration below builds in these ideas.

\begin{figure}[htpb]
\[
\begin{array}{r r c l}
  \mbox{\bf Kinds} & K & ::= & \type\ |\ \typedpi{x}{A}{K}\\[5pt]

  \mbox{\bf Canonical Type Families} & A,B & ::= &
           P\ |\ \typedpi{x}{A}{B}\\
  \mbox{\bf Atomic Type Families} & P & ::= & a\ |\ P\app M\\[5pt]
  \mbox{\bf Canonical Terms} & M,N & ::= & R\ |\ \lflam{x}{M}\\
  \mbox{\bf Atomic Terms} & R & ::= & c\ |\ x\ |\ R\app M\\[5pt]
  \mbox{\bf Signatures} & \Sigma & ::= &
  \emptysig\ |\ \Sigma,c:A\ |\ \Sigma,a:K\\[5pt]
  \mbox{\bf Contexts} & \Gamma & ::= & \emptyctx\ |\ \Gamma,x:A
\end{array}
\]
\caption{The Syntax of LF Expressions}
\label{fig:lf-terms}
\end{figure}

\subsubsection{The Syntax} The syntax of LF expressions is described in
Figure~\ref{fig:lf-terms}.
The primary interest is in three categories of expressions:
kinds, types which are indexed by kinds, and terms which are indexed
by types. 
In these expressions, $\lambda$ and $\Pi$ are binding or abstraction
operators.
Relative to these operators, we assume the principle of
equivalence under renaming that is applied as needed. 
We also assume as understood the notions of free and bound
variables that are usual to expressions involving such operators.
To ensure the absence of $\beta$-redexes, terms are stratified into
\emph{canonical} and \emph{atomic} forms.  
A similar stratification is used with types that is exploited by the
formation rules to force all well-typed terms to be in $\eta$-long
form.
We use $x$ and $y$ to represent term-level variables, which are bound
by abstraction operators or in the contexts that are associated with
terms. 
Further, we use $c$ and $d$ for term-level constants, and $a$
and $b$ for type-level constants, both of which are typed in
signatures. 
The expression $\imp{A_1}{A_2}$ is used as an alternative notation
for the type family $\typedpi{x}{A_1}{A_2}$ when $x$ does
not appear free in $A_2$. 
An atomic term has the form $(h\app M_1\app \ldots\app M_n)$ where $h$
is a variable or a constant.
We refer to $h$ as the \emph{head symbol} of such a term. 

\subsubsection{Simultaneous Hereditary Substitution}

We will need to consider substitution into LF expressions when
explicating typing and other logical notions related to these
expressions.
To preserve the form of these expressions, it is necessary to build
$\beta$-reduction into the application of such substitutions.
An important consideration in this context is that substitution
application must be a terminating operation. 
Towards ensuring this property, substitutions are indexed by types
that are eventually intended to characterize the functional structure
of expressions. 

\begin{definition}\label{def:aritytypes}
The collection of expressions that are obtained from the constant
$\oty$ using the binary infix constructor $\atyarr$ constitute the
{\it arity types}.
Corresponding to each canonical type $A$, there is an arity type
called its {\it erased form} and denoted by $\erase{A}$ and given as
follows: $\erase{P} = \oty$ and $\erase{\typedpi{x}{A_1}{A_2}} =
\arr{\erase{A_1}}{\erase{A_2}}$.  

\end{definition}

\begin{definition}\label{def:substitution}
A substitution $\theta$ is a finite set of 
the form $\{\langle x_1,M_1,\alpha_1 \rangle, \ldots, \langle x_n,
M_n, \alpha_n \rangle \}$, where, for $1 \leq i \leq n$, $x_i$ is a
distinct variable, $M_i$ is a canonical term and $\alpha_i$ is an
arity type.\footnote{Note that by a 
systematic abuse of notation, $n$ may be less than $m$ in a sequence written in
the form $s_m,\ldots,s_n$, in which case the empty sequence is
denoted. In this particular instance, a substitution can be an empty
set of triples.}
Given such a substitution, $\domain{\theta}$ denotes the set 
$\{x_1,\ldots,x_n\}$ and $\range{\theta}$ denotes the set
$\{M_1,\ldots,M_n\}$.
\end{definition}
 
Given a substitution $\theta$ and an expression $E$ that is a kind, a
type, a canonical term or a context, we wish the expression
$\hsubst{\theta}{E}$ notionally to denote the application of $\theta$
to $E$.
However, such an application is not guaranteed to exist.
We therefore use the expression $\hsub{\theta}{E}{E'}$ to indicate
when it is defined and has $E'$ as a result.
The key part of defining this relation is that of articulating
its meaning when $E$ is a canonical term.
This is done in Figure~\ref{fig:hsub} via rules for deriving this
relation. 
These rules use an auxiliary definition of substitution into an
atomic term which accounts for any normalization that is necessitated
by the replacement of a variable by a term.
The different categories of rules in this figure are distinguished by
being preceded by a box containing the judgement form they relate to.
The extension of this definition to the case where $E$ is a kind or a type
corresponds essentially to the application of the substitution to the
terms that appear within $E$. 
This idea is made explicit for types in Figure~\ref{fig:hsubtypes} and its
elaboration for kinds is similar.
A substitution is meaningfully applied to a context only when it does
not replace variables to which the context assigns types and when a
replacement does not lead to inadvertent capture.
When these conditions are satisfied, the substitution distributes to
the types that are assigned to the variables as the rules in
Figure~\ref{fig:hsubctx} make clear. 

\begin{figure}[tbhp]

\fbox{$\hsub{\theta}{M}{M'}$}

\begin{center}
\begin{tabular}{ccc}    
  \infer{\hsub{\theta}
              {R}
              {R'}}
        {\hsubr{\theta}
              {R}
              {R'}} \qquad
  & \qquad
  \infer{\hsub{\theta}
              {R}
              {M'}}
        {\hsubr{\theta}
              {R}
              {M':\alpha'}}

  & \qquad
\infer{\hsub{\theta}
            {(\lflam{x}{M})}
            {\lflam{x}{M'}}}
      {x\ \mbox{\rm not free in}\ \domain{\theta} \cup \range{\theta}
        \qquad
        \hsub{\theta}{M}{M'}}
\end{tabular}
\end{center}

\vspace{7pt}

\fbox{$\hsubr{\theta}{R}{M':\alpha'}$}

\begin{center}
\begin{tabular}{cc}
\infer{\hsubr{\theta}{x}{M:\alpha}}
      {\langle x,M,\alpha \rangle \in \theta}

& \qquad

\infer{\hsubr{\theta}{(R\app M)}{M''':\alpha''}}
      {\hsubr{\theta}{R}{\lflam{x}{M'}:\arr{\alpha'}{\alpha''}} 
       \qquad
       \hsub{\theta}{M}{M''}
       \qquad
       \hsub{\{\langle x, M'',\alpha'\rangle \}}{M'}{M'''}}
\end{tabular}
\end{center}

\vspace{7pt}

\fbox{$\hsubr{\theta}{R}{R'}$}

\begin{center}
\begin{tabular}{ccc}
\infer{\hsubr{\theta}{c}{c}}
      { }
& \qquad
\infer{\hsubr{\theta}
            {x}
            {x}}
      {x\not\in\domain{\theta}}

& \qquad 
\infer{\hsubr{\theta}{(R\app M)}{R'\app M'}}
       {\hsubr{\theta}{R}{R'} \qquad\qquad
        \hsub{\theta}{M}{M'}}
\end{tabular}
\end{center}

\caption{Applying Substitutions to Terms}
\label{fig:hsub}
\end{figure}

\begin{figure}[tbhp]
\begin{center}
\begin{tabular}{c}
\infer{\hsub{\theta}{a}{a}}
      { }

\qquad \qquad
  
\infer{\hsub{\theta}{(P\app M)}{(P'\app M')}}
      {\hsub{\theta}{P}{P'} &
       \hsub{\theta}{M}{M'}}

\\[15pt]

\infer{\hsub{\theta}{(\typedpi{x}{A_1}{A_2})}{\typedpi{x}{A_1'}{A_2'}}}
      {x\ \mbox{\rm not free in}\ \domain{\theta} \cup \range{\theta}
        \qquad
        \hsub{\theta}{A_1}{A_1'} \qquad
       \hsub{\theta}{A_2}{A_2'}}
\end{tabular}
\end{center}

\caption{Applying Substitutions to Types}
\label{fig:hsubtypes}
\end{figure}

\begin{figure}[tbhp]
\begin{center}
\begin{tabular}{c}
  \infer
    {\hsub{\theta}
          {\emptyctx}
          {\emptyctx}}
    {}

    \qquad
    
  \infer
    {\hsub{\theta}{(\Gamma, x:A)}{\Gamma', x:A'}}
    {x\ \mbox{\rm not free in}\ \domain{\theta} \cup \range{\theta} \qquad
     \hsub{\theta}{\Gamma}{\Gamma'} \qquad
     \hsub{\theta}{A}{A'}} 
\end{tabular}
\end{center}  
\caption{Applying Substitutions to Contexts}
\label{fig:hsubctx}
\end{figure}

We define a measure on substitutions that is useful in showing that
their application terminates. 

\begin{definition}\label{def:typesubsize}
The size of an arity type is the number of occurrences of $\atyarr$ in
it. The size of a substitution is the largest of the sizes of the arity types
in each of its triples.
\end{definition}

The following theorem can be proved by induction first on
the sizes of substitutions and then on the structures of expressions;
it would first be proved simultaneously for canonical and atomic terms
and then extended to types, kinds and contexts. 
\begin{theorem}\label{th:uniqueness}
For any context, kind, type or canonical term $E$ and any substitution
$\theta$, it is decidable whether there is an $E'$ such that
$\hsub{\theta}{E}{E'}$ is derivable.
Moreover, there is at most one $E'$ for which it is derivable.
Similarly, for any atomic term $R$ and substitution $\theta$, it is
decidable whether there is an $R'$ or an $M'$ and $\alpha'$ such
that $\hsubr{\theta}{R}{R'}$ or $\hsubr{\theta}{R}{M' : \alpha'}$ is
derivable, at most one of these judgements is derivable and that too
for a unique $R'$, respectively, $M'$ and $\alpha'$. 
\end{theorem}

The following property has an obvious proof by induction on the
structure of the expression.

\begin{theorem}\label{th:vacuoussubs}
If $E$ is a kind, a type or a canonical term none of whose free
variables is a member of $\domain{\theta}$, then $\hsub{\theta}{E}{E}$ has
as derivation. If $R$ is an atomic term none of whose free variables
is a member of $\domain{\theta}$ then $\hsubr{\theta}{R}{R}$ has a
derivation.
\end{theorem}

Simultaneous hereditary substitution enjoys a permutation 
property that is similar to the one described in \cite{harper07jfp}
for unitary substitution.
\begin{theorem}\label{th:subspermute}
Let $\theta_1$ be the substitution $\{\langle x_1, M_1,\alpha_1
\rangle, \ldots \langle x_n,M_n,\alpha_n \rangle \}$.
Further, let
$\theta_2$ be the substitution  $\{\langle y_1, N_1,\beta_1
\rangle, \ldots \langle y_m,N_m,\beta_m \rangle \}$ where
$y_1,\ldots,y_m$ are variables that are distinct from $x_1,\ldots,x_n$ 
and that do not appear free in $M_1,\ldots,M_n$.  
Finally, suppose that for each $i$, $1 \leq i \leq m$, there is some
$N'_i$ such that $\hsub{\theta_1}{N_i}{N'_i}$ has a derivation and let
$\theta_3 = \{ \langle y_1, N'_1,\beta_1 \rangle, \ldots \langle
y_m,N'_m,\beta_m \rangle \}$.
For every kind, type and canonical term $E$, $E_1$ and
$E_2$ such $\hsub{\theta_1}{E}{E_1}$ and $\hsub{\theta_2}{E}{E_2}$ have
derivations, there must be an $E'$ such 
that $\hsub{\theta_1}{E_2}{E'}$ and $\hsub{\theta_3}{E_1}{E'}$ have
derivations.
\end{theorem}
\begin{proof}
The proof proceeds by a primary induction on the sum of the sizes
of $\theta_1$ and $\theta_2$ and a secondary induction on the
derivation of $\hsub{\theta_2}{E}{E_2}$.
We omit the details which are similar to those for Lemma 2.10 in
\cite{harper07jfp}.
\end{proof}

\begin{figure}[tbhp]
\begin{center}

  \begin{tabular}{c}

  \infer
      {\stlctyjudgr{\STLCGamma}{c}{\alpha}}
      {c: \alpha \in \STLCGamma}

  \qquad
  \infer
      {\stlctyjudgr{\STLCGamma}{x}{\alpha}}
      {x: \alpha \in \STLCGamma}

  \qquad
  \infer
      {\stlctyjudgr{\STLCGamma}{R \app M}{\alpha}}
      {\stlctyjudgr{\STLCGamma}{R}{\alpha' \atyarr \alpha} \qquad
       \stlctyjudg{\STLCGamma}{M}{\alpha'}}

  \\[15pt]
  \infer
      {\stlctyjudg{\STLCGamma}{\lflam{x}{M}}{\alpha_1 \atyarr \alpha_2}}
      {\stlctyjudg{\aritysum{\{x:\alpha_1\}}{\STLCGamma}}{M}{\alpha_2}}

  \qquad
  \infer
      {\stlctyjudg{\STLCGamma}{R}{\oty}}
      {\stlctyjudgr{\STLCGamma}{R}{\oty}}
  \end{tabular}
\end{center}
\caption{Arity Typing for Canonical Terms}\label{fig:aritytyping}
\end{figure}

While the application of a substitution to an LF expression may not
always exist, this is guaranteed to be the case when certain arity
typing constraints are satisfied as we describe below.

\begin{definition}\label{def:aritytyping}
An \emph{arity context} $\STLCGamma$ is a set of unique assignments of
arity types to (term) constants and variables; these assignments are
written as $x:\alpha$ or $c:\alpha$.
Given two arity contexts $\STLCGamma_1$ and $\STLCGamma_2$, we write
$\aritysum{\STLCGamma_1}{\STLCGamma_2}$ to denote the collection of all the 
assignments in $\STLCGamma_1$ and the assignments in $\STLCGamma_2$ to
the constants or variables not already assigned a type in $\STLCGamma_1$. 
The rules in Figure~\ref{fig:aritytyping} define the arity typing
relation denoted by $\stlctyjudg{\STLCGamma}{M}{\alpha}$ between
a term $M$ and an arity type $\alpha$ relative to an arity context
$\STLCGamma$. 
A kind or type $E$ is said to respect an arity context $\STLCGamma$
under the following conditions: if $E$ is $\type$; if $E$ is an atomic
type and for each canonical term $M$ appearing in $E$ there is an
arity type $\alpha$ such that $\stlctyjudg{\STLCGamma}{M}{\alpha}$ is 
derivable; and if $E$ has the form $\typedpi{x}{A}{E'}$ and $A$
respects $\STLCGamma$ and $E'$ respects
$\aritysum{\{x:\erase{A}\}}{\STLCGamma}$. 
A context $\Gamma$ is said to respect $\STLCGamma$ if
for every $x:A$ appearing in $\Gamma$ it is the case that $A$ respects
$\STLCGamma$.
A substitution $\theta$ is {\it arity type preserving}
with respect to $\Theta$ if for every $\langle x,M,\alpha \rangle \in
\theta$ it is  the case that $\stlctyjudg{\STLCGamma}{M}{\alpha}$ is
derivable. 
Associated with a substitution $\theta$ is the arity context $\{ x :
\alpha\ \vert\ \langle x, M, \alpha \rangle \in \theta \}$ that is
denoted by $\context{\theta}$.
\end{definition}

\begin{theorem}\label{th:aritysubs}
Let $\theta$ be a substitution that is arity type preserving with
respect to $\STLCGamma$ and let $\STLCGamma'$ denote the arity context 
$\aritysum{\context{\theta}}{\STLCGamma}$. 
\begin{enumerate}
\item If $E$ is a canonical type or kind that respects the
  arity context $\STLCGamma'$, then there must be an $E'$ that
  respects $\STLCGamma$ and that is such that $\hsub{\theta}{E}{E'}$
  is derivable. 

\item If $M$ is a canonical term such that
$\stlctyjudg{\STLCGamma'}{M}{\alpha}$ is
derivable, then there must be an $M'$ such that
$\hsub{\theta}{M}{M'}$ and $\stlctyjudg{\STLCGamma}{M'}{\alpha}$ are
derivable.

\item If $R$ is an atomic term such that
  $\stlctyjudgr{\STLCGamma'}{R}{\alpha}$
  is derivable, then either there is an atomic term $R'$ such that
  $\hsub{\theta}{R}{R'}$ and $\stlctyjudgr{\STLCGamma}{R}{\alpha}$ are
  derivable or there is a canonical term $M$ such that
  $\hsub{\theta}{R}{M : \alpha}$ and
  $\stlctyjudg{\STLCGamma}{M}{\alpha}$ are derivable. 
\end{enumerate}
\end{theorem}
\begin{proof}
The first clause in the theorem is an easy consequence of the
second. We prove clauses (2) and (3) simultaneously by induction first
on the sizes of substitutions and then on the structure of terms.
The argument proceeds by considering the cases for the term structure,
first proving (3) and then using this in proving (2). 
\end{proof}

We will often consider expressions and substitutions that satisfy the
arity typing requirements of the theorem above, which then guarantees
that the applications of the substitutions have results.  
We introduce a notation that is convenient in this situation: we will
write $\hsubst{\theta}{E}$ to denote the unique $E'$ such that
$\hsub{\theta}{E}{E'}$ has a derivation whenever such a derivation is
known to exist.

\begin{definition}\label{def:composition}
Two substitutions $\theta_1$ and $\theta_2$ are said to be \emph{arity type
compatible} relative to the arity context $\STLCGamma$ if 
$\theta_2$ is type preserving with respect to $\STLCGamma$ and
$\theta_1$ is type preserving with respect to
$\aritysum{\context{\theta_2}}{\STLCGamma}$. The composition of two
such substitutions, written as $\theta_2 \circ \theta_1$, is 
the substitution
\begin{tabbing}
  \qquad\qquad\qquad\=\qquad\qquad\qquad\qquad\qquad\=\kill
  \> $\{ \langle x,M',\alpha \rangle\ \vert\ \langle x,M,\alpha
  \rangle \in \theta_1\ \mbox{\rm and}\ \hsub{\theta_2}{M}{M'}\ \mbox{\rm has a
    derivation}\}\ \cup$\\
  \>\>$\{ \langle y,N,\beta \rangle\ \vert\ \langle y,N,\beta \rangle
  \in \theta_2\ \mbox{\rm and}\ y \not\in \domain{\theta_1} \}$
\end{tabbing}
By Theorem~\ref{th:aritysubs} there must be an $M'$ for
which $\hsub{\theta_2}{M}{M'}$ has a derivation for each $\langle
x,M,\alpha \rangle \in \theta_1$. Moreover such an $M'$ must be 
unique. Thus, the composition described herein is well-defined. Note
also that the composition must also be arity type preserving with
respect to $\STLCGamma$.
\end{definition}

\begin{theorem}\label{th:composition}
Let $\theta_1$ and $\theta_2$ be substitutions that are arity type
compatible relative to $\STLCGamma$ and let $\STLCGamma'$ denote the
arity context $\aritysum{\context{\theta_2 \circ \theta_1}}{\STLCGamma}$.
\begin{enumerate}
\item If $E$ is a canonical kind, type or context that respects $\STLCGamma'$ and 
$E'$ and $E''$ are, respectively, canonical types or kinds such that 
$\hsub{\theta_1}{E}{E'}$ and $\hsub{\theta_2}{E'}{E''}$ have
  derivations, then $\hsub{\theta_2 \circ \theta_1}{E}{E''}$ has a
  derivation.
  
\item If $M$ is a canonical term such that, for some arity type
  $\alpha$, $\stlctyjudg{\STLCGamma'}{M}{\alpha}$ is derivable and $M'$
  and $M''$ are canonical terms such that 
$\hsub{\theta_1}{M}{M'}$ and $\hsub{\theta_2}{M'}{M''}$ have
derivations, then $\hsub{\theta_2 \circ \theta_1}{M}{M''}$ has a
derivation.

\item If $R$ is a canonical term such that, for some arity type
  $\alpha$, $\stlctyjudg{\STLCGamma'}{R}{\alpha}$ is derivable and
  \begin{enumerate}
    \item $M'$ and $M''$ are canonical terms such that
      $\hsubr{\theta_1}{R}{M' : \alpha}$ and $\hsub{\theta_2}{M'}{M''}$
      have derivations, then $\hsubr{\theta_2 \circ \theta_1}{R}{M'':\alpha}$
      has a derivation;

    \item $R'$ and $M''$ are, respectively, an atomic and a
      canonical term such that $\hsubr{\theta_1}{R}{R'}$ and
      $\hsubr{\theta_2}{R'}{M'' : \alpha}$ have derivations then
      $\hsubr{\theta_2 \circ \theta_1}{R}{M'' : \alpha}$ has a
      derivation; 

    \item $R'$ and $R''$ are atomic terms such that
      $\hsubr{\theta_1}{R}{R'}$ and $\hsubr{\theta_2}{R'}{R''}$ have
      derivations, then $\hsubr{\theta_2 \circ \theta_2}{R}{R''}$ has
      a derivation.
  \end{enumerate}
\end{enumerate}
\end{theorem}

\begin{proof}
Clause 1 of the theorem follows easily from an induction on the
structure of the canonical type or kind, assuming the property stated
in clause 2.

\smallskip
We prove clauses 2 and 3 together. These clauses are premised on the 
existence of a derivation corresponding to the application of the
substitution $\theta_1$ to either $M$ or $R$.
The argument is by induction on the size of this derivation and it
proceeds by considering the cases for the last rule in the
derivation. 

\smallskip
We consider first the cases where the derivation is for
$\hsub{\theta_1}{M}{M'}$; the clause in the theorem relevant to these
cases is 2. 
An easy argument using the induction hypothesis yields the desired
conclusion when $M$ is of the form $\lflam{x}{M_1}$.
In the case that $M$ is an atomic term, there is a shorter derivation
for $\hsubr{\theta_1}{M}{M' : \alpha}$ or $\hsubr{\theta_1}{M}{M'}$.
In the first case, the induction hypothesis, specifically clause 3(a),
allows us to conclude that $\hsub{\theta_2 \circ \theta_1}{M}{M''}$
has a derivation. 
In the second case, $M'$ must be an atomic term and there must
therefore be a derivation for $\hsubr{\theta_2}{M'}{M'':\alpha}$ or
$\hsubr{\theta_2}{M'}{M''}$.
Using the induction hypothesis, specifically clause 3(b) or 3(c), we
can again conclude that there must be a derivation for $\hsub{\theta_2
  \circ \theta_1}{M}{M''}$.

\smallskip
We consider next the cases for the last rule when the derivation is for
$\hsubr{\theta_1}{R}{M' : \alpha}$.
\begin{itemize}
\item If $M$ is a variable $x$ such that $\langle x,M',\alpha \rangle
  \in \theta_1$, it must be the case that $\langle x,M'',\alpha\rangle
  \in \theta_2 \circ \theta_1$.
Hence there must be a derivation for $\hsubr{\theta_2 \circ
  \theta_1}{M}{M'' : \alpha}$. 

\item Otherwise $M$ must be of the form $(R_1\app M_2)$ where
  $\hsub{\theta_1}{R_1}{\lflam{x}{M_3} : \alpha' \atyarr \alpha}$, 
  $\hsub{\theta_1}{M_2}{M_4}$, and
  $\hsub{\{\langle x, M_4,\alpha'\rangle \}}{M_3}{M'}$ have
  derivations for suitable choices for $M_3$, $\alpha'$ and
  $M_4$.
  We note first that
  $\aritysum{(\context{\theta_2 \circ \theta_1})}{\STLCGamma} =
  \aritysum{\context{\theta_1}}{(\aritysum{\context{\theta_2}}{\STLCGamma})}$. 
  Then, by the assumptions of the theorem and
  Theorem~\ref{th:aritysubs}, it follows that there must be terms
  $M'_3$ and $M'_4$ such that $\hsub{\theta_2}{M_3}{M'_3}$ and
  $\hsub{\theta_2}{M_4}{M'_4}$ have derivations.
  We see by using the induction hypothesis with respect to the derivation for
  $\hsub{\theta_1}{R_1}{\lflam{x}{M_3} : \alpha' \atyarr \alpha}$ that
  there must be a derivation for $\hsub{\theta_2 \circ 
    \theta_1}{R_1}{\lflam{x}{M'_3} : \alpha' \atyarr \alpha}$. 
  Using the induction hypothesis again with respect
  to the derivation for $\hsub{\theta_1}{M_2}{M_4}$, we see that
  there must be a derivation for $\hsub{\theta_2 \circ
    \theta_1}{M_2}{M'_4}$. By Theorem~\ref{th:subspermute} it follows
  that $\hsub{\{\langle  
      x, M'_4,\alpha' \rangle \}}{M'_3}{M''}$ has a derivation and,
  hence that $\hsubr{\theta_2 \circ \theta_1}{(R_1 \app M_2)}{M'' :
    \alpha}$ has one too. 
\end{itemize}

Finally we consider the cases for the last rule when the derivation is
for $\hsubr{\theta_1}{R}{R'}$.
The argument when $R$ is a constant is trivial.
The case when $R$ is a variable follows almost as immediately using
the definition of $\theta_2\circ\theta_1$.
The only remaining case is when $R$ is of the form $(R_1 \app M_2)$
and $R'$ is $(R'_1\app M'_2)$ where $\hsubr{\theta_1}{R_1}{R'_1}$ and
$\hsub{\theta_1}{M_2}{M'_2}$ have shorter derivations for suitable
terms $R'_1$ and $M'_2$.
We then have two subcases to consider with respect to the application
of $\theta_2$ to $(R'_1 \app M'_2)$:
\begin{itemize}
  \item There is a derivation for $\hsubr{\theta_2}{(R'_1
    \app M'_2)}{(R''_1 \app M''_2)}$ where $R''_1$ and $M''_2$ are terms
    such that $\hsubr{\theta_2}{R'_1}{R''_1}$ and
    $\hsubr{\theta_2}{M'_2}{M''_2}$ have derivations; note that the
    relevant clause in this case is 3(c) and $R''$ is $(R''_1 \app
    M''_2)$. 
    The induction hypothesis lets us conclude 
    that $\hsubr{\theta_2 \circ \theta_1}{R_1}{R''_1}$ and
    $\hsubr{\theta_2 \circ \theta_1}{M_2}{M''_2}$ have derivations.
    Hence, $\hsubr{\theta_2 \circ \theta_1}{(R_1\app M_2)}{(R''_1\app
      M''_2)}$ must have a derivation.

   \item There is a derivation for $\hsubr{\theta_2}{(R'_1\app
     M'_2)}{M'' : \alpha}$.
     In this case, for suitable choices for $M_3$,
     $\alpha'$ and $M_4$, there must be derivations for
     $\hsubr{\theta_2}{R'_1}{\lflam{x}{M_3} : \alpha' \atyarr \alpha}$,
     $\hsub{\theta_2}{M'_2}{M_4}$ and $\hsub{\{\langle x, M_4,\alpha'
     \rangle \}}{M_3}{M''}$.
     The induction hypothesis now lets us conclude that there are
     derivations for
     $\hsubr{\theta_2\circ\theta_1}{R_1}{\lflam{x}{M_3} : \alpha' 
       \atyarr \alpha}$ and $\hsub{\theta_2 \circ \theta_1}{M_2}{M_4}$.
     It then follows easily that there must be a derivation for
     $\hsubr{\theta_2 \circ \theta_1}{(R_1 \app M_2)}{M'' : \alpha}$.
    \end{itemize}
\end{proof}

The erased form of a type is invariant under substitution. This is the
content of the theorem below whose proof is straightforward.

\begin{theorem}\label{th:erasure}
For any type $A$ and substitution $\theta$, if $\hsub{\theta}{A}{A'}$
has a derivation, then $\erase{A} = \erase{A'}$.
\end{theorem}

\subsubsection{Wellformedness Judgements}

\begin{figure}[htbp]

\fbox{$\lfsig{\Sigma}$}
  
\begin{center}
\begin{tabular}{c}

\infer[\sigempty]{\lfsig{\emptysig}}{} \qquad
  
\infer[\sigterm]
      {\lfsig{\Sigma,c:A}}
      {\lfsig{\Sigma} \qquad \lftype{\emptyctx}{A} \qquad c\ \mbox{\rm does not
          appear in}\ \Sigma}

\\[10pt]

\infer[\sigfam]
      {\lfsig{\Sigma,a:K}}
      {\lfsig{\Sigma} \qquad \lfkind{\emptyctx}{K} \qquad a\ \mbox{\rm does not
          appear in}\ \Sigma}
\end{tabular}
\end{center}
      
\vspace{4pt}
\fbox{$\lfctx{\Gamma}$}

\begin{center}
\begin{tabular}{c}
    
\infer[\ctxempty]
      {\lfctx{\emptyctx}}{}

\qquad
      
\infer[\ctxterm]
      {\lfctx{\Gamma,x:A}}
      {\lfctx{\Gamma} \qquad \lftype{\Gamma}{A} \qquad x\ \mbox{\rm does not
          appear free in}\ \Gamma}
\end{tabular}
\end{center}

\vspace{4pt}
\fbox{$\lfkind{\Gamma}{K}$}

\begin{center}
\begin{tabular}{c}
\infer[\canonkindtype]
      {\lfkind{\Gamma}{\type}}{}

\qquad 
\infer[\canonkindpi]
      {\lfkind{\Gamma}{\typedpi{x}{A}{K}}}
      {\lftype{\Gamma}{A} \qquad
       \lfkind{\Gamma,x:A}{K}}
\end{tabular}
\end{center}

 \fbox{$\lftype{\Gamma}{A}$}

\begin{center}
\begin{tabular}{cc}       
\infer[\canonfamatom]
      {\lftype{\Gamma}{P}}
      {\lfsynthkind{\Gamma}{P}{\type}}\quad & \quad
\infer[\canonfampi]
      {\lftype{\Gamma}{\typedpi{x}{A_1}{A_2}}}
      {\lftype{\Gamma}{A_1} \qquad \lftype{\Gamma, x:A_1}{A_2}}
\end{tabular}
\end{center}

\vspace{4pt}
\fbox{$\lfsynthkind{\Gamma}{P}{K}$}

\begin{center}
\begin{tabular}{c}
\infer[\atomfamconst]
      {\lfsynthkind{\Gamma}{a}{K}}
      {a:K\in\Sigma}
\\[10pt]
\infer[\atomfamapp]
      {\lfsynthkind{\Gamma}{P\app M}{K}}
      {\lfsynthkind{\Gamma}{P}{\typedpi{x}{A}{K_1}} \qquad
       \lfchecktype{\Gamma}{M}{A} \qquad
      \hsub{\{\langle x, M, \erase{A}\rangle\}}{K_1}{K}}
\end{tabular}
\end{center}

\vspace{4pt}
\fbox{$\lfchecktype{\Gamma}{M}{A}$}

\begin{center}
\begin{tabular}{cc}
\infer[\canontermatom]
      {\lfchecktype{\Gamma}{R}{P}}
      {\lfsynthtype{\Gamma}{R}{P}} \quad & \quad 
\infer[\canontermlam]
      {\lfchecktype{\Gamma}{\lflam{x}{M}}{\typedpi{x}{A_1}{A_2}}}
      {\lfchecktype{\Gamma,x:A_1}{M}{A_2}}
\end{tabular}
\end{center}

\vspace{4pt}
\fbox{$\lfsynthtype{\Gamma}{R}{A}$}
    
\begin{center}
  \begin{tabular}{c}
\infer[\atomtermvar]
      {\lfsynthtype{\Gamma}{x}{A}}
      {x:A\in\Gamma}
\qquad
\infer[\atomtermconst]
      {\lfsynthtype{\Gamma}{c}{A}}
      {c:A\in\Sigma}
      
\\[10pt]
\infer[\atomtermapp]
      {\lfsynthtype{\Gamma}{R\app M}{A}}
      {\lfsynthtype{\Gamma}{R}{\typedpi{x}{A_1}{A_2}} \qquad
       \lfchecktype{\Gamma}{M}{A_1} \qquad
       \hsub{\{\langle x, M, \erase{A_1}\rangle\}}{A_2}{A}} 
 \end{tabular}
\end{center}

\caption{The Formation Rules for LF}
\label{fig:lf-judgements}
\end{figure}

Canonical LF includes seven judgements: $\lfsig{\Sigma}$ that ensures
that the constants declared in a signature are distinct and their type
or kind classifiers are well-formed; $\lfctx{\Gamma}$ that ensure that
the variables declared in a signature are distinct and their type
classifiers are well-formed in the preceding declarations and
well-formed signature $\Sigma$; $\lfkind{\Gamma}{K}$ that determines
that a kind $K$ is well-formed with respect to a well-formed signature
and context pair; $\lftype{\Gamma}{A}$ and
$\lfsynthkind{\Gamma}{P}{K}$ that check, respectively, the formation
of a canonical and atomic type relative to a well-formed 
signature, context and kind triple; and $\lfchecktype{\Gamma}{M}{A}$ and
$\lfsynthtype{\Gamma}{R}{A}$ that ensure, respectively, that a
canonical and atomic term are well-formed with respect to a
well-formed signature, context and canonical type triple. 
Figure~\ref{fig:lf-judgements} presents the rules for deriving
these judgements.
In the rules \canonkindpi\ and \canontermlam\ we assume $x$ to be a
variable that does not appear free in $\Gamma$.
The formation rule for type and term level application,
i.e. $\atomfamapp$ and $\atomtermapp$, require the substitution of a
term into a kind or a type.
Use is made towards this end of hereditary substitution. The index for
such a substitution is obtained by erasure from the type established
for the term.

The judgement forms other than $\lfsig{\Sigma}$ that are described
above are parameterized by a signature that remains unchanged in the
course of their derivation.
In the rest of this paper we will assume a fixed signature that has in
fact been verified to be well-formed at the outset. 
The judgement forms require some of their other components to satisfy
additional restrictions. 
For example, judgements of the forms $\lfchecktype{\Gamma}{M}{A}$ and
$\lfsynthtype{\Gamma}{R}{A}$ require that $\Sigma$, $\Gamma$ and $A$
be well-formed as an ensemble.
To be coherent, the rules in Figure~\ref{fig:lf-judgements} must
ensure that in deriving a judgement that satisfies these requirements,
it is necessary only to consider the derivation of judgements that
also accord with these requirements.
The fact that they possess this property can be verified by an
inspection of their structure, using the observation that will be made
in Theorem~\ref{th:transitivity} that hereditary substitution
preserves the property of being well-formed for kinds and types. 

Arity typing judgements for terms approximate LF typing judgements as
made precise below. 

\begin{definition}
  The arity context induced by the signature $\Sigma$  and context
  $\Gamma$ is the collection of assignments that includes $x :
  \erase{A}$ for each $x : A \in \Gamma$ and $c : \erase{A}$ for each
  $c : A \in \Sigma$.
  When the context $\Gamma$ is irrelevant or empty, we shall refer to
  the arity context as the one induced by just $\Sigma$.
\end{definition}

\begin{theorem}\label{th:arityapprox}
  Let $\STLCGamma$ be the arity context induced by the signature
  $\Sigma$ and context $\Gamma$.
  If $\lfctx{\Gamma}$ then $\Gamma$ respects $\Theta$.
  If $\lfkind{\Gamma}{K}$ or $\lftype{\Gamma}{A}$ then, respectively,
  $K$ or $A$ respect $\STLCGamma$.
  If $\lfchecktype{\Gamma}{M}{A}$ is derivable, then
  $\stlctyjudg{\STLCGamma}{M}{\erase{A}}$ must also be derivable.
  If $\lfsynthtype{\Gamma}{R}{A}$ is derivable, then
  $\stlctyjudg{\STLCGamma}{R}{\erase{A}}$ must also be derivable.
\end{theorem}

\begin{proof}
The last two parts of the theorem are proved simultaneously by
induction on the size of the derivation of
$\lfchecktype{\Gamma}{M}{A}$ and $\lfsynthtype{\Gamma}{R}{A}$. 
The first two parts follows from them, again by induction on the
derivation size.  
\end{proof}

\subsection{Formalizing Object Systems in LF}\label{ssec:informal-reasoning}

\begin{figure}[htbp]
\[
\begin{array}{lcl}
  \of{\tpty}{\type}
  & \quad\quad & 
  \of{\ofemptytm}{\ofty\app\emptytm\app\unittm}

  \\
  
  \of{\unittm}{\tpty}
  & &

  \\

  \of{\arrtm}{\arr{\tpty}{\tpty}}
  & &
  \of{\ofapptm}
     {\typedpi{E_1}{\tmty}{\typedpi{E_2}{\tmty}
     {\typedpi{T_1}{\tpty}{\typedpi{T_2}{\tpty}{}}}}}
  \\
  & &
    \qquad
    \typedpi{D_1}{\ofty\app E_1\app (\arrtm\app T_1\app T_2)}{
      \typedpi{D_2}{\ofty\app E_2\app T_1}{}}

  \\

  \of{\tmty}{\type}
  & &
  \qquad \ofty\app(\apptm\app E_1\app E_2)\app T_2

  \\

  \of{\emptytm}{\tmty}
  & &
  \\

  \of{\apptm}{\arr{\tmty}{\arr{\tmty}{\tmty}}} & &
    \of{\oflamtm}
       {\typedpi{R}
       {\arr{\tmty}{\tmty}}
       {\typedpi{T_1}{\tpty}{\typedpi{T_2}{\tpty}{}}}}
  \\

  \of{\lamtm}{\arr{\tpty}{\arr{(\arr{\tmty}{\tmty})}{\tmty}}}
  &  & 
  \qquad \typedpi{D}{(\typedpi{x}{\tmty}
          {\typedpi{y}{\ofty\app x\app T_1}
          {\ofty\app(R\app x)\app T_2}})}{}
  \\
  
  & & \qquad \ofty\app(\lamtm\app T_1\app(\lflam{x}{R\app x}))
                             \app(\arrtm\app T_1\app T_2)
  \\

  \of{\ofty}{\arr{\tmty}{\arr{\tpty}{\type}}} & & \\  

  \of{\eqty}{\arr{\tpty}{\arr{\tpty}{\type}}} & &
    \of{\refltm}{\typedpi{T}{\tpty}{\eqty\app T\app T}}
\end{array}
\]
\caption{An LF Specification for the Simply-Typed Lambda Calculus}
\label{fig:stlc-term-spec}
\end{figure}

A key use of LF is in formalizing systems that are described 
through relations between objects that are specified through a
collection of inference rules. 
In the paradigmatic approach, each such relation is represented by a
dependent type whose term arguments are encodings of objects that
might be in the relation in question.
The inference rules translate in this context into term constructors
for the type representing the relation.
We illustrate these ideas through an encoding of the typing relation
for the simply-typed $\lambda$-calculus (STLC), a running example for this
paper. 

We assume the reader to be familiar with the types and terms in the
STLC and also with the rules that define its typing relation.
Figure~\ref{fig:stlc-term-spec} presents an LF signature that serves as
an encoding of this system. 
This encoding uses the higher-order abstract syntax approach to
treating binding. 
The specification introduces two type families, $\tpty$ and $\tmty$ to 
represent the simple types and $\lambda$-terms.
Additionally, for each expression form in the object system, it
includes a constant that produces a term of type $\tpty$ or $\tmty$;
as should be apparent from the declarations, we have assumed an object
language whose terms are constructed from a single constant of
atomic type that is represented by the LF constant \emptytm\ and whose
type is represented by the LF constant \unittm.
This signature also provides a representation of two relations over
object language expressions: typing between terms and types and
equality between types. 
Specifically, the type-level constants $\ofty$ and $\eqty$ are
included towards this end. 
The rules defining the relations of interest in the object system are
encoded by constants in the signature.
The types associated with these constants ensure that well-formed
terms of atomic type that are formed using the constants correspond to
derivations of the relation in the object language that is represented
by the type. 

One of the purposes for constructing a specification is to use it to
prove properties about the object system.
For example, we may want to show that when a type can be associated
with a term in the STLC, it must be unique.
Based on our encoding, this property can be stated as the following
about typing derivations in LF:
\begin{quotation}
\noindent For any terms $M_1,M_2,E,T_1,T_2$, if there are LF
derivations for
$\lfchecktype{}{E}{\tmty}$, $\lfchecktype{}{T_1}{\tpty}$,
$\lfchecktype{}{T_2}{\tpty}$,
$\lfchecktype{}{M_1}{\ofty\app E\app T_1}$ and
$\lfchecktype{}{M_2}{\ofty\app E\app T_2}$, then there must be a term
$M_3$ such than there is a derivation for
$\lfchecktype{}{M_3}{\eqty\app T_1\app T_2}$. 
\end{quotation}
To prove this property, we would obviously need to unpack its logical
structure.
We would also need to utilize an understanding of LF in
analyzing the hypothesized typing derivations corresponding to the
STLC typing judgements. 
Considering the case where $E$ is an abstraction will lead us to actually
wanting to prove a more general property:
\begin{quotation}
\noindent For any terms $M_1,M_2,E,T_1,T_2$ and contexts $\Gamma$, if
there are LF derivations for the judgements
$\lfchecktype{\Gamma}{E}{\tmty}$, $\lfchecktype{\Gamma}{T_1}{\tpty}$,
$\lfchecktype{\Gamma}{T_2}{\tpty}$,
$\lfchecktype{\Gamma}{M_1}{\ofty\app E\app T_1}$ and
$\lfchecktype{\Gamma}{M_2}{\ofty\app E\app T_2}$, 
then there must be a term $M_3$ such than there is a derivation
for $\lfchecktype{}{M_3}{\eqty\app T_1\app T_2}$. 
\end{quotation}
Now, this property is not provable without some constraints on the form
of contexts.
In this example, it suffices to prove it when $\Gamma$ is restricted
to being of the form
\begin{center}
 $(x_1:\tmty,y_1:\ofty\app x_1\app{Ty}_1,\ldots,x_n:\tmty,y_n:\ofty\app x_n\app{Ty}_n)$.
\end{center}
\noindent In completing the argument, we would need to use properties of LF
derivability.
A property that would be essential in this case is the finiteness of
LF derivations, which enables us to use an inductive argument.

The objective in this paper is to provide a formal mechanism for
carrying out such analysis.
We do this by describing a logic that is suitable for this purpose.
One of the requirements of this logic is that it should permit the
expression of the kinds of properties that arise in the process of
reasoning.
Beyond this, it should further be possible to complement the statement
of properties with inference rules that permit the encoding of
interesting and sound forms of reasoning. 

\subsection{Meta-Theoretic Properties about LF Derivability}
\label{ssec:lf-properties}

Our reasoning system will need to embody an understanding of
derivability in LF.  
We describe some properties related to this notion here that will be 
useful in this context. 
The first three theorems, which express structural properties
about derivations, have easy proofs.
The fourth theorem states a subsitutivity property for wellformedness
judgements.
This theorem is proved in \cite{harper07jfp}.

\begin{theorem}\label{th:weakening}
If $\mathcal{D}$ is a derivation for $\lfkind{\Gamma}{K}$,
$\lftype{\Gamma}{A}$ or $\lfchecktype{\Gamma}{M}{A}$, then, for any
variable $x$ that is fresh to the judgement and for any $A'$ such that
$\lftype{\Gamma}{A'}$ is derivable, there is a derivation,
respectively, for $\lfkind{\Gamma, x:A'}{K}$,
$\lftype{\Gamma, x : A'}{A}$ or $\lfchecktype{\Gamma, x: A'}{M}{A}$
that has the same structure as $\mathcal{D}$. 
\end{theorem}

\begin{theorem}\label{th:strengthening}
  If $\mathcal{D}$ is a derivation for judgements $\lfkind{\Gamma, x:A'}{K}$,
$\lftype{\Gamma, x : A'}{A}$ or $\lfchecktype{\Gamma, x: A'}{M}{A}$
  and $x$ is a variable that does not appear free in $K$, $A$, or $M$
  and $A$ respectively, then there must be a derivation that has the
  same structure as $\mathcal{D}$ for $\lfkind{\Gamma}{K}$,
  $\lftype{\Gamma}{A}$ or $\lfchecktype{\Gamma}{M}{A}$, respectively.
\end{theorem}

\begin{theorem}\label{th:exchange}
If $x$ does not appear in $A_2$ then
$\Gamma_1,y:A_2,x:A_1,\Gamma_3$ is a well-formed context with respect
to a signature $\Sigma$ whenever $\Gamma_1,x:A_1,y:A_2,\Gamma_3$ is. 
Further, if there is a derivation $\mathcal{D}$ for
$\lfkind{\Gamma, x:A_1, y:A_2,\Gamma_2}{K}$,
$\lftype{\Gamma, x:A_1, y:A_2,\Gamma_2}{A}$ or $\lfchecktype{\Gamma,
  x:A_1, y:A_2,\Gamma_2}{M}{A}$, then there must be a derivation that
has the same structure as $\mathcal{D}$ for
$\lfkind{\Gamma, y:A_2, x:A_1,\Gamma_2}{K}$,
$\lftype{\Gamma, y:A_2, x:A_1,\Gamma_2}{A}$ or $\lfchecktype{\Gamma,
  y:A_2, x:A_1,\Gamma_2}{M}{A}$, respectively.
\end{theorem}

\begin{theorem}\label{th:transitivity}
Assume that $\lfctx{\Gamma_1,x_0:A_0,\Gamma_2}$ and
$\lfchecktype{\Gamma_1}{M_0}{A_0}$ have derivations, and let $\theta$ be
the substitution $\{ \langle x_0, M_0,\erase{A_0} \rangle \}$. 
Then there is a $\Gamma'_2$ such that
$\hsub{\theta}{\Gamma_2}{\Gamma'_2}$ and $\lfctx{\Gamma_1,\Gamma'_2}$
have derivations. 
Further,
\begin{enumerate}
  \item if $\lfkind{\Gamma_1,x_0:A_0,\Gamma_2}{K}$ has a derivation,
    then there is a $K'$ such that $\hsub{\theta}{K}{K'}$ and
    $\lfkind{\Gamma_1,\Gamma'_2}{K'}$ have derivations;

  \item if $\lftype{\Gamma_1,x_0:A_0,\Gamma_2}{A}$ has a derivation,
    then there is an $A'$ such that $\hsub{\theta}{A}{A'}$ and
    $\lftype{\Gamma_1,\Gamma'_2}{A'}$ have derivations; and

  \item if $\lfchecktype{\Gamma_1,x_0:A_0,\Gamma_2}{M}{A}$ has a
    derivation (for some well-formed type $A$), there is an $A'$ and
    an $M'$ such that $\hsub{\theta}{A}{A'}$, $\hsub{\theta}{M}{M'}$,
    and $\lfchecktype{\Gamma_1,\Gamma_2'}{M'}{A'}$ have derivations.
\end{enumerate}
\end{theorem}

The reasoning system will need to build in a means for analyzing
typing derivations of the form $\lfchecktype{\Gamma}{M}{A}$.
This analysis will be driven by the structure of the type $A$.
The decomposition when $A$ is of the form $\typedpi{x_1}{A_1}{A_2}$
has an obvious form.
The development below, culminating in Theorem~\ref{th:atomictype},
provides the basis for the analysis when $A$ is an atomic type. 

\begin{lemma}\label{lem:arityrespecting}
Let $\Gamma$ be a context such that $\lfctx{\Gamma}$ has a derivation
and let $\STLCGamma$ be the arity context induced by $\Sigma$ and
$\Gamma$. 
Suppose that $\typedpi{y_1}{A_1}{\ldots\typedpi{y_n}{A_n}{A}}$ is a
type associated with a (term) constant or variable by $\Sigma$ or
$\Gamma$, or that $\typedpi{y_1}{A_1}{\ldots\typedpi{y_n}{A_n}{K}}$ is
a kind associated with a (type) constant by $\Sigma$, where the $y_i$s
are distinct variables.  
Then, for $1 \leq i \leq n$, $A_i$ and
$\typedpi{y_{i}}{A_{i}}{\ldots\typedpi{y_n}{A_n}{A}}$ or,
respectively, $\typedpi{y_{i}}{A_{i}}{\ldots\typedpi{y_n}{A_n}{K}}$  
respect the arity context $\aritysum{\{y_1 : \erase{A_1}, \ldots, y_{i-1} :
  \erase{A_{i-1}}\}}{\STLCGamma}$.
Further, $A$ or, respectively, $K$ respects the arity context 
$\aritysum{\{y_1 : \erase{A_1}, \ldots, y_{n} : \erase{A_{n}}\}}{\STLCGamma}$.
\end{lemma}

\begin{proof}
  Since $\Sigma$ and $\Gamma$ are well-formed by assumption, depending
  on the case under consideration, either
  $\lftype{\Gamma}{\typedpi{y_1}{A_1}{\ldots\typedpi{y_n}{A_n}{A}}}$  
  or $\lfkind{\emptyctx}{\typedpi{y_1}{A_1}{\ldots\typedpi{y_n}{A_n}{K}}}$
  must have a derivation.
  The desired conclusions now follow from Theorem~\ref{th:arityapprox}
  and Definition~\ref{def:aritytyping}.
\end{proof} 

\begin{lemma}\label{lem:lfcomp}
  Let $\Gamma_1$ be a context such that $\lfctx{\Gamma_1}$ has a
  derivation, let $\STLCGamma$ be the arity context induced by $\Sigma$
  and $\Gamma_1$, and let $\theta$ be a substitution that is arity type
  preserving with respect to $\STLCGamma$. Further, let $x_0$ be a
  variable that is neither bound in $\Gamma_1$ nor a member of 
  $\domain{\theta}$,  let $A_0$ and $M_0$ be such that
  $\lftype{\Gamma_1}{A_0}$ and $\lfchecktype{\Gamma_1}{M_0}{A_0}$ are
  derivable and let $\theta' = \theta \cup \{\langle
  x_0,M_0,\erase{A_0}\rangle \}$.
  \begin{enumerate}
    \item $\theta'$ is arity type preserving with respect to
      $\STLCGamma$.

    \item Let $\Gamma_2$ be a context that respects an arity context
      $\STLCGamma'$ such that 
      $\aritysum{\context{\theta'}}{\STLCGamma}\subseteq \STLCGamma'$ and let
      $\Gamma'_2$ be a context such that
      $\hsub{\theta}{\Gamma_2}{\Gamma'_2}$, and
      $\lfctx{\Gamma_1, x_0 : A_0, \Gamma'_2}$ have derivations. Then
      there is a context $\Gamma_2''$ such that the following hold:

      \begin{enumerate}
      \item $\hsub{\{\langle x_0,M_0,\erase{A_0}\rangle \}}
                  {\Gamma'_2}
                  {\Gamma''_2}$,
            $\hsub{\theta'}{\Gamma_2}{\Gamma''_2}$ and
            $\lfctx{\Gamma, \Gamma''_2}$ have derivations;

      \item if $K$ is a kind that also respects $\STLCGamma'$ and $K'$
        is a kind such that there are derivations for $\hsub{\theta}{K}{K'}$ and
        $\lfkind{\Gamma_1,x_0:A_0,\Gamma'_2}{K'}$, 
        then there is a kind $K''$ such that
        $\hsub{\{\langle x_0,M_0,\erase{A_0}\rangle \}}{K'}{K''}$,
        $\hsub{\theta'}{K}{K''}$ and
        $\lfkind{\Gamma_1, \Gamma''_2}{K''}$ are derivable; and

      \item if $A$ is a type that also respects $\STLCGamma'$ and $A'$
        is a type such that there are derivations for
        $\hsub{\theta}{A}{A'}$ and
        $\lftype{\Gamma_1,x_0:A_0,\Gamma'_2}{A'}$,
        then there is a type $A''$ such that
        $\hsub{\{\langle x_0,M_0,\erase{A_0}\rangle \}}{A'}{A''}$,
        $\hsub{\theta'}{A}{A''}$ and $\lftype{\Gamma_1,
        \Gamma''_2}{A''}$ have derivations.
      \end{enumerate}
   \end{enumerate}
\end{lemma}    

\begin{proof}
Since $\lfchecktype{\Gamma_1}{M_0}{A_0}$ has a derivation, it follows
from Theorem~\ref{th:arityapprox} that $\{\langle
x_0,M_0,\erase{A_0}\rangle\}$ is type preserving with respect to $\STLCGamma$.
It then follows from the assumptions in the lemma that $\theta'$ is in fact
$\{\langle x_0, M_0, \erase{A_0} \rangle \} \circ \theta$ and type
preserving with respect to $\STLCGamma$. The various observations in clause 2 
now follow from Theorems~\ref{th:composition} and \ref{th:transitivity}.
\end{proof}

\begin{theorem}\label{th:atomictype}
Let $\Gamma$ be a context such that $\lfctx{\Gamma}$ has a derivation.
\begin{enumerate}
\item $\lfsynthtype{\Gamma}{R}{A'}$  has a derivation 
  if
  \begin{enumerate}
  \item $R$ is of the form $(c \app M_1 \app \ldots\app M_n)$ for some
    $c:\typedpi{y_1}{A_1}{\ldots \typedpi{y_n}{A_n}{A}} \in \Sigma$ or
    of the form $(x \app M_1 \app \ldots\app M_n)$ for some
    $x:\typedpi{y_1}{A_1}{\ldots \typedpi{y_n}{A_n}{A}} \in \Gamma$,
    
  \item there is a sequence of types $A'_1,\ldots,A'_n$ such that, for
    $1\leq i\leq n$, there are derivations for
    $\hsub{\{\langle y_1, M_1,\erase{A_1}\rangle, \ldots,
             \langle y_{i-1}, M_{i-1}, \erase{A_{i-1}} \rangle\}}
          {A_i}
          {A'_i}$
    and $\lfchecktype{\Gamma}{M_i}{A'_i}$, and
         
  \item $\hsub{\{\langle y_1, M_1,\erase{A_1}\rangle, \ldots,
                 \langle y_n, M_n, \erase{A_n} \rangle\}}
              {A}
              {A'}$
    and $\lftype{\Gamma}{A'}$ have derivations.
  \end{enumerate}
  
\item $\lfsynthtype{\Gamma}{R}{A'}$  has a derivation of height $h$
  only if
  \begin{enumerate}
  \item $R$ is of the form $(c \app M_1 \app \ldots\app M_n)$
    for some $c:\typedpi{y_1}{A_1}{\ldots \typedpi{y_n}{A_n}{A}} \in \Sigma$
    or of the form $(x \app M_1 \app \ldots\app M_n)$ for some
    $x:\typedpi{y_1}{A_1}{\ldots \typedpi{y_n}{A_n}{A}} \in \Gamma$,

  \item there is a sequence of types $A'_1,\ldots,A'_n$ such that, for
    $1 \leq i \leq n$, there is a derivation for 
    $\hsub{\{\langle y_1, M_1,\erase{A_1}\rangle, \ldots,
             \langle y_{i-1}, M_{i-1}, \erase{A_{i-1}} \rangle\}}
          {A_i}
          {A'_i}$
    and a derivation of height less than $h$ for
    $\lfchecktype{\Gamma}{M_i}{A'_i}$, and  

  \item $\hsub{\{\langle y_1, M_1,\erase{A_1}\rangle, \ldots,
                 \langle y_n, M_n, \erase{A_n} \rangle\}}
              {A}
              {A'}$
    and $\lftype{\Gamma}{A'}$ have derivations.
  \end{enumerate}      
\end{enumerate}
\end{theorem}

\begin{proof} At the outset, we should check the coherence of clauses
  1(b) and 2(b) 
  in the theorem statement by verifying that, for $1 \leq i \leq n$, it
  is the case that $\lftype{\Gamma}{A'_i}$ has a derivation.
  Towards this end, we first note that there must be a derivation for
  $\lftype{\Gamma, y_1 : A_1, \ldots, y_{i-1} : A_{i-1}}{A_i}$ since
  $\Sigma$ and $\Gamma$ are well-formed. 
  The desired conclusion then follows from using
  Lemma~\ref{lem:lfcomp} repeatedly and observing, via
  Theorem~\ref{th:erasure}, that erasure is preserved under
  substitution.

\smallskip
\noindent We now introduce some notation that will be useful in the
arguments that follow.
We will use $\STLCGamma$ to denote the arity context induced by
$\Sigma$ and $\Gamma$.
Further, for $1 \leq i \leq n+1$, we will write $\theta_i$ for the  
substitution $\{ \langle y_1,M_1,\erase{A_1} \rangle, \ldots, \langle
y_{i-1}, M_{i-1}, \erase{A_{i-1}}\rangle \}$.
An observation that we will make use of below is that if for $1 \leq j
< i$ it is the case that $\lfchecktype{\Gamma}{M_i}{A'_i}$ has a
derivation, then $\theta_i$ is type preserving with respect to
$\STLCGamma$.
This is an easy consequence of Theorems~\ref{th:arityapprox} and
\ref{th:erasure}.

\smallskip
\noindent {\it Proof of (1).}
  We will consider explicitly only the case where $R$ is $(c\app M_1 \app
  \ldots \app M_n)$; the argument for the case when $R$ is $(x\app M_1
  \app \app \ldots \app M_n)$ is similar.
  We will show for $1 \leq i \leq n+1$ that, under the conditions
  assumed for $M_1,\ldots, M_{i-1}$, there is a type $A''_{i}$ such
  that $\hsub{\theta_{i}}
             {(\typedpi{y_{i}}{A_{i}}{\ldots\typedpi{x_n}{A_n}{A}})}
             {A''_{i}}$,
  $\lftype{\Gamma}{A''_{i}}$ and
  $\lfsynthtype{\Gamma}{(c\app M_1 \app \ldots \app M_{i-1})}{A''_{i}}$ have
  derivations. 
  The desired conclusion follows from noting that $A'$ must be
  $A''_{n+1}$ because the result of substitution application is
  unique.  

  The claim is proved by induction on $i$.
  Consider first the case when $i$ is $1$.
  Since $\theta_1 = \emptyset$, $A''_{1}$ is
  $\typedpi{y_1}{A_1}{\ldots \typedpi{y_n}{A_n}{A}}$. 
  The wellformedness of $\Sigma$ ensures that $\lftype{\Gamma}{A''_1}$
  has a derivation and we get a derivation for
  $\lfsynthtype{\Gamma}{c}{A'}$ by using an \atomtermconst\ rule.

  Let us then assume the claim for $i$ and show that it
  must also hold for $i+1$.
  By the hypothesis, there is an $A''_i$ of the form 
  $\typedpi{y_i}{A'_i}{A''}$ where $A''$ is such that 
  $\hsub{\theta_i}{\typedpi{y_{i+1}}{A_{i+1}}{\ldots\typedpi{x_n}{A_n}{A}}}{A''}$
  has a derivation.
  Since $\lftype{\Gamma}{A''_i}$ has a derivation, so must
  $\lftype{\Gamma, y_i : A'_i}{A''}$.
  By Lemma~\ref{lem:arityrespecting},
  $\typedpi{y_{i+1}}{A_{i+1}}{\ldots\typedpi{x_n}{A_n}{A}}$
  respects the arity context
  $\aritysum{\{y_1 : \erase{A_1},\ldots, y_i : \erase{A_i}\}}{\STLCGamma}$.
  Since there are derivations for $\lfchecktype{\Gamma}{M_j}{A'_j}$ for $1 \leq
  j < i$, $\theta_i$ is type preserving over $\STLCGamma$.
  We now invoke Lemma~\ref{lem:lfcomp} to conclude that there is a
  term $A'''$ such that there are derivations for 
  $\hsub{\{ \langle y_i, M_i,\erase{A'_i} \rangle \}}{A''}{A'''}$,
  $\hsub{\theta_{i+1}}{\typedpi{y_{i+1}}{A_{i+1}}{\ldots\typedpi{x_n}{A_n}{A}}}{A'''}$,
  and $\lftype{\Gamma}{A'''}$.
  By the hypothesis, there is a derivation for
  $\lfsynthtype{\Gamma}{c\app M_1 \app \ldots \app
    M_{i-1}}{\typedpi{y_i}{A'_i}{A''}}$. 
  Using an \atomtermapp\ rule together with this derivation and the
  ones for  
  $\lfchecktype{\Gamma}{M_i}{A'_i}$, and 
  $\hsub{\{ \langle y_i, M_i,\erase{A'_i} \rangle \}}{A''}{A'''}$, we
  get a derivation for $\lfsynthtype{\Gamma}{(c\app M_1 \app \ldots \app
    M_i)}{A'''}$.
  Letting $A''_{i+1}$ be $A'''$ we see that all the requirements are satisfied.

\smallskip
\noindent {\it Proof of (2).} We prove the claim by induction on
the height of the derivation of $\lfsynthtype{\Gamma}{R}{A'}$. We 
consider the cases for the last rule used in the derivation.
If this rule is \atomtermvar\ or \atomtermconst, the argument is
straightforward. 
The only case to be considered further, then, is that when the rule is
\atomtermapp.

In this case, we know that $R$ must be of the form $(R'\app M')$
where there is a shorter derivation for $\lfsynthtype{\Gamma}{R'}{B'}$
for some type $B'$.
From the induction hypothesis, it follows that $R'$ has the form
$(c\app M_1\app \ldots\app M_n)$ or $(x\app M_1\app \ldots\app M_n)$ 
for some $c : \typedpi{y_1}{A_1}{\ldots \typedpi{y_n}{A_n}{B}} \in
  \Sigma$ or $x : \typedpi{y_1}{A_1}{\ldots \typedpi{y_n}{A_n}{B}} \in
  \Gamma$ and that there must be a sequence of types
  $A'_1,\ldots,A'_n$ that, together with the terms $M_1,\ldots,M_n$
  satisfy the requirements stated in clause 2(b).
Moreover, $B'$ must be such that $\hsub{\theta_{n+1}}{B}{B'}$ and
$\lftype{\Gamma}{B'}$ have derivations.
Since the rule is an \atomtermapp, $B'$ must have the structure of an
abstracted type.
From this it follows that $B$ must be of the form
$\typedpi{y_{n+1}}{A_{n+1}}{A}$ and, correspondingly, $B'$ must be of
the form $\typedpi{y_{n+1}}{A'_{n+1}}{A''}$ where
$\hsub{\theta_{n+1}}{A_{n+1}}{A'_{n+1}}$ and $\hsub{\theta_{n+1}}{A}{A''}$ have
derivations.
Noting that the type of $c$ or $x$ is really of the form
$\typedpi{y_1}{A_1}{\ldots \typedpi{y_{n+1}}{A_{n+1}}{A}}$ it follows
  from Lemma~\ref{lem:arityrespecting} that $A$
respects the arity context $\aritysum{\{y_1 : \erase{A_1}, \ldots,
  y_{n+1} : \erase{A_{n+1}}\}}{\STLCGamma}$.
Also, since $\lftype{\Gamma}{B'}$ has a derivation, it must be the
case that $\lftype{\Gamma, y_{n+1} : A'_{n+1}}{A''}$ has one.
Since the derivation concludes with a \atomtermapp\ rule, it must be
the case that $\lfchecktype{\Gamma}{M'}{A'_{n+1}}$ and $\hsub{\{
  \langle y_{n+1}, M', \erase{A'_{n+1}} \rangle \}}{A''}{A'}$ have
shorter derivations than the one for $\lfsynthtype{\Gamma}{R}{A'}$.
Since $\theta_{n+1}$ is type preserving with respect to $\STLCGamma$, we
may now use Lemma~\ref{lem:lfcomp} and Theorem~\ref{th:erasure} to
conclude that $\hsub{\theta_{n+1} \cup \{\langle y_{n+1}, M_2, \erase{A_{n+1}} 
  \rangle \}}{A}{A'}$ and $\lftype{\Gamma}{A'}$ have derivations. 
Renaming $M'$ to $M_{n+1}$ we see that all the requirements of clause
2 are satisfied.
\end{proof}

Theorem~\ref{th:atomictype} gives us an alternative means for deriving
judgements of the form $\lfchecktype{\Gamma}{R}{P}$, in the process
dispensing with judgements of the form $\lfsynthtype{\Gamma}{R}{A}$.
Note also that in \emph{analyzing} judgements of the form
$\lfchecktype{\Gamma}{R}{P}$, it is necessary to consider only
\emph{shorter} derivations for subterms of $R$.
This observation will be used in developing a means for arguing
inductively on the heights of LF derivations. 

A property similar to that in Theorem~\ref{th:atomictype} can be
observed for wellformedness judgements for atomic types.
Theorem~\ref{th:atomickind} presents a version that suffices for this
paper. 
A proof of this theorem can be constructed based essentially on the
one for Theorem~\ref{th:atomictype}.
  
\begin{theorem}\label{th:atomickind}
Let $\Gamma$ be a context such that $\lfctx{\Gamma}$ is derivable.
Then $\lfsynthkind{\Gamma}{P}{K'}$ has a derivation if and only if
there is an
$a : \typedpi{y_1}{A_1}{\ldots \typedpi{y_n}{A_n}{K}} \in \Sigma$
such that 
\begin{enumerate}
\item $P$ is of the form $(a \app M_1 \app \ldots\app M_n)$;
  
\item there is a sequence of types $A'_1,\ldots,A'_n$ such that, for $1
  \leq i \leq n$, there are derivations for
  $\hsub{\{\langle y_1,  M_1,\erase{A_1}\rangle, \ldots,
           \langle y_{i-1}, M_{i-1},  \erase{A_{i-1}} \rangle\}}
        {A_i}
        {A'_i}$ and
  $\lfchecktype{\Gamma}{M_i}{A'_i}$; and 

\item $\hsub{\{\langle y_1, M_1,\erase{A_1}\rangle, \ldots, \langle
  y_n, M_n, \erase{A_n} \rangle\}}{K}{K'}$ and $\lfkind{\Gamma}{K'}$
  have derivations. 
\end{enumerate}
\end{theorem}

\section{A Logic for Expressing Properties of LF Specifications}
\label{sec:logic}

We turn now to the task of designing a logic in which we can express
properties of an object system that has been specified in LF.
The discussions in Section~\ref{ssec:informal-reasoning} suggest a 
possible structure for such a logic.
The logic would be parameterized by an LF signature that has been
determined to be well-formed at the outset.
The basic building blocks for the properties that are to be described
would be typing judgements.
More specifically, the logic would use such judgements as its atomic
formulas and would interpret them using LF derivability.
More complex formulas would be then be constructed using 
logical connectives and quantifiers over LF terms.
As the example in Section~\ref{ssec:informal-reasoning} illustrates,
it would be necessary to also permit a quantification over LF contexts.

To develop an actual logic based on these ideas, we need to describe a
more precise correspondence between LF typing judgements and atomic
formulas. 
The judgement forms that need to be considered in this context are
those for typing canonical and atomic terms, \ie, the
$\lfchecktype{\Gamma}{M}{A}$ and $\lfsynthtype{\Gamma}{R}{A}$ forms. 
The main judgement form is in fact the first one: the second form
serves mainly to explicate judgements of the first kind when the type
is atomic and, as we have noted already, Theorem~\ref{th:atomictype}
provides the basis for circumventing such an explicit treatment
through a special ``focused'' typing rule. 
In light of this, it suffices to describe an encoding of only the
first judgement form.
The judgement in the LF setting assumes the wellformedness of the
context $\Gamma$ and the type $A$.
In the logic, the context and, therefore, also the type can be
dynamically determined by instantiations for context variables.
To deal with this situation, we will build the wellformedness of
$\Gamma$ and $A$ into the interpretation of the encoding of the
judgement. 
There is, however, an aspect of the wellformedness checking that we
would like to extract into a static pre-processing phase.
The LF typing rules combine the checking of canonicity of terms with
the determination of inhabitation that relies on the semantically more
meaningful aspect of dependencies in types.
To allow the focus in the logic to be on the latter aspect, we will
build the former into a wellformedness criterion for formulas using
arity types.

Another aspect that needs further consideration is the representation of
LF contexts in atomic formulas.
To support typing derivations that use the
\canontermlam\ rule, this representation must allow for the explicit
association of types with variables.
These variables may appear free in the terms and types in the atomic
formula.
However, their interpretation in this context must be different from
the variables that are bound by quantifiers: in particular, these
variables cannot be instantiated and each of them must be treated as
being distinct within the atomic formula. 
The necessary treatment of these variables can be realized by
representing them by \emph{nominal constants} in the style of
\cite{gacek11ic,tiu06lfmtp}.
Context expressions must, in addition, allow for an unspecified part
whose exact extent is to be determined by instantiation of an external
context quantifier.
To support this ability, we will allow context variables to appear in
these expressions.
However, as observed in Section~\ref{ssec:informal-reasoning}, we
would like to be able to restrict the instantiation of such variables
to blocks of declarations adhering to specified forms.
To impose such constraints, the logic will permit context variables to
be typed by \emph{context schemas} that are motivated by regular world
descriptions used in the Twelf
system~\cite{Pfenning02guide,schurmann00phd}. 

In the rest of this section, we present a logic called \logic\ that
substantiates the ideas outlined above.
The first two subsections present the well-formed formulas of \logic and
identify their intended meaning. 
The end result of this discussion is a means for describing properties
of a specification given by an LF signature and for assessing the
validity of such properties.
The third subsection illuminates this capability through a
collection of examples.
The last subsection observes the counterpart in \logic\ of the
property of irrelevance of the particular names that are chosen for
the variables bound by the context in an LF judgement.
The particular expression of this property takes the form of the
invariance of validity of formulas under permutations of
nominal constants. 

\subsection{The Formulas of the Logic}

We begin by considering the representation of LF terms and types in \logic.
Figure~\ref{fig:logic-terms-and-types} presents the syntax of these
expressions.
As with LF syntax, we use $c$ and $d$ to represent term level
constants, $a$ and $b$ to represent type level constants and $x$ and
$y$ to represent term-level variables.
We also use $n$ to represent a special category of symbols called the
nominal constants.
LF terms and types are obviously a subset of the expressions presented
here.
Going the other way, there are two main additions to the LF
counterparts in the collection of expressions described here.
First, nominal constants may be used in constructing terms.
Second, as we shall soon see, variables may be bound not only by
term and type level abstractions but also by formula level
quantifiers.

\begin{figure}[tbhp]
\[
\begin{array}{r r c l}
  \mbox{\bf Terms} & M,N & ::= & R\ |\ \lflam{x}{M}\\
  \mbox{\bf Atomic Terms} & R & ::= & c\ |\ x\ |\ n\ |\ R\app M\\[5pt]
  \mbox{\bf Types} & A & ::= &
           P\ |\ \typedpi{x}{A_1}{A_2}\\
  \mbox{\bf Atomic Types} & P & ::= & a\ |\ P\app M\\[5pt]
\end{array}
\]
\caption{Terms and Types in \logic}
\label{fig:logic-terms-and-types}
\end{figure}

The logic \logic\ is parameterized by an LF style signature $\Sigma$
that assigns kinds to type-level constants and types to term-level ones.
This signature is assumed to be well-formed in the sense described in
Section~\ref{sec:lf}.
The logic also assumes as given a set $\noms$ of nominal constants,
each specified with an arity type, with a countably infinite supply of
such constants for each arity type $\alpha$.
The types of the nominal constants are fixed once and for all by
$\noms$.
This allows us to treat sets of such constants ambiguously as
collections whose elements are of the form $n : \alpha$ or simply
$n$. 

\begin{figure}[tbhp]

\begin{center}
\begin{tabular}{c}

\infer{\akindingp{\STLCGamma}{a}{K}}
      {a:K \in \Sigma}

\qquad\qquad

\infer{\akindingp{\STLCGamma}{P\app M}{K}}
      {\akindingp{\STLCGamma}{P}{\typedpi{x}{A}{K}} \qquad
       \stlctyjudg{\STLCGamma}{M}{\erase{A}}}

\\[10pt]

\infer{\wftype{\STLCGamma}{P}}
      {\akindingp{\STLCGamma}{P}{\type}}

\qquad\qquad

\infer{\wftype{\STLCGamma}{\typedpi{x}{A_1}{A_2}}} 
      {\wftype{\STLCGamma}{A_1} \qquad \wftype{\aritysum{\{x :
            \erase{A_1} \}}{\STLCGamma}}{A_2}}
\end{tabular}
\end{center}

\caption{Arity Kinding for Canonical Types}
\label{fig:arity-kinding}
\end{figure}

As explained earlier, expressions in \logic\ will be expected to
satisfy typing constraints that check for canonicity.
At the term level, these constraints will be realized through arity
typing relative to a suitable arity context.
At the type level, we must additionally ensure that (type) constants
have been supplied with an adequate number of arguments.
We make these notions precise below; we assume the obvious extension
of erasure to types in \logic\ here and elsewhere.
\begin{definition}
The typing relation between an arity context, a term and an arity type
that is described in Definition~\ref{def:aritytyping} is extended to
the present context by permitting terms to contain nominal constants
and by allowing arity contexts to contain assignments to such
constants.
The rules in Figure~\ref{fig:arity-kinding} define an arity kinding
property denoted by $\wftype{\STLCGamma}{A}$ for a type $A$ relative to
an arity context $\STLCGamma$.
In these rules, $\Sigma$ is the signature parameterizing \logic.
We will often need to refer to the arity context induced by $\Sigma$.
We call this the \emph{initial constant context} and we reserve the symbol
$\STLCGamma_0$ to denote it.
\end{definition}

Hereditary substitution extends naturally to the terms and types in
\logic\ by treating nominal constants like other constants.
The following theorem relating to such substitutions has an obvious
proof.  

\begin{theorem}\label{th:aritysubs-ty}
If $\theta$ is type preserving with respect to $\STLCGamma$ and
$\wftype{\aritysum{\context{\theta}}{\STLCGamma}}{A}$ and
$\hsub{\theta}{A}{A'}$ have derivations, then
$\wftype{\STLCGamma}{A'}$ has a derivation.
\end{theorem}

\begin{figure}[tbhp]
  \[\begin{array}{rrcl}
\mbox{\bf Block Declarations} & \Delta & ::= & \emptybb\ \vert\ \Delta, y : A \\
\mbox{\bf Block Schema}   & \mathcal{B} & ::= & \{x_1:\alpha_1,\ldots, x_n:\alpha_n\}\Delta\\
\mbox{\bf Context Schema} & \mathcal{C} & ::= & \emptycs\ \vert\ \mathcal{C}, \mathcal{B}
\end{array}\]
\caption{Block Schemas and Context Schemas}
\label{fig:context-schemas}
\end{figure}

The logic allows for quantifiers over LF contexts.
In the intended interpretation, such quantifiers are meant to be
instantiated with context expressions that assign LF types to nominal
constants. 
However, it will be necessary to be able to constrain the possible
instantiations in real applications.
This ability is supported by typing
context quantifiers using \emph{context schemas} whose structure is
presented in Figure~\ref{fig:context-schemas}.
In essence, a context schema comprises a collection of \emph{block
  schemas}.
A block schema consists of a header of variables annotated with
arity types and a body of declarations associating types with
variables. 
Each variable in the header and that is assigned a type in the body of
a block schema is required to be distinct. 
A block is intended to serve as a template for generating a sequence
of bindings for nominal constants through an instantiation process
that will be made clear in the next subsection.
A context expression corresponding to a context schema is to be obtained
by some number of instantiations of its block schemas.
Block and context schemas are required to satisfy typing constraints
towards ensuring that the context expressions generated from them will
be well-formed in the manner required by the logic.
These constraints are represented by the typing judgements
$\abstyping{\mathcal{B}}$ and $\acstyping{\mathcal{C}}$, respectively,
that are defined by the rules in Figure~\ref{fig:schematyping}.

\begin{figure}[tbhp]

\begin{center}
\begin{tabular}{c}

\infer{\wfdecls{\STLCGamma}{\emptybb}{\STLCGamma}}{}

\qquad\qquad

\infer{\wfdecls{\STLCGamma}{\Delta, y:A}{\STLCGamma' \cup \{y:\erase{A}\}}} 
      {\wfdecls{\STLCGamma}{\Delta}{\STLCGamma'} \qquad 
       y\ \mbox{\rm is not assigned by}\ \STLCGamma' \qquad
       \wftype{\STLCGamma'}{A}}

\\[15pt]

\infer{\abstyping{\{x_1:\alpha_1,\ldots, x_n:\alpha_n\}\Delta}}
      {x_1,\ldots,x_n\ \mbox{\rm are distinct variables}
       \qquad
       \wfdecls{\STLCGamma_0 \cup \{x_1 : \alpha_1, \ldots,
                                    x_n : \alpha_n\}}
               {\Delta}
               {\STLCGamma'}}

\\[15pt]

\infer{\acstyping{\emptycs}}{}

\qquad

\infer{\acstyping{\mathcal{C},\mathcal{B}}}
      {\acstyping{\mathcal{C}} \qquad \abstyping{\mathcal{B}}}

\end{tabular}
\end{center}
\caption{Wellformedness Judgements for Block and Context Schemas}
\label{fig:schematyping}
\end{figure}

\begin{figure}[tbhp]
\[\begin{array}{lrcl}
\mbox{\bf Context Expressions} & G & ::= &
    \emptyce\ |\ \Gamma\ |\ G,n:A\\
\mbox{\bf Formulas} & F & ::= & 
    \fatm{G}{\of{M}{A}}\ |\ \ftrue\ |\ \ffalse\ |\ \fimp{F_1}{F_2}\ |\ \fand{F_1}{F_2}\ |\\
& & &\for{F_1}{F_2}\ |\ \fctx{\Gamma}{\mathcal{C}}{F}\ |\ \fall{x:\alpha}{F}\ |\ \fexists{x:\alpha}{F}
\end{array}\]
\caption{The Formulas of \logic}
\label{fig:formula-syntax}
\end{figure}

We are finally in a position to describe the formulas of \logic.
The syntax of these formulas is presented in
Figure~\ref{fig:formula-syntax}.
The symbol $\Gamma$ is used in these formulas to represent context
variables.
Atomic formulas, which represent LF typing judgements, have the form
$\fatm{G}{M:A}$.
The context in these formulas is constituted by a sequence of type
associations with nominal constants, possibly preceded by a context
variable.
Included in the collection are the logical constants $\ftrue$ and
$\ffalse$ and the familiar connectives for constructing more complex
formulas.
Universal and existential quantification over term variables is also
permitted and these are written as $\fall{x:\alpha}{F}$ and
$\fexists{x:\alpha}{F}$, respectively.
Such quantification is indexed, as might be expected, by arity types.
The collection also includes universal quantification over context
variables that is typed by context schemas, written as
$\fctx{\Gamma}{\mathcal{C}}{F}$.
We assume the usual principle of equivalence under renaming with
respect to the term and context quantifiers and apply them as needed. 

\begin{figure}[tbhp]

\begin{center}
\begin{tabular}{c}

\infer{\wfctx{\STLCGamma}{\Xi}{\emptyce}}
      {} 
\qquad

\infer{\wfctx{\STLCGamma}{\Xi}{\Gamma}}
      {\Gamma \in\Xi}

\\[10pt]

\infer{\wfctx{\STLCGamma}{\Xi}{G,n:A}}
      {\wfctx{\STLCGamma}{\Xi}{G} \qquad
       n:\erase{A}\in \STLCGamma \qquad
       \wftype{\STLCGamma}{A}}
\\[15pt]

\infer{\wfform{\STLCGamma}{\Xi}{\fatm{G}{M:A}}}
      {\wfctx{\STLCGamma}{\Xi}{G} \qquad
       \wftype{\STLCGamma}{A} \qquad
       \stlctyjudg{\STLCGamma}{M}{\erase{A}}}

\\[10pt]

\infer{\wfform{\STLCGamma}{\Xi}{\ftrue}}{} 

\qquad
      
\infer{\wfform{\STLCGamma}{\Xi}{\ffalse}}{}

\qquad
      
\infer[\bullet \in \{\supset,\land,\lor\}]
      {\wfform{\STLCGamma}{\Xi}{F_1 \bullet F_2}}
      {\wfform{\STLCGamma}{\Xi}{F_1} \qquad 
       \wfform{\STLCGamma}{\Xi}{F_2}}

\\[10pt]

\infer{\wfform{\STLCGamma}{\Xi}{\fctx{\Gamma}{\mathcal{C}}{F}}}
      {\acstyping{\mathcal{C}} \qquad
       \wfform{\STLCGamma}{\Xi \cup \{ \Gamma \}}{F}}

\qquad\qquad
      
\infer[\genericq \in \{\forall, \exists \}]
      {\wfform{\STLCGamma}{\Xi}{\fgeneric{x:\alpha}{F}}}
      {\wfform{\aritysum{\{x:\alpha\}}{\STLCGamma}}{\Xi}{F}}
 
\end{tabular}
\end{center}
 
\caption{The Wellformedness Judgement for Formulas}
\label{fig:wfform}
\end{figure}

A formula $F$ is determined to be well-formed or not relative to an arity
context $\STLCGamma$ and a collection of context variables $\Xi$.
This judgement is written concretely as $\wfform{\STLCGamma}{\Xi}{F}$
and the rules defining it are presented in Figure~\ref{fig:wfform}.
At the top-level, formulas are expected to be \emph{closed}, \ie, to not have
any free term or context variables.
More specifically, we expect $\wfform{\noms \cup \STLCGamma_0}{\emptyset}{F}$ to be
derivable for such formulas. 
The analysis within the scope of term and context
quantifiers augments these sets in the expected way.
For context quantifiers, this analysis must also check that the
annotating context schema is well-formed.
An atomic formula $\fatm{G}{M:A}$ is deemed well-formed if its
components $G$, $M$ and $A$ are well-formed and if $M$ can be assigned
the erased form of $A$ as its arity type.
The context expression $G$ is well-formed if any context variable used
in it is bound in the overall formula and if the types assigned to
nominal constants in the explicit part of $G$ are well-formed and
such that their erased forms match the arity types of the nominal
constants they are assigned to.
Note that these types may use nominal constants without paying
attention to dependency ordering; assessing
whether they are used in a manner that respects this ordering is a part
of the meaning of the atomic formula.

The following theorem, whose proof is obvious, shows that
the wellformedness judgement for formulas continues to
hold under the augmentation of the two contexts that  parameterize it. 
\begin{theorem}\label{th:wfsupset}
If $\wfform{\STLCGamma}{\Xi}{F}$ has a derivation and $\STLCGamma
\subseteq \STLCGamma'$ and $\Xi \subseteq \Xi'$, then
$\wfform{\STLCGamma'}{\Xi'}{F}$ also has a derivation.
\end{theorem}

\subsection{The Interpretation of Formulas}

A key component to understanding the meanings of formulas is
understanding the interpretation of the quantifiers over term and
context variables.
These quantifiers are intended to range over closed expressions of the
relevant categories.
For a quantifier over a term variable, this translates concretely into
closed terms of the relevant arity type.
For a quantifier over a context variable, we must first explain when
a context expression satisfies a context schema. 

\begin{figure}[tbhp]

\begin{center}
\begin{tabular}{c}

\infer{\declinst{\mathbb{N}}{\emptybb}{\emptyce}{\emptyset}}{}

\qquad

\infer{\declinst{\mathbb{N}}{\Delta,y:A}{G, n : A'}
                 {\theta \cup \{\langle y,n,\erase{A}\rangle\}}}
      {\declinst{\mathbb{N}}{\Delta}{G}{\theta} \qquad
        n : \erase{A} \in {\mathbb{N}} \qquad
        \hsub{\theta}{A}{A'}}

\\[15pt]

\infer{\bsinst{\mathbb{N}}{\Psi}{\{x_1 : \alpha_1,\ldots, x_n : \alpha_n\}\Delta}{G}}
      {\declinst{\mathbb{N}}{\Delta}{G'}{\theta}
       \quad\ 
       \{ \stlctyjudg{{\mathbb{N}} \cup \Psi \cup \STLCGamma_0}{t_i}{\alpha_i}\ \vert\ 1 \leq i \leq n \}
       \quad\ 
       \hsub{\{\langle x_i,t_i,\alpha_i\rangle \ \vert\ 1 \leq i \leq n \}} 
             {G'}
             {G}
       }

\\[15pt]

\infer{\csinstone{\mathbb{N}}{\Psi}{\mathcal{C},\mathcal{B}}{G}}
      {\bsinst{\mathbb{N}}{\Psi}{\mathcal{B}}{G}}

\qquad

\infer{\csinstone{\mathbb{N}}{\Psi}{\mathcal{C},\mathcal{B}}{G}}
      {\csinstone{\mathbb{N}}{\Psi}{\mathcal{C}}{G}}

\\[15pt]

\infer{\csinst{\mathbb{N}}{\Psi}{\mathcal{C}}{\emptyce}}
      {}

\qquad 

\infer{\csinst{\mathbb{N}}{\Psi}{\mathcal{C}}{G, G'}}
      {\csinst{\mathbb{N}}{\Psi}{\mathcal{C}}{G} \qquad
       \csinstone{\mathbb{N}}{\Psi}{\mathcal{C}}{G'}}

\end{tabular}
\end{center}

\caption{Instantiating a Context Schema}
\label{fig:ctx-schema}
\end{figure}

We do this by describing the relation of ``being an instance of''
between a closed context expression $G$ and a context schema $\mathcal{C}$.
This relation is indexed by a nominal constant context $\mathbb{N}$
that is a subset of $\noms$ and a \emph{term variables context} $\Psi$
that identifies a finite collection of such variables together with their
arity types: in combination with  the constants in $\STLCGamma_0$,
these collections, circumscribe the symbols that can be used in the
declarations in the context expressions.\footnote{In  
  determining closed instances of context schemas, $\mathbb{N}$ will
  be $\noms$ and $\Psi$ will be the empty set. The more general form
  for this relation, which includes a parameterization by these sets,
  will be useful in later sections.} 
The relation is written as $\csinst{\mathbb{N}}{\Psi}{\mathcal{C}}{G}$
and it is defined by the rules in Figure~\ref{fig:ctx-schema}.
This relation is defined via the repeated use of a ``one-step''
instantiation relation written as
$\csinstone{\mathbb{N}}{\Psi}{\mathcal{C}}{G}$; note that by $G, G'$
we mean a context expression that is obtained by adding the bindings
corresponding to $G'$ in front of those in $G$.
The definition of the one-step instantiation relation for context
schemas uses an auxiliary judgement
$\bsinst{\mathbb{N}}{\Psi}{\mathcal{B}}{G}$ that denotes the relation
of ``being an instance of'' between a block 
schema and a context expression fragment.
This relation holds when the context expression is obtained by
generating a sequence of bindings for nominal constants from
$\mathbb{N}$ using the body of the block schema and then instantiating
the variables in the header of the block schema with terms of the
right arity types.
The former task is realized through the relation
$\declinst{\mathbb{N}}{\Delta}{G}{\theta}$ that holds between a block of
declarations $\Delta$, a context expression $G$ that is obtained
by replacing the variables assigned in $\Delta$ with
suitable nominal constants, and a substitution $\theta$ that
corresponds to this replacement.
We assume here and elsewhere that the application of a hereditary
substitution to a sequence of declarations  corresponds to its
application to the type in each assignment. 

\begin{theorem}\label{th:schemainst}
Let $\mathcal{C}$ and $G$ be a context schema and a context expression
such that $\acstyping{\mathcal{C}}$ and
$\csinst{\mathbb{N}}{\Psi}{\mathcal{C}}{G}$ are derivable. Then for
any arity context $\STLCGamma$ such that
$\mathbb{N} \cup \Psi \cup \STLCGamma_0 \subseteq \STLCGamma$, it is the case that 
$\wfctx{\STLCGamma}{\emptyset}{G}$ has a derivation. 
\end{theorem}

\begin{proof}
We first show that for any block declaration
$\Delta$ and any arity context $\STLCGamma$ such that $\mathbb{N}
\subseteq \STLCGamma$, if $\wfdecls{\STLCGamma}{\Delta}{\STLCGamma'}$ and
$\declinst{\mathbb{N}}{\Delta}{G'}{\theta'}$ are derivable for some $\STLCGamma'$
and $\theta'$, then (a)~$\theta'$ is type preserving with respect to
$\Theta$, (b)~$\Theta'$ is $\aritysum{\context{\theta'}}{\Theta}$, and
(c)~each binding in $G'$ is of the form $n : A$ where $n:\erase{A}\in
\Theta$ and $\wftype{\Theta}{A}$ has a derivation.
This claim is proved by induction on the derivation of
$\wfdecls{\STLCGamma}{\Delta}{\STLCGamma'}$; properties (a) and (b)
are included in the claim because they are useful together with
Theorem~\ref{th:aritysubs-ty} in showing property (c) in the
induction step. 
Next we show, through an easy inductive argument, that if
$\wfdecls{\STLCGamma_0 \cup \{x_1 : \alpha_1, \ldots,
                              x_n : \alpha_n\}}
         {\Delta}
         {\STLCGamma'}$
has a derivation and $\STLCGamma$ is such that
$\mathbb{N} \cup \Psi \cup \STLCGamma_0 \subseteq \STLCGamma$, then, for
some $\STLCGamma''$, it is the case that 
$\wfdecls{\aritysum{\{x_1 : \alpha_1, \ldots,x_n : \alpha_n\}}
                   {\STLCGamma}}           
         {\Delta}
         {\STLCGamma''}$
has a derivation.
Using Theorem~\ref{th:aritysubs-ty} with these two
observations, we can show easily that if
$\abstyping{\{x_1 : \alpha_1,\ldots, x_n : \alpha_n\}\Delta}$ and 
$\bsinst{\mathbb{N}}{\Psi} 
        {\{x_1 : \alpha_1,\ldots, x_n : \alpha_n\}\Delta}
        {G}$ 
have derivations then for each binding of the form $n:A$ in $G$ it is
the case that $n:\erase{A} \in \Theta$ and $\wftype{\Theta}{A}$.
The theorem follows easily from this observation.
\end{proof}

In defining validity for formulas, we will need to consider
substitutions for context and term variables.
Context variables substitutions have the form
$\{G_1/\Gamma_1,\ldots,G_n/\Gamma_n\}$ where, for $1 \leq i \leq n$,
$\Gamma_i$ is a context variable and $G_i$ is a context expression. 
If $\sigma$ is such a substitution, we will write $\domain{\sigma}$ to
denote the set $\{\Gamma_1,\ldots,\Gamma_n\}$.
Further, the application of $\sigma$ to a formula
$F$, which is denoted by $\subst{\sigma}{F}$, will correspond to the 
replacement of the free occurrences of the variables
$\Gamma_1,\ldots,\Gamma_n$ in $F$ by the corresponding context
expressions, renaming bound context variables appearing in $F$ away
from those appearing in $G_1,\ldots,G_n$.
For term variables, the replacement must also ensure the
transformation of the resulting expression to normal form.
Towards this end, we adapt hereditary substitution to formulas.
The application of this substitution simply distributes over
quantifiers and logical symbols, respecting the scopes of quantifiers
through the necessary renaming.
The application to the atomic formula $\fatm{G}{M:A}$ also distributes
to the component parts.
We have already discussed the application to terms and types.
The application to context expressions leaves context variables
unaffected and simply distributes to the types in the explicit
bindings.
Note that no check is mandated in the process for clashes in the names
of nominal constants appearing in the context expression being
substituted into and the substitution terms. 
In this respect, this application is unlike that to LF contexts
that is defined in Figure~\ref{fig:hsubctx}.

\begin{theorem}\label{th:subst-formula}
Let $\STLCGamma$ be an arity context and let $\Xi$ be a collection of
context variables.
\begin{enumerate}
\item If $\theta$ is a term variables substitution that is arity type preserving with
respect to $\STLCGamma$ and $F$ is a formula such that there is a
derivation for $\wfform{\aritysum{\context{\theta}}{\STLCGamma}}{\Xi}{F}$,
then there is a unique formula $F'$ such that
$\hsub{\theta}{F}{F'}$ has a derivation.
Moreover, for this $F'$ it is the case that
$\wfform{\STLCGamma}{\Xi}{F'}$ is derivable.  

\item If $\sigma=\{G_1/\Gamma_1,\ldots,G_n/\Gamma_n\}$ is a context variables
substitution which is such that all judgements in the collection
$\left\{\wfctx{\STLCGamma}
             {\Xi\setminus\{\Gamma_1,\ldots,\Gamma_n\}}{G_i}\ |\ 
                    1\leq i\leq n\right\}$
are derivable and $F$ is a formula such that there is a derivation for
$\wfform{\STLCGamma}{\Xi}{F}$, then there is a derivation for
$\wfform{\STLCGamma}{\Xi\setminus\{\Gamma_1,\ldots,\Gamma_n\}}{\subst{\sigma}{F}}$.
\end{enumerate}
\end{theorem}

\begin{proof}
The first clause follows from an induction on the derivation of 
$\wfform{\aritysum{\context{\theta}}{\STLCGamma}}{\Xi}{F}$, using  
Theorems~\ref{th:uniqueness} and \ref{th:aritysubs} in the atomic case
to ensure the appropriate arity typing judgements will be derivable under 
the substitution $\theta$.
The second clause follows from an induction on the derivation of
$\wfform{\STLCGamma}{\Xi}{F}$, using the assumption that
$\wfctx{\STLCGamma}{\Xi\setminus\{\Gamma_1,\ldots,\Gamma_n\}}{G_i}$
is derivable to ensure wellformedness under the substitution $\sigma$
in the atomic case.
\end{proof}

\noindent Following the notation introduced after Theorem~\ref{th:aritysubs}, 
if $F$ and $\theta$ are a formula and a substitution that together satisfy
the requirements of the first part of the theorem, we will write
$\hsubst{\theta}{F}$ to denote the $F'$ for which 
$\hsub{\theta}{F}{F'}$ is derivable.
As is implicit in the preceding discussion, term and context variables
substitutions may introduce new nominal constants. 
If $\theta$ is a term variables substitution, we will write
$\supportof{\theta}$ to denote the collection of such constants that
appear in the terms in $\range{\theta}$.
Similarly, if $\sigma$ is the context variables substitution
$\{G_1/\Gamma_1,\ldots,G_n/\Gamma_n\}$, 
we will write $\supportof{\sigma}$ to denote the collection of nominal
constants that appear in $G_1,\ldots,G_n$.

A closed atomic formula of the form
$\fatm{G}{M:A}$ is intended to encode an LF judgement of the
form $\lfchecktype{\Gamma}{M}{A}$.
In this encoding, nominal constants that appear in terms represent
free variables for which bindings appear in the context in LF
judgements.
To substantiate this interpretation, the rules \canonkindpi, \canonfampi\ and
\canontermlam\ must introduce fresh nominal constants into contexts in
typing derivations and they must replace bound variables appearing in
terms and types with these constants. 
We use this interpretation to define validity for closed atomic
formulas with one further qualification: unlike in the LF judgement,
for the atomic formula we must also ascertain the wellformedness of
the context and the type.
This notion of validity is then extended to all closed formulas by
recursion on formula structure.

\begin{definition}\label{def:semantics}
Let $F$ be a formula such that $\wfform{\noms \cup
  \STLCGamma_0}{\emptyset}{F}$ is derivable. 
\begin{itemize}
\item If $F$ is $\fatm{G}{M:A}$, then it is valid exactly when all of $\lfctx{G}$,
  $\lftype{G}{A}$, and $\lfchecktype{G}{M}{A}$ are derivable in LF,
  under the interpretation of nominal constants as variables bound in
  a context and with the modification of the rules \canonkindpi, \canonfampi\ and
  \canontermlam\ to introduce fresh nominal constants into contexts and to
  instantiate the relevant bound variables in kinds, types and terms with
  these constants.

\item If $F$ is $\ftrue$ it is valid and if it is $\ffalse$ it is not valid.

\item If $F$ is $\fimp{F_1}{F_2}$, it is valid if $F_2$ is valid in
  the case that $F_1$ is valid.

\item If $F$ is $\fand{F_1}{F_2}$, it is valid if both $F_1$ and
  $F_2$ are valid.

\item If $F$ is $\for{F_1}{F_2}$, it is valid if either $F_1$ or $F_2$ is valid.

\item If $F$ is $\fctx{\Gamma}{\mathcal{C}}{F}$, it is valid if
  $\subst{\{G/\Gamma\}}{F}$ is valid for every $G$ such that
  $\csinst{\noms}{\emptyset}{\mathcal{C}}{G}$ is derivable.

\item If $F$ is $\fall{x:\alpha}{F}$, it is valid if
  $\hsubst{\{\langle x, M,\alpha\rangle \}}{F}$ is valid for every $M$ such that
  $\stlctyjudg{\noms \cup \STLCGamma_0}{M}{\alpha}$ is derivable.

\item If $F$ is $\fexists{x:\alpha}{F}$, it is valid if
  $\hsubst{\{\langle x, M,\alpha\rangle \}}{F}$ is valid for some $M$ such that 
  $\stlctyjudg{\noms \cup \STLCGamma_0}{M}{\alpha}$ is derivable.
\end{itemize}
Theorems~\ref{th:schemainst} and \ref{th:subst-formula} ensure the
coherence of this definition. 
\end{definition}

\subsection{Understanding the Notion of Validity}\label{ssec:logic-examples}

In the examples we consider below, we assume an instantiation of
\logic\ 
based on the signature presented in
Section~\ref{ssec:informal-reasoning}.
Obviously, any LF typing judgement based on that signature is expressable
in the logic.
Moreover, the corresponding formula will be valid exactly when the
typing judgement is derivable in LF.
Thus, the formulas
$\fatm{\emptyce}{\emptytm : \tmty}$,
$\fatm{\emptyce}{(\lamtm\app \unittm\app (\lflam{x}{x})) : \tmty}$ and
$\fatm{n : \tmty}{n:\tmty}$ are all valid. 
Similarly, the formulas
$\fexists{d:\oty}{\fatm{\emptyce}{d : (\ofty \app \emptytm\app
    \unittm)}}$ and
\begin{tabbing}
\qquad\=\kill
\>$\fexists{d:\oty}{\fatm{\emptyce}{d :\ofty \app (\lamtm \app \unittm\app
    (\lflam{x}{x}))\app (\arrtm\app \unittm\app \unittm)}}$
\end{tabbing}
are valid but the formula
$\fexists{d:\oty}{\fatm{\emptyce}{d:\ofty \app
    (\lamtm \app \unittm\app (\lflam{x}{x}))\app \unittm}}$
is not.
Note that the arity type associated with the quantified variable in
each of these formulas provides only a rough constraint on the
instantiation needed to verify the validity of the formula; to do
this, the instance must also satisfy LF typeability requirements
represented by formula that appears within the scope of the
quantifier.

Wellformedness conditions for formulas ensure only that the terms
appearing within formulas satisfy canonicity requirements, i.e. that
these terms are in $\beta$-normal form and that variables and
constants are applied to as many arguments as they can take.
Arity typing does not distinguish between terms in different
expression categories.
For example, the formula
\begin{tabbing}
\qquad\=\kill
\> $\fexists{d:\oty}{\fatm{\emptyce}{d:\ofty \app (\lamtm \app \emptytm \app
    (\lflam{x}{x}))\app (\arrtm\app \unittm\app \unittm)}}$
\end{tabbing}
is well-formed but not valid.
An alternative design choice, with equivalent consequences from the
perspective of the valid properties that can be expressed in the
logic, might have been to let the fact that $\lamtm$ is ill-applied to
$\emptytm$ to impact on the wellformedness of the formula.
The wellformedness conditions do not also enforce a distinctness
requirement for bindings in a context.
Thus, the formula $\fatm{n : \tmty, n : \tpty}{\emptytm : \tmty}$ is
well-formed.
However, it is not valid because $\lfctx{n : \tmty, n : \tpty}$ is not
derivable in LF under the described interpretation for nominal
constants.
An implication of these observations is that a naive form of weakening 
does not hold with respect to the encoding of LF derivability in
\logic; additional conditions similar to this described in 
Theorem~\ref{th:weakening} must be verified for this principle to
apply.

To provide a more substantive example of the kinds of properties that
can be expressed in \logic, let us consider the formal statement of
the property of uniqueness of type assignments for the STLC.
As noted in Section~\ref{ssec:informal-reasoning}, this property is
best described in a form that considers typing expressions in
contexts that have a particular kind of structure. 
That structure can be formalized in \logic\ by a context
schema comprising the single block
\begin{tabbing}
  \qquad\=\kill
  \> $\{t : o\}x:tm,y:\ofty\app x\app t$.
\end{tabbing}
Let us denote this context schema by $c$.
Observe that a context that instantiates this schema will not
provide a variable that can be used to construct an atomic term of
type $\tpty$.
Thus, the strengthening property for expressions representing types
that is expressed by the formula 
\begin{tabbing}
  \qquad\=\kill
  \> $\fctx{\Gamma}{c}
           {\fall{t :\oty}
                 {\fimp{\fatm{\Gamma}{t : \tpty}}
                   {\fatm{\emptyce}{t:\tpty}}}}$.
\end{tabbing}
should hold.
We can in fact easily show this formula to be valid by
using Theorem~\ref{th:atomictype} and an induction on the height of
the derivation for $\fatm{G}{t : \tpty}$ for a closed term $t$ and a
closed instance $G$ of $c$.
Using the validity of this formula, we can also easily argue that the
following formula that expresses a strengthening property pertaining to
the equality of types is also valid:
\begin{tabbing}
  \qquad\=\kill
  \> $\fctx{\Gamma}{c}
           {\fall{d :\oty}
           {\fall{t_1 :\oty}
           {\fall{t_2 : \oty}
                 {\fimp{\fatm{\Gamma}{d : \eqty\app t_1\app t_2}}
                 {\fatm{\emptyce}{d:\eqty\app t_1\app t_2}}}}}}$.
\end{tabbing}

The property of uniqueness of type assignments for the STLC can be expressed through the
following formula:
\begin{tabbing}
\qquad\=\qquad\=\qquad\=\kill
\>$\fctx{\Gamma}{c}{\fall{e:\oty}{\fall{t_1:\oty}{\fall{t_2:\oty}{\fall{d_1:\oty}{\fall{d_2:\oty}{}}}}}}$\\
\>\>$\fimp{\fatm{\Gamma}{e:\tmty}}
          {\fimp{\fatm{\Gamma}{t_1:\tpty}}
                {\fimp{\fatm{\Gamma}{t_2:\tpty}}{}}}$\\
\>\>\> $\fimp{\fatm{\Gamma}{d_1:\ofty\app e\app t_1}}
             {\fimp{\fatm{\Gamma}{d_2:\ofty\app e\app t_2}} 
                   {\fexists{d_3:\oty}
                            {\fatm{.}{d_3:\eqty\app t_1\app t_2}}}}$.
\end{tabbing}
This formula can be seen to be valid using the strengthening property
just described if we can establish the validity of the formula
\begin{tabbing}
\qquad\=\qquad\=\qquad\=\kill
\>$\fctx{\Gamma}{c}{\fall{e:\oty}{\fall{t_1:\oty}{\fall{t_2:\oty}{\fall{d_1:\oty}{\fall{d_2:\oty}{}}}}}}$\\
\>\>$\fimp{\fatm{\Gamma}{e:\tmty}}
          {\fimp{\fatm{\Gamma}{t_1:\tpty}}
                {\fimp{\fatm{\Gamma}{t_2:\tpty}}{}}}$\\
\>\>\> $\fimp{\fatm{\Gamma}{d_1:\ofty\app e\app t_1}}
             {\fimp{\fatm{\Gamma}{d_2:\ofty\app e\app t_2}} 
                   {\fexists{d_3:\oty}
                            {\fatm{\Gamma}{d_3:\eqty\app t_1\app t_2}}}}$.
\end{tabbing}
To show this, it suffices to argue that, for a closed context
expression $G$ that instantiates the schema $c$ and for closed
expressions $d_1$, $d_2$, $e$, $t_1$, and $t_2$, if the formulas
$\fatm{G}{e:\tmty}$, $\fatm{G}{t_1:\tpty}$, 
$\fatm{G}{t_2:\tpty}$, $\fatm{G}{d_1 : \ofty\app e\app t_1}$ and 
$\fatm{G}{d_2 : \ofty\app e\app t_2}$ are valid, then there must be a
closed expression $d_3$ such that
$\fatm{G}{d_3 : \eqty\app t_1\app t_2}$ is also valid.
Such an argument can be constructed by induction on the height
of the LF derivation of $\lfchecktype{G}{d_1}{\ofty\app e\app t_1}$,
which we analyze using Theorem~\ref{th:atomictype} in the manner
discussed earlier. 
There are essentially four cases to consider, corresponding to
whether the head symbol of $d_1$ is \ofemptytm, \ofapptm, 
\oflamtm, or a nominal constant that is assigned the type
$(\ofty\app n\app t_1)$ in $G$ where $n$ is also a nominal constant that
is bound in $G$.
In the last case, we use the fact that the validity of
$\fatm{G}{d_1 : \ofty\app e\app t_1}$ implies that $\lfctx{G}$ is
derivable to conclude the uniqueness of $n$ and, hence, of the typing.
The argument when $d_1$ is \ofemptytm\ has an obvious form.
The argument when $d_1$ has \ofapptm\ or \oflamtm\ as its head symbol
will invoke the induction hypothesis. 
In the case where the head symbol is \oflamtm, we will need
to consider a shorter derivation of a typing judgement in which the
context has been enhanced.
However, we will be able to use the induction hypothesis by observing
that the enhancements to the context conform to the constraints
imposed by the context schema. 
Note that the form of $d_1$ also constrains the form of $e$ in all the
cases, a fact that is used implicitly in the analysis outlined.

\subsection{Nominal Constants and Invariance Under Permutations}
\label{ssec:permutations}

The particular choices for bound variable names in the kinds, types
and terms that comprise LF expressions are considered irrelevant.
This understanding is built in concretely through the notion of
$\alpha$-conversion that renders equal expressions that differ
only in the names used for such variables.
Typing derivations transform expressions with bound variables into
ones where variables are ostensibly free but in fact bound implicitly
in the associated contexts.
The lack of importance of name choices is reflected in this case in an
invariance in the validity of typing judgements under a suitable
renaming of variables appearing in the judgements.
In a situation where context variables are represented by nominal
constants, this property can be expressed via an invariance of formula
validity under permutations of nominal constants as we describe here.
We begin with a definition of the notions of permutations of
nominal constants and their applications to expressions. 
\begin{definition}

A permutation of the nominal constants is an arity type preserving
bijection from $\noms$ to $\noms$ that differs from the identity
map at only a finite number of constants. 
The permutation that maps $n_1,\ldots, n_m$ to $n_1',\ldots,n_m'$,
respectively, and is the identity everywhere else is written as 
$\{n_1'/n_1,\ldots,n_m'/n_m\}$.
The support of a permutation $\pi=\{n_1'/n_1,\ldots,n_m'/n_m\}$, 
denoted by $\supp{\pi}$, is the collection of nominal constants
$\{n_1,\ldots, n_m\}$ or, identically, $\{n_1',\ldots,n_m'\}$.
Every permutation $\pi$ has an obvious inverse that is written as
$\inv{\pi}$. 
\end{definition}

\begin{definition}
The application of a permutation $\pi$ to an expression $E$ of a
variety of kinds is described below and is denoted in all cases by
$\permute{\pi}{E}$.
If $E$ is a term, type, or kind then the
application consists of replacing each nominal constant $n$ that
appears in $E$ with $\pi(n)$. 
If $E$ is a context then the application of $\pi$ to $E$
replaces each explicit binding $n:A$ in $E$ with 
$\pi(n):\permute{\pi}{A}$.
If $E$ is an LF judgement $\mathcal{J}$ then the permutation is applied
to each component of the judgement in the way described above.
If $E$ is a formula then the permutation is applied to its component parts.
The application of $\pi$ to a term variables substitution
$\{\langle x_1,M_1,\alpha_1\rangle,\ldots,
   \langle x_n,M_n,\alpha_n\rangle\}$
yields the substitution 
$\{\langle x_1,\permute{\pi}{M_1},\alpha_1\rangle,\ldots,
   \langle x_n,\permute{\pi}{M_n},\alpha_n\rangle\}$.
The application of $\pi$ to a context variables substitution
$\{G_1/\Gamma_1,\ldots, G_n/\Gamma_n\}$ yields
$\{\permute{\pi}{G_1}/\Gamma_1,\ldots, \permute{\pi}{G_n}/\Gamma_n\}$.
\end{definition}

The following theorem expresses the property of interest
concerning LF judgements cast in the form relevant to \logic. 
\begin{theorem}\label{th:perm-lf}
Let LF judgements and derivations be recast in the form discussed
earlier in this section: variables that are bound in a context are
represented by nominal constants and the rules \canonkindpi,
\canonfampi\ and \canontermlam\ introduce fresh nominal constants into
contexts and replace variables in kinds, types and terms with these
constants.
In this context, let $\mathcal{J}$ be an LF judgement which has a
derivation. 
Then for any permutation $\pi$, $\permute{\pi}{\mathcal{J}}$ is derivable.
Moreover, the structure of this derivation is the same as that for
$\mathcal{J}$.
\end{theorem} 
\begin{proof}
This proof is by induction on the derivation for $\mathcal{J}$.
Perhaps the only observation worthy of note is that the freshness of
nominal constants used in \canonkindpi, \canonfampi, and
\canontermlam\ rules is preserved under permutations of nominal
constants. 
\end{proof}
The above observation underlies the main theorem of this section.
\begin{theorem}\label{th:perm-form}
Let $F$ be a closed formula and let $\pi$ be a permutation.
Then $F$ is valid if and only if $\permute{\pi}{F}$ is valid.
\end{theorem}
\begin{proof}
Noting that $\inv{\pi}$ is also a permutation and that
$\permute{\inv{\pi}}{\permute{\pi}{F}}$ is $F$, it suffices to prove
the claim in only one direction.
We do this by induction on the structure of $F$.

The desired result follows easily from
Theorem~\ref{th:perm-lf} and the relationship of validity to LF
derivability when $F$ is atomic.
The cases where $F$ is $\ftrue$ or $\ffalse$ are trivial and the
ones in which $F$ is $\fimp{F_1}{F_2}$, $\fand{F_1}{F_2}$ or
$\for{F_1}{F_2}$ are easily argued with recourse to the induction
hypothesis and by noting that the permutation distributes to the
component formulas.

In the case where $F$ is $\fctx{\Gamma}{\mathcal{C}}{F'}$, we first
note that if $\csinst{\noms}{\emptyset}{\mathcal{C}}{G}$ has a
derivation then
$\csinst{\noms}{\emptyset}{\mathcal{C}}{\permute{\inv{\pi}}{G}}$ must also
have one.
From this and the validity of $F$ it follows that
$\subst{\{(\permute{\inv{\pi}}{G})/\Gamma\}}{F'}$ must be valid.
Moreover, $\subst{\{(\permute{\inv{\pi}}{G})/\Gamma\}}{F'}$ has the same
structural complexity as $F'$.
Hence, by the induction hypothesis,
$\permute{\pi}{(\subst{\{\permute{(\inv{\pi}}{G})/\Gamma\}}{F'})}$ is valid.
Noting that this formula is the same as
$\subst{\{G/\Gamma\}}{(\permute{\pi}{F'})}$ and that
$\fctx{\Gamma}{\mathcal{C}}{\permute{\pi}{F'}}$ is identical to 
$\permute{\pi}{(\fctx{\Gamma}{\mathcal{C}}{F'})}$, the validity of
$\permute{\pi}{F}$ easily follows.

Suppose that $F$ has the form $\fall{x:\alpha}{F'}$.
We observe here that if
$\stlctyjudg{\noms \cup \STLCGamma_0}{M}{\alpha}$ has a derivation
then
$\stlctyjudg{\noms \cup \STLCGamma_0}{\permute{\inv{\pi}}{M}}{\alpha}$
has one too and that
$\permute{\pi}{(\hsubst{\{\langle x,\permute{\inv{\pi}}{M},\alpha\rangle\}}{F'})}$
is the same formula as
$\hsubst{\{\langle x,M,\alpha\rangle\}}{(\permute{\pi}{F'})}$.
Using the definition of validity, the induction hypothesis and the
fact that permutation distributes to the component formula together
with the above observations, we may easily conclude that
$\permute{\pi}{F}$ is valid.

Finally, suppose that $F$ is of the form $\fexists{x:\alpha}{F'}$.
Here we note that if
$\stlctyjudg{\noms \cup \STLCGamma_0}{M}{\alpha}$ has a derivation
then
$\stlctyjudg{\noms \cup \STLCGamma_0}{\permute{\pi}{M}}{\alpha}$
has one too and that
$\permute{\pi}{(\hsubst{\{\langle x, M,\alpha\rangle\}}{F'})}$
is the same formula as
$\hsubst{\{\langle
  x,\permute{\pi}{M},\alpha\rangle\}}{(\permute{\pi}{F'})}$.
Using the definition of validity and the induction hypothesis, it is
now easy to conclude that $\permute{\pi}{F}$ must be valid.
\end{proof}

\section{A Proof System for the Logic} 
\label{sec:proof-system}

In this section we describe a proof system that provides a formal
mechanism for demonstrating validity for formulas in \logic.
This proof system is oriented around sequents that represent
assumption and conclusion formulas augmented with devices that capture
additional aspects of states that arise in the process of reasoning. 
The syntax for sequents is more liberal than is meaningful at the
outset, and this is rectified by imposing wellformedness requirements
on them.
We associate a semantics with sequents that is consistent with their
intended use.
We then present a collection of proof rules that can be used to derive
sequents.
These rules belong to two broad categories.
The first category comprises rules that embody logical aspects such as
the meanings of sequents and of the logical symbols that appear in formulas.
A key aspect of \logic\ is that its atomic formulas represent the
notion of derivability in LF that is also open to analysis.
The second category of proof rules builds in capabilities for such
analysis.
An important property for our proof rules is that they should require
the proofs only of well-formed sequents in constructing derivations
for well-formed sequents.
One concern in their presentation is therefore to check that they
satisfy this property. 
At a more substantive level, the proof rules must support
a reasoning process that is both sound and effective.
We focus in this paper on the issue of soundness, leaving the
demonstration of effectiveness to other work, e.g. see
\cite{southern21lfmtp}.
In the proofs of soundness, we will assume the wellformedness of
sequents, as guaranteed by the complementary consideration.

The first subsection below presents the sequents underlying the proof
system and identifies a semantics with them. 
The remaining subsections develop the collection of proof
rules. 
In Section~\ref{ssec:core-logic}, we present a collection of rules that
encapsulate the meanings of the logical symbols and also certain
structural aspects of sequents.
We then turn to the rules that internalize aspects of LF derivability
that are intrinsic to the understanding of the atomic formulas.
Section~\ref{ssec:atomic} develops rules for analyzing atomic formulas.
An important component of these rules is the interpretation of typing
judgements involving atomic types via the particular LF specification
that parameterizes the logic: this interpretation leads, in
particular, to a case analysis rule for such atomic formulas that
appear as assumptions in a sequent.
Section~\ref{ssec:induction} presents rules that enable reasoning based on
induction on the heights LF derivations.
In the final subsection, we introduce proof rules that
encode meta-theorems concerning LF derivability that often find use in
reasoning about LF specifications.

\subsection{The Structure of Sequents}\label{ssec:sequents}

A sequent in our proof system is characterized by a collection of
assumption formulas and a conclusion or goal formula.
The formulas may contain free term and context variables that are to
be interpreted as being implicitly universally quantified over the  
sequent and, therefore, its proof.
We find it useful also to identify with the sequent a collection of
nominal constants that circumscribes the ones that appear in its
formulas. 

The nominal constants and term variables that appear in the
sequent have arity types associated with them.
Context variables are also typed and their types are, in spirit, based
on context schemas.
However, subproofs may require a partial elaboration of a context
variable and the types associated with such variables accommodates
this possibility.
More specifically, these types have the form
$\ctxty{\mathcal{C}}{G_1;\ldots; G_n}$ where $\mathcal{C}$ is a context 
schema and $G_1,\ldots, G_n$ are context expressions.
Such a type is intended to represent the collection of context
expressions obtained by interspersing $G_1,\ldots,G_n$ with
instantiations of the context schema $\mathcal{C}$ and possibly
prefixed by a context variable of suitable type that represents a yet
to be elaborated sequence of declarations.
Additionally, context variables are annotated with a collection of
nominal constants that express the constraint that the elaborations of
these variables must not use names in these collections; the ability
to express such constraints is an essential part of the mechanism for
analyzing typing judgements involving abstractions as we will see
later in this section.

The ideas pertaining to context variable typing are made precise
through the following definition. 

\begin{figure}
\begin{center}
\begin{tabular}{c}
\infer{\wfctxvarty{\mathbb{N}}{\Psi}{\ctxty{\mathcal{C}}{\cdot}}}
      {}

\qquad\qquad      
\infer{\wfctxvarty{\mathbb{N}}{\Psi}{\ctxty{\mathcal{C}}{\mathcal{G};G}}}
      {\csinstone{\mathbb{N}}{\Psi}{\mathcal{C}}{G}
       \qquad
       \wfctxvarty{\mathbb{N}}{\Psi}{\ctxty{\mathcal{C}}{\mathcal{G}}}} 

\\[15pt]

\infer{\ctxtyinst{\mathbb{N}}{\Psi}{\Xi}{\ctxty{\mathcal{C}}{\emptycb}}{\emptyce}}
      {}

\qquad
\infer{\ctxtyinst{\mathbb{N}}{\Psi}{\Xi}{\ctxty{\mathcal{C}}{\cdot}}{\Gamma}}
      {\ctxvarty{\Gamma}
                {\mathbb{N}_{\Gamma}}
                {\ctxty{\mathcal{C}}{\mathcal{G}}} \in\Xi
         &
       (\noms\setminus\mathbb{N}_{\Gamma})\subseteq\mathbb{N}}

\\[15pt]
      
\infer{\ctxtyinst{\mathbb{N}}{\Psi}{\Xi}{\ctxty{\mathcal{C}}{\mathcal{G}}}{G,G'}}
      {\ctxtyinst{\mathbb{N}}{\Psi}{\Xi}{\ctxty{\mathcal{C}}{\mathcal{G}}}{G}
      \qquad
      \csinstone{\mathbb{N}}{\Psi}{\mathcal{C}}{G'}}

\qquad\qquad

\infer{\ctxtyinst{\mathbb{N}}{\Psi}{\Xi}{\ctxty{\mathcal{C}}{\mathcal{G};G'}}{G,G'}}
      {\ctxtyinst{\mathbb{N}}{\Psi}{\Xi}{\ctxty{\mathcal{C}}{\mathcal{G}}}{G}}

\end{tabular}
\end{center}
\caption{Well-Formed Context Variable Types and their Instantiations}
\label{fig:wfctxvar}
\end{figure}

\begin{definition}\label{def:cvar-types}
A \emph{context variable type} is a expression of the form
$\ctxty{\mathcal{C}}{\mathcal{G}}$ where $\mathcal{C}$ is a context
schema such that $\acstyping{\mathcal{C}}$ is derivable and
$\mathcal{G}$ represents a sequence of \emph{context blocks} given as
follows: 
\[ \mathcal{G} ::= \emptycb \ | \ \mathcal{G}; n_1 : A_1, \ldots, n_k : A_k. \]
Such a type is said to be well-formed with respect to a nominal
constant set $\mathbb{N} \subseteq \noms$ and a term variables context
$\Psi$ if it is the case that the relation
$\wfctxvarty{\mathbb{N}}{\Psi}{\ctxty{\mathcal{C}}{\mathcal{G}}}$
that is defined by the rules in Figure~\ref{fig:wfctxvar} holds;
intuitively, if $\mathcal{G}$ is a listing of blocks obtained by
instantiating block schemas comprising $\mathcal{C}$.
A \emph{context variables context} is a collection of associations with
context variables of sets of nominal constants and context variable types,
each written in the
form $\ctxvarty{\Gamma}{\mathbb{N}_\Gamma}{\ctxty{\mathcal{C}}{\mathcal{G}}}$.
Given a context variables context $\Xi$, we write $\ctxsanstype{\Xi}$ for
the set
$\{ \Gamma\ |\ \ctxvarty{\Gamma}
                        {\mathbb{N}_\Gamma}
                        {\ctxty{\mathcal{C}}{\mathcal{G}}}
               \in \Xi\}$, \ie, the collection of context variables
assigned types by $\Xi$.
A context expression $G$ is said to be an instance of a context type
$\ctxty{\mathcal{C}}{\mathcal{G}}$ relative to $\mathbb{N}$, $\Psi$
and $\Xi$ if the relation
$\ctxtyinst{\mathbb{N}}{\Psi}{\Xi}{\ctxty{\mathcal{C}}{\mathcal{G}}}{G}$,
that is also defined in Figure~\ref{fig:wfctxvar}, holds.
\end{definition}

The following theorem, whose proof is straightforward,
shows that an instance of a well-formed context variable type is a
well-formed context relative to the relevant arity context and context
variable collection. 

\begin{theorem}\label{th:ctx-var-type-instance}
If $\wfctxvarty{\mathbb{N}}{\Psi}{\ctxty{\mathcal{C}}{\mathcal{G}}}$
and
$\ctxtyinst{\mathbb{N}}{\Psi}{\Xi}{\ctxty{\mathcal{C}}{\mathcal{G}}}{G}$
have derivations then so does
$\wfctx{\mathbb{N}\cup \Psi}{\ctxsanstype{\Xi}}{G}$.
\end{theorem}

The following theorem, whose proof is also straightforward, provides
us a means that we will often use for adjusting contexts that
parameterize the relevant wellformedness judgements. 

\begin{theorem}\label{th:ctx-ty-wk}
Let $\mathbb{N} \subseteq \mathbb{N}'$, $\Psi \subseteq \Psi'$, and
$\Xi \subseteq \Xi'$.  

\begin{enumerate}
\item If $\wfctxvarty{\mathbb{N}}{\Psi}{\ctxty{\mathcal{C}}{\mathcal{G}}}$
is derivable, then so is
$\wfctxvarty{\mathbb{N}'}{\Psi'}{\ctxty{\mathcal{C}}{\mathcal{G}}}$.

\item If 
$\ctxtyinst{\mathbb{N}}{\Psi}{\Xi}{\ctxty{\mathcal{C}}{\mathcal{G}}}{G}$
is derivable, then so is
$\ctxtyinst{\mathbb{N}'}{\Psi'}{\Xi'}{\ctxty{\mathcal{C}}{\mathcal{G}}}{G}$.
\end{enumerate}
\end{theorem}

The wellformedness judgements in Figure~\ref{fig:wfctxvar} are
preserved under meaningful substitutions as the theorem below
explicates.

\begin{theorem}\label{th:ctxtyinst-hsubst}
Let $\theta$ be an arity type preserving substitution with respect to
$\noms \cup \STLCGamma_0 \cup \Psi$ and let
$\wfctxvarty{\mathbb{N}}
            {\aritysum{\context{\theta}}{\Psi}}
            {\ctxty{\mathcal{C}}{\mathcal{G}}}$
have a derivation.
Then
\begin{enumerate}
\item there must be a derivation for
$\wfctxvarty{\mathbb{N} \cup \supportof{\theta}}{\Psi}{\ctxty{\mathcal{C}}{\hsubst{\theta}{\mathcal{G}}}}$,
and 
\item if
$\ctxtyinst{\mathbb{N}}{\aritysum{\context{\theta}}{\Psi}}{\Xi}{\ctxty{\mathcal{C}}{\mathcal{G}}}{G}$
also has a derivation, there must be a derivation for
$\ctxtyinst{\mathbb{N}\cup \supportof{\theta}}{\Psi}{\Xi}{\ctxty{\mathcal{C}}{\hsubst{\theta}{\mathcal{G}}}}{\hsubst{\theta}{G}}$.
\end{enumerate}
\end{theorem}
\begin{proof}
Since $\theta$ is arity type preserving with respect to 
$\noms \cup \STLCGamma_0 \cup \Psi$, using Theorem~\ref{th:aritysubs} we see that 
if there is a derivation for 
$\stlctyjudg{\mathbb{N} \cup (\aritysum{\context{\theta}}{\Psi}) \cup \STLCGamma_0}
            {t}
            {\alpha}$ 
then $\hsubst{\theta}{t}$ must be well-defined and there must be a
derivation for 
$\stlctyjudg{(\mathbb{N} \cup \supportof{\theta}
               \cup \Psi
               \cup \STLCGamma_0}
            {\hsubst{\theta}{t}}
            {\alpha}$.
An induction on the derivation of 
$\wfctxvarty{\mathbb{N}}
            {\aritysum{\context{\theta}}{\Psi}}
            {\ctxty{\mathcal{C}}{\mathcal{G}}}$
using these observations allows us to confirm the first part of the
theorem.
Now suppose that there is a derivation for
$\csinstone{\mathbb{N}}{\aritysum{\context{\theta}}{\Psi}}{\mathcal{C}}{G'}$.
By an induction on this derivation using the facts observed earlier,
it can be concluded that there must be a derivation for
$\csinstone{\mathbb{N} \cup \supportof{\theta}}
           {\Psi}
           {\mathcal{C}}
           {\hsubst{\theta}{G'}}$.
The second part of the theorem follows by another obvious induction from this.
\end{proof}

We now define the notion of sequents that underlies our proof system.

\begin{definition}\label{def:sequent}
A sequent, written as $\seq[\mathbb{N}]{\Psi}{\Xi}{\Omega}{F}$, is a
judgement that relates a finite subset $\mathbb{N}$ of $\noms$, a term
variables context $\Psi$, a
context variables context $\Xi$, a finite set $\Omega$ of
\emph{assumption formulas} and a \emph{conclusion} or \emph{goal} formula $F$.
The sequent is well-formed if (a)~for each
$\ctxvarty{\Gamma}{\mathbb{N}_\Gamma}{\ctxty{\mathcal{C}}{\mathcal{G}}}$ in
$\Xi$ it is the case that $\mathbb{N}_\Gamma \subseteq \mathbb{N}$ and
that
$\wfctxvarty{\mathbb{N} \setminus \mathbb{N}_\Gamma}
            {\Psi}
            {\ctxty{\mathcal{C}}{\mathcal{G}}}$ 
is derivable and (b)~for each formula $F'$ in $\{F\} \cup \Omega$ it is
the case that $\wfform{\mathbb{N} \cup \Psi \cup \Theta_0}{\ctxsanstype{\Xi}}{F'}$
is derivable. 
Given a well-formed sequent $\seq[\mathbb{N}]{\Psi}{\Xi}{\Omega}{F}$,
we will refer to $\mathbb{N}$ as its \emph{support set}, to $\Psi$ as its
\emph{term variables context} and $\Xi$ as its \emph{context variables
  context}, and we will denote the collection of variables assigned
types by $\Psi$ and $\Xi$ by $\domain{\Psi}$ and $\domain{\Xi}$,
respectively. 
We will use a comma to denote set union in representing sequents, writing
$\setand{\Omega}{F_1,\ldots,F_n}$ to denote the set
$\Omega\cup\{F_1,\ldots,F_n\}$.
\end{definition}

We will need to consider substitutions for term variables in
sequents. 
We will require legitimate substitutions to not use the nominal
constants in the support set of the sequent; this restriction will be
part of a mechanism for controlling dependencies in context
declarations.  
We will further require substitutions to satisfy arity typing
constraints for their applications to be well-defined.
These considerations are formalized below in the notion of
substitution compatibility. 

\begin{definition}\label{def:seq-term-subst}
A pair $\langle \theta, \Psi' \rangle$ comprising a
substitution and a term variables context is said
to be \emph{substitution compatible} with a well-formed sequent
$\mathcal{S} = \seq[\mathbb{N}]{\Psi}{\Xi}{\Omega}{F}$ if
\begin{enumerate}
\item $\theta$ is arity type preserving with respect to the context
  $\noms \cup \Theta_0 \cup \Psi'$, 
  
\item $\supportof{\theta} \cap \mathbb{N} = \emptyset$, and 

\item for any variable $x$, if $x : \alpha \in \Psi$ and
  $x : \alpha' \in \aritysum{\context{\theta}}{\Psi'}$, then $\alpha =
  \alpha'$. 
\end{enumerate}
\end{definition}

The application of a substitution may introduce new nominal constants
into a sequent.
When this happens, substitutions for the term variables in the
resulting sequent must be permitted to contain these constants. 
We use the technique of raising to realize this
requirement~\cite{miller92jsc}. 
The following definition is useful in formalizing this idea.
\begin{definition}\label{def:raising-subst}
Let $\Psi$ be the term variables context $\{x_1:\alpha_1,\ldots,x_m:\alpha_m\}$,
let $n_1,\ldots,n_k$ be a listing of the elements of a finite 
collection of the nominal constants $\mathbb{N}$,
and let $\beta_1,\ldots,\beta_k$ be the arity types
associated with these constants.  
Then a version of $\Psi$ raised over $\mathbb{N}$ is a set
$\{y_1 : \gamma_1,\ldots, y_m : \gamma_m\}$ where, for $1 \leq i \leq
m$, $y_i$ is a distinct variable that is also different from the
variables in $\{x_1,\ldots,x_m\}$ and $\gamma_i$ is
$\beta_1 \atyarr \cdots \atyarr \beta_k \atyarr \alpha_i$.
Further, the raising substitution associated with this version is the set
$\{\langle x_i, (y_i \app n_1 \app \ldots \app n_k),
\alpha_i\rangle\ |\ 1 \leq i \leq m\}$.
\end{definition}

The basis for using raising in the manner described is the content of
the following theorem.
We say here and elsewhere that an arity context $\STLCGamma$ is
compatible with $\noms$ if the types that $\STLCGamma$ assigns to
nominal constants are identical to their assignments in $\noms$.

\begin{theorem}\label{th:raised-subs}
Let $\theta$ be a substitution that is arity type preserving with respect to
an arity context $\STLCGamma$ that is compatible with $\noms$.
Further, let $\Psi$ be a version of $\context{\theta}$ raised over
some listing of a collection
$\mathbb{N}$ of nominal constants and let $\theta_r$ be the associated 
raising substitution.
Then there is a substitution $\theta'$ with
$\supportof{\theta'} = \supportof{\theta}\setminus \mathbb{N}$
and $\context{\theta'} = \Psi$ that is
arity type preserving with respect to $\STLCGamma$
and such that for any $E$ for which
$\wftype{\aritysum{\context{\theta}}{\STLCGamma}}{E}$ or, for some 
arity type $\alpha$,
$\stlctyjudg{\aritysum{\context{\theta}}{\STLCGamma}}{E}{\alpha}$ has a  
derivation, it is the case that 
 $\hsubst{\theta'}{\hsubst{\theta_r}{E}} =
\hsubst{\theta}{E}$.  
\end{theorem}

\begin{proof}
Each of the substitutions involved in the
expression
$\hsubst{\theta'}{\hsubst{\theta_r}{E}} = \hsubst{\theta}{E}$
will have a result under the conditions described, thereby justifying
the use of the notation introduced after Theorem~\ref{th:aritysubs}.
Now, let $n_1,\ldots,n_k$ be the listing of the constants in
$\mathbb{N}$ in the raising substitution, let
$\alpha_1,\ldots,\alpha_k$ be the respective types of these constants,
let $\langle x, t, \alpha \rangle$ be a tuple in $\theta$ and let
$\langle x, (y \app n_1 \app \cdots \app n_k), \alpha \rangle$ be the
tuple corresponding to $x$ in $\theta_r$. 
Let $x_1,\ldots,x_k$ be a listing of distinct variables that do not
appear in $t$ and let $t'$ be the result of replacing $n_i$ by $x_i$
in $t$, for $1 \leq i \leq k$.
We construct $\theta'$ by including in it the substitution
$\langle y, \lflam{x_1}{\ldots\lflam{x_k}{t'}},
         \alpha_1 \atyarr \cdots \atyarr \alpha_k \atyarr \alpha 
\rangle$
for each case of the kind considered.
It is easy to see that $\supportof{\theta'} = \supportof{\theta}\setminus
\mathbb{N}$ and that $\theta_r$ and $\theta'$ are arity type
compatible with respect to $\STLCGamma \cup \noms$.
The remaining part of the theorem follows from noting that
$\theta = \theta' \circ \theta_r$ and using
Theorem~\ref{th:composition}. 
\end{proof} 

The following definition formalizes the application of a term
substitution to a well-formed sequent when the conditions of
substitution compatibility are met.
We assume here and elsewhere that the application of a term
substitution  to a set of formulas distributes to each member of the
set, 
its application to a context variables context
distributes to each context variable type in the context and its
application to a context variable type
$\ctxty{\mathcal{C}}{\mathcal{G}}$ distributes to each context block
in $\mathcal{G}$. 
\begin{definition}\label{def:seq-term-subst-app}
Let $\mathcal{S}$ be the well-formed sequent
$\seq[\mathbb{N}]{\Psi}{\Xi}{\Omega}{F}$
and let $\langle \theta, \Psi'\rangle$
be substitution compatible with $\mathcal{S}$.
Further, let $\Psi''$ be a version of $(\Psi \setminus
\context{\theta}) \cup \Psi'$  raised over $\supportof{\theta}$ and
let $\theta_r$ be the corresponding raising substitution. 
Then the application of $\theta$ to $\mathcal{S}$ relative to $\Psi'$,
denoted by $\hsubstseq{\Psi'}{\theta}{\mathcal{S}}$, is the sequent
$\seq[\mathbb{N} \cup \supportof{\theta}]
     {\Psi''} 
     {\hsubst{\theta_r}{\hsubst{\theta}{\Xi}}}
     {\hsubst{\theta_r}{\hsubst{\theta}{\Omega}}}
     {\hsubst{\theta_r}{\hsubst{\theta}{F}}}$.
\end{definition}

The definition and notation above are obviously ambiguous since they 
depend on the particular choices of $\Psi''$ and $\theta_r$.
We shall mean $\hsubstseq{\Psi'}{\theta}{\mathcal{S}}$ to denote any
one of the sequents so determined, referring to $\Psi''$ and
$\theta_r$ as the raised context and the raising substitution
associated with the application of the substitution where
disambiguation is needed.
Note also that the definition assumes that the application of the
substitutions $\theta$ and $\theta_r$ to the relevant context variable
types and formulas is well-defined.
We show this to be the case in the theorem below. 
\begin{theorem}\label{th:seq-term-subs-ok}
Let $\langle \theta, \Psi' \rangle$ be substitution compatible with
the well-formed sequent $\mathcal{S}$. Then
$\hsubstseq{\Psi}{\theta}{\mathcal{S}}$ is well-defined and it yields
a well-formed sequent.
\end{theorem}

\begin{proof}
As in Definition~\ref{def:seq-term-subst-app}, let
$S = \seq[\mathbb{N}]{\Psi}{\Xi}{\Omega}{F}$, let 
$\Psi''$ be the raised context and let $\theta_r$ be the corresponding 
raising substitution.
To prove the theorem, it suffices to show the following:
\begin{enumerate}
\item for each 
  $\ctxvarty{\Gamma}{\mathbb{N}_\Gamma}{\ctxty{\mathcal{C}}{\mathcal{G}}}$
  in $\Xi$ it is the case that $\mathbb{N}_\Gamma \subseteq \mathbb{N}
  \cup \supportof{\theta}$ and there is a derivation for the judgement
$\wfctxvarty{(\mathbb{N}\cup \supportof{\theta}) \setminus \mathbb{N}_\Gamma}
    {\Psi''}
    {\ctxty{\mathcal{C}}{\hsubst{\theta_r}{\hsubst{\theta}{\mathcal{G}}}}}$,
  and
\item for each formula $F'$ in $\{F\} \cup \Omega$ it is
the case that $\wfform{(\mathbb{N} \cup \supportof{\theta}) \cup
  \Psi'' \cup
  \Theta_0}{\ctxsanstype{\Xi}}{\hsubst{\theta_r}{\hsubst{\theta}{F'}}}$ is
derivable.
\end{enumerate}
Implicit here is the requirement that the relevant substititions are
well-defined.

\medskip
We first show (1).
Since $\mathcal{S}$ is well-formed, it must be the case that
$\mathbb{N}_\Gamma \subseteq \mathbb{N}$ so the first part of this
requirement obviously holds.
The wellformedness of $\mathcal{S}$ implies that 
$\wfctxvarty{\mathbb{N}\setminus\mathbb{N}_{\Gamma}}
            {\Psi}
            {\ctxty{\mathcal{C}}{\mathcal{G}}}$
has a derivation.
The substitution compatibility of $\seqsub{\theta}{\Psi'}$ with $\mathcal{S}$
implies that $\context{\theta}$ and $\Psi'\setminus\context{\theta}$ agree with
$\Psi$ on the type assignments to the variables that are common to
them, from which it follows that 
$\Psi \subseteq \aritysum{\context{\theta}}
                         {((\Psi\setminus\context{\theta})\cup\Psi')}$.
Theorem~\ref{th:ctx-ty-wk} now allows us to conclude that there must
be a derivation for 
$\wfctxvarty{(\mathbb{N}\setminus\mathbb{N}_{\Gamma}}
            {\aritysum{\context{\theta}}{((\Psi\setminus\context{\theta})\cup\Psi')}}
            {\ctxty{\mathcal{C}}{\mathcal{G}}}$.
From the substitution compatibility of $\seqsub{\theta}{\Psi'}$ with
$\mathcal{S}$ it also follows that $\theta$ must be type preserving
with respect to
$(\mathbb{N} \cup \supportof{\theta}) \cup
    \STLCGamma_0 \cup ((\Psi \setminus \context{\theta}) \cup \Psi')$
    and therefore with respect to
$\noms \cup \STLCGamma_0 \cup ((\Psi\setminus\context{\theta})\cup\Psi')$
and, because $\supportof{\theta}$ is disjoint from $\mathbb{N}$, that 
$(\mathbb{N}\setminus\mathbb{N}_{\Gamma}) \cup \supportof{\theta} = 
(\mathbb{N}\cup \supportof{\theta}) \setminus\mathbb{N}_{\Gamma}$.
Using Theorem~\ref{th:ctxtyinst-hsubst} we now determine that
$\wfctxvarty{(\mathbb{N}\cup \supportof{\theta})\setminus\mathbb{N}_{\Gamma}}
            {(\Psi\setminus\context{\theta})\cup\Psi'}
            {\ctxty{\mathcal{C}}{\hsubst{\theta}{\mathcal{G}}}}$
has a derivation.
From Definitions~\ref{def:raising-subst} and
\ref{def:seq-term-subst-app}, it follows that $\theta_r$ is type 
preserving with respect to $\noms\cup\STLCGamma_0\cup\Psi''$, that
$\context{\theta_r}=(\Psi\setminus\context{\theta})\cup\Psi'$, and
that $\supportof{\theta_r} \subseteq
(\mathbb{N}\cup \supportof{\theta})\setminus\mathbb{N}_{\Gamma}$.
Thus, using Theorem~\ref{th:ctxtyinst-hsubst} again, we determine that
$\wfctxvarty{(\mathbb{N}\cup \supportof{\theta})\setminus\mathbb{N}_{\Gamma}}
            {\Psi''}
            {\ctxty{\mathcal{C}}{\hsubst{\theta_r}{\hsubst{\theta}{\mathcal{G}}}}}$
has a derivation.

\medskip
We now consider requirement (2).
Let $F'$ be a formula in $\{F\} \cup \Omega$.
As we have seen already, $\theta$ is type preserving with respect to
the arity context 
$(\mathbb{N} \cup \supportof{\theta}) \cup
\STLCGamma_0 \cup ((\Psi \setminus \context{\theta}) \cup \Psi')$.
Since $S$ is well-formed, there must be a derivation for
$\wfform{\mathbb{N} \cup \Psi \cup \Theta_0}{\ctxsanstype{\Xi}}{F'}$.
Since $\context{\theta}$ and $\Psi$ agree on the type assignments to
common variables, we can conclude, using Theorem~\ref{th:wfsupset},
that there must be a derivation for
$\wfform{\aritysum{\context{\theta}}
                  {((\mathbb{N} \cup\supportof{\theta})
                   \cup ((\Psi \setminus \context{\theta}) \cup \Psi')
                   \cup \STLCGamma_0)}} 
        {\ctxsanstype{\Xi}}{F'}$.
From Theorem~\ref{th:subst-formula}, it then follows that
$\wfform{(\mathbb{N} \cup \supportof{\theta})
             \cup (\Psi \setminus \context{\theta}) \cup \Psi')
             \cup \STLCGamma_0}
        {\ctxsanstype{\Xi}}
        {\hsubst{\theta}{F'}}$
has a derivation.
This last judgement may be rewritten as
$\wfform{\aritysum{\context{\theta_r}}{((\mathbb{N} \cup
    \supportof{\theta}) \cup
    \STLCGamma_0})}{\ctxsanstype{\Xi}}{\hsubst{\theta}{F'}}$ from which, using
Theorem~\ref{th:wfsupset}, it follows that
$\wfform{\aritysum{\context{\theta_r}}
                  {((\mathbb{N} \cup \supportof{\theta})
                    \cup \Psi'' \cup \STLCGamma_0})}
        {\ctxsanstype{\Xi}}
        {\hsubst{\theta}{F'}}$
has a derivation.
Noting that $\theta_r$ is type preserving with respect to
$\supportof{\theta} \cup \Psi''$ and hence with respect to
$((\mathbb{N} \cup \supportof{\theta}) \cup \Psi'' \cup \STLCGamma_0$,
we may use Theorem~\ref{th:subst-formula} to conclude that there must
be a derivation for
$\wfform{(\mathbb{N} \cup \supportof{\theta})
            \cup \Psi'' \cup \STLCGamma_0}
        {\ctxsanstype{\Xi}}
        {\hsubst{\theta_r}{\hsubst{\theta}{F'}}}$.
\end{proof}

We will also need to consider the application of context variables
substitutions to sequents. 
We require such substitutions to respect the types associated with
the variables. 
In contrast to term variables substitutions, context variables
substitutions are not permitted to introduce new term variables into
the sequent and they may use nominal constants that are already
present in the sequent.
We formalize these ideas in the definition of appropriateness for such
subsitutions.
\begin{definition}\label{def:seq-ctxt-subst}
Let $\Xi$ be a context variables context and let $\sigma$ be the
context variables substitution $\{G_1/\Gamma_1,\ldots,G_n/\Gamma_n\}$.
We write $\ctxvarminus{\Xi}{\sigma}$ to denote the result of trimming
$\Xi$ to exclude associations for the variables in $\domain{\sigma}$. 
We say that $\sigma$ is appropriate for a context variables context $\Xi$
with respect to a term variables context $\Psi$ and a collection of nominal
constants $\mathbb{N}$ if, for $1 \leq i \leq n$, it is the case that
$\ctxvarty{\Gamma_i}{\mathbb{N}_i}{\ctxty{\mathcal{C}_i}{\mathcal{G}_i}}
\in \Xi$ and, further,
$\wfctxvarty{\mathbb{N}}{\Psi}{\ctxty{\mathcal{C}_i}{\mathcal{G}_i}}$
and 
$\ctxtyinst{\supportof{\sigma} \setminus \mathbb{N}_i}
           {\Psi}
           {\ctxvarminus{\Xi}{\sigma}}
           {\ctxty{\mathcal{C}_i}{\mathcal{G}_i}}{G_i}$
have derivations.
The substitution $\sigma$ is additionally said to cover $\Xi$ if
$\ctxvarminus{\Xi}{\sigma}= \emptyset$. 
We say that $\sigma$ is appropriate for a well-formed sequent
$\mathcal{S}=\seq[\mathbb{N}]{\Psi}{\Xi}{\Omega}{F}$ if it is
appropriate for $\Xi$ with respect to $\Psi$ and $\mathbb{N}$  and
that it covers $\mathcal{S}$ if it covers $\Xi$.
\end{definition}

Context types are unaffected by context variables substitutions.
Context expressions may be impacted by such substitutions but, for the
right kind of substitutions, they continue to be instances of relevant
context types.
This is made precise in the following theorem.

\begin{theorem}\label{th:ctxtyinst-subst}
Let $\sigma$ be a context variables substitution that is appropriate for 
$\Xi$ with respect to $\Psi$ and $\mathbb{N}$ and let 
$\supportof{\sigma} \subseteq \mathbb{N}$.
If 
$\ctxtyinst{\mathbb{N}}{\Psi}{\Xi}{\ctxty{\mathcal{C}}{\mathcal{G}}}{G}$
has a derivation, then 
$\ctxtyinst{\mathbb{N}}
           {\Psi}
           {\ctxvarminus{\Xi}{\sigma}}
           {\ctxty{\mathcal{C}}{\mathcal{G}}}
           {\subst{\sigma}{G}}$
also has a derivation.
\end{theorem}
\begin{proof}
By induction on the derivation of 
$\ctxtyinst{\mathbb{N}}
           {\Psi}
           {\Xi}
           {\ctxty{\mathcal{C}}{\mathcal{G}}}
           {G}$.
The only case that perhaps needs explicit consideration is that where
$G$ is a context variable $\Gamma$ for which $\sigma$ substitutes the
context expression $G'$.
In this case $\mathcal{G}$ must be empty and there must be an
assignment of the form
$\ctxvarty{\Gamma}
          {\mathbb{N}_\Gamma}
          {\ctxty{\mathcal{C}}}
          {\mathcal{G}'}$
in $\Xi$ for some $\mathbb{N}_\Gamma$ and $\mathcal{G'}$.
Further, by the appropriateness of $\sigma$, there
must be a derivations for
$\wfctxvarty{\mathbb{N}}{\Psi}{\ctxty{\mathcal{C}}{\mathcal{G}'}}$
and
$\ctxtyinst{\supportof{\sigma}\setminus\mathbb{N}_\Gamma}
           {\Psi}
           {\ctxvarminus{\Xi}{\sigma}}
           {\ctxty{\mathcal{C}}{\mathcal{G}'}}
           {G_i}$.
Using Theorem~\ref{th:ctx-ty-wk} and the assumption that
$\supportof{\sigma} \subseteq \mathbb{N}$, it follows 
that 
$\ctxtyinst{\mathbb{N}}
           {\Psi}
           {\ctxvarminus{\Xi}{\sigma}}
           {\ctxty{\mathcal{C}}{\mathcal{G}'}}
           {G_i}$
has a derivation.
An easy induction on the first derivation suffices to show that if
$\ctxtyinst{\mathbb{N}}
           {\Psi}
           {\ctxvarminus{\Xi}{\sigma}}
           {\ctxty{\mathcal{C}}{\mathcal{G}'}}
           {G_i}$
and
$\wfctxvarty{\mathbb{N}}{\Psi}{\ctxty{\mathcal{C}}{\mathcal{G}'}}$
have derivations then there must be one for 
$\ctxtyinst{\mathbb{N}}
           {\Psi}
           {\ctxvarminus{\Xi}{\sigma}}
           {\ctxty{\mathcal{C}}{\cdot}}
           {G_i}$.
\end{proof}

We formalize the application of a context variables substitution to a well-formed 
sequent when the conditions of appropriateness are met in the
definition below.
Note that context substitutions can also introduce new nominal
constants, and we use the technique of raising once again to 
permit these constants to be used in substitutions for term variables
in the resulting sequent. 
We assume here and elsewhere that the application of a context
variables substitution to a set of  formulas distributes to each member
of the set. 

\begin{definition}\label{def:cvar-subst-app}
Let $\mathcal{S}$ be a well-formed sequent $\seq[\mathbb{N}]{\Psi}{\Xi}{\Omega}{F}$
and $\sigma$ be appropriate for $\mathcal{S}$.
Further let $\Psi'$ be a version of $\Psi$ raised over $\supportof{\sigma}$
and let $\theta_r$ be the corresponding raising substitution.
Then the application of $\sigma$ to $\mathcal{S}$, 
denoted by $\subst{\sigma}{\mathcal{S}}$, is
the sequent
\[\seq[\mathbb{N}\cup\supportof{\sigma}]
     {\Psi'}
     {\hsubst{\theta_r}{\ctxvarminus{\Xi}{\sigma}}}
     {\hsubst{\theta_r}{\subst{\sigma}{\Omega}}}
     {\hsubst{\theta_r}{\subst{\sigma}{F}}}.\]
\end{definition}

Once again, we must justify the use in the above definition of the
notation introduced after Theorem~\ref{th:aritysubs}.
This is part of the content of the following theorem.

\begin{theorem}\label{th:seq-ctx-var-subst-okay}
Let $\mathcal{S}$ be a well-formed sequent and let $\sigma$ be a
context variables substitution that is appropriate for
$\mathcal{S}$. Then $\subst{\sigma}{\mathcal{S}}$ is a well-defined
and it yields a well-formed sequent.
\end{theorem}

\begin{proof}
Let $\mathcal{S}$ be $\seq[\mathbb{N}]{\Psi}{\Xi}{\Omega}{F}$ and let
$\subst{\sigma}{\mathcal{S}}$ be
$\seq[\mathbb{N}\cup\supportof{\sigma}]
     {\Psi'}
     {\hsubst{\theta_r}{\ctxvarminus{\Xi}{\sigma}}}
     {\hsubst{\theta_r}{\subst{\sigma}{\Omega}}}
     {\hsubst{\theta_r}{\subst{\sigma}{F}}}$,
where $\theta_r$ and $\Psi'$ are the raising substitution and the
raised term variables context as per Definition~\ref{def:cvar-subst-app}. 
Using an argument similar to that in the proof of
Theorem~\ref{th:seq-term-subs-ok}, we can show that
$\wfctxvarty{(\mathbb{N} \cup \supportof{\sigma}) \setminus \mathbb{N}_\Gamma}
            {\Psi'}
            {\hsubst{\theta_r}{(\ctxty{\mathcal{C}}{\mathcal{G}})}}$
must have a derivation for each
$\ctxvarty{\Gamma}{\mathbb{N}_\Gamma}{\ctxty{\mathcal{C}}{\mathcal{G}}}$
in $\ctxvarminus{\Xi}{\sigma}$.
From the wellformedness of $\mathcal{S}$ and the appropriateness of
$\sigma$ for $\mathcal{S}$, it is easy to see that 
$\wfform{(\mathbb{N} \cup \supportof{\sigma})
            \cup \Psi \cup \STLCGamma_0}
        {\ctxsanstype{\ctxvarminus{\Xi}{\sigma}}}
        {\subst{\sigma}{F'}}$
must have a derivation for each $F' \in \Omega \cup \{F\}$.
Noting that the variables assigned types by $\Psi'$ are distinct from
those assigned types by $\Psi$ and that $\context{\theta_r} = \Psi$,
we may use Theorems~\ref{th:wfsupset} and \ref{th:subst-formula} to
conclude that 
$\hsubst{\theta_r}{\subst{\sigma}{F'}}$ is well-defined and that
$\wfform{(\mathbb{N} \cup \supportof{\sigma})
            \cup \Psi' \cup \STLCGamma_0}
        {\ctxsanstype{\ctxvarminus{\Xi}{\sigma}}}
        {\hsubst{\theta_r}{\subst{\sigma}{F'}}}$
has a derivation.
The theorem follows easily from these observations. 
\end{proof}

Term and context substitutions will often be used in tandem to produce
instances of sequents.
In this context, we shall say that \emph{$\theta$ and $\sigma$ are
proper for the sequent $\mathcal{S}$} if there is a term
variables context $\Psi$ such that $\langle \theta, \Psi \rangle$ is
substitution compatible with $\mathcal{S}$ and $\sigma$ is appropriate
for $\hsubstseq{\Psi}{\theta}{S}$. 

We now define validity for sequents.
For a sequent with empty term variables and context variables contexts,
we base the definition on the validity of closed formulas.
For a sequent that has formulas with free term and context variables,
this is done by considering all their relevant substitution instances.

\begin{definition}\label{def:seq-validity}
A well-formed sequent of the form
$\seq[\mathbb{N}]{\emptyset}{\emptyset}{\Omega}{F}$, referred to as a
\emph{closed} sequent, is valid if either some formula in $\Omega$ is
not valid or $F$ is valid.
A well-formed sequent $\mathcal{S}$ of the form
$\seq[\mathbb{N}]{\Psi}{\Xi}{\Omega}{F}$ is valid if
for every term substitution $\theta$ that is type preserving with
respect to $\noms \cup \Theta_0$ and such that
$\Psi = \context{\theta}$ and
$\langle \theta, \emptyset \rangle$ is 
substitution compatible with $\mathcal{S}$, and for every context
substitution $\sigma$ that is appropriate for and covers
$\hsubstseq{\emptyset}{\theta}{\mathcal{S}}$, it
is the case that
$\subst{\sigma}{\hsubstseq{\emptyset}{\theta}{\mathcal{S}}}$ is
valid. 
Note that each such
$\subst{\sigma}{\hsubstseq{\emptyset}{\theta}{\mathcal{S}}}$
will be a well-formed and closed sequent in these circumstances and we 
shall refer to it as the closed instance of $\mathcal{S}$ identified
by $\theta$ and $\sigma$. 
\end{definition}

The following theorem, whose proof is obvious, provides the basis for
using our proof system for determining the validity of formulas.

\begin{theorem}\label{seq-fmla-validity}
Let $F$ be a formula such that
$\wfform{\noms \cup \Theta_0}{\emptyset}{F}$ is derivable and let
$\mathbb{N}$ be the set of nominal constants that appear in $F$.
Then the sequent
$\seq[\mathbb{N}]{\emptyset}{\emptyset}{\emptyset}{F}$ is
well-formed.
Moreover, $F$ is valid if and only if
$\seq[\mathbb{N}]{\emptyset}{\emptyset}{\emptyset}{F}$ is.
\end{theorem}

We will need in later discussions to consider the application of
permutations of the nominal constants to sequents.
We make this notion precise below. 

\begin{definition}\label{def:perm-seq}
We extend the application of a permutation $\pi$ of nominal constants
to a variety of expressions $E$ associated with sequents,
denoting the result as before by $\permute{\pi}{E}$.
The application of $\pi$ to a set $\mathbb{N}$ of
nominal constants yields the set
$\{n' |\ n\in\mathbb{N}\ \mbox{\rm and}\ \pi(n)=n'\}$.
The application of $\pi$ to a context variable type distributes to the
constituent block instances, its application to an association of the
form 
$\ctxvarty{\Gamma}{\mathbb{N}_\Gamma}{\ctxty{\mathcal{C}}{\mathcal{G}}}$
with context variables distributes to $\mathbb{N}_\Gamma$ and
$\ctxty{\mathcal{C}}{\mathcal{G}}$,
and its application to a collection of formulas or a context variables
context distributes to the members of the collection. 
Finally, the application of $\pi$ to a sequent
$\seq[\mathbb{N}]{\Psi}{\Xi}{\Omega}{F}$ 
yields the sequent
$\seq[\permute{\pi}{\mathbb{N}}]
     {\Psi}
     {\permute{\pi}{\Xi'}}
     {\permute{\pi}{\Omega'}}
     {\permute{\pi}{F'}}$. 
\end{definition}

The following theorem notes that wellformedness of sequents is
preserved under these permutations and that permutations commute in a
suitable sense with term and context variables substitutions.
The proof of the theorem is routine even if tedious; we omit the
details. 

\begin{theorem}\label{th:perm-subst}
Let $\pi$ be a permutation of the nominal constants.
If $\mathcal{S}$ is a well-formed sequent then so is
$\permute{\pi}{\mathcal{S}}$.
Further, if $\langle \theta, \Psi \rangle$ is 
substitution compatible with $\permute{\pi}{\mathcal{S}}$ then
$\langle \permute{\inv{\pi}}{\theta}, \Psi \rangle$ is substitution
compatible with $\mathcal{S}$ and
$\hsubstseq{\Psi}{\theta}{(\permute{\pi}{\mathcal{S}})} =
\permute{\pi}{(\hsubstseq{\Psi}{\permute{\inv{\pi}}{\theta}}{\mathcal{S}})}$.
Finally, if $\sigma$ is a context variables substitution that is
appropriate for $\permute{\pi}{\mathcal{S}}$ then
$\inv{\pi}{\sigma}$ is appropriate for $\mathcal{S}$ and
$\subst{\sigma}{(\permute{\pi}{\mathcal{S}})} =
\permute{\pi}{(\subst{\permute{\inv{\pi}}{\sigma}}{\mathcal{S}})}$.
\end{theorem}

Theorem~\ref{th:perm-valid} provides a counterpart to 
Theorem~\ref{th:perm-form} at the level of sequents. 
Its proof, whose details we again omit, uses the earlier result and
the observation about substitutions contained in Theorem~\ref{th:perm-subst}.

\begin{theorem}\label{th:perm-valid}
If $\pi$ is a permutation of the nominal constants and $\mathcal{S}$
is a well-formed closed sequent that is valid, then
$\permute{\pi}{\mathcal{S}}$ is also a closed, valid sequent.
\end{theorem}

\subsection{The Core Proof Rules}
\label{ssec:core-logic}

The base for our proof system is provided by a collection of rules that
embody the interpretation of the logical symbols and also certain
aspects of the meanings of sequents. 
We present these rules in this subsection: we first identify some
structural rules, we then consider axioms and the cut rule, and we
finally introduce rules for the logical symbols.

\subsubsection{Structural Rules}

This subcollection of rules is presented in
Figure~\ref{fig:rules-structural}.
These rules can be subcategorized into those
that allow for weakening and contracting the assumption set in a
sequent and those that permit the weakening and strengthening of the
support set, the term variables context, and the context variables context.
Rules of the second subcategory encode the fact that vacuous
(well-formed) extensions to the bindings manifest in a sequent will
not impact its validity.
The strengthening and weakening rules for contexts include premises
that force modifications to context variable types and the
satisfaction of typing judgements that are necessary to ensure the 
well-formedness of the sequents in any 
application of the rule. 

\begin{figure}

\begin{center}
\begin{tabular}{cc}    
\infer[\weakening]
      {\seq[\mathbb{N}]{\Psi}{\Xi}{\setand{\Omega}{F_2}}{F_1}}
      {\seq[\mathbb{N}]{\Psi}{\Xi}{\Omega}{F_1}}
\qquad &
\qquad
\infer[\contraction]
      {\seq[\mathbb{N}]{\Psi}{\Xi}{\setand{\Omega}{F_2}}{F_1}}
      {\seq[\mathbb{N}]{\Psi}{\Xi}{\setand{\Omega}{F_2, F_2}}{F_1}}
\end{tabular}
\end{center}

\begin{center}
\begin{tabular}{c}
\infer[\sstr]
      {\seq[\mathbb{N}]{\Psi}{\Xi}{\Omega}{F}}
      {\begin{array}{c}
         \Xi=\{
         \ctxvarty{\Gamma_i}
                  {(\mathbb{N}_i\setminus\mathbb{N}')}
                  {\ctxty{\mathcal{C}_i}{\mathcal{G}_i}}
               \ |\ 
               \ctxvarty{\Gamma_i}
                        {\mathbb{N}_i}
                        {\ctxty{\mathcal{C}_i}{\mathcal{G}_i}}\in\overline{\Xi}
             \}
         \\
         \left\{\wfctxvarty{(\mathbb{N},\mathbb{N}')\setminus\mathbb{N}_i}
                           {(\Psi,\Psi')}
                           {\ctxty{\mathcal{C}_i}{\mathcal{G}_i}}\ \middle|\ 
                 \ctxvarty{\Gamma_i}
                          {\mathbb{N}_i} 
                          {\ctxty{\mathcal{C}_i}{\mathcal{G}_i}}\in\Xi'\right\}
         \\
         \seq[\mathbb{N},\mathbb{N}']{\Psi,\Psi'}{\overline{\Xi},\Xi'}{\Omega}{F}
       \end{array}} 

\\[10pt]

\infer[\sweak]
      {\seq[\mathbb{N},\mathbb{N}']{\Psi,\Psi'}{\overline{\Xi},\Xi'}{\Omega}{F}}
      {\begin{array}{c}
         \Xi=\{
               \ctxvarty{\Gamma_i}{(\mathbb{N}_i\setminus\mathbb{N}')}{\ctxty{\mathcal{C}_i}{\mathcal{G}_i}}
               \ |\ 
               \ctxvarty{\Gamma_i}
                        {\mathbb{N}_i}
                        {\ctxty{\mathcal{C}_i}{\mathcal{G}_i}}\in\overline{\Xi}
             \}
         \\
         \left\{ \wfctxvarty{(\mathbb{N}\setminus\mathbb{N}_i)}
                            {\Psi}
                            {\ctxty{\mathcal{C}_i}{\mathcal{G}_i}}
                      \ \middle|\ 
                     \ctxvarty{\Gamma_i}
                              {\mathbb{N}_i}
                              {\ctxty{\mathcal{C}_i}{\mathcal{G}_i}}\in\Xi \right\}                            
         \\
         \left\{\wfform{\mathbb{N}\cup\STLCGamma_0\cup\Psi}
                       {\{\Gamma\ |\ \ctxvarty{\Gamma}
                                       {\mathbb{N}}
                                       {\ctxty{\mathcal{C}}
                                              {\mathcal{G}}}\in\Xi\}}
                       {F'}
                  \ \middle|\ 
                  F'\in\Omega\cup\{F\} \right\}
         \\
         \seq[\mathbb{N}]{\Psi}{\Xi}{\Omega}{F}
       \end{array}}
\end{tabular}
\end{center}

\caption{The Structural Rules}\label{fig:rules-structural}
\end{figure}

The following theorem shows that these rules require the proof of only
well-formed sequents in constructing a proof of a well-formed sequent.

\begin{theorem}\label{th:structural-wf}
The following property holds for each rule in
Figure~\ref{fig:rules-structural}: if the conclusion sequent is
well-formed, the premises expressing typing conditions have
derivations and the conditions expressed by the other, non-sequent
premises are satisfied, then all the sequent premises must
be well-formed.
\end{theorem}

\begin{proof}
The claim is obvious for the \weakening, \contraction, and \sweak\ rules.
This leaves only the \sstr\ rule.
The welformedness requirements as they pertain to the formulas in
$\Omega \cup \{F\}$ in the premise sequent in this rule follow
easily from Theorem~\ref{th:wfsupset}.
The second premise ensures that the context types in $\Xi'$ in the
rule meet the necessary requirements.
For the remaining context types, we use the relationship between
$\overline{\Xi}$ and $\Xi$ and Theorem~\ref{th:ctx-ty-wk} to derive
the needed property from the wellformedness of the conclusion
sequent.
\end{proof}

The following lemma will be useful in showing the soundness of the
structural rules.

\begin{lemma}\label{lem:seq-equiv}
Let $\mathcal{S} = \seq[\mathbb{N}]{\Psi}{\Xi}{\Omega}{F}$ and 
$\mathcal{S}' = \seq[\mathbb{N}']{\Psi'}{\Xi'}{\Omega}{F}$ be
well-formed sequents such that $\mathbb{N}'\subseteq\mathbb{N}$, $\Psi'\subseteq\Psi$, and there is some
subset $\overline{\Xi}$ for the context variables context $\Xi$ where
\[\Xi'=
  \left\{
    \ctxvarty{\Gamma_i}
             {\left(\mathbb{N}_i\setminus(\mathbb{N}\setminus\mathbb{N}')\right)}
             {\ctxty{\mathcal{C}_i}{\mathcal{G}_i}}
      \ \middle|\ 
    \ctxvarty{\Gamma_i}{\mathbb{N}_i}{\ctxty{\mathcal{C}_i}{\mathcal{G}_i}}\in\overline{\Xi}
  \right\}.\]
Then $\mathcal{S}$ is valid if and only if $\mathcal{S}'$ is valid. 
\end{lemma}

\begin{proof}
We will make use in the proof of the easily confirmed fact that a
closed and well-formed sequent that differs from another such sequent
only in its support set is valid exactly when the other sequent is
valid. 
We note now that if $\theta$ and $\sigma$ are proper for $\mathcal{S}$
and $\sigma'$ is a restriction of $\sigma$ to $\domain{\Xi}$, then
$\theta$ and $\sigma'$ must be proper for $\mathcal{S}'$.
From this it follows easily that corresponding to each closed instance
of $\mathcal{S}$ there must be a closed instance of $\mathcal{S}'$
that differs at most in its support set.
Thus, if $\mathcal{S}'$ is valid, $\mathcal{S}$ must be too.

In the ``only if'' direction, we note first that a term substitution
$\theta$ and a context substitution $\sigma$ that identify a closed
instance of $\mathcal{S}'$ may not be proper for $\mathcal{S}$
because $\theta$ and $\sigma$ use nominal constants from
$\mathbb{N} \setminus \mathbb{N}'$ in the terms they substitute for variables.
To rectify this problem, we use a permutation $\pi$ that swaps the
nominal constants in $\mathbb{N} \setminus \mathbb{N}'$ with ones that
do not appear in $\mathbb{N}$, $\theta$ or $\sigma$, leaving
the other nominal constants unaffected.
It is easily seen that $\permute{\pi}{\theta}$ and
$\permute{\pi}{\sigma}$ identify a closed instance of $\mathcal{S'}$
and are proper with respect to $\mathcal{S}$.
These substitutions can now be extended to substitutions $\theta'$ and
$\sigma'$ that substitute closed expressions of the appropriate kinds
for the variables in $\Psi\setminus\Psi'$ and $\Xi\setminus\Xi'$,
respectively, and that thereby identify a closed instance of
$\mathcal{S}$. 
This closed instance is valid by assumption.
Further, it differs from
$\subst{\permute{\pi}{\sigma}}
       {\hsubstseq{\emptyset}
                  {\permute{\pi}{\theta}}
                  {\mathcal{S}'}}$
only in its support set.
Thus, the latter sequent must also be valid.
Noting that $\mathcal{S}' = \permute{\pi}{\mathcal{S}'}$ and using
Theorem~\ref{th:perm-subst}, we see that
$\subst{\permute{\pi}{\sigma}}
       {\hsubstseq{\emptyset}
                  {\permute{\pi}{\theta}}
                  {\mathcal{S}'}}$
is the same sequent as
$\permute{\pi}
         {\subst{\sigma}
                {\hsubstseq{\emptyset}
                           {\theta}
                           {\mathcal{S}'}}}$.
Using Theorem~\ref{th:perm-valid}, we now conclude that 
$\subst{\sigma}{\hsubstseq{\emptyset}{\theta}{\mathcal{S}'}}$ is
valid.
Since this argument is independent of the choice of $\theta$
and $\sigma$, it follows that $\mathcal{S}'$ must be valid if
$\mathcal{S}$ is.
\end{proof}

The soundness of the structural rules is the content of the following theorem.

\begin{theorem}\label{th:structural-sound}
The following property holds for every instance of each of the rules
in Figure~\ref{fig:rules-structural}: if the premises expressing
typing judgements are derivable, the conditions described in the other
non-sequent premises are satisfied and the premise sequent is valid,
then the conclusion sequent must also be valid. 
\end{theorem}

\begin{proof}
The argument is obvious for the \weakening\ and \contraction\ rules.
For the \sstr\ and \sweak\ rules, we make use of
Lemma~\ref{lem:seq-equiv}.
\end{proof}

\subsubsection{The Axiom and the Cut Rule}\label{sssec:initial-rules}

The \cut\ rule codifies the notion of lemmas: if we can show the
validity of a formula relative to a given assumption set, then this
formula can be included in the assumptions to simplify the
reasoning process.
The \id\ rule recognizes the validity of a sequent in which the
conclusion formula appears in the assumption set.
In its simplest form, the \id\ rule would require the conclusion
formula to be included as is in the assumption set. 
It is possible, and also pragmatically useful, to generalize this form
to allow also for a permutation of nominal constants in the formulas
in the process of matching.
However, this has to be done with care to ensure that identity under
the considered permutations continues to hold even after later
instantiations of term and context variables appearing in the
formulas.
The following definition presents a notion of equivalence for formulas
under permutations that encodes this care. 

\begin{definition}[Formula Equivalence]\label{def:form-eq}
The equivalence of two context expressions $G_1$ and $G_2$ with respect to a context
variables context $\Xi$ and a permutation $\pi$, written $\ctxeq{\Xi}{\pi}{G_2}{G_1}$,
is a relation defined by the following three clauses:
\begin{enumerate}
\item $\ctxeq{\Xi}{\pi}{\emptyctx}{\emptyctx}$ holds for any $\Xi$ and $\pi$.

\item $\ctxeq{\Xi}{\pi}{\Gamma}{\Gamma}$ holds if for some 
$\ctxvarty{\Gamma'}{\mathbb{N}}{\ctxty{\mathcal{C}}{\mathcal{G}}}\in\Xi$
  it is the case that $\Gamma = \Gamma'$ and $\supp{\pi}\subseteq\mathbb{N}$.

\item If $G_1=(G_1',n_1:A_1)$ and $G_2=(G_2',n_2:A_2)$ then 
$\ctxeq{\Xi}{\pi}{G_2}{G_1}$ holds if
$\permute{\pi}{n_2}$ is identical to $n_1$, $\permute{\pi}{A_2}$ is
  the same type as $A_1$, and $\ctxeq{\Xi}{\pi}{G'_2}{G'_1}$ holds.
\end{enumerate}
Two atomic formulas $\fatm{G'}{M':A'}$ and $\fatm{G}{M:A}$ are
considered equivalent with respect to $\Xi$ and $\pi$ if
$\ctxeq{\Xi}{\pi}{G'}{G}$ holds, $\permute{\pi}{M'}$ and $M$ are the
same terms, and $\permute{\pi}{A'}$ and $A$ are the same types.
Two arbitrary formulas $F'$ and $F$ are considered equivalent with
respect to $\Xi$ and $\pi$ if their component parts are so equivalent,
allowing, of course, for a renaming of variables bound by quantifiers.
We denote this equivalence by the expression $\formeq{\Xi}{\pi}{F'}{F}$.
\end{definition}

\begin{figure}
\[
\infer[\id]
      {\seq[\mathbb{N}]
           {\Psi}
           {\Xi}
           {\Omega}
           {F}}
      {F'\in\Omega 
         & 
       \pi\ \mbox{is a permutation of nominal constants such that}\ \supportof{\pi}\subseteq\mathbb{N}
         &
       \formeq{\Xi}
              {\pi}
              {F'}
              {F}}
\]
\smallskip
\[
\infer[\cut]{\seq[\mathbb{N}]{\Psi}{\Xi}{\Omega}{F_1}}
            {\seq[\mathbb{N}]{\Psi}{\Xi}{\Omega}{F_2} &
             \seq[\mathbb{N}]{\Psi}{\Xi}{\setand{\Omega}{F_2}}{F_1} &
             \wfform{\mathbb{N}\cup\STLCGamma_0\cup\Psi}{\dom{\Xi}}{F_2}}
\]
\caption{The Axiom and the Cut Rule}
\label{fig:rules-other}
\end{figure}

The \id\ and the \cut\ rules are presented in
Figure~\ref{fig:rules-other}.
The \id\ rule limits the permutations to be considered to 
ones that rename only nominal constants appearing in the support set
of the sequent.
The \cut\ rule includes a premise that ensures the wellformedness of
the cut formula.

\begin{theorem}\label{th:other-wf}
The following property holds of the \id\ and the \cut\ rule: if the
conclusion sequent is well-formed, the premises expressing typing
conditions have derivations and the conditions expressed by the other,
non-sequent premises are satisfied, then the premise sequents must be
well-formed.
\end{theorem}
\begin{proof}
The requirement is vacuously true for the \id\ rule and it has an
obvious proof for the \cut\ rule.
\end{proof}

In showing the soundness of the \id\ rule, we will need the
observation contained in the two lemmas below that the equivalence of
formulas modulo permutations is preserved under the kinds of
substitutions that have to be considered in determining the validity
of sequents. 

\begin{lemma}\label{lem:equiv-hsub}
Suppose that $\formeq{\Xi}{\pi}{F_2}{F_1}$
is holds for some formulas $F_1$ and $F_2$ and some $\Xi$ and $\pi$ of 
the right kinds.
If $\theta$ is a term substitution such that
$\supp{\theta}\cap\supp{\pi}=\emptyset$ and
both $\hsub{\theta}{F_2}{F_2'}$ and $\hsub{\theta}{F_1}{F_1'}$ 
have derivations for some $F_1'$ and $F_2'$, then
$\formeq{\hsubst{\theta}{\Xi}}{\pi}{F'_2}{F'_1}$
holds.
\end{lemma}
\begin{proof}
Let $E$ be a type or term on which the application of the substitution
$\theta$ is defined and let $\pi$ be a permutation of the nominal
constants. 
We observe then that 
$\permute{\pi}{(\hsubst{\theta}{E})}$ is the same type or term as 
$\hsubst{\permute{\pi}{\theta}}{(\permute{\pi}{E})}$.
This observation, which was implicit in Theorem~\ref{th:perm-subst},
can be established by induction on the derivation of
$\hsub{\theta}{E}{E'}$ for the relevant $E'$. 
If $\supp{\theta}$ and $\supp{\pi}$ are disjoint as is  
assumed in this lemma, it further follows that
$\permute{\pi}{(\hsubst{\theta}{E})}$ is the same term or type as 
$\hsubst{\theta}{(\permute{\pi}{E})}$.

The lemma itself is proved by induction on the structure of $F_1$ or,
equivalently, $F_2$. 
Most of the cases are straightforward, the only one perhaps needing
explicit consideration being that when $F_1$ and $F_2$ are the atomic
formulas $\fatm{G_1}{M_1:A_1}$ and $\fatm{G_2}{M_2:A_2}$, respectively.
In this case, by assumption, $\permute{\pi}{M_2}$ and $M_1$ are the
same terms, $\permute{\pi}{A_2}$ and $A_1$ are the same types and 
$\ctxeq{\Xi}{\pi}{G_2}{G_1}$ holds.
Using the observation at the beginning of the proof, we see 
that  $\permute{\pi}{(\hsubst{\theta}{M_2})}$ is the same term as
$\hsubst{\theta}{M_1}$ and $\permute{\pi}{(\hsubst{\theta}{A_2})}$ is
the same type as $\hsubst{\theta}{A_1}$. 
Thus, the lemma would follow if we can show that 
$\ctxeq{\Xi}{\pi}{\hsubst{\theta}{G_2}}{\hsubst{\theta}{G_1}}$ holds.
However, this is easily done by an induction on the structure of $G_1$.
The argument is obvious when $G_1$ is $\emptyctx$ or a context
variable.
If $G_1$ is of the form $(G_1',n_1:A_1)$, then it must be the case that 
$G_2$ has the form $(G_2',n_2:A_2)$, where $\permute{\pi}{n_2}$ is the
same nominal constant as $n_1$, $\permute{\pi}{A_2}$ is equal to $A_1$
up to renaming of bound variables, and $\ctxeq{\Xi}{\pi}{G_2'}{G_1'}$ holds.
Drawing again the observation at the beginning of the proof, we may
conclude that $\permute{\pi}{\hsubst{\theta}{A_2}}$ is the same type
as $\hsubst{\theta}{A_1}$.
From the induction hypothesis, we know that
$\ctxeq{\Xi}{\pi}{\hsubst{\theta}{G_2'}}{\hsubst{\theta}{G_1'}}$
holds.
It is now evident that 
$\ctxeq{\Xi}{\pi}{\hsubst{\theta}{G_2}}{\hsubst{\theta}{G_1}}$ must hold.
\end{proof}

\begin{lemma}\label{lem:equiv-sub}
Let $\sigma$ be an appropriate substitution for $\Xi$ with respect
to some term variables context $\Psi$ and set of nominal constants
$\mathbb{N}$, and let $\formeq{\Xi}{\pi}{F'}{F}$ hold.  
Then 
$\formeq{\ctxvarminus{\Xi}{\sigma}}{\pi}{\subst{\sigma}{F'}}{\subst{\sigma}{F}}$
must also hold.
\end{lemma}
\begin{proof}
The key observation underlying the proof is that if $G_1$ and $G_2$
are two context expressions such that $\ctxeq{\Xi}{\pi}{G_1}{G_2}$
holds and $\sigma$ is a context substitution that is appropriate for
$\Xi$ with respect to $\Psi$ and $\mathbb{N}$, then
$\ctxeq{\ctxvarminus{\Xi}{\sigma}}
       {\pi}
       {\subst{\sigma}{G_1}}
       {\subst{\sigma}{G_2}}$
must hold.
This fact can be verified by an induction on the structure of $G_1$.
The only complex case is that where $G_1$ is a context variable
$\Gamma$ for which $\sigma$ substitutes the context expression $G'$.
Here, $\Xi$ must contain an association of the form
$\ctxvarty{\Gamma}
          {\mathbb{N}'}
          {\ctxty{\mathcal{C}}
          {\mathcal{G}}}$,
with $\supportof{\pi} \subseteq \mathbb{N}'$
and with there being a derivation for
$\ctxtyinst{\supportof{\sigma} \setminus \mathbb{N}'}
           {\Psi}
           {\ctxvarminus{\Xi}{\sigma}}
           {\ctxty{\mathcal{C}}{\mathcal{G}}}{G'}$.
By an obvious inductive argument on the derivation, we can show that
$\ctxeq{\ctxvarminus{\Xi}{\sigma}}{\pi}{G'}{G'}$
must hold, as we need to do in this case.

The lemma itself is proved by an obvious induction on the structure of
$F$.
In the case where $F$ is the atomic formula $\fatm{G}{M:A}$, we make
use of the above observation. 
\end{proof}

We can now show the soundness of the \id\ and \cut\ rules.

\begin{theorem}\label{th:other-sound}
The following property holds for every instance of the \id\ and \cut\ rules:
if the premises expressing typing judgements are derivable, the
conditions described in the other non-sequent premises are satisfied
and all the premise sequents are valid, 
then the conclusion sequent must also be valid. 
\end{theorem}
\begin{proof}
Let $\theta$ and $\sigma$ identify a closed instance of the conclusion
of the \id\ rule.
The requirements that $\theta$ and $\sigma$ must satisfy and the
assumption that the premise conditions $\supportof{\pi} \subseteq 
\mathbb{N}$ and  $\formeq{\Xi}{\pi}{F'}{F}$ hold allow us to use
Lemmas~\ref{lem:equiv-hsub} 
and \ref{lem:equiv-sub} to conclude that 
$\formeq{\emptyset}
        {\pi}
        {\subst{\sigma}{\hsubst{\theta}{F'}}}
        {\subst{\sigma}{\hsubst{\theta}{F}}}$
holds.
Now, for any context variables context $\Xi'$, permutation $\pi'$, and
closed formulas $G$ and $H$ it is easily seen that if
$\formeq{\Xi'}{\pi'}{G}{H}$ holds then $H$ must be the same formula as 
$\permute{\pi'}{G}$.
We may therefore conclude that the closed instance of the conclusion
of the \id\ rule identified by $\theta$ and $\sigma$ is valid.
Since the argument was independent of the choice of $\theta$ and
$\sigma$, we conclude that the theorem must be true for this case.

For the \cut\ rule, we note that, given the wellformedness premise
pertaining to the formula $F_2$, any substitutions $\theta$ and
$\sigma$ that identify a closed instance of the conclusion sequent
must also identify closed instances of the two premise sequents.
It is easy to see now that the validity of the latter two sequents
must entail the validity of the former.
We conclude that the theorem holds for this rule as well by noting that
the argument is again independent of the choice of $\theta$ and $\sigma$.
\end{proof}

\subsubsection{Rules for the Logical Symbols}

Figure~\ref{fig:rules-base} presents proof rules that facilitate
reasoning based on the meanings of the logical symbols that are
permitted in formulas. 
That we may focus attention only on well-formed sequents in the
context of these rules is the content of the following theorem.
\begin{theorem}\label{th:core-wf}
The following property holds of the rules in Figure~\ref{fig:rules-base}: if the
conclusion sequent is well-formed, the premises expressing typing
conditions have derivations and the conditions expressed by the other,
non-sequent premises are satisfied, then the premise sequents must be
well-formed.
\end{theorem}

\begin{proof}
The argument pertaining to $\ftrue$, $\ffalse$ and all the
propositional connectives is straightforward, leaving us only to argue
for the quantifier rules.

\begin{figure}[htb]

\begin{center}
\begin{tabular}{c}    
\infer[\topR]{\seq[\mathbb{N}]{\Psi}{\Xi}{\Omega}{\ftrue}}{}

\qquad 

\infer[\botL]{\seq[\mathbb{N}]{\Psi}{\Xi}{\setand{\Omega}{\ffalse}}{F}}{}

\\[10pt]

\infer[\andR]{\seq[\mathbb{N}]{\Psi}{\Xi}{\Omega}{\fand{F_1}{F_2}}}
             {\seq[\mathbb{N}]{\Psi}{\Xi}{\Omega}{F_1} &
               \seq[\mathbb{N}]{\Psi}{\Xi}{\Omega}{F_2}}

\qquad 

\infer[\andL]{\seq[\mathbb{N}]{\Psi}{\Xi}{\setand{\Omega}{\fand{F_1}{F_2}}}{F}}
             {\seq[\mathbb{N}]{\Psi}{\Xi}{\setand{\Omega}{F_i}}{F} &
              i\in\{1,2\}}

\\[10pt]
             
\infer[\orR]{\seq[\mathbb{N}]{\Psi}{\Xi}{\Omega}{\for{F_1}{F_2}}}
            {\seq[\mathbb{N}]{\Psi}{\Xi}{\Omega}{F_i} &
              i\in\{1,2\}}

\qquad
            
\infer[\orL]{\seq[\mathbb{N}]{\Psi}{\Xi}{\setand{\Omega}{\for{F_1}{F_2}}}{F}}
            {\seq[\mathbb{N}]{\Psi}{\Xi}{\setand{\Omega}{F_1}}{F} &
              \seq[\mathbb{N}]{\Psi}{\Xi}{\setand{\Omega}{F_2}}{F}}

\\[10pt]

\infer[\impR]{\seq[\mathbb{N}]{\Psi}{\Xi}{\Omega}{\fimp{F_1}{F_2}}}
      {\seq[\mathbb{N}]{\Psi}{\Xi}{\setand{\Omega}{F_1}}{F_2}}

\qquad
      
\infer[\impL]{\seq[\mathbb{N}]{\Psi}{\Xi}{\setand{\Omega}{\fimp{F_1}{F_2}}}{F}}
             {\seq[\mathbb{N}]{\Psi}{\Xi}{\Omega}{F_1} &
              \seq[\mathbb{N}]{\Psi}{\Xi}{\setand{\Omega}{F_2}}{F}}
\end{tabular}

\medskip\smallskip

\begin{tabular}{c}
\infer[\allR]
      {\seq[\mathbb{N}]{\Psi}{\Xi}{\Omega}{\fall{x:\alpha}{F}}}
      {\begin{array}{c}
           \mathbb{N}=
                  \{n_1:\alpha_1,\ldots,n_m:\alpha_m\}
              \qquad
           y\not\in\dom{\Psi}
              \qquad
           \hsub{\{\langle x, y\app n_1\ldots n_m,\alpha\rangle\}}{F}{F'}
           \\[2pt]
           \seq[\mathbb{N}]
               {\Psi\cup\{y:(\arr{\alpha_1}{\arr{\ldots}{\arr{\alpha_m}{\alpha}}})\}}
               {\Xi}
               {\Omega}
               {F'}
      \end{array}}

\\[10pt]

\infer[\allL]{\seq[\mathbb{N}]{\Psi}{\Xi}{\setand{\Omega}{\fall{x:\alpha}{F_1}}}{F_2}}
             {\stlctyjudg{\mathbb{N}\cup\STLCGamma_0\cup\Psi}{t}{\alpha}
               &
              \hsub{\{\langle x, t, \alpha\rangle\}}{F_1}{F_1'}
               &
               \seq[\mathbb{N}]{\Psi}{\Xi}{\setand{\Omega}{F_1'}}{F_2}
             }

\\[10pt]
             
\infer[\existsR]{\seq[\mathbb{N}]{\Psi}{\Xi}{\Omega}{\fexists{x:\alpha}{F}}}
                {\stlctyjudg{\mathbb{N}\cup\STLCGamma_0\cup\Psi}{t}{\alpha}
                  &
                 \hsub{\{\langle x, t, \alpha\rangle\}}{F}{F'}
                  &
                 \seq[\mathbb{N}]{\Psi}{\Xi}{\Omega}{F'}
                }

\\[10pt]

\infer[\existsL]
      {\seq[\mathbb{N}]{\Psi}{\Xi}{\setand{\Omega}{\fexists{x:\alpha}{F_1}}}{F_2}}
      {\begin{array}{c}
          \mathbb{N}=\{n_1:\alpha_1,\ldots,n_m:\alpha_m\}
           \qquad
           y\not\in\dom{\Psi}
           \qquad
           \hsub{\{\langle x, y\app n_1\ldots n_m, \alpha\rangle\}}{F_1}{F_1'}
         \\[2pt]
         \seq[\mathbb{N}]
             {\Psi\cup\{y:(\arr{\alpha_1}{\arr{\ldots}{\arr{\alpha_m}{\alpha}}})\}}
             {\Xi}
             {\setand{\Omega}{F_1'}}
             {F_2} 
      \end{array}}

\\[10pt]

\infer[\ctxR]{\seq[\mathbb{N}]{\Psi}{\Xi}{\Omega}{\fctx{\Gamma}{\mathcal{C}}{F}}}
             {\seq[\mathbb{N}]
                  {\Psi}
                  {\Xi,\ctxvarty{\Gamma'}
                                {\emptyset}
                                {\ctxty{\mathcal{C}}{\cdot}}}
                  {\Omega}
                  {\subst{\{\Gamma'/\Gamma\}}{F}} 
                &
              \Gamma'\not\in\dom{\Xi}}

\\[10pt]

\infer[\ctxL]{\seq[\mathbb{N}]
                  {\Psi}
                  {\Xi}
                  {\setand{\Omega}{\fctx{\Gamma}{\mathcal{C}}{F_1}}}
                  {F_2}}
             {\ctxtyinst{\mathbb{N}}{\Psi}{\Xi}{\ctxty{\mathcal{C}}{\emptycb}}{G}
               &
              \seq[\mathbb{N}]
                  {\Psi}
                  {\Xi}
                  {\setand{\Omega}{\subst{\{G/\Gamma\}}{F_1}}}
                  {F_2}
             }
\end{tabular}
\end{center}
\caption{The Logical Rules}
\label{fig:rules-base}
\end{figure}

The considerations for the rules \allR\ and \existsL\ are similar and
we therefore discuss only the case for the former rule in detail.
In showing that the context variables context for the premise sequent
satisfies the wellformedness conditions, we observe that
the corresponding context for the conclusion sequent does by assumption
and we then use Theorem~\ref{th:ctx-ty-wk}.
In showing that the assumption formulas in the premise sequent satisfy
the required conditions, we again use the fact that these formulas in
the conclusion sequent do and then invoke Theorem~\ref{th:wfsupset}.
Thus, it only remains to show that $F'$ satisfies the conditions
required of it.
From the wellformedness of $\fall{x:\alpha}{F}$ relative to the
conclusion sequent, we see that 
$\wfform{\mathbb{N}\cup\STLCGamma_0\cup\Psi,x:\alpha}{\Xi^-}{F}$
has a derivation.
The substitution $\{\langle x,y\app n_1\ldots
n_m,\alpha\rangle\}$ is obviously type preserving with respect to the
arity typing context 
$\mathbb{N}\cup\STLCGamma_0\cup(\Psi,y:\arr{\alpha_1}{\arr{\ldots}{\arr{\alpha_m}{\alpha}}})$.
We may therefore invoke Theorem~\ref{th:subst-formula} to conclude that 
the judgement
$\wfform{\mathbb{N}\cup\STLCGamma_0\cup
             \Psi,y:\arr{\alpha_1}{\arr{\ldots}{\arr{\alpha_m}{\alpha}}}}
        {\Xi^-}
        {F'}$
has a derivation.

The considerations for the rules $\allL$ and $\existsR$ are also
similar and so we discuss only the case for $\allL$ in detail.
Most of the wellformedness requirements for the premise sequent in
fact follow trivially from the fact they hold for the conclusion
sequent.
The only perhaps nontrivial part is showing that $F'_1$ is
well-formed.
From the wellformedness of the conclusion sequent, we know that
$\wfform{\mathbb{N}\cup\STLCGamma_0\cup\Psi \cup \{x:\alpha\}}
        {\Xi^-}
        {F_1}$ 
has a derivation.
Since
$\stlctyjudg{\mathbb{N}\cup\STLCGamma_0\cup\Psi}
            {t}
            {\alpha}$
has a derivation, $\{\langle x, t, \alpha\rangle\}$ is an arity type
preserving substitution with respect to
$\mathbb{N}\cup\STLCGamma_0\cup\Psi$. 
We can now invoke Theorem~\ref{th:subst-formula} to conclude that
there must be a derivation for
$\wfform{\mathbb{N}\cup\STLCGamma_0\cup\Psi}
        {\Xi^-}
        {F'_1}$.

We consider the case for \ctxR\ next.
Here we must show that the addition
$\ctxvarty{\Gamma'}{\emptyset}{\ctxty{\mathcal{C}}{\cdot}}$
to the context variables context in the premise sequent satisfies the
wellformedness requirements.
That it does follows easily from the observations that
$\acstyping{\mathcal{C}}$ must be derivable because the formula
$\fctx{\Gamma}{\mathcal{C}}{F}$ is well-formed (relative to the 
conclusion sequent) and that 
$\wfctxvarty{\mathbb{N}}{\Psi}{\ctxty{\mathcal{C}}{\cdot}}$
has a trivial derivation.
The wellformedness of all the other associations in the context
variables context follows immediately from the wellformedness of the
conclusion sequent and the wellformedness of the formulas in $\Omega$
can be argued similarly, additionally using Theorem~\ref{th:wfsupset}.
Finally, using the fact that 
$\wfform{\mathbb{N}\cup\STLCGamma_0\cup\Psi}
        {\Xi^-}
        {\fctx{\Gamma}{\mathcal{C}}{F}}$
has a derivation, we can easily show that there must be a derivation
for 
$\wfform{\mathbb{N}\cup\STLCGamma_0\cup\Psi}
        {\Xi^-\cup\{\Gamma'\}}
        {\subst{\{\Gamma'/\Gamma\}}{F}}$.

The only remaining case is that of \ctxL.
Moreover, the only requirement whose validation is nontrivial in
this case is that there is a derivation for
$\wfform{\mathbb{N}\cup\STLCGamma\cup\Psi}{\Xi^-}{\subst{\{G/\Gamma\}}{F_1}}$.
There is a trivial derivation for
$\wfctxvarty{\mathbb{N}}{\Psi}{\ctxty{\mathcal{C}}{\cdot}}$ and, from
one of the premises, we know that there is a derivation for
$\ctxtyinst{\mathbb{N}}{\Psi}{\Xi}{\ctxty{\mathcal{C}}{\cdot}}{G}$.
But then, from Theorem~\ref{th:ctx-var-type-instance}, it follows that
$\wfctx{\mathbb{N}}{\Psi}{G}$ has a derivation.
By the well-formedness of the conclusion sequent, 
$\wfform{\mathbb{N}\cup\STLCGamma\cup\Psi}
        {\Xi^-}
        {\fctx{\Gamma}{\mathcal{C}}{F_1}}$
has a derivation, from which it follows that 
$\wfform{\mathbb{N}\cup\STLCGamma\cup\Psi}{\Xi^-\cup\{\Gamma\}}{F_1}$
must also have one.
We now use Theorem~\ref{th:subst-formula} to reach the desired conclusion.
\end{proof}

The proof of soundness of the rules $\ctxL$, $\allL$ and $\existsR$
will use inversions in the order of application of substitutions that
are justified by the three lemmas below.

\begin{lemma}\label{lem:form-hsubperm}
Let $\theta_1$ and $\theta_2$ be term variables substitutions that are arity
type preserving with respect to $\Theta$ and
$\aritysum{\context{\theta_1}}{\Theta}$,
respectively, and such that the variables in $\domain{\theta_2}$ are
distinct from those in $\domain{\theta_1}$ and do not appear free in
the terms in $\range{\theta_1}$. 
Further, let $\theta_2'$ be the substitution
\[\{\langle x,M',\alpha\rangle\ |\
         \langle x,M,\alpha\rangle\in\theta_2\mbox{ and }
         \hsub{\theta_1}{M}{M'}\}.\]
If $F$ is a formula such that
$\wfform{\aritysum{\context{\theta_2}}{(\aritysum{\context{\theta_1}}{\Theta}})}
        {\ctxsanstype{\Xi}}
        {F}$ 
is derivable for some context variables context $\Xi$, then
both $\hsubst{\theta_1}{\hsubst{\theta_2}{F}}$ and
$\hsubst{\theta_2'}{\hsubst{\theta_1}{F}}$ are defined and
are in fact the same formulas.
\end{lemma}

\begin{proof}
We note first that the substitution $\theta'_2$ is well-defined by
Theorem~\ref{th:aritysubs}.
The lemma is proved by an induction on the structure of formulas, 
making use eventually of Theorem~\ref{th:subspermute}.
\end{proof}

\begin{lemma}\label{lem:form-subperm}
Let $\sigma_1$ and $\sigma_2$ be context variables substitutions
such that the variables substituted for by $\sigma_2$ are distinct
from those substituted for by $\sigma_1$ and they also do not
appear in the context expressions in the range of $\sigma_1$.
Further, let $\sigma_2'$ be the context variables substitution
\[\{G'/\Gamma\ |\
    G/\Gamma\in\sigma_2\mbox{ and }\subst{\sigma_1}{G}=G'\}.\]
Then, for any formula $F$, $\subst{\sigma_1}{\subst{\sigma_2}{F}}$ and
$\subst{\sigma_2'}{\subst{\sigma_1}{F}}$ represent identical formulas.
\end{lemma}

\begin{proof}
Using an inductive argument, we first show the property for context
expressions and then extend it to formulas.
\end{proof}

\begin{lemma}\label{lem:form-perm}
Let $\sigma$ be a context variables substitution appropriate for $\Xi$
with respect to $\context{\theta}\cup\Psi$ and $\mathbb{N}$, let
$\theta$ be an arity type preserving substitution with respect to  
$\mathbb{N}\cup\STLCGamma_0\cup\Psi$, and let $F$ be a formula such
that
$\wfform{\mathbb{N}\cup\STLCGamma_0\cup
                    (\aritysum{\context{\theta}}{\Psi})}
        {\ctxsanstype{\Xi}}
        {F}$ is derivable.
Then, if $\sigma'$ be the substitution
$\{G'/\Gamma\ |\ G/\Gamma\in\sigma\mbox{ and }\hsub{\theta}{G}{G'}\}$,
$\hsubst{\theta}{\subst{\sigma}{F}}$ and
$\subst{\sigma'}{\hsubst{\theta}{F}}$ denote the same formulas.
\end{lemma}

\begin{proof}
Note first that for any $G/\Gamma \in \sigma$ it must be the case
that
$\wfctx{\mathbb{N}\cup\STLCGamma_0\cup(\aritysum{\context{\theta}}{\Psi})}
       {\ctxvarminus{\Xi}{\sigma}}
       {G}$
is derivable; this follows from the appropriateness property for
$\sigma$, using Theorem~\ref{th:ctx-var-type-instance}.
It is easily seen from this that $\hsubst{\theta}{G}$ must be defined
and, hence, that the substitution $\sigma'$ is well-defined. 
Using Theorem~\ref{th:subst-formula}, it also follows that 
$\wfform{\mathbb{N}\cup\STLCGamma_0\cup
                    (\aritysum{\context{\theta}}{\Psi})}
        {\ctxsanstype{\Xi}}
        {\subst{\sigma}{F}}$ is derivable.
Using Theorem~\ref{th:subst-formula} again, it follows that
$\hsubst{\theta}{\subst{\sigma}{F}}$ must be defined.
Finally, by Theorem~\ref{th:aritysubs}, $\hsubst{\theta}{F}$ must be
defined.
This establishes the coherence of the lemma statement.
Its proof itself is based on an induction on the structure of the
formula $F$, noting that term substitutions leave context variables
unchanged. 
\end{proof}

The soundness of the logical rules is the content of the following
theorem.

\begin{theorem}\label{th:core-sound}
The following property holds for every instance of each of the rules
in Figure~\ref{fig:rules-base}:
if the premises expressing typing judgements are derivable, the
conditions described in the other non-sequent premises are satisfied
and all the premise sequents are valid, 
then the conclusion sequent must also be valid. 
\end{theorem}

\begin{proof}
The argument is straightforward for the rules pertaining to the
logical constants and propositional connectives.
We therefore focus on the quantifier rules in the rest of the proof.
The wellformedness of sequents ensures that the applications of
substitutions to formulas and terms must be defined in actual
instances of these rules. 
In light of this, we will freely use the notation introduced after 
Theorems~\ref{th:aritysubs} in the discussions below. 

The arguments for the \allR\ and \existsL\ rules are similar, so we 
present only the argument for the \allR\ rule.
Suppose that the claim is false in this case.
Then there must be substitutions $\theta$ and $\sigma$ that identify a
closed instance of the conclusion sequent that is not valid.
We may assume without loss of generality that $y \notin
\domain{\theta}$; if this is not true at the outset, since $y
\notin \domain{\Psi}$, we may drop the tuple pertaining to $y$ from
$\theta$ to get another term substitution on which to base the
argument. 
Now, in the situation under consideration, all the formulas in
$\subst{\sigma}{\hsubst{\theta}{\Omega}}$ must be valid and
$\subst{\sigma}{\hsubst{\theta}{(\fall{x:\alpha}{F})}}$ must not be
valid.
The latter means that there is some term $t$ for which 
$\stlctyjudg{\noms\cup\STLCGamma_0}{t}{\alpha}$ has a derivation that
is such that 
$\hsubst{\{\langle x,t,\alpha\rangle\}}
        {\subst{\sigma}{\hsubst{\theta}{F}}}$
is not valid; note that $x$ is a variable that does not 
appear in $\theta$ or $\sigma$.
Using Lemma~\ref{lem:form-perm}, we see that this is the same formula
as
$\subst{\sigma}
       {\hsubst{\{\langle x,t,\alpha \rangle\}}
               {\hsubst{\theta}{F}}}$.
Using Theorem~\ref{th:raised-subs}, we can conclude that there must be
a term $t'$ that is devoid of nominal constants from $\mathbb{N}$ for
which the typing judgement 
$\stlctyjudg{\noms\cup\STLCGamma_0}
            {t'}
            {\arr{\arr{\alpha_1}{\arr{\ldots}{\alpha_m}}}{\alpha}}$
is derivable and that is such that 
$\subst{\sigma}
       {\hsubst{\{\langle x,t,\alpha \rangle\}}
               {\hsubst{\theta}{F}}}$
is the same formula as
$\hsubst{\{\langle y,t',\alpha\rangle\}}
        {\hsubst{\{\langle x, y\app n_1\ldots n_m,\alpha\rangle\}}
                {\hsubst{\theta}{F}}}$.
Since $x$ does not appear in $\theta$ and $y$ is not in
$\domain{\theta}$, we can conclude using Lemma~\ref{lem:form-hsubperm}
that the latter formula is the same as 
$\hsubst{\{\langle y,t',
             \arr{\alpha_1}{\arr{\ldots}{\arr{\alpha_m}{\alpha}}} \rangle\}}
        {\hsubst{\theta}
                {\hsubst{\{\langle x, y\app n_1\ldots n_m,\alpha\rangle\}}{F}}}$.
Using Theorem~\ref{th:composition}, we see further that this formula
is the same as
$\hsubst{\theta'}
        {\hsubst{\{\langle x, y\app n_1\ldots n_m,\alpha\rangle\}}{F}}$,
where $\theta'$ is the substitution 
$\theta \cup \{\langle y, t',
                 \arr{\alpha_1}{\arr{\ldots}{\arr{\alpha_m}{\alpha}}} \rangle\}$.
We observe at this point that $\theta'$ and $\sigma$ are substitutions
that identify a closed instance of the premise sequent:
since $t'$ is a closed term of the right type that is devoid of the
nominal constants in $\mathbb{N}$, it follows easily that $\theta'$
and $\sigma$ must be ``proper'' for the premise sequent if
$\theta$ and $\sigma$ are proper for the conclusion
sequent, the formulas in $\subst{\sigma}{\hsubst{\theta'}{\Omega}}$
are identical to the ones in $\subst{\sigma}{\hsubst{\theta}{\Omega}}$
and hence all closed, and the formula
$\subst{\sigma}
       {\hsubst{\theta'}
               {\hsubst{\{\langle x, y\app n_1\ldots n_m,\alpha\rangle\}}
                       {F}}}$
is identical to
$\hsubst{\{\langle x,t,\alpha\rangle\}}
        {\subst{\sigma}{\hsubst{\theta}{F}}}$
and hence also closed.
However, we now have a contradiction since this close instance of the
premise sequent must also be one that is not valid.

The arguments for the \allL\ and \existsR\ rules are also similar, so
we present only the one for the former rule.
Let $\theta$ and $\sigma$ be substitutions that identify a closed
instance of the conclusion sequent. 
Now, if any formula in
$\subst{\sigma}{\hsubst{\theta}{(\Omega \cup \fall{x:\alpha}{F_1})}}$
is not valid, then the conclusion sequent must be valid.
Let us therefore assume that all these formulas are valid.
Our task is to show that then
$\subst{\sigma}{\hsubst{\theta}{F_2}}$ must also be valid.
Clearly, $\theta$ and $\sigma$ also identify a closed instance of the
premise sequent.
If we can show that the assumption formula
$\subst{\sigma}
        {\hsubst{\theta}{\hsubst{\{\langle x,t,\alpha\rangle\}}{F_1}}}$ 
in this sequent is valid, then all its assumption formulas would be
valid, thereby enabling us to conclude that 
$\subst{\sigma}{\hsubst{\theta}{F_2}}$ is also valid, as desired.
Using Lemmas~\ref{lem:form-hsubperm} and \ref{lem:form-perm}, we see
that the assumption formula in question is the same as 
$\hsubst{\{\langle x,\hsubst{\theta}{t},\alpha\rangle\}}
        {\subst{\sigma}{\hsubst{\theta}{F_1}}}$.
Since $\stlctyjudg{\mathbb{N}\cup\STLCGamma_0\cup\Psi}{t}{\alpha}$ has
a derivation, by Theorem~\ref{th:aritysubs} there must be one for 
$\stlctyjudg{\mathbb{N}\cup\STLCGamma_0}{\hsubst{\theta}{t}}{\alpha}$.
But this means that
$\hsubst{\{\langle x,\hsubst{\theta}{t},\alpha\rangle\}}
        {\subst{\sigma}{\hsubst{\theta}{F_1}}}$
is a closed instance of
$\subst{\sigma}{\hsubst{\theta}{\fall{x:\alpha}{F_1}}}$ and must
therefore be valid by assumption.

Let us consider the case for \ctxR\ next.
If the claim is false, then there must be substitutions $\theta$ and
$\sigma$ that identify a closed instance of the conclusion sequent
that is not valid.
For this to happen, it must be the case that all the formulas in
$\subst{\sigma}{\hsubst{\theta}{\Omega}}$ are valid and the formula
$\subst{\sigma}{\hsubst{\theta}{\fctx{\Gamma}{\mathcal{C}}{F}}}$ is
not valid.
From the latter, it follows that for some context expression $G$ it is
the case that $\csinst{\noms}{\emptyset}{\mathcal{C}}{G}$ is derivable
and
$\subst{\{G/\Gamma\}}{\subst{\sigma}{\hsubst{\theta}{F}}}$
is not valid.
But now consider the substitutions $\theta$ and $\sigma'$ where
$\sigma' = \sigma \cup \{G/\Gamma'\}$.
It is easily checked that these are proper for the premise sequent.
Further, $\subst{\sigma'}{\hsubst{\theta}{F}}$ is the same formula as
$\subst{\{G/\Gamma\}}{\subst{\sigma}{\hsubst{\theta}{F}}}$ and, since
$\Gamma' \notin \domain{\Xi}$, $\subst{\sigma'}{\hsubst{\theta}{\Omega}}$
is the same collection of formulas as
$\subst{\sigma}{\hsubst{\theta}{\Omega}}$.
Thus $\theta$ and $\sigma'$ identify a closed instance of the premise
sequent that is not valid, thereby contradicting the assumption of
falsity of the theorem. 

The only remaining case is that of \ctxL.
Let $\theta$ and $\sigma$ be substitutions that identify a closed
instance of the conclusion sequent.
If any formula in 
$\subst{\sigma}
       {\hsubst{\theta}
               {(\Omega \cup \{\fctx{\Gamma}{\mathcal{C}}{F}\})}}$
is not valid, then the conclusion sequent must be valid.
Let us therefore assume that all these formulas are valid.
We must show that $\subst{\sigma}{\hsubst{\theta}{F_2}}$ is then also
valid. 
By assumption, there is a derivation for
$\ctxtyinst{\mathbb{N}}{\Psi}{\Xi}{\ctxty{\mathcal{C}}{\emptycb}}{G}$.
Using Theorem~\ref{th:ctx-ty-wk}, it follows that there must be a
derivation for 
$\ctxtyinst{\noms}{\Psi}{\Xi}{\ctxty{\mathcal{C}}{\emptycb}}{G}$.
Using Theorems~\ref{th:ctxtyinst-hsubst} and~\ref{th:ctxtyinst-subst}
we can show from this that 
$\ctxtyinst{\noms}
           {\emptyset}
           {\emptyset}
           {\ctxty{\mathcal{C}}{\emptycb}}
           {\subst{\sigma}{\hsubst{\theta}{G}}}$
is derivable; to invoke Theorem~\ref{th:ctxtyinst-subst}, we would
need $\sigma$ to be appropriate for $\hsubst{\theta}{\Xi}$ with respect
to $\emptyset$ and $\noms$ but this is ensured by the properness of
$\theta$ and $\sigma$ with respect to the conclusion sequent.
By Theorem~\ref{th:ctx-var-type-instance}, there must be a derivation for
$\csinst{\noms}{\emptyset}{\mathcal{C}}{\subst{\sigma}{\hsubst{\theta}{G}}}$.
Clearly, $\theta$ and $\sigma$ must identify a closed instance of the
premise sequent.
Moreover, all the assumption formulas in this instance must be valid:
we have assumed this to be the case for the formulas in
$\subst{\sigma}{\hsubst{\theta}{\Omega}}$ and
$\subst{\sigma}{\hsubst{\theta}{\subst{G/\Gamma}{F_1}}}$ must be valid
because 
$\subst{\sigma}{\hsubst{\theta}{\fctx{\Gamma}{\mathcal{C}}{F}}}$ is
valid and
$\csinst{\noms}{\emptyset}{\mathcal{C}}{\subst{\sigma}{\hsubst{\theta}{G}}}$
is derivable.
But then, by the assumption of validity of the premise sequent, it
must be the case that $\subst{\sigma}{\hsubst{\theta}{F_2}}$ is valid.
\end{proof}

\subsection{Proof Rules that Interpret Atomic Formulas}
\label{ssec:atomic}

We now present rules that build in an analysis of atomic formulas
based on the understanding that they represent LF typing judgements. 
When an atomic formula appears as the conclusion of a sequent, the
analysis takes an obvious form: the derivability of the sequent
can be based on that of a sequent in which the conclusion judgement
has been unfolded using an LF rule.
The treatment of an atomic assumption formula is more complex: we must
consider all the ways in which this formula could be valid in
assessing the validity of the sequent.
This ``case analysis'' can be driven by the structure of the type in
the formula.
When the type is of the form $\typedpi{x}{A}{B}$, the term must be an
abstraction and there is exactly one way in which a purported typing
derivation could have concluded.
When the type is atomic, Theorem~\ref{th:atomictype}
provides us information about the different cases that need to be
considered.
However, a complicating factor is that the formula may contain term
and context variables in it and we must consider also all the possible
instantiations of these variables that could made the judgement true
in the analysis. 

The analysis of atomic assumption formulas is obviously somewhat
intricate and we devote the first part of this subsection to its
discussion.
We then use the resulting understanding to develop proof rules for
atomic formulas. 

\subsubsection{Analyzing an Atomic Assumption Formula with an Atomic Type}
\label{ssec:cases}

Consider an atomic formula $\fatm{G}{R:P}$ where, obviously, $P$ has
the form  $(a\app M_1\app \ldots\app M_n)$.
Let us suppose initially that this formula is closed.
In this case, for the formula to be valid, $R$ would need to have as a
head a constant declared in $\Sigma$ or a nominal constant assigned a
type in $G$. 
If the arguments of $R$ do not satisfy the constraints imposed by the
type associated with the head, then the typing judgement will not be
derivable and hence we can conclude that the sequent is in fact
valid.
On the other hand, if the arguments of $R$ do satisfy the required
constraints, then Theorem~\ref{th:atomictype} gives us a means for
decomposing the given typing judgement into ones pertaining to
$M_1,\ldots, M_n$.
The validity of the given sequent can therefore be reduced to the
validity of a sequent that results from replacing the atomic formula
under consideration by ones that represent the mentioned typing
judgements.

In the general case, the formula $\fatm{G}{R:P}$ may not be closed.
This could happen in two different ways.
First, the context expression may have a part that is yet to be
determined, i.e., $G$ may be of the form
$\Gamma, n_1 : A_1,\ldots, n_m : A_m$ where $\Gamma$ has a set of
names $\mathbb{N}$ and a context
variable type of the form $\ctxty{\mathcal{C}}{G_1;\ldots;G_\ell}$
associated with it in the sequent.
Second, the expressions in the atomic formula and the context variable
type may contain variables in them that are bound in the term variables
context.
To articulate a proof rule around the atomic formula in this
situation, it is necessary to develop a means for analyzing the
formula in a way that pays attention to the validity of the sequent
under all acceptable instantiations of the context and term variables.

The analysis that we describe proceeds in two steps.
We first describe a finite way to consider elaborations of the context
variable that make explicit all the heads that need to be considered
for the term in an analysis of the closed instances of the atomic
formula. 
This process yields a finite collection of pairs
comprising a sequent in which the context variable may have been
partially instantiated, and a specifically identified possibility for
the head that is either drawn from the signature or that appears
explicitly in the context; the intent here, which is verified in
Lemma~\ref{lem:heads-cover}, is that considering just
the second components of these pairs as the heads of the term in the
typing judgement will suffice for a complete analysis based on
Theorem~\ref{th:atomictype}. 
The second step actually carries out the analysis in each of these
cases, using the idea of unification in the application of
Theorem~\ref{th:atomictype} to accommodate all possible closed
instantiations of the term variables in the sequent.

\paragraph{Elaborating Context Variables and Identifying Head Possibilities}

We first note that context expressions have may have implicit and
explicit parts, the former being subject to elaboration via context
substitutions. 

\begin{definition}\label{def:implicit-and-explicit}
Let $\mathcal{S} = \seq[\mathbb{N}]{\Psi}{\Xi}{\Omega}{F}$ be a
well-formed sequent.
If $G$ is a context expression appearing in $\mathcal{S}$ then
it must be of either the form $n_1 : A_1,\ldots, n_m:A_m$ or of the form
$\Gamma, n_1 : A_1,\ldots, n_m:A_m$ where $\Gamma$ is a context 
variable with an associated declaration 
$\ctxvarty{\Gamma}{\mathbb{N}_{\Gamma}}{\ctxty{\mathcal{C}}{G_1; \ldots; G_n}}$ 
in $\Xi$.
In the latter case, we say that $G$ has an implicit part relative to
$\mathcal{S}$ that is given by
$\ctxvarty{\Gamma}{\mathbb{N}_{\Gamma}}{\ctxty{\mathcal{C}}{G_1; \ldots; G_n}}$. 
Further, we refer to $n_1 : A_1,\ldots, n_m:A_m$ in the former case
and to the sequence formed by listing the bindings in $G_1,\ldots,G_n$
followed by $n_1 : A_1,\ldots, n_m:A_m$ in the latter case as the
explicit bindings in $G$ relative to $\mathcal{S}$.
\end{definition}

Let $\Gamma$ be a context variable that has the type
$\ctxty{\mathcal{C}}{G_1;\ldots;G_\ell}$.
Closed instances of $\Gamma$ are then generated by interspersing
$G_1,\ldots,G_\ell$ with blocks of declarations generated from the
block schema comprising $\mathcal{C}$.
In determining possibilities for the head of $R$ from the implicit
part of $G$ in an atomic formula of the form $\fatm{G}{R:P}$, we need
to consider an elaboration of $G$ with only one such block; of course,
for a complete analysis, we will need to consider all the
possibilities for such an elaboration.
The function $\addblksans$ defined below formalizes such an
elaboration, returning a modified sequent and a potential head for the
term in the typing judgement.
Note that in an elaboration based on a block schema 
of the form
$\{x_1:\alpha_1,\ldots,x_n:\alpha_n\}y_1 : A_1, \ldots, y_k : A_k$,
it would be necessary to consider a choice of nominal constants for
the schematic variables $y_1,\ldots,y_k$.
The function is parameterized by such a choice.
We must also accommodate all possible instantiations for the variables
$x_1,\ldots,x_n$, subject to the proviso that these instantiations do
not use nominal constants that appear in a later part of the context
expression.
This is done by introducing new term
variables for $x_1,\ldots,x_n$ and by raising such variables over the
nominal constants that are not prohibited from appearing in the
instantiations; to support the latter requirement, the function is
parameterized by a collection of nominal constants.  
Finally, we observe that the elaboration process may introduce new
nominal constants into the sequent, necessitating a raising of the
term variables over the new constants. 

\begin{definition}\label{def:addblock}
Let $\mathcal{S}$ be the the well-formed sequent 
$\seq[\mathbb{N}]{\Psi}{\Xi}{\Omega}{F}$, let
$\ctxvarty{\Gamma}
          {\mathbb{N}_{\Gamma}}
          {\ctxty{\mathcal{C}}{G_1;\ldots; G_n}}$
be an assignment in $\Xi$, and let $\mathcal{B} = \{x_1:\alpha_1,\ldots,
x_n:\alpha_n\}y_1 : A_1, \ldots, y_k : A_k$ be one of the block
schemas comprising $\mathcal{C}$.
Further, let $\mathbb{N'} \subseteq (\mathbb{N}\setminus\mathbb{N}_{\Gamma})$
be a collection of nominal
constants, let $ns$ be a list $n_1,\ldots,n_k$ of distinct nominal
constants that are also different from the constants in
$\mathbb{N'}$ and that are such that, for $1 \leq i \leq k$,
$n_i : \erase{A_i} \in (\noms\setminus\mathbb{N}_{\Gamma})$. 
Finally, for $0 \leq j \leq n$, let $\mathbb{N}_j$ be the collection
of nominal constants assigned types in $G_1,\ldots,G_j$.
Then, letting
\begin{enumerate}
  \item $\Psi'_j$ be a version of $\Psi$ raised over
    $\{n_1,\ldots,n_k\} \setminus \mathbb{N}$ and $\theta'_j$ be the
    associated raising substitution, 
  \item $A_1',\ldots,A_k'$ be the types $A_1,\ldots,A_k$ with the schematic variables
  $y_1,\ldots,y_k$ replaced with the names $n_1,\ldots,n_k$,
  \item  $\Psi''_j$ be a version of $\{x_1:\alpha_1,\ldots,x_n:\alpha_n\}$
    raised over $\mathbb{N'} \cup \mathbb{N}_j \cup
    (\{n_1,\ldots,n_k\} \setminus \mathbb{N})$ with the new variables 
    chosen to be distinct from those in $\Psi'_j$, $\theta''_j$ be the
    associated raising substitution, and $G$ be the context expression
    $n_1 : \hsubst{\theta''_j}{A'_1},\ldots, n_k : \hsubst{\theta''_j}{A'_k}$, and
  \item $\Xi'_j$ be the context variables context
   \[\hsubst{\theta'_j}
           {(\Xi \setminus \{\ctxvarty{\Gamma}{\mathbb{N}_{\Gamma}}{\ctxty{\mathcal{C}}{G_1; \ldots; G_n}}\})}
   \cup
   \{\ctxvarty{\Gamma}{\mathbb{N}_{\Gamma}}{\ctxty{\mathcal{C}}
                    {\hsubst{\theta'_j}{G_1};\ldots;\hsubst{\theta'_j}{G_j};
                     G;\hsubst{\theta'_j}{G_{j+1}};\ldots;\hsubst{\theta'_j}{G_n}}}\},\]
\end{enumerate}
for $0 \leq j \leq n$ and $1 \leq i \leq k$,
$\addblk{\mathcal{S}}{\ctxvarty{\Gamma}{\mathbb{N}_{\Gamma}}{\ctxty{\mathcal{C}}{G_1; \ldots; G_n}}}{\mathcal{B}}{ns}{\mathbb{N'}}{j}{i}$ is defined to be the tuple
\[\langle \seq[\mathbb{N} \cup ns]
             {\Psi'_j \cup \Psi''_j}
             {\Xi'_j}
             {\hsubst{\theta'_j}{\Omega}}
             {\hsubst{\theta'_j}{F}},
         n_i : \hsubst{\theta''_j}{A'_i}
\rangle.\]
Note that the conditions in the definition ensure that all the
substitutions involved in it will have a result, thereby permitting us
to use the notation introduced after Theorem~\ref{th:aritysubs}.           
\end{definition}

The elaboration just described is parameterized by the choice of
nominal constants for the variables assigned types in the block
schema.
In identifying the choices that have to be considered, it is useful to
partition the members of $(\noms\setminus\mathbb{N}_{\Gamma})$ 
into two sets: those that appear in
the support set of the sequent whose elaboration is being considered
and those that do not.
It is necessary to consider all possible assignments that satisfy
arity typing constraints from the first category.
From the second category, as we shall soon see, it suffices to
consider exactly one representative assignment.
Note also that we may insist that the nominal constant in each
assignment of the block be disinct; if this is not the case, the
sequent is easily seen to be valid.
The function $\namessans$ defined below embodies these ideas.
The function is parameterized be a sequence of arity types
corresponding to the declarations in the block schema, a collection of
``known'' nominal constants that are available for use in an
elaboration of the block schema and a collection of nominal constants
that are already bound in the context expressions and hence must not
be used again. 

\begin{definition}\label{def:nameslists}
Let $tys$ be a sequence of arity types and let $\mathbb{N}_o$ and
$\mathbb{N}_b$ be finite sets of nominal constants. 
Further, let $\emptyseq$ denote an empty sequence and 
$\consseq{x}{xs}$ denote a sequence that starts with $x$ and
continues with the sequence $xs$.
Then the collection of name choices for $tys$ relative to $\mathbb{N}_o$
and away from $\mathbb{N}_b$ is denoted by
$\names{tys}{\mathbb{N}_o}{\mathbb{N}_b}$ and defined by recursion
on $tys$ as follows:
\begin{center}
\begin{minipage}{5in}
\begin{tabbing}
\=\qquad\qquad\=\kill
\>
$\names{tys}{\mathbb{N}_o}{\mathbb{N}_b)} =$\\
\>\>
     $\begin{cases}
        \{\emptyseq\} & \mbox{\rm if}\ tys = \emptyseq\\[10pt]
        \{\consseq{n}{nl}\ |\ 
                  n : \alpha \in \noms,
                  n \in \mathbb{N}_o\setminus \mathbb{N}_b,
                  \ \mbox{\rm and}\\
          \qquad\qquad
               nl \in \names{tys'}{\mathbb{N}_o}{\mathbb{N}_b \cup
                 \{n\}}\}\ \cup\\                     
        \{\consseq{n}{nl}\ |\ 
               n\ \mbox{\rm is the first nominal
                            constant}& \mbox{\rm if}\ tys = \consseq{\alpha}{tys'} \\ 
          \qquad\qquad  \mbox{\rm such that}\ n : \alpha \in \noms\
                    \mbox{\rm and}\
                  n \not\in \mathbb{N}_o \cup \mathbb{N}_b,\\
          \qquad\qquad  \mbox{\rm and}\ nl \in
                    \names{tys'}{\mathbb{N}_o}{\mathbb{N}_b \cup \{n\}} \}
      \end{cases}$
\end{tabbing}
\end{minipage}
\end{center}
We assume in this definition the existence of an ordering on the
nominal constants that allows us to select the first of these
constants that satisfies a criterion of interest.
\end{definition}  

We can now identify a finite collection of elaborations of the
implicit part of a context expression that must be considered in the
analysis of an assumption formula of the form $\fatm{G}{R:P}$ that
appears in a sequent $\mathcal{S}$.
We do this below through the definition of the function
$\implheadssans$. 

\begin{definition}\label{def:implicitheads}
Let $\mathcal{S}$ be the well-formed sequent 
$\seq[\mathbb{N}]{\Psi}{\Xi}{\Omega}{F}$, let $G$ be a context expression appearing in a
formula in $\mathcal{S}$ that has an implicit part relative to
$\mathcal{S}$ that is given by 
$\ctxvarty{\Gamma}{\mathbb{N}_{\Gamma}}{\ctxty{\mathcal{C}}{G_1; \ldots; G_n}}$, 
and let
$\mathcal{B} = \{x_1:\alpha_1,\ldots, x_n:\alpha_n\}
                      y_1 : A_1, \ldots, y_k : \alpha_k$
be one of the block schemas comprising $\mathcal{C}$.
Further, let $\mathbb{N}_b$ be the collection of nominal constants
assigned types by the explicit bindings of $G$ relative to
$\mathcal{S}$ and let 
$\mathbb{N}_o = (\mathbb{N}\setminus\mathbb{N}_{\Gamma})\setminus\mathbb{N}_b$.
Finally, let $\allblks{\mathcal{S}}
              {\ctxvarty{\Gamma}{\mathbb{N}_{\Gamma}}{\ctxty{\mathcal{C}}{G_1; \ldots; G_n}}}
              {\mathcal{B}}$ denote the set
\begin{tabbing}
\qquad\=\qquad\qquad\=\kill
\> $\{ \addblk{\mathcal{S}}
             {\ctxvarty{\Gamma}{\mathbb{N}_{\Gamma}}{\ctxty{\mathcal{C}}{G_1; \ldots; G_n}}}
             {\mathcal{B}}
             {ns}
             {\mathbb{N}_o}
             {j}
             {i} \ \vert$ \\
\>\>  $0 \leq j \leq n,
       1 \leq i \leq k,
       ns \in \names{(\erase{A_1},\ldots,\erase{A_k})}
                    {\mathbb{N}_o}
                    {\mathbb{N}_{\Gamma}\cup\mathbb{N}_b}\}$.
\end{tabbing}
If $\{\mathcal{B}_1,\ldots,\mathcal{B}_m \}$ is the collection of 
of block schemas comprising $\mathcal{C}$, then the implicit 
heads in $G$ relative to $\mathcal{S}$ is defined to be the set
\[ \bigcup \{ \allblks{\mathcal{S}}
                      {\ctxvarty{\Gamma}{\mathbb{N}_{\Gamma}}{\ctxty{\mathcal{C}}{G_1; \ldots; G_n}}}
                      {\mathcal{B}} \ \vert\
                           \mathcal{B} \in
                           \{\mathcal{B}_1,\ldots,\mathcal{B}_m \}\}. \]
This set is denoted by $\implheads{\mathcal{S}}{G}$.
\end{definition}

The complete set of heads and corresponding (elaborated) sequents that
must be considered in the analysis of an atomic formula of the form
$\fatm{G}{R:P}$ is identified through the function \headssans\ that is
defined below.

\begin{definition}\label{def:heads}
Let $\mathcal{S}$ be a well-formed sequent and let $G$ be a context
expression appearing in a formula in $\mathcal{S}$.
Let $\mbox{\sl NewHds}$ be the set $\implheads{\mathcal{S}}{G}$ if $G$ has an
implicit part relative to $\mathcal{S}$ and the empty set otherwise.
Then the heads in $G$ relative to $\mathcal{S}$ is defined to be the
set
\begin{tabbing}
\ \=\kill
\> $\{ \langle \mathcal{S}, c : A\rangle\ |\ c : A \in \Sigma\} \cup
\{ \langle \mathcal{S}, n : A\rangle\ |\ n : A\ 
               \mbox{\rm is an explicit binding in G relative to}\ \mathcal{S} \} \cup \mbox{\sl
        NewHds}$.
\end{tabbing}
This set is denoted by $\hds{\mathcal{S}}{G}$.
\end{definition}

The first property that we observe of the elaboration process
described is that it requires us to consider only well-formed
sequents. 

\begin{lemma}
\label{lem:heads-wf}
Let $\mathcal{S}=\seq[\mathbb{N}]{\Psi}{\Xi}{\Omega}{F}$ be a well-formed sequent
and let $\fatm{G}{R:P}$ be an atomic formula in $\Omega$.
Then for each $(\mathcal{S}',h:A)\in \hds{\mathcal{S}}{G}$ it must be the
case that $\mathcal{S}'$ is a well-formed sequent. 
Further, if $\mathcal{S}'$ is
$\seq[\mathbb{N'}]{\Psi'}{\Xi'}{\Omega'}{F'}$, it must be the
case that $\wftype{(\mathbb{N}'\cup\STLCGamma_0\cup\Psi')}{A}$ is derivable.
\end{lemma}

\begin{proof}
The claim is not immediately obvious only 
when $(\mathcal{S}',h:A)\in \implheads{\mathcal{S}}{G}$.
For these cases, it suffices to show that every pair generated by
\addblksans\ satisfies the requirements of the lemma.
However, this is easily argued.
The main observation--that gets used twice---is that if $\Psi_2$ is a
version of $\Psi_1$ raised over some collection of nominal constants
$\mathbb{N}_2$ with $\theta$ being the associated raising
substitution, and 
$\wftype{\STLCGamma \cup \Psi_1}{A'}$ holds for some arity context
$\STLCGamma$ that is disjoint from $\Psi_1$ and $\Psi_2$, then 
$\wftype{\STLCGamma \cup \Psi_2 \cup
  \mathbb{N}_2}{\hsubst{\theta}{A'}}$ also holds.
\end{proof}

We want next to show that the collection of pairs of sequents and
heads identified by the elaboration process are
sufficient for the analysis of validity for a sequent with an
assumption formula of the form $\fatm{G}{R:P}$.
The following lemma provides the basis for this observation.
We obviously do to identify all possible elaborations because we consider
only a single representative for a ``new name'' for a binding in a
block instance.
However, we show that we cover all elaborations up to a permutation;
Theorem~\ref{th:perm-valid} then ensures that this is sufficient. 

\begin{lemma}
\label{lem:heads-cover}
Let $\mathcal{S}=\seq[\mathbb{N}]{\Psi}{\Xi}{\Omega}{F}$ be a
well-formed sequent and let $\fatm{G}{R:P}$ be a formula in $\Omega$.
Further, let $\theta$ and $\sigma$ be term and context
variables substitutions that identify a closed instance of $\mathcal{S}$ and
that are such that $\subst{\sigma}{\hsubst{\theta}{\fatm{G}{R:P}}}$ is
valid.
If the term $\hsubst{\theta}{R}=(h\app M_1\ldots M_n)$, then 
there is a pair $\langle \mathcal{S}', h':A' \rangle$ in $\hds{\mathcal{S}}{G}$
such that
\begin{enumerate}
\item there is a formula $\fatm{G'}{R':P'}$ amongst the assumption
formulas of $\mathcal{S}'$ with $h':A'$ appearing in either $\Sigma$
or in the explicit bindings in $G'$ relative to $\mathcal{S'}$, and 

\item there is a closed instance of $\mathcal{S}'$ identified by 
  term and context variables substitutions $\theta'$ and $\sigma'$ and a
  permutation $\pi$ such that $\permute{\pi}{h'} = h$,
$\permute{\pi}{\subst{\sigma'}{\hsubst{\theta'}{\fatm{G'}{R':P'}}}}=
  \subst{\sigma}{\hsubst{\theta}{\fatm{G}{R:P}}}$, and
$\permute{\pi}{\subst{\sigma'}{\hsubstseq{\emptyset}{\theta'}{\mathcal{S}'}}}= 
    \subst{\sigma}{\hsubstseq{\emptyset}{\theta}{\mathcal{S}}}$.  
\end{enumerate}
\end{lemma}

\begin{proof}
Since $\subst{\sigma}{\hsubst{\theta}{\fatm{G}{R:P}}}$ is valid, it
must be the case that there are LF derivations for 
$\lfctx{\subst{\sigma}{\hsubst{\theta}{G}}}$,
$\lftype{\subst{\sigma}{\hsubst{\theta}{G}}}
        {\hsubst{\theta}{P}}$, and
$\lfchecktype{\subst{\sigma}{\hsubst{\theta}{G}}}
             {\hsubst{\theta}{R}}
             {\hsubst{\theta}{P}}$.
Using Theorem~\ref{th:atomictype} together with the fact that
$\hsubst{\theta}{R}=(h\app M_1\ldots M_n)$, we see that, for an
appropriate $A$, $h: A$ must be a member of $\Sigma$ or it must appear
in $\subst{\sigma}{\hsubst{\theta}{G}}$.
Our argument distinguishes two ways that this could happen: it could
be because $h:A$ is a member of $\Sigma$ or it is an instance of a
declaration in the explicit part of $G$ or because it is introduced
into the context $\subst{\sigma}{\hsubst{\theta}{G}}$ by the
substitution $\sigma$.

The first collection of cases is easily dealt with: essentially, we
pick $h'$, $\mathcal{S}'$, $\theta'$ and $\sigma'$ to be identical to
$h$, $\mathcal{S}$, $\theta$ and $\sigma$, respectively, and we let
$\pi$ be the identity permutation.
The requirements of the lemma then follow easily from the definition
of the \headssans\ function.

In the cases that remain, $G$ must have the form $\Gamma, n^G_1
: A^G_1,\ldots, n^G_p : A^G_p$ for some context variable $\Gamma$ that
has the set of names $\mathbb{N}_{\Gamma}$ and the type 
$\ctxty{\mathcal{C}}{G_1;\ldots; G_\ell}$ assigned to it
in $\Xi$ and $h$ must be introduced by the substitution that $\sigma$
makes for $\Gamma$ as the $i^{th}$ binding, for
some $i$, in a block of declarations resulting from instantiating
one of the block schemas constituting $\mathcal{C}$. 
Let us suppose  the relevant block schema is $\mathcal{B}$ and
it has the form $\{x_1:\alpha_1,\ldots,x_n:\alpha_n\}
(y_1:B_1,\ldots, y_k:B_k)$.
Moreover, let us suppose that this block of declarations appears in
$\subst{\sigma}{\hsubst{\theta}{G}}$ somewhere 
between the instances of $G_j$ and $G_{j+1}$, for some $j$ between
$0$ and $\ell$.
We may, without loss of generality, assume $x_1,\ldots,x_n$ to be
distinct from the variables assigned types by $\Psi$.
We can then visualize the block introducing $h$ as
$(n_1:\hsubst{\theta^h}{B_1'},\ldots,n_k:\hsubst{\theta^h}{B_k'})$ for some
$n_1,\ldots,n_k$ of the requisite types, for some types $B_1',\ldots,B_k'$ 
which are the types $B_1,\ldots,B_k$ with the schematic variables of the 
schema replaced by these names, and for a closed substitution
$\theta^h$ whose domain is $x_1,\ldots,x_n$ and, since
$\lfctx{\subst{\sigma}{\hsubst{\theta}{G}}}$ is derivable, whose
support does not contain the nominal constants in $\mathbb{N}_{\Gamma}$, 
$n^G_1,\ldots,n^G_p$, or those that are assigned a type in
$G_{j+1},\ldots,G_\ell$.
It follows from this that if we can associate the type
$\ctxty{\mathcal{C}}
       {\hsubst{\theta}{G_1};\ldots;\hsubst{\theta}{G_j};
        n_1:\hsubst{\theta^h}{B_1'},\ldots,n_k:\hsubst{\theta^h}{B_k'};
        \hsubst{\theta}{G_{j+1}};\ldots; \hsubst{\theta}{G_\ell}}$
with $\Gamma$, then the context expression that $\sigma$ substitutes
for $\Gamma$ can still be generated from the changed type.
The key to our showing that the requirements of the lemma are met in
these cases will be to establish that $\hds{\mathcal{S}}{G}$
contains a sequent and head pair such that the type of $\Gamma$ is
elaborated to a form from which the above type can be obtained, up to
a permutation of nominal constants, by a well-behaved substitution and
the head is identified as the $i^{th}$ item in the introduced block of
declarations.  

Towards this end, let us consider the tuple $\langle \mathcal{S''},
h'': A'' \rangle$ 
that is generated by
\[\addblk{\mathcal{S}}{\ctxvarty{\Gamma}{\mathbb{N}_{\Gamma}}{\ctxty{\mathcal{C}}{G_1; \ldots; G_\ell}}}
  {\mathcal{B}}{(n_1,\ldots,n_k)}{\mathbb{N}_o}{j}{i},\] 
where $\mathbb{N}_o$ is the collection of nominal constants obtained
by leaving out of $\mathbb{N}$ the constants in $\mathbb{N}_{\Gamma}$ and 
the constants that appear amongst the
explicit bindings of $G$ relative to $\mathcal{S}$.
In this case, $\mathcal{S''}$ will have the form
$\seq[\mathbb{N}'']{\Psi''}{\Xi''}{{\Omega''}}{F''}$ with the following
properties.
First, $\mathbb{N}''$ will be identical to
$\mathbb{N} \cup \{n_1,\ldots,n_k \}$. 
Second, $\Psi''$ will comprise two disjoint parts $\Psi^\mathcal{S}_r$
and $\Psi^\mathcal{B}_r$, where $\Psi^\mathcal{S}_r$ is a 
version of $\Psi$ raised over the nominal constants in
$\{n_1,\ldots,n_k\}$ that are not members of $\mathbb{N}$ with a
corresponding raising substitution $\theta^\Psi_r$, and
$\Psi^\mathcal{B}_r$ is a version of $\{x_1,\ldots,x_n\}$ raised over
all the nominal constants in $\mathbb{N} \cup \{n_1,\ldots,n_k \}$
except the ones that are assigned a type in $G_{j+1},\ldots,G_\ell$ or
that appear in $n^G_1,\ldots,n^G_p$ with the corresponding raising
substitution $\theta^{\mathcal{B}}_r$.
Third, $\Xi''$ will be
\[  \hat{\Xi} \cup 
   \{ \ctxvarty{\Gamma}
               {\mathbb{N}_{\Gamma}}
               {\ctxty{\mathcal{C}}
                     {\hsubst{\theta^\Psi_r}{G_1};\ldots,\hsubst{\theta^\Psi_r}{G_j};
                      n_1 : \hsubst{\theta^\mathcal{B}_r}{B_1'},
                      \ldots, n_k :\hsubst{\theta^\mathcal{B}_r}{B_k'};
                     \hsubst{\theta^\Psi_r}{G_{j+1}},\ldots\hsubst{\theta^\Psi_r}{G_\ell}}}
   \}
\]
where $\hat{\Xi} =
       \hsubst{\theta^\Psi_r}{(\Xi \setminus \{\ctxvarty{\Gamma}{\mathbb{N}_{\Gamma}}{\ctxty{\mathcal{C}}{G_1; \ldots; G_\ell}}\})}$.
Finally, each formula in $\Omega'' \cup \{F''\}$ is obtained by
applying the raising substitution $\theta^\Psi_r$ to a corresponding
one in $\Omega \cup \{F\}$. 
Using Theorem~\ref{th:raised-subs} we observe that, because
$\supportof{\theta}$ is disjoint from the set $\mathbb{N}$, there is a 
(closed) raising substitution $\theta_r$  with $\context{\theta_r} =
\Psi^\mathcal{S}_r$ whose support is
disjoint from the set $\mathbb{N} \cup \{n_1,\ldots,n_k\}$ and which is such
that $\hsubst{\theta_r}{\Omega''} = \hsubst{\theta}{\Omega}$,
$\hsubst{\theta_r}{F''} = \hsubst{\theta}{F}$,
$\hsubst{\theta_r}{\hat{\Xi}}$ is equal to
$\hsubst{\theta}{(\Xi \setminus \{\ctxvarty{\Gamma}{\mathbb{N}_{\Gamma}}{\ctxty{\mathcal{C}}{G_1; \ldots; G_\ell}}\})}$,
and, for each $q$, $1 \leq q \leq \ell$,
$\hsubst{\theta_r}{\hsubst{\theta^\Psi_r}{G_q}}=\hsubst{\theta}{G_q}$.
Using Theorem~\ref{th:raised-subs} again, we see that there is a
(closed) substitution $\theta^h_r$ with
$\context{\theta^h_r} = \Psi^\mathcal{B}_r$ whose support is
disjoint from $\mathbb{N} \cup \{n_1,\ldots,n_k\}$ and that is such
that, for $1 \leq q \leq k$, it is the case that
$\hsubst{\theta^h_r}{\hsubst{\theta^\mathcal{B}_q}{B_q'}} =
\hsubst{\theta^h}{B_q'}$.  
Based on all these observations, it is easy to see that if we let
$\theta''= \theta_r \cup \theta^h_r$, then $\langle
\theta'',\emptyset\rangle$ is substitution compatible with
$\mathcal{S}''$ and $\hsubstseq{\emptyset}{\theta''}{\mathcal{S}''}$ is
identical to $\hsubstseq{\emptyset}{\theta}{\mathcal{S}}$ except for
the fact that the type associated with $\Gamma$ in its context
variables context is
$\ctxty{\mathcal{C}}
       {\hsubst{\theta}{G_1};\ldots;\hsubst{\theta}{G_j};
        n_1:\hsubst{\theta^h}{B_1'},\ldots,n_k:\hsubst{\theta^h}{B_k'};
        \hsubst{\theta}{G_{j+1}};\ldots; \hsubst{\theta}{G_\ell}}$.
By the earlier observation, $\sigma$ is appropriate for 
$\hsubstseq{\emptyset}{\theta''}{\mathcal{S}''}$ and, in fact
$\subst{\sigma}{\hsubstseq{\emptyset}{\theta''}{\mathcal{S}''}} =
 \subst{\sigma}{\hsubstseq{\emptyset}{\theta}{\mathcal{S}}}$.
 Noting also that $h'':A''$ must, by the definition of \addblksans, be
 $n_i:\hsubst{\theta^\mathcal{B}_r}{B_i'}$, if 
 $\hds{\mathcal{S}}{G}$ includes in it a pair obtained by this
 particular call to \addblksans, then we can pick $\mathcal{S}'$ to be
 $\mathcal{S}''$, $h'$ to be $h''$, $A'$ to be $A''$, $\theta'$ to be
 $\theta''$, $\sigma'$ to be $\sigma$ and $\pi$ to be the identity
 permutation to satisfy the requirements of the lemma.

We are, of course, not assured that there will be a pair in
$\hds{\mathcal{S}}{G}$ corresponding to the use of \addblksans\ with 
exactly the arguments considered above.
Specifically, the sequences of nominal constants that are considered
for the block instance may not include $n_1,\ldots, n_k$.
However, we know that some sequence $n'_1,\ldots, n'_k$
will be considered that is identical to $n_1,\ldots,n_k$ except for
constants in identical locations in the two sequences that are not
drawn from $\mathbb{N}$.
Since the constants in any sequence must be distinct, it follows
easily that we can describe a permutation $\pi'$ on the nominal
constants that is the identity map on $\mathbb{N}$ and that maps
$n'_1,\ldots,n'_k$ to $n_1,\ldots,n_k$. 
It can also be seen then that $\hds{\mathcal{S}}{G}$ will include a
tuple $\langle \mathcal{S}''', h''' : A''' \rangle$ such that
$\permute{\pi'}{\mathcal{S}'''} = \mathcal{S}''$, $\permute{\pi'}{h'''} = h''$, and
$\permute{\pi'}{A'''} = A''$.
Picking $\mathcal{S}'$ to be $\mathcal{S}'''$, $h'$ to be $h'''$, $A'$ to be $A'''$,
$\theta'$ to be $\permute{\inv{\pi'}}{\theta''}$, $\sigma'$ to be
$\permute{\inv{\pi'}}{\sigma''}$, $\pi$ to be $\pi'$ and using
Theorem~\ref{th:perm-subst}, we can once
again see that the requirements of the lemma are met.
\end{proof}

\paragraph{Generating a Covering Set of Premise Sequents} 

We have, at this stage, a way for identifying all the possible heads
to consider for $R$ in analyzing an atomic assumption formula
$\fatm{G}{R:P}$ in a sequent $\mathcal{S}$
However, we are still need a systematic approach for considering all
the term and context substitutions that yield closed instances of
$\mathcal{S}$ in which the term component of $F$ has the relevant
head. 
We now turn to this task.
Rather than identifying the closed instances immediately, we will
think of taking a step in this direction that also allows us to reduce
the typing judgement represented by $F$ based on the typing rule for
the LF judgement it represents. 
The first step in this direction will be to determine a substitution
that makes the head of $R$ identical to the one it needs to be in its
closed form.
We introduce a particular form of unification towards this end. 
The next step will be to reduce the sequent based on the observations
in Theorem~\ref{th:atomictype}.
The eventual proof rule will then combine the identification of
relevant heads using $\hds{\mathcal{S}}{G}$, the solving of a
unification problem based on each such head, and the 
reduction of the sequent.

We begin this development by describing the relevant notion of
unification.

\begin{definition}\label{def:unification}
A \emph{unification problem} $\mathcal{U}$ is a tuple
$\unif{\mathbb{N}}{\Psi}{\mathcal{E}}$
in which $\mathbb{N}$ is a collection of nominal constants, $\Psi$ is
a term variables context, and $\mathcal{E}$ is a set of equations
$\{\eqn{E_1}{E_1'},\ldots,\eqn{E_n}{E_n'}\}$ where, for
$1 \leq i \leq n$, either
$\wftype{\mathbb{N}\cup\Psi\cup\STLCGamma_0}{E_i}$ and 
$\wftype{\mathbb{N}\cup\Psi\cup\STLCGamma_0}{E'_i}$ have derivations or
there is an arity type $\alpha$ such that 
$\stlctyjudg{\mathbb{N}\cup\Psi\cup\STLCGamma_0}{E_i}{\alpha}$ and 
$\stlctyjudg{\mathbb{N}\cup\Psi\cup\STLCGamma_0}{E_i'}{\alpha}$ 
have derivations.
A \emph{solution} to such a unification problem is a pair
$\solun{\theta}{\Psi'}$ of a substitution and a term variables context
such that
\begin{enumerate}
\item $\theta$ is type preserving with respect to $\noms\cup\STLCGamma_0\cup\Psi'$,

\item $\supportof{\theta} \cap \mathbb{N} = \emptyset$,

\item for any $x$ if
$x : \alpha \in \Psi$ and $x : \alpha' \in \aritysum{\context{\theta}}{\Psi'}$ then
  $\alpha = \alpha'$, and

\item for each $i$, $1\leq i\leq n$, expressions $E_i$ and $E_i'$ from 
$\eqn{E_i}{E_i'}$ are such that $\hsubst{\theta}{E_i}=\hsubst{\theta}{E_i'}$.
\end{enumerate}
Note that the typing contraints validate the use of the notation
$\hsubst{\theta}{E_i}$ and $\hsubst{\theta}{E'_i}$.
\end{definition}

Given an atomic term $R$, we can determine its instances that have a
particular head $h$ through the unification of $R$ with $h$ applied to
a sequence of fresh variables.
We will use this idea to narrow down the set of instances of a sequent
that must be considered once we have determined what the head of the
term in an atomic goal of the form $\fatm{G}{R:P}$ must be.
However, we must first build into our notion of a fresh variable the
ability to instantiate it with nominal constants appearing in the
sequent.
We do this below by using the mechanism of raising.

\begin{definition}\label{def:generalized-variable}
Let $\Psi$ be a term variables context and let $\mathbb{N}$ be a finite
subset of $\noms$. 
Further, let $n_1,\ldots,n_k$ be a listing of the nominal constants in
$\mathbb{N}$ and let $\alpha_1,\ldots,\alpha_k$ be the respective
types of these constants.
Then, for any variable $z$ that does not appear in $\Psi$, $z :
\alpha_1 \atyarr \cdots \atyarr \alpha_k \atyarr \beta$ is 
said to be a variable of arity type $\beta$ away from $\Psi$ and
raised over $\mathbb{N}$.
Moreover, $(z \app n_1\app \ldots \app n_k)$ is said to be the
generalized variable term corresponding to $z$.
\end{definition}

The following lemma now formalizes the described refinement of the sequent.

\begin{lemma}\label{lem:inst-solun}
Let $\mathcal{S} = \seq[\mathbb{N}]{\Psi}{\Xi}{\setand{\Omega}{\fatm{G}{R : P}}}{F'}$ be a
well-formed sequent.
Further, let $\theta$ be a term substitution that together with a
context substitution $\sigma$ identifies a closed instance of
$\mathcal{S}$ and is such that
$\hsubst{\theta}{R} = (h\app M_1\app \ldots\app M_n)$ and
$\hsubst{\theta}{P} = \hsubst{\{\langle x_1,M_1,\erase{A_1}\rangle,
  \ldots, \langle x_n,M_n,\erase{A_n}\rangle \}}{P'}$
for a head $h$ that $\Sigma$ or
$\subst{\sigma}{\hsubst{\theta}{G}}$ assigns the
type $\typedpi{x_1}{A_1}{\ldots\typedpi{x_n}{A_n}{P'}}$.
Finally for $1 \leq i \leq n$ let $z_i:\alpha_i'$ be a distinct
variable of type $\erase{A_i}$ away from $\Psi$ and raised over
$\mathbb{N}$, and let $t_i$ be the generalized variable term
corresponding to $z_i$. 
Then $\solun{\theta}{\emptyset}$ is a
solution to the unification problem
\begin{tabbing}
\qquad\quad\=\quad\qquad\qquad\=\kill
\>
$\langle \mathbb{N},
         \Psi \cup \{z_1:\alpha_1',\ldots, z_n:\alpha_n'\},$\\
\>\>
$\left\{P = \hsubst{\{\langle x_1,t_1,\erase{A_1}\rangle,
            \ldots, \langle x_n,t_n,\erase{A_n}\rangle \}}{P'}, 
         R = (h \app t_1 \app \ldots\app t_n)\right\}\rangle$.
\end{tabbing}
\end{lemma}
\begin{proof}
As $\theta$ and $\sigma$ identify a closed instance of $\mathcal{S}$
is must be that $\seqsub{\theta}{\emptyset}$ is substitution compatible with
$\mathcal{S}$.
But then, $\seqsub{\theta}{\emptyset}$ satisfies the first three
clauses of the definition for a solution. 
The assumptions $\hsubst{\theta}{R} = (h\app M_1\app \ldots\app M_n)$ and
$\hsubst{\theta}{P} = \hsubst{\{\langle x_1,M_1,\erase{A_1}\rangle,
  \ldots, \langle x_n,M_n,\erase{A_n}\rangle \}}{P'}$ 
ensure that the final clause is also satisfied, implying that 
$\solun{\theta}{\emptyset}$ solves the given unification problem.
\end{proof}

The reduction of a sequent based on an atomic assumption formula is 
the content of the next definition.
Note that the type and the term in the formula are required to be
atomic for this reduction to be applicable and the type of the head of
the term determines the reduction.
\begin{definition}\label{def:decompseq}
Let
$\mathcal{S}=\seq[\mathbb{N}]{\Psi}{\Xi}{\setand{\Omega}{F}}{F'}$ be a
well-formed sequent, where $F$ is the formula $\fatm{G}{h\app
  M_1\ldots M_n:P}$ for some $h$ that is assigned the type 
$A=\typedpi{x_1}{A_1}{\ldots\typedpi{x_n}{A_n}{P'}}$ in $\Sigma$ or
the explicit bindings in $G$ relative to $\mathcal{S}$.
The sequent obtained by decomposing the assumption formula $F$ based on the
type $A$ is 
$\seq[\mathbb{N}]
     {\Psi}
     {\Xi}
     {\setand{\Omega}{\fatm{G}{M_1:A_1'},\ldots,\fatm{G}{M_n:A_n'}}}
     {F'}$
where, for $1 \leq i \leq n$, $A_i'$ is
$\hsubst{\{\langle x_1,M_1,\erase{A_1}\rangle,\ldots,
               \langle x_{i-1}, M_{i-1},\erase{A_{i-1}}\rangle\}}{A_i}$.
This sequent is denoted by $\decompseq{F}{\mathcal{S}}$.
Note that the well-formedness of $\mathcal{S}$ justifies the use of
the notation
$\hsubst{\{\langle x_1,M_1,\erase{A_1}\rangle,\ldots,
           \langle x_{i-1}, M_{i-1},\erase{A_{i-1}}\rangle\}}
        {A_i}$.
\end{definition}

The following lemma expresses the soundness of the idea of reducing
a sequent. Additionally, it identifies a measure with atomic
formulas that diminishes with the replacements effected by
a reduction step; this property will be useful in showing soundness
for an induction rule that we present in Section~\ref{ssec:induction}.

\begin{lemma}\label{lem:decomp-decr}
Let $\mathcal{S} = \seq[\mathbb{N}]{\Psi}{\Xi}{\setand{\Omega}{F}}{F'}$
be a well-formed sequent with $F$ an atomic formula of the form
$\fatm{G}{R:P}$.
Further, let $h$ be assigned the type
$\typedpi{x_1}{A_1}{\ldots\typedpi{x_n}{A_n}{P'}}$ in either $\Sigma$
or the explicit bindings in $G$ relative to $\mathcal{S}$ and 
for each $i$, $1 \leq i \leq n$, let $z_i:\alpha_i'$ be a distinct
variable of type $\erase{A_i}$ away from $\Psi$ and raised over
$\mathbb{N}$, and let $t_i$ be the generalized variable term
corresponding to $z_i$. 
Finally, let $\mathcal{U}$ be the unification problem 
\begin{tabbing}
\qquad\quad\=\quad\qquad\qquad\=\kill
\>
$\langle \mathbb{N}, 
         \Psi \cup \{z_1:\alpha_1',\ldots, z_n:\alpha_n'\},$\\
\>\>
$\left\{P = \hsubst{\{\langle x_1,t_1,\erase{A_1}\rangle,
            \ldots, \langle x_n,t_n,\erase{A_n}\rangle \}}{P'}, 
         R = (h \app t_1 \app \ldots\app t_n)\right\}\rangle$.
\end{tabbing}
Then any solution to $\mathcal{U}$ is substitution compatible with
$\mathcal{S}$.
Further, for any $\solun{\theta}{\Psi_\theta}$ that is a solution to
$\mathcal{U}$ and $\theta_r$ that is a raising substitution associated
with the application of $\theta$ to $\mathcal{S}$ relative to
$\Psi_\theta$, there must be terms $M_1, \ldots, M_n$ such that the
following hold:
\begin{enumerate}
\item $\hsubst{\theta_r}{\hsubst{\theta}{R}}$ is $(h\app M_1 \app \ldots \app M_n)$.
\item For any $\theta'$ and $\sigma'$ identifying a closed instance
  of $\hsubstseq{\Psi_\theta}{\theta}{\mathcal{S}}$, if 
$\subst{\sigma'}{\hsubst{\theta'}{\hsubst{\theta_r}{\hsubst{\theta}{\fatm{G}{R:P}}}}}$
is valid and there is a derivation for 
$\subst{\sigma'}{\hsubst{\theta'}{\hsubst{\theta_r}{\hsubst{\theta}{(\lfchecktype{G}{R}{P})}}}}$
of height $k$, then for each $i$, $1 \leq i \leq n$, letting
$A'_i = \hsubst{\{\langle x_1,M_1,\erase{A_1}\rangle,\ldots,
                                       \langle x_{i-1}, M_{i-1},\erase{A_{i-1}}\rangle\}}
               {A_i}$, it must be
the case that 
$\subst{\sigma'}{\hsubst{\theta'}
                       {(\fatm{\hsubst{\theta_r}{\hsubst{\theta}{G}}}
                              {M_i: A'_i})}}$  
is valid and that there is a derivation of height less than $k$ for 
$\subst{\sigma'}{\hsubst{\theta'}{(\lfchecktype{\hsubst{\theta_r}{\hsubst{\theta}{G}}}
                                {M_i}{A'_i})}}.$
\item There is an $\mathcal{S'}$ such that 
$\mathcal{S}' = 
    \decompseq{\hsubst{\theta_r}{\hsubst{\theta}{F}}}
              {\hsubstseq{\Psi_\theta}{\theta}{\mathcal{S}}}$
and $\mathcal{S}'$ is valid only if
$\hsubstseq{\Psi_\theta}{\theta}{\mathcal{S}}$  is. 
\end{enumerate}
\end{lemma}
\begin{proof}
A straightforward examination of Definitions~\ref{def:seq-term-subst}
and~\ref{def:unification} suffices to verify that solutions to
$\mathcal{U}$ must be substitution compatible with $\mathcal{S}$.
Any solution $\solun{\theta}{\Psi_{\theta}}$ to the unification
problem $\mathcal{U}$ must be such that
$\hsubst{\theta}{R}=\hsubst{\theta}{(h\app t_1\ldots t_n)}$.
From this it follows that
$\hsubst{\theta_r}{\hsubst{\theta}{R}}=
\hsubst{\theta_r}{\hsubst{\theta}{(h\app t_1\ldots t_n)}}$.
Since $h$ is unaffected by substitutions, it is easy to see that
$\hsubst{\theta_r}{\hsubst{\theta}{(h\app t_1\ldots t_n)}} = (h \app
(\hsubst{\theta_r}{\hsubst{\theta}{t_1}})\app \ldots \app
  (\hsubst{\theta_r}{\hsubst{\theta}{t_n}}))$.
Picking  $M_i$ to be the term $\hsubst{\theta_r}{\hsubst{\theta}{t_i}}$
for each $i$, $1 \leq i \leq n$, we see that clause (1) in the lemma
is satisfied.

For the second clause we note first that the typing judgements in question must 
be closed and hence the consideration is meaningful.
Consider an arbitrary closed instance of $\hsubstseq{\Psi_{\theta}}{\theta}{\mathcal{S}}$
identified by $\theta'$ and $\sigma$.
If $\subst{\sigma}{\hsubst{\theta'}{\hsubst{\theta_r}{\hsubst{\theta}{F}}}}$
is a valid formula then using the definition of validity as well as
clause (1) in the lemma we can extract a derivation for
$\lfctx{\subst{\sigma}{\hsubst{\theta'}{\hsubst{\theta_r}{\hsubst{\theta}{G}}}}}$
and a derivation of height $k$ for
$\lfchecktype{\subst{\sigma}{\hsubst{\theta'}{\hsubst{\theta_r}{\hsubst{\theta}{G}}}}}
             {\hsubst{\theta'}{(h\app M_1\ldots M_n)}}
             {\hsubst{\theta'}{\hsubst{\theta_r}{\hsubst{\theta}{P}}}}$.
Application of Theorems~\ref{th:atomictype} and~\ref{th:subspermute} 
are then sufficient to conclude
that there is a derivation of height less than $k$ for
$\subst{\sigma}{\hsubst{\theta'}
       {(\lfchecktype{\hsubst{\theta_r}{\hsubst{\theta}{G}}}
                     {M_i}{A'_i})}}.$
For us to be able to conclude that, for each $i$, $1 \leq i \leq n$,
the formula
$\subst{\sigma}{\hsubst{\theta'}
                       {(\fatm{\hsubst{\theta_r}{\hsubst{\theta}{G}}}
                              {M_i: A'_i})}}$  
is valid, it only remains to show that there is a derivation for
$\lftype{\subst{\sigma}{\hsubst{\theta'}{\hsubst{\theta_r}{\hsubst{\theta}{G}}}}}
        {\hsubst{\theta'}{A_i'}}$.
However, this has been done in the proof of Theorem~\ref{th:atomictype}.

We finish by proving the third clause.
From clause (1) we know that 
$\hsubst{\theta_r}{\hsubst{\theta}{R}}$ will be of the form $(h\app M_1\ldots M_n)$,
thus by the definition of \decompseqsans\ there must exist an 
$\mathcal{S'}$ such that 
$\mathcal{S}' = 
    \decompseq{\hsubst{\theta_r}{\hsubst{\theta}{F}}}
              {\hsubstseq{\Psi_\theta}{\theta}{\mathcal{S}}}$.
Suppose $\mathcal{S}'$ is valid.
Consider an arbitrary closed instance of $\hsubstseq{\Psi_{\theta}}{\theta}{\mathcal{S}}$ 
identified by $\theta'$ and $\sigma$.
These same $\theta'$ and $\sigma$ also identify a closed instance of $\mathcal{S}'$
as these sequents only differ in that the assumption formula
$\hsubst{\theta_r}{\hsubst{\theta}{F}}$
has been replaced with the collection of reduced formulas
$\left\{
  \fatm{\hsubst{\theta_r}{\hsubst{\theta}{G}}}{M_i: A'_i}
\ \middle|\ 
  1\leq i\leq n\right\}$.
If any formula in the set of assumption formulas of 
$\subst{\sigma}{\hsubstseq{\emptyset}{\theta'}{\hsubstseq{\Psi_{\theta}}{\theta}{\mathcal{S}}}}$
were not valid then this instance would be vacuously valid, so suppose all
such formulas are valid.
Then in particular, 
$\subst{\sigma}{\hsubst{\theta'}{\hsubst{\theta_r}{\hsubst{\theta}{F}}}}$
must be valid and thus by clause (2), for each $i$, $1\leq i\leq n$, the formula
$\subst{\sigma}{\hsubst{\theta'}{\fatm{\hsubst{\theta_r}{\hsubst{\theta}{G}}}
      {M_i: A'_i}}}$
will be valid.
But then all of the assumption formulas of $\subst{\sigma}{\hsubst{\theta'}{\mathcal{S}'}}$
must be valid and since this is a closed instance of a valid sequent
we can conclude that the goal formula
$\subst{\sigma}{\hsubst{\theta'}{\hsubst{\theta_r}{\hsubst{\theta}{F'}}}}$ 
is valid.
Therefore any closed instance of 
$\hsubstseq{\Psi_{\theta}}{\theta}{\mathcal{S}}$ will be valid, 
and clause (3) in the lemma is satisfied.
\end{proof}

Lemmas~\ref{lem:inst-solun} and \ref{lem:decomp-decr} yield the
following possibility for analyzing the derivability of a sequent with
$\fatm{G}{R:P}$ as an atomic assumption formula relative to an
identified head for $R$:
use the unification problem identified in
Lemma~\ref{lem:inst-solun} to generate a collection of term
substitutions and analyze the derivability of the
reduced sequent under these substitutions.
Unfortunately, this would not be a very effective strategy if we have
to treat each unifier separately.
Towards dealing with this issue, we introduce the idea of a covering
set of solutions to a unification problem.
The next three definitions culminate in a formulation of this notion. 

\begin{definition}\label{def:substitution-restriction}
The restriction of a substitution $\theta$ to the term variables context $\Psi$ is
the substitution
$
\left\{
  \langle x,M,\alpha\rangle\ \middle|\ 
  \langle x,M,\alpha\rangle\in\theta\mbox{ and } x:\alpha'\in\Psi
\right\}.
$
This substitution is denoted by $\restrict{\theta}{\Psi}$.
\end{definition}

\begin{definition}\label{def:covering-subst}
Let $\Psi$, $\Psi_1$ and $\Psi_2$ be term variables contexts,
and let $\theta_1$ and $\theta_2$ be substitutions that are
arity type preserving with respect to
$\noms \cup \STLCGamma_0 \cup \Psi_1$ and
$\noms \cup \STLCGamma_0 \cup \Psi_2$, respectively.
Then $\solun{\theta_2}{\Psi_2}$ is said to cover
$\solun{\theta_1}{\Psi_1}$ relative to $\Psi$ if there exists a pair 
$\solun{\theta_3}{\Psi_3}$ of a substitution and a term variables
context such that
\begin{enumerate}
\item $\theta_3$ is type preserving with respect to 
$\noms\cup\STLCGamma_0\cup\Psi_3$,
\item for any
$x : \alpha \in \Psi_2$, if $x : \alpha' \in \aritysum{\context{\theta_3}}{\Psi_3}$ then
  $\alpha = \alpha'$, and 
\item The substitutions $\restrict{\theta_1}{\Psi}$ and 
$\restrict{(\comp{\theta_3}{\theta_2})}{\Psi}$ are identical.
\end{enumerate} 
Note that the second condition ensures that $\noms \cup \STLCGamma_0
\cup ((\Psi_2\setminus\context{\theta_3}))\cup \Psi_3)$ determines a
valid arity context and that $\theta_2$ and $\theta_3$ are arity type 
compatible with respect to this context.
Thus, the composition of $\theta_2$ and
$\theta_3$ in the third condition is well-defined.
\end{definition}

\begin{definition}\label{def:covering-solns}
A collection $S$ of solutions to a unification problem
$\mathcal{U} = \unif{\mathbb{N}}{\Psi}{\mathcal{E}}$ 
is said to be \emph{covering set of solutions for $\mathcal{U}$} if every
solution to $\mathcal{U}$ is covered by some solution in $S$ relative
to $\Psi$.
\end{definition}

The following lemma provides the basis for using the reduced forms
generated by just a covering set of solutions for the relevant
unification problem in analyzing the derivability of the sequent.

\begin{lemma}\label{lem:covers-seq}
Let $\mathcal{S}$ be the well-formed sequent
$\seq[\mathbb{N}]{\Psi}{\Xi}{\Omega}{F'}$ and let $\theta_1$ and
$\sigma$ identify a closed instance $\mathcal{S}'$ of $\mathcal{S}$. 
Further, let $\solun{\theta_2}{\Psi_2}$ be substitution
compatible with $\mathcal{S}$ and such that it covers 
$\solun{\theta_1}{\emptyset}$ relative to $\Psi$.
Then there is a term substitution $\theta$ that together with $\sigma$
identifies a closed instance of
$\hsubstseq{\Psi_2}{\theta_2}{\mathcal{S}}$ that is valid if and only
if $\mathcal{S}'$ is.
\end{lemma}
\begin{proof}
We argue below that, under the assumptions of the lemma, there is a
substitution $\theta_3$ and a term 
variables context $\Psi_3$ such that $\solun{\theta_3}{\Psi_3}$ is
substitution compatible with
$\hsubstseq{\Psi_2}{\theta_2}{\mathcal{S}}$ and for any term $M$ 
such that $\stlctyjudg{\mathbb{N}\cup\STLCGamma_0\cup\Psi}{M}{\alpha}$
is derivable it is the case that 
$\hsubst{\theta_3}{\hsubst{\theta_{2r}}{\hsubst{\theta_2}{M}}} =
\hsubst{\theta_1}{M}$, where $\theta_{2r}$ is the raising substitution
associated with the application of $\theta_2$ to $\mathcal{S}$ relative to $\Psi_2$.
It follows from this that the formulas and context types appearing in
$\hsubstseq{\Psi_3}{\theta_3}{\hsubstseq{\Psi_2}{\theta_2}{S}}$ must be
identical to the ones in $\hsubstseq{\emptyset}{\theta_1}{S}$.
It is then easily seen that $\theta_3$ can be extended into a
substitution $\theta$ that together with $\sigma$ identifies a closed
instance of $\hsubstseq{\Psi_2}{\theta_2}{\mathcal{S}}$ whose
formulas are identical to those of $\mathcal{S}'$.
The lemma is an immediate consequence.

Since $\solun{\theta_2}{\Psi_2}$ covers $\solun{\theta_1}{\emptyset}$ relative to $\Psi$
we know that there exists a pair $\solun{\theta'_3}{\Psi_3}$ satisfying the conditions of
Definition~\ref{def:covering-subst}.
We claim that we may further assume of $\theta'_3$ that
(1)~$\context{\theta'_3}=(\Psi\setminus\context{\theta_2})\cup\Psi_2$,
and (2)~$\supportof{\theta'_3}\cap\mathbb{N}=\emptyset$.
Condition (1) may initially be violated because
there are tuples in $\theta'_3$ for variables that are not assigned
types by $(\Psi\setminus\context{\theta_2})\cup\Psi_2$ or because 
$\context{\theta'_3}$ does not span all of $\Psi_2$.
The first is easily fixed by dropping such tuples and the second by  
adding tuples of the form $\langle x,x,\alpha\rangle$ to
$\theta'_3$ and, if needed, the type assignment $x:\alpha$ to
$\Psi_3$. 
As for the second condition, suppose that it is violated at first.
Consider then a permutation $\pi$ of nominal constants that swaps the
constants in $\supportof{\theta'_3} \cap \mathbb{N}$ with ones that
are not contained in
$\supportof{\theta_1} \cup \supportof{\theta_2} \cup \mathbb{N}$ and
is the identity map on all other constants.  
By the choice of $\pi$ and the fact that $\langle
\theta_1,\emptyset \rangle$ and $\langle \theta_2, \Psi_2\rangle$ are
both substitution compatible with $\mathcal{S}$, it is the case that
$\permute{\pi}{M}$ is identical to $M$ for any $M$ such that $\langle
x, M, \alpha \rangle$ is a member of $\theta_1$ or $\theta_2$.
Now, we could replace $\theta'_3$ with $\permute{\pi}{\theta'_3}$
if the latter substitution also satisfies the conditions needed of it
by Definition~\ref{def:covering-subst}. 
It is easy to check that $\permute{\pi}{\theta'_3}$ satisfies the
first two of these conditions.
From the fact that
$\restrict{\theta'_3}{\Psi\setminus\context{\theta_2}} =
\restrict{\theta_1}{\Psi\setminus\context{\theta_2}}$, it follows that
  the tuples in $\restrict{\theta'_3}{\Psi\setminus\context{\theta_2}}$
  will be unaffected by the permutation and hence
$\restrict{\permute{\pi}{\theta'_3}}{\Psi\setminus\context{\theta_2}} =
  \restrict{\theta_1}{\Psi\setminus\context{\theta_2}}$.
Thus, to conclude that
$\restrict{\comp{\permute{\pi}{\theta'_3}}{\theta_2}}{\Psi} =
\restrict{\theta_1}{\Psi}$, it suffices to show that for any tuple
of the form $\langle x, M, \alpha \rangle$ that belongs to $\theta_2$,
if it is the case that
$\langle x, \hsubst{\theta'_3}{M}, \alpha \rangle$ belongs to
$\theta_1$, then
$\hsubst{\theta'_3}{M} = \hsubst{\permute{\pi}{\theta'_3}}{M}$.
A simple inductive argument establishes the fact that
$\permute{\pi}{\hsubst{\theta'_3}{M}} =
\hsubst{\permute{\pi}{\theta'_3}}{(\permute{\pi}{M})}$.
The desired conclusion follows by utilizing the observation that
$\hsubst{\theta'_3}{M}$ and $M$ are unaffected by the permutation.

We now use Theorem~\ref{th:raised-subs} to obtain a ``raised'' version
of $\theta'_3$ that together with $\Psi_3$ will constitute the
pair $\langle \theta_3,\Psi_3 \rangle$ that we desired at the outset.
Specifically, the theorem allows us to conclude that there is a 
substitution $\theta_3$ satisfying the following properties:
\begin{enumerate}
\item $\supportof{\theta_3}$ is disjoint from $\mathbb{N} \cup
  \supportof{\theta_2}$,

\item $\context{\theta_3}$ is identical to the
raised version of $(\Psi\setminus\context{\theta_2})\cup\Psi_2$
corresponding to the raising substitution $\theta_{2r}$,

\item $\theta_3$ is arity type preserving with respect to
  $\noms\cup\STLCGamma\cup\Psi_3$, and

\item for every term $M$ such that
$\stlctyjudg{\mathbb{N}\cup\STLCGamma_0\cup\Psi}{M}{\alpha}$,
$\hsubst{\theta'_3}{\hsubst{\theta_2}{M}} =
\hsubst{\theta_3}{\hsubst{\theta_{2r}}{\hsubst{\theta_2}{M}}}$.
\end{enumerate}
The argument for the first three of these properties is obvious.
For the last property, we observe, using Theorem~\ref{th:aritysubs},
that
$\stlctyjudg{(\mathbb{N}\cup\supportof{\theta_2})
              \cup\STLCGamma_0\cup((\Psi\setminus\context{\theta_2})\cup
\Psi_2)}{\hsubst{\theta_2}{M}}{\alpha}$ has a derivation under the
condition described; 
Theorem~\ref{th:raised-subs} can then be invoked in an obvious way.
It follows immediately from the first three properties that $\langle
\theta_3,\Psi_3 \rangle$ is substitution compatible with
$\hsubstseq{\Psi_2}{\theta_2}{\mathcal{S}}$.
It therefore only remains to show that for every $M$ such that
$\stlctyjudg{\mathbb{N}\cup\STLCGamma_0\cup\Psi}{M}{\alpha}$, it is
the case that $\hsubst{\theta_1}{M} =
\hsubst{\theta_3}{\hsubst{\theta_{2r}}{\hsubst{\theta_2}{M}}}$.
An easy inductive argument shows that for any $M$ of the kind
described and any $\theta$, if $\hsubst{\theta}{M}=M'$ is derivable
exactly when $\hsubst{\restrict{\theta}{\Psi}}{M}=M'$ is derivable.
It follows from this that
$\hsubst{\theta_1}{M} = \hsubst{\theta'_3}{\hsubst{\theta_2}{M}}$.
Property (4) then yields the desired result.

\end{proof}

We now use the observations in Lemmas~\ref{lem:heads-cover},
\ref{lem:inst-solun}, \ref{lem:decomp-decr} and \ref{lem:covers-seq}
to describe a complete analysis of the derivability of a sequent
around an atomic assumption formula.

\begin{definition}\label{def:def-cases}
Let $\mathcal{S}$ be the well-formed sequent
$\seq[\mathbb{N}]{\Psi}{\Xi}{\Omega}{F'}$, let $F = \fatm{G}{R:P}$ be
a formula in $\Omega$ and let
$h:\typedpi{x_1}{A_1}{\ldots\typedpi{x_n}{A_n}}{P'}$ be a type
assignment in $\Sigma$ or in the explicit bindings in $G$. 
Further, for $1 \leq i \leq n$, let $z_i :\alpha_i$ be a
distinct variable of type $\erase{A_i}$ away from $\Psi$ and raised
over $\mathbb{N}$, and let $t_i$ be the generalized variable term
corresponding to $z_i$. 
Finally, let
$\mathcal{U}$ be the unification problem 
\begin{tabbing}
\qquad\quad\=\quad\qquad\qquad\=\kill
\>
$\langle \mathbb{N}, 
         \Psi \cup \{z_1:\alpha_1,\ldots, z_n:\alpha_n\},$\\
\>\>
$\left\{P = \hsubst{\{\langle x_1,t_1,\erase{A_1}\rangle,
            \ldots, \langle x_n,t_n,\erase{A_n}\rangle \}}{P'}, 
         R = (h \app t_1 \app \ldots\app h_n)\right\}\rangle$
\end{tabbing}
and let $C$ be a covering set of solutions for $\mathcal{U}$.
Then the \emph{analysis of $\mathcal{S}$ based on $F$ and $h$} is
denoted by 
$\makecases[F]{\mathcal{S}}{h:A}$ and is given by the set of sequents
\begin{tabbing}
\qquad\=\kill
\> $\left\{ \decompseq{F'} 
                      {\hsubstseq{\Psi_{\theta}}{\theta}{\mathcal{S}}}
      \ \middle| \ 
      \solun{\theta}{\Psi_{\theta}}\in C\ 
           \mbox{and}\ F'\ \mbox{is the formula in}
           \ \hsubstseq{\Psi_{\theta}}{\theta}{\mathcal{S}}
           \ \mbox{resulting from}\ F \right\}$.
\end{tabbing}
If $\mathcal{S}$ is a well-formed sequent and $F=\fatm{G}{R:P}$ is an
assumption formula in $\mathcal{S}$, then the \emph{complete analysis
  of $\mathcal{S}$ based on $F$} is the set of sequents
\begin{tabbing}
\qquad\qquad\=$\bigcup \{ \makecases[F]{\mathcal{S}'}{h:A}\ \vert\ $\=\kill
\>$\bigcup \{ \makecases[F']{\mathcal{S}'}{h:A}
                \ \vert
                \ (\mathcal{S}';\ h:A)\in\hds{\mathcal{S}}{G}$\\
\>\>\ $\mbox{and}\ F'\ \mbox{is the formula in}\ \mathcal{S'}\ \mbox{resulting
  from}\ F \}.$
\end{tabbing} 
This collection is denoted by $\casesfn[F]{\mathcal{S}}$.
Note that the notations $\makecases[F]{\mathcal{S}}{h:A}$ and
$\casesfn[F]{\mathcal{S}}$ are both ambiguous---for instance, the
first notation leaves out mention of the covering set of solutions
that plays a role in generating the set it denotes.
We will assume them to denote any of the set of sequents that can be
generated in the respective ways described in this definition.
\end{definition}

The following lemmas will be useful in showing that the proof rule
that we will present using \casessans\ is well-defined and sound.

\begin{lemma}\label{lem:cases-seq-ok}
If $\mathcal{S}$ is a well-formed sequent and $F$ is an atomic
assumption formula of the form $\fatm{G}{R:P}$ that appears in
$\mathcal{S}$, then every sequent in $\casesfn[F]{\mathcal{S}}$ is well-formed.
\end{lemma}
\begin{proof}
By Lemma~\ref{lem:heads-wf} we know that every 
$\langle\mathcal{S}',h:A\rangle\in\hds{\mathcal{S}}{G}$
is such that $\mathcal{S}'$ is a well-formed sequent.
Further, if $\mathcal{S}'$ is the sequent
$\seq[\mathbb{N}']{\Psi'}{\Xi'}{\Omega'}{F'}$, then
$\wftype{\mathbb{N}'\cup\STLCGamma_0\cup\Psi'}{A}$ must have a
derivation. 
Let $F''$ be the formula from $\mathcal{S}'$ corresponding to $F$
and let us denote the unification problem that is considered by 
$\makecases[F']{\mathcal{S}'}{h:A}$ by $\mathcal{U}$. 
Then, by Theorem~\ref{th:seq-term-subs-ok}, the application of any
solution for $\mathcal{U}$ to $\mathcal{S}'$ must be well-formed.
Now, if $\fatm{G'}{h\app M_1\ldots M_n:P'}$ is a formula that is
well-formed with respect to $\mathbb{N}\cup\STLCGamma_0\cup\Psi$ and  
$\ctxsanstype{\Xi}$ where
$h:\typedpi{x_1}{A_1}{\ldots\typedpi{x_n}{A_n}{P''}}$ appears 
in $\Sigma$ or the explicit bindings of $G'$, then it is easily argued
that, for $1\leq i\leq n$,
\[\fatm{G}
      {M_i:\hsubst{\{\langle x_1,M_1,\erase{A_1}\rangle,\ldots,
                     \langle x_{i-1},M_{i-1},\erase{A_{i-1}}\rangle\}}
                  {A_i}}\]
is also a well-formed formula with respect to 
$\mathbb{N}\cup\STLCGamma_0\cup\Psi$ and $\ctxsanstype{\Xi}$.
It follows from this that each sequent in
$\makecases[F']{\mathcal{S}'}{h:A}$ must be well-formed.
Since this is true regardless of the pair
$\langle\mathcal{S}',h:A\rangle$ that is considered, it must be the
case that all the sequents in $\casesfn[F]{\mathcal{S}}$ are well-formed.
\end{proof}

\begin{lemma}\label{lem:cases-cover}
Let $\mathcal{S}$ be a well-formed sequent and let $F$ be an atomic
assumption formula of the form $\fatm{G}{R:P}$ in $\mathcal{S}$.
If all the sequents in $\casesfn[F]{\mathcal{S}}$ are valid then
$\mathcal{S}$ must be valid.
\end{lemma}
\begin{proof}
Consider an arbitrary closed instance of the sequent $\mathcal{S}$ identified
by $\theta$ and $\sigma$.
By Lemma~\ref{lem:heads-cover} and Theorem~\ref{th:perm-valid}, there
is a pair $\langle\mathcal{S}',h':A'\rangle$ in $\hds{\mathcal{S}}{G}$
that has in it an atomic assumption formula $F'$ of the form
$\fatm{G'}{R':P'}$ and 
that is such that if every closed instance of $\mathcal{S}'$ identified by
substitutions $\theta'$ and $\sigma'$ in which the head of
$\hsubst{\theta'}{R'}$ is identical to $h'$ is valid then
$\subst{\sigma}{\hsubstseq{\emptyset}{\theta}{\mathcal{S}}}$ must be
valid.
By Lemmas~\ref{lem:inst-solun}, \ref{lem:decomp-decr}, and
\ref{lem:covers-seq}, there is a sequent $\mathcal{S}''$ in
$\makecases[F']{\mathcal{S'}}{h':A'}$ that is such that if
$\mathcal{S}''$ is valid then the mentioned instance of $\mathcal{S}'$
must be valid.
Finally, since $\casesfn[F]{\mathcal{S}}$ collects all these sequents
for each pair in $\hds{\mathcal{S}}{G}$, it follows that the validity
of the sequents in it ensures the validity of every closed instance of
$\mathcal{S}$ and, hence, of $\mathcal{S}$.
\end{proof}

\subsubsection{Proof Rules that Introduce Atomic Formulas}\label{ssec:atomic-rules}

Figure~\ref{fig:rules-atom} presents rules for introducing atomic
formulas that internalize their understanding as encodings of typing
derivations in LF. 
The rule for introducing an atomic goal formula with an atomic type
has a proviso on the context expression.
This proviso ensures that the corresponding LF context is well-formed
if the assumption formulas of the sequent instance are valid.
This rule can be complemented by others, such as ones that check
(closed) contexts for wellformedness or that check closed goal
formulas for LF derivability.
We do not discuss these alternatives explicitly here, but their
soundness should be easy to verify.
The treatment of an atomic assumption and goal formula with a type
corresponding to an abstraction is identical:
a nominal constant of the right type that does not appear in the
support set of the sequent is picked and use for the abstracted
variable as per the LF rule for typing abstractions, changing the
association with any relevant context variable to ensure that only
those of its instantiations are considered that do not use the chosen
nominal constant.

\begin{figure}[tbh]

\begin{center}
\begin{tabular}{c}

\infer[\appL]
      {\seq[\mathbb{N}]{\Psi}
           {\Xi}
           {\setand{\Omega}{\fatm{G}{R:P}}}
           {F}}
      {\casesfn{\seq[\mathbb{N}]{\Psi}
                    {\Xi}
                    {\setand{\Omega}{\fatm{G}{R:P}}}
                    {F}}}
\\[12pt] 

\infer[\appR]
      {\seq[\mathbb{N}]{\Psi}
           {\Xi}
           {\Omega}
           {\fatm{G}{h\app M_1\ldots M_n:P'}}}
      {\begin{array}{c}
           h:\typedpi{x_1}{A_1}{\ldots \typedpi{x_n}{A_n}{P}}
                      \in\Sigma\mbox{ or the explicit bindings in }G
             \\
           G\ \mbox{appears as the context of some atomic formula in}\ \Omega\\
           \hsub{\{\langle x_1, M_1, \erase{A_1}\rangle,\ldots,\langle x_n, M_n, \erase{A_n}\rangle\}}
                {P}
                {P'}
         \\
         \begin{array}{l}
         \left\{\seq[\mathbb{N}]
               {\Psi}
               {\Xi}
               {\Omega}
               {}\right.\\
               \qquad
               \left.\fatm{G}
                     {M_i:
                       \hsubst{\{\langle x_1, M_1, \erase{A_1}\rangle,\ldots,
                                 \langle x_{i-1}, M_{i-1}, \erase{A_{i-1}}\rangle\}}
                              {A_i}}
               \ \mid\ 1\leq i\leq n \right\}
         \end{array}
       \end{array} 
      } 

\\[12pt]

\infer[\absL]
      {\seq[\mathbb{N}]{\Psi}{\Xi}{\setand{\Omega}{\fatm{G}{\lflam{x}{M}:\typedpi{x}{A_1}{A_2}}}}{F}}
      {\begin{array}{c}
           n\ \mbox{is new to}\ \mathbb{N} \mbox{ and has type } \erase{A_1}
         \\
         \Xi' =
           \begin{cases}
             \left(\Xi \setminus
             \left\{\ctxvarty{\Gamma}
                             {\mathbb{N}_{\Gamma}}
                             {\ctxty{\mathcal{C}}
                             {\mathcal{G}}}\right\}\right)
             \cup 
             \left\{\ctxvarty{\Gamma}
                             {(\mathbb{N}_{\Gamma},n:\erase{A_1})}
                             {\ctxty{\mathcal{C}}{\mathcal{G}}}\right\}
               & \mbox{if }\Gamma\mbox{ appears in }G \\
             \Xi & \mbox{otherwise}
           \end{cases}
         \\
         \seq[\mathbb{N},n:\erase{A_1}]
             {\Psi}
             {\Xi'}
             {\setand{\Omega}
                     {\fatm{G,n:A_1}
                           {\hsubst{\{\langle x,n,\erase{A_1}\rangle\}}{M}:
                              \hsubst{\{\langle x,n,\erase{A_1}\rangle\}}{A_2}}}}
             {F}
       \end{array}}

\\[12pt]

\infer[\absR]
      {\seq[\mathbb{N}]{\Psi}{\Xi}{\Omega}{\fatm{G}{\lflam{x}{M}:\typedpi{x}{A_1}{A_2}}}}
      {\begin{array}{c}
           n\ \mbox{is new to}\ \mathbb{N} \mbox{ and has type } \erase{A_1}
         \\
         \Xi' =
           \begin{cases}
             \left(\Xi \setminus
             \left\{\ctxvarty{\Gamma}
                             {\mathbb{N}_{\Gamma}}
                             {\ctxty{\mathcal{C}}
                             {\mathcal{G}}}\right\}\right)
             \cup 
             \left\{\ctxvarty{\Gamma}
                             {(\mathbb{N}_{\Gamma},n:\erase{A_1})}
                             {\ctxty{\mathcal{C}}{\mathcal{G}}}\right\}
               & \mbox{if }\Gamma\mbox{ appears in }G \\
             \Xi & \mbox{otherwise}
           \end{cases}
         \\
           \seq[\mathbb{N},n:\erase{A_1}]
               {\Psi}
               {\Xi'}
               {\Omega}
               {\fatm{G,n:A_1}
                     {\hsubst{\{\langle x,n,\erase{A_1}\rangle\}}{M}:
                          \hsubst{\{\langle x,n,\erase{A_1}\rangle\}}{A_2}}}
       \end{array}}
\end{tabular}
\end{center}
\caption{ Rules that Introduce Atomic Formulas}
\label{fig:rules-atom}
\end{figure}

As usual, we show that these rules preserve the wellformedness of
sequents and are also sound.

\begin{theorem}\label{th:atom-wf}
The following property holds for each rule in
Figure~\ref{fig:rules-atom}: if the conclusion sequent is
well-formed, the premises expressing typing conditions have
derivations and the conditions expressed by the other, non-sequent
premises are satisfied, then all the sequent premises must
be well-formed.
\end{theorem}
\begin{proof}
Lemma~\ref{lem:cases-seq-ok} ensures that the property holds for the
\appL\ rule.
For the \appR\ rule, we note first that the
wellformedness of the conclusion sequent immediately assures us that
the context variables contexts and the assumption formulas of each of
the $n$ premise sequents satisfy the conditions required of them for
the wellformedness of the corresponding sequent. 
This leaves us needing to show only that
\[\wfform{\mathbb{N} \cup \STLCGamma_0 \cup \Psi}
        {\Xi}
        {\fatm{G}
              {M_i:\hsubst{\{\langle x_1, M_1, \erase{A_1}\rangle,\ldots,
                                  \langle x_{i-1}, M_{i-1}, \erase{A_{i-1}}\rangle\}}
              {A_i}}}\]
has a derivation for $1 \leq i \leq n$.
Without loss of generality, we may assume that the variables
$x_1,\ldots,x_n$ are chosen to be distinct and also different from 
the ones in $\Psi$.
For $1\leq i \leq (n+1)$, let $\theta_i$ denote the term
substitution
$\{\langle x_1, M_1, \erase{A_1}\rangle,\ldots,
\langle x_{i-1}, M_{i-1}, \erase{A_{i-1}}\rangle\}$.
From the wellformedness of the conclusion sequent and the premise
conditions for the rule we easily see that there must be derivations
for
$\wftype{\mathbb{N}\cup\STLCGamma_0\cup\Psi \cup \context{\theta_i}}
        {A_i}$ for $1 \leq i \leq n$ 
and
$\stlctyjudg{\mathbb{N}\cup\STLCGamma_0\cup\Psi}
            {h\app M_1\ldots M_n}
            {\erase{\hsubst{\theta_{n+1}}{{P}}}}$,
where
$h : \erase{\typedpi{x_1}{A_1}{\ldots \typedpi{x_n}{A_n}{P}}}$ is
in $\mathbb{N} \cup \STLCGamma_0 \cup \Psi$.
Using an induction on $i$ and invoking Theorems~\ref{th:erasure} and
\ref{th:aritysubs-ty}, we can show from this that, for
$1 \leq i \leq n$, it must be the case that $\theta_i$ is type
preserving with respect to $\mathbb{N}\cup\STLCGamma_0\cup\Psi$, that
$\hsubst{\theta_i}{A_i}$ is well-defined, and that 
$\wftype{\mathbb{N}\cup\STLCGamma_0\cup\Psi}{\hsubst{\theta_i}{A_i}}$
and 
$\stlctyjudg{\mathbb{N}\cup\STLCGamma_0\cup\Psi}
            {M_i}
            {\erase{\hsubst{\theta_i}{A_i}}}$
have derivations.
The desired conclusion is an easy consequence of this.

We consider the two remaining cases, for \absL\ and \absR,
simultaneously.
Using the newness of $n$ and the wellformedness of the conclusion
sequent, it is easily shown that the context variables context, all
the formulas in $\Omega$ and, in the case of \absL, the formula $F$
all satisfy the conditions required of them for the premise sequent to
be well-formed.
This leaves it only to be shown that
\[\wfform{(\mathbb{N} \cup \{n:\erase{A_1}\}) \cup \STLCGamma_0 \cup \Psi}
        {\ctxsanstype{\Xi'}}
        {\fatm{G, n:A_1}
              {\hsubst{\{\langle x, n, \erase{A_1}\}\rangle}{M} :
                  \hsubst{\{\langle x, n, \erase{A_1}\rangle\}}{A_2}}} \]
has a derivation.
Noting that $\ctxsanstype{\Xi'} = \ctxsanstype{\Xi}$ and the
wellformedness of the conclusion sequent, we can easily see that
$\wfctx{(\mathbb{N} \cup \{n:\erase{A_1}\}) \cup \STLCGamma_0 \cup \Psi}
       {\ctxsanstype{\Xi'}}
       {G}$
must have a derivation.
What remains to be shown, then, is that 
$\wftype{(\mathbb{N} \cup \{n:\erase{A_1}\})\cup\STLCGamma_0\cup\Psi}
        {\hsubst{\{\langle x, n, \erase{A_1}\rangle\}}{A_2}}$
and
$\stlctyjudg{(\mathbb{N}\cup \{n:\erase{A_1}\}) \cup\STLCGamma_0\cup\Psi}
            {\hsubst{\{\langle x, n, \erase{A_1}\}\rangle}{M}}
            {\erase{\hsubst{\{\langle x, n, \erase{A_1}\rangle\}}{A_2}}}$
have derivations. 
However, this is also easily argued for using the facts that, by the
wellformedness of the conclusion sequent, 
$\wftype{\mathbb{N}\cup\STLCGamma_0\cup\Psi}
        {\typedpi{x}{A_1}{A_2}}$
and
$\stlctyjudg{\mathbb{N} \cup\STLCGamma_0\cup\Psi}
            {\lflam{x}{M}}
            {\erase{\typedpi{x}{A_1}{A_2}}}$
must have derivations and that $\{\langle x,n,\erase{A_1}\rangle\}$ is
an arity type preserving substitution; this argument would invoke
Theorems~\ref{th:erasure} and \ref{th:aritysubs-ty} at
relevant points.
\end{proof}

\begin{theorem}\label{th:atom-sound}
The following property holds for every instance of each of the rules
in Figure~\ref{fig:rules-atom}: if the premises expressing
typing judgements are derivable, the conditions described in the other
non-sequent premises are satisfied and the premise sequents are valid,
then the conclusion sequent must also be valid. 
\end{theorem}
\begin{proof}
Lemma~\ref{lem:cases-cover} verifies that the property holds for the
\appL\ rule.
Consider a closed instance of the conclusion sequent of the
\appR\ rule that is identified by the substitutions $\theta$ and
$\sigma$.
If any formula in $\subst{\sigma}{\hsubst{\theta}{\Omega}}$ is not
valid, then this instance would be valid.
We may therefore assume that all these formulas are valid and our task
is then to show that the formula
$\subst{\sigma}{\hsubst{\theta}{\fatm{G}{h\app M_1\ldots M_n:P'}}}$ 
  must also be valid under the assumptions of the theorem.
Since $G$ appears as the context in one of the formulas in $\Omega$,
$\lfctx{\subst{\sigma}{\hsubst{\theta}{G}}}$
must have a derivation.
Clearly, the substitutions $\theta$ and $\sigma$ also identify closed
instances of the premise sequents all of whose assumption formulas
are valid.
From the validity of these instances, the definition of substitution
application and Theorem~\ref{th:subspermute}, it follows that, for 
$1 \leq i \leq n$, the formulas
$\fatm{\subst{\sigma}{\hsubst{\theta}{G}}}
      {\hsubst{\theta}{M_i}:
            \hsubst{\{\langle x_1, \hsubst{\theta}{M_1}, \erase{A_1}\rangle,\ldots,
                       \langle x_{i-1}, \hsubst{\theta}{M_{i-1}}, \erase{A_{i-1}}\rangle\}}
                   {\hsubst{\theta}{A_i}}}$
must be valid.
Their validity means that, for $1 \leq i \leq\ n$, the LF judgements
\[\lfchecktype{\subst{\sigma}{\hsubst{\theta}{G}}}
             {\hsubst{\theta}{M_i}}
             {\hsubst{\{\langle x_1, \hsubst{\theta}{M_1}, \erase{A_1}\rangle,\ldots,
                        \langle x_{i-1}, \hsubst{\theta}{M_{i-1}}, \erase{A_{i-1}}\rangle\}}
                     {\hsubst{\theta}{A_i}}}\]
must have derivations; these judgements are coherent because 
the validity of the atomic formula also assures us of the
wellformedness of the corresponding types. 
From the assumption concerning $h$ and the definition of substitution
application, it follows that 
$h:\typedpi{x_1}
           {\hsubst{\theta}{A_1}}
           {\ldots \typedpi{x_n}{\hsubst{\theta}{A_n}}{\hsubst{\theta}{P}}}$
is a member of $\Sigma$ or $\subst{\sigma}{\hsubst{\theta}{G}}$.
We may now invoke Theorem~\ref{th:atomictype}, part (1), to conclude
that there must be derivations for 
$\lftype{\subst{\sigma}{\hsubst{\theta}{G}}}
        {\hsubst{\{\langle x_1, \hsubst{\theta}{M_1}, \erase{A_1}\rangle,\ldots,
                     \langle x_{n}, \hsubst{\theta}{M_{n}}, \erase{A_{n}}\rangle\}}
                {\hsubst{\theta}{P}}}$
and
\[\lfchecktype{\subst{\sigma}{\hsubst{\theta}{G}}}
              {h\app \hsubst{\theta}{M_1}\app \ldots \app \hsubst{\theta}{M_n}}
              {\hsubst{\{\langle x_1, \hsubst{\theta}{M_1}, \erase{A_1}\rangle,\ldots,
                         \langle x_{i-1}, \hsubst{\theta}{M_{i-1}}, \erase{A_{i-1}}\rangle\}}
                      {\hsubst{\theta}{A_i}}};\]
the wellformedness of the type is established in the proof.
Using Theorem~\ref{th:subspermute}, we see that the type 
$\hsubst{\{\langle x_1, \hsubst{\theta}{M_1}, \erase{A_1}\rangle,\ldots,
                     \langle x_{n}, \hsubst{\theta}{M_{n}}, \erase{A_{n}}\rangle\}}
        {\hsubst{\theta}{P'}}$
is identical to $\hsubst{\theta}{P'}$.
By the definition of substitution 
$(h\app \hsubst{\theta}{M_1}\app \ldots \app \hsubst{\theta}{M_n})$
is the same term as $\hsubst{\theta}{(h \app M_1 \app \ldots \app M_n)}$.
Thus, there are derivations for all three judgements determining the
validity of
$\subst{\sigma}{\hsubst{\theta}{\fatm{G}{h\app M_1\ldots M_n:P'}}}$,
implying that it must be valid. 

We consider next the case for the \absL\ rule.
Let $\theta$ and $\sigma$ identify a closed instance of the conclusion
sequent.
We must show that this instance is valid provided the assumptions of
the theorem are satisfied.
We may assume without loss of generality that $\supportof{\theta}$ and
$\supportof{\sigma}$ \emph{does not} include the nominal constant $n$.
If this condition is initially violated, we may consider instead
variants of these substitutions under a permutation that swaps $n$
with a nominal constant that does not appear in
$\supportof{\theta} \cup \supportof{\sigma} \cup \mathbb{N}$,
show the validity of the closed instance under these variants and then
invoke Theorem~\ref{th:perm-valid}.
If any formula in  
$\subst{\sigma}
       {\hsubst{\theta}
       {(\setand{\Omega}{\fatm{G}{\lflam{x}{M}:\typedpi{x}{A_1}{A_2}}})}}$ 
is not valid then the conclusion sequent must be valid.
Thus it only remains for us to show the validity of
$\subst{\sigma}{\hsubst{\theta}{F}}$ when all these formulas are
valid. 
From the assumption about $\theta$ and $\sigma$, it is easy to
see that they must also determine a closed instance of the premise
sequent.
Using the definition of validity and the assumption of newness for
$n$, it is easy to see that the validity of
$\subst{\sigma}{\hsubst{\theta}{(\lflam{x}{M}:\typedpi{x}{A_1}{A_2})}}$
entails that there must be derivations for all three judgements that
determine the validity of 
$\subst{\sigma}
       {\hsubst{\theta}
               {(\fatm{G,n:A_1}
                     {\hsubst{\{\langle x,n,\erase{A_1}\rangle\}}{M}:
                        \hsubst{\{\langle x,n,\erase{A_1}\rangle\}}{A_2}})}}$.
From this it follows that all the assumption formulas in the instance of
the premise sequent under consideration must be valid.
From the assumed validity of this sequent, it then follows that
$\subst{\sigma}{\hsubst{\theta}{F}}$ must be valid, as we needed to show.

The only case left to be considered is that of the \absR\ rule.
Let $\theta$ and $\sigma$ identify a closed instance of the conclusion
sequent.
Our task is to show that the instance they identify is valid under the
assumptions of the theorem.
As before, we may assume that $\supportof{\theta}$ and
$\supportof{\sigma}$ do not contain $n$.
These substitutions must then identify a closed instance of the
premise sequent.
The conclusion sequent would be trivially valid if any formula in
$\subst{\sigma}{\hsubst{\theta}{\Omega}}$ is not valid.
Thus, we only need to show that
$\subst{\sigma}
       {\hsubst{\theta}{\lflam{x}{M}:\typedpi{x}{A_1}{A_2}}}$
when all these formulas are valid.
Noting that the assumption formulas in the relevant instance of the
premise sequent are exactly $\subst{\sigma}{\hsubst{\theta}{\Omega}}$,
its validity implies the validity of
$\subst{\sigma}
       {\hsubst{\theta}
               {\fatm{G,n:A_1}
                     {\hsubst{\{\langle x,n,\erase{A_1}\rangle\}}{M}:
                          \hsubst{\{\langle x,n,\erase{A_1}\rangle\}}{A_2}}}}$. 
Using the definitions of validity and LF derivability, it is easy to
see that the validity of this formula implies the derivability of all
three judgements that determine the validity of 
$\subst{\sigma}{\hsubst{\theta}{\Omega}}$.
Thus, this formula must be valid, thereby concluding the proof.
\end{proof}

\subsection{Proof Rules for Induction on the Heights of LF Derivations}
\label{ssec:induction}

The idea we use to build in a means for reasoning on the heights of LF
derivations is borrowed from the Abella proof
assistant~\cite{baelde14jfr,gacek09phd}. 
This idea is based on an annotation scheme that serves to determine
when an atomic formula represents a typing derivation in LF that has a
height less than that of the corresponding LF derivation represented by 
an atomic formula that appears in a formula being proved and, hence,
when a property in which this atomic formula appears negatively can be
assumed to hold in an inductive argument.
The first part of this subsection describes these annotations and defines a
refined semantics for sequents that pays attention to their intended
meaning.  
The second part introduces an induction rule which uses
annotations.
The last part presents auxiliary forms of the
rules for atomic formulas as well as the \id\ rule that are applicable
to annotated formulas; these forms are needed to make effective use of
the induction rule. 

\subsubsection{Extending Formula Syntax with Annotations}
We use a denumerable collection of annotations that go in pairs: $@$
and $*$, $@@$ and $**$, and so on.
We will write $@^n$ (resp. $*^n$) to denote a sequence in which the
character $@$ (resp. $*$) is repeated $n$ times.
We use $\eqannaux{i}{F}$ on an atomic formula $F$ to indicate that it
has a certain height and $\ltannaux{i}{F}$ to indicate that it has a 
strictly smaller height; we will explain what a height means shortly.
This height annotation is decreased whenever we decompose a derivation into 
sub-derivations based on its structure, as is done in the $\appL$ rule
of the previous section.

To understand the meaning of the annotations recall that an 
atomic formula $\fatm{G}{M:A}$, is valid if then there are LF derivations for
$\lfctx{G}$, $\lftype{G}{A}$ and $\lfchecktype{G}{M}{A}$.
It is the height of the typing judgement $\lfchecktype{G}{M}{A}$ that
is the basis for our scheme for inductive reasoning.
Thus, when we talk of the height of an atomic formula, we mean the
height of the derivation of this typing judgement. 
In particular, the valid closed instances of the annotated atomic
formula $\eqannaux{i}{\fatm{G}{M:A}}$ are the ones for which the
corresponding instances of $\lfchecktype{G}{M}{A}$ have derivations of
height up to some particular size $m$, while the closed instances of the relatedly
annotated formula $\ltannaux{i}{\fatm{G'}{M':A'}}$ will be valid only
if the corresponding instances of $\lfchecktype{G'}{M'}{A'}$ have 
derivations of a height strictly smaller than $m$.
Having available a denumerable collection of pairs of
such annotations allows us to simultaneously relate the heights of
different pairs of atomic formulas in this manner. 
We may of course also want to consider atomic formulas without any
annotations. 
We use the notation $\fatmann{G}{M:A}{Ann}$ to denote a formula which
may be unannotated ($\fatm{G}{M:A}$) or have an annotation 
($\fatmann{G}{M:A}{@^n}$ or $\fatmann{G}{M:A}{*^n}$).
Note that only the syntax of atomic formulas is extended with these annotations.

\begin{definition}\label{def:ann-wf}
A formula $F$ containing annotations is well-formed with respect to
$\Theta$ and $\Xi$ if the formula $F'$
obtained by erasing all annotations in $F$ is such that
$\wfform{\Theta}{\Xi}{F'}$ holds.
We overload notation by denoting this property also by
$\wfform{\Theta}{\Xi}{F}$.
A sequent $\mathcal{S}$ containing annotations is well-formed if the sequent
$\mathcal{S}'$ obtained from $\mathcal{S}$ by erasing all annotations
is well-formed by virtue of Definition~\ref{def:sequent}.
\end{definition}

We now refine the semantics of sequents to take into account the
annotations on formulas. 
The basis for doing so is provided by an association between
annotations and actual heights.

\begin{definition}\label{def:ann-and-heights}
A \emph{height assignment} $\Upsilon$ maps each annotation $@^i$
to a natural number $m_i$. 
Given a particular $\Upsilon$, a height assignment that is identical
to $\Upsilon$ except that $@^i$ is mapped to $m$ is denoted by
$\assn{\Upsilon}{@^i}{m}$. 
\end{definition}

\begin{definition}\label{def:ann_form_valid}
We define validity only for closed annotated formulas that are
well-formed, i.e, for formulas $F$ such that 
$\wfform{\noms \cup \STLCGamma_0}{\emptyset}{F}$ is derivable.
For such formulas, we first define formula validity with respect to a
height assignment $\Upsilon$, written $\fvalid{\Upsilon}{F}$, as
follows: 
\begin{itemize}
\item $\fvalid{\Upsilon}{\fatmann{G}{M:A}{Ann}}$ holds if $\lfctx{G}$
  and $\lftype{G}{A}$ are derivable, and
  \begin{itemize}
  \item $\lfchecktype{G}{M}{A}$ has a derivation, if $Ann$ is the empty annotation,

  \item $\lfchecktype{G}{M}{A}$ has a derivation of  height less than or equal to
  $\Upsilon(@^i)$, if $Ann$ is $@^i$, and
  
  \item $\lfchecktype{G}{M}{A}$ has a derivation of height less than
  $\Upsilon(@^i)$, if $Ann$ is $*^i$. 
  \end{itemize}
  
\item $\fvalid{\Upsilon}{\ftrue}$ holds.

\item $\fvalid{\Upsilon}{\ffalse}$ does not hold.

\item $\fvalid{\Upsilon}{\fimp{F_1}{F_2}}$ holds if 
$\fvalid{\Upsilon}{F_2}$ holds in the case that $\fvalid{\Upsilon}{F_1}$ holds.

\item $\fvalid{\Upsilon}{\fand{F_1}{F_2}}$ holds if both
$\fvalid{\Upsilon}{F_1}$ and $\fvalid{\Upsilon}{F_2}$ hold.

\item $\fvalid{\Upsilon}{\for{F_1}{F_2}}$ holds if either 
$\fvalid{\Upsilon}{F_1}$ or $\fvalid{\Upsilon}{F_2}$ holds.

\item $\fvalid{\Upsilon}{\fctx{\Gamma}{\mathcal{C}}{F}}$ holds if
$\fvalid{\Upsilon}{\subst{\{G/\Gamma\}}{F}}$ holds for every context expression
$G$ such that
$\csinst{\noms}{\emptyset}{\mathcal{C}}{G}$ is derivable.

\item $\fvalid{\Upsilon}{\fall{x:\alpha}{F}}$ holds if 
$\fvalid{\Upsilon}{\hsubst{\{\langle x,M,\alpha\rangle\}}{F}}$ holds for every 
$M$ such that \\$\stlctyjudg{\noms\cup\STLCGamma_0}{M}{\alpha}$ is derivable.

\item $\fvalid{\Upsilon}{\fexists{x:\alpha}{F}}$ holds if
$\fvalid{\Upsilon}{\hsubst{\{\langle x,M,\alpha\rangle\}}{F}}$ holds for some
$M$ such that \\$\stlctyjudg{\noms\cup\STLCGamma_0}{M}{\alpha}$ is derivable.
\end{itemize}
As with Definition~\ref{def:semantics}, the coherence of
this definition is assured by Theorems~\ref{th:schemainst} and
\ref{th:subst-formula}. 
\end{definition}

The validity of a sequent containing annotations corresponds
to the validity of each of its closed instances relative to every
height assignment. 
In formalizing this idea, we assume the adaptation of the
notions of the compatibility of term substitutions
(Definition~\ref{def:seq-term-subst}), the appropriateness of context substitutions
(Definition~\ref{def:seq-ctxt-subst}) and the applications of these substitutions
(Definitions~\ref{def:seq-term-subst-app} and
\ref{def:cvar-subst-app}) to annotated sequents that is 
obtained by ignoring the annotations on formulas.
Further, we refer to every instance of an annotated well-formed
sequent that is determined by term and context substitutions as in
Definition~\ref{def:seq-validity} as one of its closed instances; the
wellformedness of these closed instances follows easily from the
results of Section~\ref{ssec:sequents}.
\begin{definition}\label{def:ann_seq_valid}
A well-formed sequent of the form $\seqsans[\mathbb{N}]{\Omega}{F}$ is
valid with respect to a height 
assignment $\Upsilon$ if $\fvalid{\Upsilon}{F}$ holds whenever
$\fvalid{\Upsilon}{F'}$ holds for every $F'\in\Omega$.
A well-formed sequent $\mathcal{S}$ is valid with respect to
$\Upsilon$ if every closed 
instance of $\mathcal{S}$ is valid with respect to $\Upsilon$ and it
is valid, without qualifications, under the extended semantics if
it is valid with respect to every height assignment.
\end{definition}

While height assignments associate heights with every annotation,
the assignments to only a finite subset of annotations matter in
determining the validity of a formula or sequent.
This observation is an immediate consequence of the following lemma.

\begin{lemma}\label{lem:ann_ind}
If $F$ is a formula in which the annotations $@^i$ and $*^i$ do not
occur, then $\fvalid{\Upsilon}{F}$ holds for a height assignment 
$\Upsilon$ if and only if $\fvalid{\assn{\Upsilon}{@^i}{m}}{F}$ holds
for every choice of $m$.
Similarly, if $\mathcal{S}$ is a sequent in which the annotations $@^i$
and $*^i$ do not occur, then $\mathcal{S}$ is valid with respect to
$\Upsilon$ if and only if $\mathcal{S}$ is valid with respect to
$\assn{\Upsilon}{@^i}{m}$ for every choice of $m$.
\end{lemma}
\begin{proof}
The first clause is shown by a straightforward induction on the
formation of $F$ where in the base case we know that no atomic formula
is annotated by $@^i$ or $*^i$ and therefore the $m_i$ assigned to
that annotation cannot play a role in determining its validity.
The second clause follows from the first using the definition of validity
for sequents of Definition~\ref{def:ann_seq_valid}.
\end{proof}

Our ultimate interest is in determining that formulas devoid of
annotations are valid in the sense articulated in
Definition~\ref{def:semantics}.
This also means that we are eventually interested in the validity of
sequents devoid of annotations in the sense described in
Definition~\ref{def:seq-validity}.
However, in building in capabilities for inductive reasoning, we will
consider rules that will introduce annotated formulas into sequents.
Establishing the soundness of the resulting proof system requires us
to refine the definition of validity for sequents to the one presented
in Definition~\ref{def:ann_seq_valid}.
That this is acceptable from the perspective of our eventual goal is
the content of the following theorem.

\begin{theorem}\label{th:valid__valid}
A well-formed sequent that is devoid of annotations is valid 
in the sense of Definition~\ref{def:ann_seq_valid} if and only if it 
is valid in the sense of Definition~\ref{def:seq-validity}.
\end{theorem}
\begin{proof}
Using an argument similar to that for Lemma~\ref{lem:ann_ind}, we show
that a closed instance of a sequent devoid of annotations is
valid in the sense of Definition~\ref{def:seq-validity} if and only if
it is valid by virtue of Definition~\ref{def:ann_form_valid} with
respect to any height assignment. 
The theorem is an immediate consequence.
\end{proof}

Our attention henceforth will be on sequents that potentially contain
annotated formulas.
For this reason, absent qualifications, we shall interpret 
``wellformedness'' by virtue of Definition~\ref{def:ann-wf} and
``validity'' by virtue of  Definitions~\ref{def:ann_form_valid} and
\ref{def:ann_seq_valid}. 
As we show below, the proof rules that we have discussed up to this
point that apply to non-atomic formulas and that are different from
the \id\ rule lift in an obvious way to the situation where formulas
carry annotations. 
Similarly, the \id\ rule still applies as before when the formula that
is the focus of the rule is unannotated, and no change is needed to the
rules for atomic formulas when these formulas are unannotated.
The extension of the \id\ rule and the rules for atomic formulas to
the situation when the focus formula is annotated needs some care.
We discuss this matter after the presentation of the induction rule.

\begin{theorem}\label{th:nonatom_rules_valid}
The following properties hold for the lifted forms of the proof rules
in Figure~\ref{fig:rules-structural}, Figure~\ref{fig:rules-base},
and the \cut\ proof rule from Figure~\ref{fig:rules-other}, and to the
\id\ rule and the rules from Figure~\ref{fig:rules-atom} when the formula
they pertain to are unannotated:
\begin{enumerate}
\item If the conclusion sequent is well-formed, the premises expressing typing
  conditions have derivations and the conditions expressed by the
  other, non-sequent premises are satisfied, then the premise sequents
  must be well-formed.

\item If the premises expressing typing judgements are derivable, the
  conditions described in the other non-sequent premises are satisfied
  and all the premise sequents are valid, then the conclusion sequent must 
  also be valid.
\end{enumerate} 
\end{theorem}
\begin{proof}
The definition of wellformedness for annotated sequents is based on
the original notion via the erasure of annotations.
Hence, the first claim follows immediately from the earlier results
for unannotated sequents.
The second claim can be proved relative to an arbitrary height
assignment; this then generalizes to all possible height assignments.
For the proof rules not relating to atomic formulas, the heights which may
be assigned to annotations can play no role in the soundness argument
as there are no atomic formulas being interpreted within the proof.
For the atomic proof rules the claim is restricted to the case where the 
formulas being analyzed in the rules are not annotated, and thus again
the height assignment has no impact on the previously presented argument for
soundness. 
\end{proof}

\subsubsection{The Induction Proof Rule}

\begin{figure}
\[
\infer[\ind]
      {\seq[\mathbb{N}]{\Psi}
           {\Xi}
           {\Omega}
           {\mathcal{Q}_1.
             \fimp{F_1}{\fimp{\ldots}
                       {\mathcal{Q}_{k-1}.
                         \fimp{F_{k-1}}
                              {\mathcal{Q}_k.
                                  \fimp{\fatm{G}{M:A}}
                                        {\mathcal{Q}_{k+1}.
                                           \fimp{F_{k+1}}
                                                {\fimp{\ldots}{F_n}}}}}}}}
      {\deduce{\mathcal{Q}_1.
                 \fimp{F_1}{\fimp{\ldots}
                           {\mathcal{Q}_{k-1}.
                             \fimp{F_{k-1}}
                                  {\mathcal{Q}_k.
                                      \fimp{\fatmann{G}{M:A}{@^i}}
                                           {\mathcal{Q}_{k+1}.
                                               \fimp{F_{k+1}}
                                                    {\fimp{\ldots}{F_n}}}}}}}
              {\seq[\mathbb{N}]{\Psi}
                   {\Xi}
                   {\setand{\Omega}
                           {\mathcal{Q}_1.
                              \fimp{F_1}{\fimp{\ldots}
                                        {\mathcal{Q}_{k-1}.
                                          \fimp{F_{k-1}}
                                               {\mathcal{Q}_k.
                                                   \fimp{\fatmann{G}{M:A}{*^i}}
                                                        {\mathcal{Q}_{k+1}.
                                                           \fimp{F_{k+1}}
                                                                {\fimp{\ldots}{F_n}}}}}}}}
                   {}}}
\]
\caption{The Induction Proof Rule}
\label{fig:rule-ind}
\end{figure}

The induction rule is presented in Figure~\ref{fig:rule-ind}.
In this rule, $\mathcal{Q}_j$,  for $1 \leq j < n$, represents a
sequence of context and universal term quantifiers. 
There is a proviso on the rule: the annotations $@^i$ and $*^i$
must be fresh, i.e., they must not already appear in the sequent that is the
conclusion of the rule.

Given the definition of well-formedness for formulas and sequents that
contain annotations, the following theorem has an obvious proof. 

\begin{theorem}\label{th:ind-wf}
If the conclusion sequent of an instance of the \ind\ rule is
well-formed, then the premise sequent must be well-formed.
\end{theorem}

A derivation of the premise sequent of the induction rule 
provides a schema for constructing a concrete, valid derivation for any
closed instance of the conclusion sequent based on the height $m$ of the
assumption derivation in LF.
Since this proof will be general with respect to the choice of $m$ it ensures
that the property holds for all natural numbers, and so the sequent without an annotated 
atomic formula will be valid.
The soundness of this rule is shown by formalizing these ideas as a
meta-level argument using induction on natural numbers.
The following two lemmas are useful towards this end.

\begin{lemma}\label{lem:all_eq__lt}
Let $F_1,\ldots, F_{k-1},F_{k+1},\ldots F_n$ be formulas in which the
annotations $@^i$ and $*^i$ do not appear, let $F^{@^i}$ be 
a closed well-formed formula of the form
\[
\mathcal{Q}_1.
   \fimp{F_1}{\fimp{\ldots}
             {\mathcal{Q}_{k-1}.
               \fimp{F_{k-1}}
                    {\mathcal{Q}_k.
                        \fimp{\eqannaux{i}{\fatm{G}{M:A}}}
                             {\mathcal{Q}_{k+1}.
                                 \fimp{F_{k+1}}
                                      {\fimp{\ldots}{F_n}}}}}}
\]
in which, for $1 \leq j < n$, $\mathcal{Q}_j$ denotes a sequence of
context or universal term quantifiers, and let $F^{*^i}$ be the
(obviously closed and well-formed) formula
\[
\mathcal{Q}_1.
   \fimp{F_1}{\fimp{\ldots}
             {\mathcal{Q}_{k-1}.
               \fimp{F_{k-1}}
                    {\mathcal{Q}_k.
                        \fimp{\ltannaux{i}{\fatm{G}{M:A}}}
                             {\mathcal{Q}_{k+1}.
                                 \fimp{F_{k+1}}
                                      {\fimp{\ldots}{F_n}}}}}}
\]
For any height assignment $\Upsilon$ and natural number $m$, if it
is the case that $\fvalid{\assn{\Upsilon}{@^i}{l}}{F^{@^i}}$ holds for
every $l < m$, then $\fvalid{\assn{\Upsilon}{@^i}{m}}{F^{*^i}}$ must
also hold.
\end{lemma}

\begin{proof}
We will prove the lemma by an induction on the number of top-level
quantifiers and implications in $F^{@^i}$ prior to the occurrence of
the formula $\eqannaux{i}{\fatm{G}{M:A}}$ in it; the argument will
implicitly make use of the closed and well-formed nature of the
formulas under consideration.

In the base case, $F^{@^i}$ is a formula of the form
$\fimp{\eqannaux{i}{\fatm{G}{M:A}}}{F''}$, where $F''$ is such that
the annotations $@^{i}$ and $*^i$ do not appear in it. 
The formula $F^{*^i}$ correspondingly has the form
$\fimp{\ltannaux{i}{\fatm{G}{M:A}}}{F''}$.
We may assume here that
$\fvalid{\assn{\Upsilon}{@^i}{m}}{\ltannaux{i}{\fatm{G}{M:A}}}$ holds
because the lemma is obviously true if it does not. 
From this, it easily follows that there is an $l$ less than $m$ for which
$\fvalid{\assn{\Upsilon}{@^i}{l}}{\eqannaux{i}{\fatm{G}{M:A}}}$ 
holds.
From the assumption about $F^{@^i}$, it follows then that
$\fvalid{\assn{\Upsilon}{@^i}{l}}{F''}$
holds.
But then, by Lemma~\ref{lem:ann_ind}, it follows that
$\fvalid{\assn{\Upsilon}{@^i}{m}}{F''}$
must also hold, from which we conclude easily that 
$\fvalid{\assn{\Upsilon}{@^i}{m}}{F^{*^i}}$ holds.

The cases in which $F^{@^i}$ begins with a context and universal term 
quantifier are similar, so we consider only the latter explicitly.
In this case, $F^{@^i}$ is a formula of the form $\fall{x:\alpha}{F'^{@^i}}$
and $F^{*^i}$ is correspondingly a formula of the form
$\fall{x:\alpha}{F'^{*^i}}$, where $F'^{*^i}$ is a formula that is
identical to $F'^{@^i}$ except that the annotation on the formula
$\fatm{G}{M:A}$ within it is changed to $*^i$.
Note that $F'^{@^i}$ and $F'^{*^i}$ must be formulas such that any of
their substitution instances taken in pairs meet the 
descriptions in the lemma statement and all the substitution instances
of $F'^{@^i}$ must be less complex than $F^{@^i}$.
Now consider a term $t$ such that $\stlctyjudg{\noms\cup\STLCGamma_0}{t}{\alpha}$.
By the definition of validity and the assumption in the lemma, 
$\fvalid{\assn{\Upsilon}{@^i}{l}}{\hsubst{\{\langle x,t,\alpha\rangle\}}{\eqannaux{i}{F'}}}$
must hold for any $l<m$.
But then, by the induction hypothesis, 
$\fvalid{\assn{\Upsilon}{@^i}{m}}{\hsubst{\{\langle x,t,\alpha\rangle\}}{\ltannaux{i}{F'}}}$
must also hold. 
Since this is true regardless of the choice of $t$, it follows from
the definition of validity that 
$\fvalid{\assn{\Upsilon}{@^i}{m}}{\fall{x:\alpha}{\ltannaux{i}{F'}}}$,
\ie $\fvalid{\assn{\Upsilon}{@^i}{m}}{\ltannaux{i}{F}}$,
must hold. 

Finally, suppose that $F^{@^i}$ is a formula of the form
$F'' \supset F'^{@^i}$ and that $F^{*^i}$ is correspondingly a formula
of the form $F'' \supset F'^{*^i}$, where $F'^{*^i}$ differs from
$F'^{@^i}$ only in that the annotation on the occurrence of the
formula $\fatm{G}{M:A}$ within it is changed to $*^i$.
Note that $F'^{@^i}$ and $F'^{*^i}$ must be formulas that, as a pair, meet the
descriptions in the lemma statement and $F'^{@^i}$ must also be less
complex than $F^{@^i}$.
We may assume in this case that $\fvalid{\assn{\Upsilon}{@^i}{m}}{F''}$
holds because the lemma has a trivial proof if it does not.
Using Lemma~\ref{lem:ann_ind}, we then see that
$\fvalid{\assn{\Upsilon}{@^i}{l}}{F''}$ must hold for any $l < m$.
From the assumption in the lemma, it then easily follows that it
must be the case that 
$\fvalid{\assn{\Upsilon}{@^i}{l}}{F'^{@^i}}$
holds for any $l < m$.
Using the induction hypothesis, we may now conclude that 
$\fvalid{\assn{\Upsilon}{@^i}{m}}{F'^{*^i}}$ holds, from which it
follows easily that $\fvalid{\assn{\Upsilon}{@^i}{m}}{F^{*^i}}$ must
hold.
\end{proof}

\begin{lemma}\label{lem:ann__unann}
Let $F$ be a closed formula of the form
\[
\mathcal{Q}_1.
       \fimp{F_1}{\fimp{\ldots}
                       {\mathcal{Q}_{k-1}.
                         \fimp{F_{k-1}}
                              {\mathcal{Q}_k.
                                   \fimp{\fatm{G}{M:A}}
                                        {\mathcal{Q}_{k+1}.
                                            \fimp{F_{k+1}}{\fimp{\ldots}{F_n}}}}}}
\]
in which the annotations $@^i$ or $*^i$ do not occur and let $F^{@^i}$
correspondingly be the formula
\[
\mathcal{Q}_1.
       \fimp{F_1}{\fimp{\ldots}
                       {\mathcal{Q}_{k-1}.
                         \fimp{F_{k-1}}
                              {\mathcal{Q}_k.
                                   \fimp{\eqannaux{i}{\fatm{G}{M:A}}}
                                        {\mathcal{Q}_{k+1}.
                                            \fimp{F_{k+1}}{\fimp{\ldots}{F_n}}}}}}
\]
If $\fvalid{\assn{\Upsilon}{@^i}{m}}{F^{@^i}}$ holds for every natural
  number $m$, then $\fvalid{\Upsilon}{F}$ must also hold. 
\end{lemma}

\begin{proof}
This lemma is also proved by an induction on the number of top-level
quantifiers and implications in $F$ prior to the occurrence of
the formula $\fatm{G}{M:A}$ in it.
The argument again implicitly uses the information that the $F$ and
$F^{@^i}$ are closed and well-formed.

In the base case, $F$ is a formula of the form
$\fimp{\fatm{G}{M:A}}{F''}$, where $F''$ is a formula that is devoid
of the annotations $@^{i}$ and $*^i$, and $F^{@^i}$ is correspondingly
the formula $\fimp{\eqannaux{i}{\fatm{G}{M:A}}}{F''}$.
If $\fvalid{\assn{\Upsilon}{@^i}{m}}{\fatm{G}{M:A}}$ does not hold for
any $m$, then $\lfchecktype{G}{M}{A}$ must not have a derivation and
hence $\fvalid{\Upsilon}{\fatm{G}{M:A}}$ must not hold. 
It is easy to see that the lemma must be true in this situation.
On the other hand, if
$\fvalid{\assn{\Upsilon}{@^i}{m}}{\fatm{G}{M:A}}$ does hold for some
$m$, then, by the assumptions of the lemma,
$\fvalid{\assn{\Upsilon}{@^i}{m}}{F''}$ must also  hold.
But then, by Lemma~\ref{lem:ann_ind}, $\fvalid{\Upsilon}{F''}$ must
hold and, hence, so also should $\fvalid{\Upsilon}{F}$.

The argument in the inductive cases follows a pattern that is very
similar to that in the proof of Lemma~\ref{lem:all_eq__lt}.
We avoid repeating the details that should be obvious at this point. 
\end{proof}

\begin{theorem}\label{th:ind-sound}
If the conclusion sequent of an instance of the \ind\ rule is
well-formed, the premise sequent is valid and
the requirement of non-occurrence of the annotations $@^i$ and $*^i$
is satisfied, then the conclusion sequent of the rule instance must be
valid. 
\end{theorem}
\begin{proof}
Let $\theta$ and $\sigma$ be substitutions that identify a closed
instance of the conclusion sequent and let $\Upsilon$ be a height
assignment.
The instance of the conclusion sequent would obviously be valid with
respect to $\Upsilon$ if any formula in
$\subst{\sigma}{\hsubst{\theta}{\Omega}}$ is not valid with respect to
$\Upsilon$, so let us assume that they are all in fact so valid. 
Our task then is to show that the instance of the goal formula in the
conclusion sequent determined by $\theta$ and $\sigma$ is valid with
respect to $\Upsilon$.
Using a harmless overloading of notation, we will identify this
instance as the formula
\[
\mathcal{Q}_1.
       \fimp{F_1}{\fimp{\ldots}
                       {\mathcal{Q}_{k-1}.
                         \fimp{F_{k-1}}
                              {\mathcal{Q}_k.
                                   \fimp{\fatm{G}{M:A}}
                                        {\mathcal{Q}_{k+1}.
                                            \fimp{F_{k+1}}{\fimp{\ldots}{F_n}}}}}},
\]
and refer to it as $F$.
We will also use the notation $\eqannaux{i}{F}$ and $\ltannaux{i}{F}$
below to denote the formulas that differ from $F$ only in that the
distinguished occurrence of the subformula $\fatm{G}{M:A}$ is
annotated with $@^{i}$ and $*^{i}$, respectively.

Our objective is to show thar $\fvalid{\Upsilon}{F}$ holds.
We will do this by showing that the judgement
$\fvalid{\assn{\Upsilon}{@^i}{m}}{\eqannaux{i}{F}}$ holds for  
every $m$ and then invoking Lemma~\ref{lem:ann__unann}.
In proving the latter, we will use the observation that 
$\theta$ and $\sigma$ identify a closed instance also of the premise 
sequent of the rule and the assumption that all instances of the
premise sequent are valid.
The argument is carried out by (strong) induction on $m$. 
Given the structure of the formula $F^{@^i}$ and the fact that no
typing derivation in LF has height $0$, it is obvious that
$\fvalid{\assn{\Upsilon}{@^i}{0}}{\eqannaux{i}{F}}$ holds.
Now assume that $\fvalid{\assn{\Upsilon}{@^i}{l}}{\eqannaux{i}{F}}$
holds for every $l$ less than $m$.
By Lemma~\ref{lem:all_eq__lt}, it follows then that
$\fvalid{\assn{\Upsilon}{@^i}{m}}{\ltannaux{i}{F}}$ holds.
Since every formula in $\subst{\sigma}{\hsubst{\theta}{\Omega}}$ is
valid with respect to $\Upsilon$ by assumption, it follows by virtue
of Lemma~\ref{lem:ann_ind} and the proviso that the annotations $@^i$
and $*^i$ does not appear in the formulas of $\Omega$ that
every assumption formula in the instance of the premise
sequent identified by $\theta$ and $\sigma$ must be true under the
height assignment $\assn{\Upsilon}{@^i}{m}$.
That $\fvalid{\assn{\Upsilon}{@^i}{m}}{\eqannaux{i}{F}}$ must hold is
now a consequence of the validity of this instance of the premise sequent. 
\end{proof}

\subsubsection{Additional Proof Rules that Interpret Annotations}

\begin{figure}[htb]

\begin{center}
\begin{tabular}{c}

\infer[\annappL]
      {\seq[\mathbb{N}]{\Psi}
           {\Xi}
           {\setand{\Omega}{\fatmann{G}{R:P}{Ann}}}
           {F}}
      {\begin{array}{c}
         \mathcal{CS} = 
               \casesfn[\fatm{G}{R:P}]
                       {\seq[\mathbb{N}]
                            {\Psi}
                            {\Xi}
                            {\setand{\Omega}{\fatm{G}{R:P}}}
                            {F}}
         \\
         \begin{array}{l} 
                \{\seq[\mathbb{N}']
                      {\Psi'}
                      {\Xi'}
                      {\setand{\Omega'}
                            {\fatmann{G_1}{M_1:A_1}{*^i},\ldots,\fatmann{G_k}{M_k:A_k}{*^i}}}
                      {F'} 
          \\ 
           \qquad\qquad
                  \mid\ 
                  \seq[\mathbb{N}']
                      {\Psi'}{\Xi'}
                      {\setand{\Omega'}{\fatm{G_1}{M_1:A_1},\ldots,\fatm{G_k}{M_k:A_k}}}
                      {F'} \in \mathcal{CS}
                  \}
         \end{array}    
       \end{array}}

\\[12pt] 

\infer[\annappR]
      {\seq[\mathbb{N}]{\Psi}
           {\Xi}
           {\Omega}
           {\fatmann{G}{h\app M_1\ldots M_n:P'}{@^i}}}
      {\begin{array}{c}
           h:\typedpi{x_1}{A_1}{\ldots \typedpi{x_n}{A_n}{P}}
                      \in\Sigma\mbox{ or the explicit bindings in }G
             \\
           G\ \mbox{appears as the context of some atomic formula in}\ \Omega\\
           \hsub{\{\langle x_1, M_1, \erase{A_1}\rangle,\ldots,\langle x_n, M_n, \erase{A_n}\rangle\}}
                {P}
                {P'}
             \\
         \begin{array}{l}
            \{\seq[\mathbb{N}]
                  {\Psi}
                  {\Xi}
                  {\Omega}
                  {}
            \\
            \qquad
            \fatmann{G}
                    {M_i:
                      \hsubst{\{\langle x_1, M_1, \erase{A_1}\rangle,\ldots,
                                \langle x_{i-1}, M_{i-1}, \erase{A_{i-1}}\rangle\}}
                             {A_i}}
                    {*^i}
               \ \mid\ 1\leq i\leq n \}
         \end{array}
       \end{array} 
      } 

\\[12pt]

\infer[\annabsL]
      {\seq[\mathbb{N}]
           {\Psi}
           {\Xi}
           {\setand{\Omega}
                   {\fatmann{G}{\lflam{x}{M}:\typedpi{x}{A_1}{A_2}}{Ann}}}
           {F}}
      {\begin{array}{c}
           n\ \mbox{is new to}\ \mathbb{N} \mbox{ and has type } \erase{A_1}
         \\
         \Xi' =
           \begin{cases}
             \left(\Xi \setminus
             \left\{\ctxvarty{\Gamma}
                             {\mathbb{N}_{\Gamma}}
                             {\ctxty{\mathcal{C}}
                             {\mathcal{G}}}\right\}\right)
             \cup 
             \left\{\ctxvarty{\Gamma}
                             {(\mathbb{N}_{\Gamma},n:\erase{A_1})}
                             {\ctxty{\mathcal{C}}{\mathcal{G}}}\right\}
               & \mbox{if }\Gamma\mbox{ appears in }G \\
             \Xi & \mbox{otherwise}
           \end{cases}
         \\
         \begin{array}{l}
           \seqsansfmlas{\mathbb{N},n:\erase{A_1}}
                        {\Psi}
                        {\Xi'}
           \\
           \qquad
           \seqsansctxts{\setand{\Omega}
                                {\fatmann{G,n:A_1}
                                         {\hsubst{\{\langle x,n,\erase{A_1}\rangle\}}{M}:
                                           \hsubst{\{\langle x,n,\erase{A_1}\rangle\}}{A_2}}
                                {*^i}}}
                        {F}
         \end{array}
       \end{array}}

\\[12pt]

\infer[\absR]
      {\seq[\mathbb{N}]
           {\Psi}
           {\Xi}
           {\Omega}{\fatmann{G}
                            {\lflam{x}{M}:\typedpi{x}{A_1}{A_2}}
                            {@^i}}}
      {\begin{array}{c}
           n\ \mbox{is new to}\ \mathbb{N} \mbox{ and has type } \erase{A_1}
         \\
         \Xi' =
           \begin{cases}
             \left(\Xi \setminus
             \left\{\ctxvarty{\Gamma}
                             {\mathbb{N}_{\Gamma}}
                             {\ctxty{\mathcal{C}}
                             {\mathcal{G}}}\right\}\right)
             \cup 
             \left\{\ctxvarty{\Gamma}
                             {(\mathbb{N}_{\Gamma},n:\erase{A_1})}
                             {\ctxty{\mathcal{C}}{\mathcal{G}}}\right\}
               & \mbox{if }\Gamma\mbox{ appears in }G \\
             \Xi & \mbox{otherwise}
           \end{cases}
         \\
         \begin{array}{l}
           \seqsansfmlas{\mathbb{N},n:\erase{A_1}}
                        {\Psi}
                        {\Xi'}
           \\
           \qquad             
           \seqsansctxts{\Omega}
                        {\fatmann{G,n:A_1}
                            {\hsubst{\{\langle x,n,\erase{A_1}\rangle\}}{M}:
                                 \hsubst{\{\langle x,n,\erase{A_1}\rangle\}}{A_2}}
                            {*^i}}
         \end{array}
       \end{array}}
\end{tabular}
\end{center}

\caption{Rules that Introduce Atomic Formulas with Annotations}\label{fig:ann-atomic}
\end{figure}

With the addition of annotations to the syntax, we may consider 
additional rules for introducing annotated atomic formulas.
The form that these rules take must pay attention to the heights of
the LF derivations that these formulas encode; this is essential to
being able to make effective use of the induction rule.
In this spirit, our proof system includes the variants of the existing
rules for atomic formulas that are shown in
Figure~\ref{fig:ann-atomic}.
In the $\annappL$ and $\annabsL$ rules, $Ann$ must be either $@^i$ or
$*^i$.
In the $\annappL$ rule, $\Omega'$ denotes the formulas that result from
$\Omega$ and $\fatm{G_1}{M_1:A_1}, \ldots, \fatm{G_k}{M_k:A_k}$
denote the formulas that result from $\fatm{G}{R:P}$ in the reduction
of each of the instances of the conclusion sequent that are considered
by \casessans.
Note that the $@^i$ annotation on the assumption formula that is the
focus of the rule is replaced by the $*^i$ annotation in the
assumption formulas in the premise sequent, as is needed for the
hypothesis added by the induction rule to be useful. 

The \id\ rule described in Section~\ref{sssec:initial-rules} requires
the assumption and conclusion formulas that are matched to be
equivalent in terms of validity.
With the inclusion of annotations, it is possible to weaken this
requirement.
For example, the assumption formula $\ltann{\fatm{G}{M:A}}$ ensures
the validity of a sequent in which the conclusion formula is 
$\eqann{\fatm{G}{M:A}}$.
Similarly, the presence of either of these formulas as an assumption
in a sequent suffices to ensure its validity if its conclusion formula
is $\fatm{G}{M:A}$.
These observations can be expanded to include non-atomic formulas with
the proviso that the polarity of the occurrence of the formula must be
paid attention to.
For example, it is the validity of $\fimp{\eqann{\fatm{G}{M:A}}}{F}$
that implies that $\fimp{\ltann{\fatm{G}{M:A}}}{F}$ is valid, and the
validity of both of these formulas is implied by the validity of
$\fimp{\fatm{G}{M:A}}{F}$.  
We make these thought precise through a definition of comparative
strengths of formulas.
\begin{definition}\label{def:fmla-strength}
We first identify an ordering over annotations: $*^i$ is stronger than
$@^i$ which, in turn, is stronger than no annotation.
Next we define a strength relation between atomic formulas relative 
to a context variables context $\Xi$ and a permutation $\pi$ of
nominal constants: $\fatmann{G}{M:A}{Ann}$ is at least as strong as 
$\fatmann{G'}{M':A'}{Ann'}$ in this context if 
$\formeq{\Xi}{\pi}{\fatm{G}{M:A}}{\fatm{G'}{M':A'}}$ holds
and $Ann$ is stronger than or identical to $Ann'$. 
Finally, we extend this strength relation to arbitrary formulas.
The formula $\fimp{F_2}{F_2'}$ is at least as strong as $\fimp{F_1}{F_1'}$
with respect to $\Xi$ and $\pi$ if $F_1$ is at least as strong as
$F_2$ with respect to $\Xi$ and $\inv{\pi}$ and $F_2'$ is at least as
strong as $F_1'$ with respect to $\Xi$ and $\pi$.
For all other non-atomic formulas, $F_2$ is at least as strong as
$F_1$ with respect to $\Xi$ and  $\pi$ if their components satisfy the
same relation, under a possibly extended $\Xi$ in  the case of context 
quantification, allowing for renaming of variables bound by
quantifiers. 
The ``at least as strong as'' relation for formulas relative to $\Xi$
and $\pi$ is represented by the judgement
$\streq{\Xi}{\pi}{F_2}{F_1}$. 
\end{definition}

The strength relation between formulas is preserved by term and
context variables substitutions under some provisos.
\begin{lemma}\label{lem:streq-subs}
Suppose that $\streq{\Xi}{\pi}{F_2}{F_1}$ holds for suitable $F_1$,
$F_2$, $\Xi$ and $\pi$. Then,
\begin{enumerate}
\item if $\theta$ is a term substitution such that
$\supp{\theta}\cap\supp{\pi}=\emptyset$ and
both $\hsub{\theta}{F_2}{F_2'}$ and $\hsub{\theta}{F_1}{F_1'}$ 
have derivations for some $F_1'$ and $F_2'$, then
$\streq{\hsubst{\theta}{\Xi}}{\pi}{\hsubst{\theta}{F_2}}{\hsubst{\theta}{F_1}}$
holds; and

\item if $\sigma$ is an appropriate substitution for $\Xi$ with respect to some 
term variables context $\Psi$ and set of nominal constants
$\mathbb{N}$ then
$\streq{\ctxvarminus{\Xi}{\sigma}}{\pi}{\subst{\sigma}{F_2}}{\subst{\sigma}{F_1}}$ 
holds.
\end{enumerate}
\end{lemma}
\begin{proof}
Both clauses can be shown to be true by an inductive argument on the
structure of $F_2$, using the definition of substitution and invoking
Lemmas~\ref{lem:equiv-hsub} and \ref{lem:equiv-sub} in the case of
atomic formulas. 
\end{proof}

The main property for the strength definition that we will use is that
validity of a closed formula is ensured by that of another closed
formula that is at least as strong as it is. 
\begin{lemma}\label{lem:streq-valid}
Let $F_1$ and $F_2$ be closed well-formed formulas, let $\Xi$ be a 
context variables context and let $\pi$ be a permutation of nominal
constants such that $\streq{\Xi}{\pi}{F_2}{F_1}$ holds.
If $F_2$ is valid with respect to a height assignment $\Upsilon$, then
$F_1$ must also be valid with respect to $\Upsilon$.
\end{lemma}

\begin{proof}
The proof by an induction on the structure of $F_2$ using the
definition of $\streq{\Xi}{\pi}{F_2}{F_1}$.
The argument takes an obvious form in the inductive cases and we
consider only the case for atomic formulas, where the annotations
become important, in detail.
In this case, $F_2$ and $F_1$ must, respectively, be of the form
$\fatmann{G_2}{M_2:A_2}{Ann_2}$ and $\fatmann{G_1}{M_1:A_1}{Ann_1}$,
where $Ann_2$ is an annotation that is stronger than or identical to
$Ann_1$, and
$\formeq{\Xi}{\pi}{\fatm{G_2}{M_2:A_2}}{\fatm{G_1}{M_1:A_1}}$ must hold.
Since the formulas are closed, the last observation leads to the
conclusion that $\permute{\pi}{(\fatm{G_2}{M_2:A_2})}$ and 
$\fatm{G_1}{M_1:A_1}$ are identical.
Using Theorem~\ref{th:perm-lf}, it follows from this that if there are
derivations for $\lfctx{G_2}$, $\lftype{G_2}{A_2}$ and
$\lfchecktype{G_2}{M_2}{A_2}$, then there must be derivations that
have the \emph{same structure} for $\lfctx{G_1}$, $\lftype{G_1}{A_1}$
and $\lfchecktype{G_1}{M_1}{A_1}$, respectively.
The claim in the lemma can now be verified using the definition of
validity.
\end{proof}

\begin{figure}
\[
\infer[\id]
      {\seq[\mathbb{N}]{\Psi}{\Xi}{\Omega}{F_1}}
      {F_2\in\Omega
         &
       \pi\ \mbox{is a permutation of nominal constants such that}\ \supportof{\pi}\subseteq\mathbb{N}
         &
       \streq{\Xi}{\pi}{F_2}{F_1}}
\]
\caption{The Axiom Rule Generalized to Annotated Formulas}\label{fig:ann-id}
\end{figure}

Figure~\ref{fig:ann-id} presents a generalized version of the
\id\ rule based on the discussions above.
We show the preservation of wellformedness and the soundness
properties for the new rules below.

\begin{theorem}\label{th:ann-sound}
The following properties hold for every instance of the rules
in Figures~\ref{fig:ann-atomic} and~\ref{fig:ann-id}:

\begin{enumerate}
\item If the conclusion sequent is
well-formed, the premises expressing typing conditions have
derivations and the conditions expressed by the other, non-sequent
premises are satisfied, then all the sequent premises must
be well-formed.

\item If the conclusion sequent is well-formed, the premises expressing
typing judgements are derivable, the conditions described in the other
non-sequent premises are satisfied and the premise sequent is valid,
then the conclusion sequent must also be valid. 
\end{enumerate}
\end{theorem}

\begin{proof}
The truth of the first clause can be verified by using
Theorems~\ref{th:other-wf} and \ref{th:atom-wf} and noting that the
wellformedness of sequents with annotations is determined by the
wellformedness of the sequent that results from erasing all the
annotations.

The truth of the second clause for the new \id\ rule follows easily
from Lemmas~\ref{lem:streq-subs} and \ref{lem:streq-valid} and the
definition of validity for sequents.
For the rules in Figure~\ref{fig:ann-atomic}, we adapt the argument
for Theorem~\ref{th:atom-sound} to the new semantics.
The refinement to the semantics to take into account the heights of
derivations for typing judgements associated with atomic formulas does
not impact the consideration of any formula other than the one
introduced by the rule and the ones in the premise sequents associated
with it.
For the formula introduced by the $\annappL$ rule, we use the fact
verified in Lemma~\ref{lem:decomp-decr} that if the typing 
judgement that determines the height of a closed instance of the formula
$\fatm{G}{R : P}$ has a derivation of height $k$, then the typing judgements
that determine the heights of the corresponding closed instances of the new
assumption formulas in the reduced form of the sequent must have
derivations of height less than $k$.
For the formula introduced by the $\annappR$ rule, we use the
observation contained in Theorem~\ref{th:atomictype} that the typing
judgement that determines the height of a closed instance of the formula
$\fatm{G}{h\app M_1 \app \ldots \app M_n : P'}$ must have a derivation
of height at most $(k + 1)$ if the typing judgement associated with the
corresponding instances of the goal formulas of the premise sequents
have height at most $k$.
Finally, for the formula introduced by the $\annabsL$ and $\annabsR$
rules, we use the easily verified fact that an LF judgement of the
form $\lfchecktype{\Gamma}{\lflam{x}{M}}{\typedpi{x}{A_1}{A_2}}$ 
has a derivation of height $(k+1)$ if and only if the
judgement $\lfchecktype{\Gamma, x:A_1}{M}{A_2}$ has a derivation of
height $k$.
\end{proof}

\subsection{Proof Rules Encoding LF Meta-Theorems}
\label{ssec:meta-theorems}

We now present proof rules that encode the meta-theorems discussed in 
Section~\ref{ssec:lf-properties} are often used in informal
arguments about the properties of LF specifications.
The strengthening, permutation and instantiation meta-theorems can be
encoded as axioms, \ie, as proof rules that have no premise sequents
although they may have side-conditions that determine applicability.
Given the interpretation of atomic formulas, the weakening rule, which
introduces a new type assignment into the LF context, requires the
type to be checked for wellformedness.
This checking is built into the proof rule that encodes it via
suitable premise sequents. 

To determine the premise sequents for the weakening rule, we use the 
process of ``type decomposition'' that is embodied in the following
definition.

\begin{definition}\label{def:typdecomp}
The decomposition of an LF type $A$ with respect to a
collection of nominal constants $\mathbb{N}$, a context variables
context $\Xi$ and a context expression $G$, denoted by
$\typdecomp{\mathbb{N}}{\Xi}{G}{A}$, is defined by recursion on the
type $A$ as follows: 
\begin{enumerate}
\item If $A$ is a type of the form $(a\app M_1\ldots M_n)$
  where $a:\typedpi{x_1}{A_1}{\ldots\typedpi{x_n}{A_n}{\type}}
  \in\Sigma$, then $\typdecomp{\mathbb{N}}{\Xi}{G}{A}$ is the
  collection
\[\left\{
     \langle
      \mathbb{N},
      \Xi,
      \fatm{G}
           {M_i:\hsubst{\{\langle x_1,M_1,\erase{A_1}\rangle,\ldots,
                           \langle x_{i-1},M_{i-1},\erase{A_{i-1}}}
                       {A_i}}
    \rangle\ \vert\ 1 \leq i \leq n \right\}.\]

\item If $A=\typedpi{x}{A_1}{A_2}$, then letting 
$G'$ be $G,n:A_1$, $\mathbb{N}'$ be
  $\mathbb{N} \cup \{ n \}$, and $\Xi'$ be the set 
  \[\left\{\ctxvarty{\Gamma}
                      {\mathbb{N}'_{\Gamma}}
                      {\ctxty{\mathcal{C}}
                        {\mathcal{G}}}\ 
             \middle|\ 
             \ctxvarty{\Gamma}{\mathbb{N}_{\Gamma}}{\ctxty{\mathcal{C}}{\mathcal{G}}}\in\Xi
                          \mbox{ and }
               \mathbb{N}'_{\Gamma}\ \mbox{is}\ \mathbb{N}_{\Gamma} \cup \{n\}\
                          \mbox{ if }\Gamma\ \mbox{occurs in}\ G\ \mbox{and}\
                          \mathbb{N}_{\Gamma}\ \mbox{ otherwise}
         \right\}\]
for some nominal constant $n:\erase{A_1} \in \noms \setminus \mathbb{N}$, 
$\typdecomp{\mathbb{N}}{\Xi}{G}{A}$ is the collection
$\typdecomp{\mathbb{N}}{\Xi}{G}{A_1}\cup
    \typdecomp{\mathbb{N}'}{\Xi'}{G'}{\hsubst{\{\langle x,
        n,\erase{A_1}\rangle\}}{A_2}}.$
\end{enumerate}
We shall refer to the set
$\bigcup \{ \mathbb{N}'\ \vert\ \langle \mathbb{N}',\Xi',F'\rangle \in
                                        \typdecomp{\mathbb{N}}{\Xi}{G}{A}
                                        \}
\setminus \mathbb{N}$ as the \emph{new name set} of
$\typdecomp{\mathbb{N}}{\Xi}{G}{A}$; we circumvent the ambiguity in
this identification by assuming that new names are chosen in the same
way each time.
\end{definition}
\noindent The definition above is abusive of notation because it uses
the convention discussed after Theorem~\ref{th:aritysubs} for denoting
the result of applying a substitution without the certainty that such
a result will exist.
However, this abuse is harmless: the result will exist in all the
situations that we will use \typdecompsans.
The following lemma describes the conditions under which we will use
this function and shows the harmlessness of the abuse in addition to
establishing additional properties of \typdecompsans\ that we will
need in justifying its use.

\begin{lemma}\label{lem:typdecomp}
Let $\Xi$ be a context variables context, $\Psi$ a term
variables context, and $\mathbb{N}$ a set
of nominal constants such that for each
$\ctxvarty{\Gamma_i}{\mathbb{N}_i}{\ctxty{\mathcal{C}_i}{\mathcal{G}_i}}$
in $\Xi$ the judgement
$\wfctxvarty{\mathbb{N} \setminus \mathbb{N}_i}
            {\Psi}
            {\ctxty{\mathcal{C}_i}{\mathcal{G}_i}}$
has a derivation.
Further, let $G$ be a context expression and $B$ be a type such that
there are derivations for
$\wfctx{\mathbb{N}\cup\STLCGamma_0\cup\Psi}{\ctxsanstype{\Xi}}{G}$ and
$\wftype{\mathbb{N}\cup\STLCGamma_0\cup\Psi}{B}$. 
Then $\typdecomp{\mathbb{N}}{\Xi}{G}{B}$ must be defined.
Moreover, letting  $\overline{\mathbb{N}}$ be the new name set of
$\typdecomp{\mathbb{N}}{\Xi}{G}{B}$, the following properties hold of
the set $\typdecomp{\mathbb{N}}{\Xi}{G}{B}$ identifies: 
\begin{enumerate}
\item For each 
$\langle\mathbb{N}',\Xi',F'\rangle\in\typdecomp{\mathbb{N}}{\Xi}{G}{B}$ it is the
  case that 
\begin{enumerate}

\item
  $\ctxvarty{\Gamma_i}{\mathbb{N}_i}{\ctxty{\mathcal{C}_i}{\mathcal{G}_i}}\in
  \Xi$ if and only if 
  $\ctxvarty{\Gamma_i}{\mathbb{N}_i'}{\ctxty{\mathcal{C}_i}{\mathcal{G}_i}}
  \in \Xi'$ for some $\mathbb{N}'_i$ such  
  $\mathbb{N}'\setminus\mathbb{N} \supseteq \mathbb{N}'_i\setminus \mathbb{N}_i$,
  
\item each member
  $\ctxvarty{\Gamma_i}{\mathbb{N}_i'}{\ctxty{\mathcal{C}_i}{\mathcal{G}_i}}$
  of $\Xi'$ is such that 
$\wfctxvarty{\mathbb{N}' \setminus \mathbb{N}_i'}
            {\Psi}
            {\ctxty{\mathcal{C}_i}{\mathcal{G}_i}}$
has a derivation,

\item there is a derivation for
  $\wfform{\mathbb{N}'\cup\STLCGamma_0\cup\Psi}{\ctxsanstype{\Xi'}}{F'}$.
\end{enumerate}

\item For any closed term substitution $\theta$ with $\domain{\theta}
  = \Psi$ and closed context substitution $\sigma$ with
  $\domain{\sigma} = \Xi$ and $\supportof{\sigma}$ disjoint from
  $\overline{\mathbb{N}}$, if it is the case that
  $\subst{\sigma}{\hsubst{\theta}{G}}$ and $\hsubst{\theta}{B}$ are
  defined and $\lfctx{\subst{\sigma}{\hsubst{\theta}{G}}}$ has a
  derivation, and, further, that $\subst{\sigma}{\hsubst{\theta}{F'}}$
  is defined and is a valid formula for every 
  $\langle \mathbb{N}',\Xi',F'\rangle \in \typdecomp{\mathbb{N}}{\Xi}{G}{B}$,
  then there must be a derivation for the LF judgement
  $\lftype{\subst{\sigma}{\hsubst{\theta}{G}}}{\hsubst{\theta}{B}}$.
\end{enumerate}
\end{lemma}

\begin{proof}
We must first determine the coherence of the lemma by checking that
$\subst{\sigma}{\hsubst{\theta}{G}}$, $\hsubst{\theta}{B}$ and
$\subst{\sigma}{\hsubst{\theta}{F'}}$ are closed expressions for each
of the substitutions $\theta$, $\sigma$ and formulas $F'$ considered
in clause (3). 
This must be the case because the substitutions are closed,
all the free variables in $B$ are contained in $\Psi$, and all the
free variables in $F'$ and $G$ must be contained in $\Psi$ and $\Xi$.

The lemma is now proved by induction on the structure of the type $B$.
We consider the cases for this structure below.

\medskip
\noindent {\it $B$ is the atomic type $(a\app M_1 \app \ldots \app M_n)$.}

\smallskip
\noindent Since $\wftype{\mathbb{N}\cup\STLCGamma_0\cup\Psi}{B}$ is
derivable, $\Sigma$ must assign a kind of the form
$\typedpi{x_1}{A_1}{\ldots\typedpi{x_n}{A_n}{\type}}$ to $a$ 
and there must be derivations for
$\stlctyjudg{\mathbb{N}\cup\STLCGamma_0\cup\Psi}{M_i}{\erase{A_i}}$
for $1 \leq i \leq n$.
From the wellformedness of $\Sigma$ it follows that
$\wftype{\mathbb{N}\cup\STLCGamma_0\cup\Psi\cup
                  \{x_1:\erase{A_1},\ldots, x_{i-1}:\erase{A_{i-1}}\}}
        {A_i}$
has a derivation.
We may now use Theorem~\ref{th:aritysubs-ty} to conclude that
$\hsubst{\{\langle x_1,M_1,\erase{A_1}\rangle,\ldots,
              \langle x_{i-1}, M_{i-1}, \erase{A_{i-1}}\rangle \}}
        {A_i}$ 
is defined, thereby ensuring that $\typdecomp{\mathbb{N}}{\Xi}{G}{B}$
is defined in this case.

We now turn to showing the various clauses.
Clauses 1(a) is trivial and clause 1(b) follows from the assumptions
of the lemma. 
For the remainder of the argument, let us write
$A'_i$ to denote the type
$\hsubst{\{\langle x_1,M_1,\erase{A_1}\rangle,\ldots,
              \langle x_{i-1}, M_{i-1}, \erase{A_{i-1}}\rangle \}}
        {A_i}$, for $1 \leq i \leq n$.
Theorem~\ref{th:aritysubs-ty} yields the further observation
that there are derivations for
$\wftype{\mathbb{N}\cup\STLCGamma_0\cup\Psi}{A'_i}$.
We also note that $\typdecomp{\mathbb{N}}{\Xi}{G}{B}$ is the set
$\{\langle \mathbb{N},\Xi,\fatm{G}{M_i : A'_i} \rangle
       \ \vert\ 1 \leq i \leq n \}$.
Now, for clause 1(c), it suffices to show that
$\wfctx{\mathbb{N}\cup\STLCGamma_0\cup\Psi}{\ctxsanstype{\Xi}}{G}$,
$\wftype{\mathbb{N}\cup\STLCGamma_0\cup\Psi}{A'_i}$,
and 
$\stlctyjudg{\mathbb{N}\cup\STLCGamma_0\cup\Psi}{M_i}{\erase{A'_i}}$
have derivations.
The first is true by assumption and the latter two follow from earlier
observations and the fact that erasure is preserved under
substitution~(Theorem~\ref{th:erasure}).

This leaves only clause (2).
The validity of
$\fatm{\subst{\sigma}{\hsubst{\theta}{G}}}
      {\hsubst{\theta}{M_i} : \hsubst{\theta}{A'_i}}$,
for $1 \leq i \leq n$, implies that 
there must be LF derivations for 
$\lfchecktype{\subst{\sigma}{\hsubst{\theta}{G}}}
             {\hsubst{\theta}{M_i}}{\hsubst{\theta}{A'_i}}$.
Using Theorem~\ref{th:subspermute}, we see that
$\hsubst{\theta}{A'_i}$ is identical to
$\hsubst{\{ \langle x_1, \hsubst{\theta}{M_1}, \erase{A_1} \rangle, \ldots, 
             \langle x_{i-1},\hsubst{\theta}{M_{i-1}},\erase{A_{i-1}}\rangle\}}
        {\hsubst{\theta}{A_i}}$;
we assume here that the variables $x_1,\ldots,x_n$ are chosen so as to not
be members of $\domain{\theta}$ or to appear free in the terms in
$\range{\theta}$, as is possible to ensure through renaming.
Using Theorems~\ref{th:atomickind} and \ref{th:erasure}, we can
conclude now that there must be a derivation for
$\lfsynthkind{\subst{\sigma}{\hsubst{\theta}{G}}}
             {(a\app \hsubst{\theta}{M_1} \app \ldots \app \hsubst{\theta}{M_n})}
             {\type}$.
The claim follows easily from this by noting that $\hsubst{\theta}{B}$
is identical to
$(a\app \hsubst{\theta}{M_1} \app \ldots \app \hsubst{\theta}{M_n})$.

\medskip
\noindent {\it $B$ is the type $\typedpi{x}{A_1}{A_2}$.}

\smallskip
\noindent We first show that the parameters of the two recursive uses
of $\typdecompsans$ satisfy what is required of them for the lemma
statement to apply.
The assumption that
$\wftype{\mathbb{N}\cup\STLCGamma_0\cup\Psi}{\typedpi{x}{A_1}{A_2}}$
has a derivation implies that 
$\wftype{\mathbb{N}\cup\STLCGamma_0\cup\Psi}{A_1}$
and
$\wftype{\mathbb{N}\cup\STLCGamma_0\cup\Psi\cup \{x: \erase{A_1}\}}{A_2}$
have derivations.
From the latter, using Theorem~\ref{th:aritysubs-ty}, we see that
$\hsubst{\{\langle x, n, \erase{A_1}\rangle \}}{A_2}$ must be defined
and that         
$\wftype{\mathbb{N}\cup\STLCGamma_0\cup\Psi}
        {\hsubst{\{\langle x, n, \erase{A_1}\rangle \}}{A_2}}$.
must have a derivation for any nominal constant $n$ that is assigned
the type $\erase{A_1}$.
Given these observations, it is straightforward to verify that the
required conditions are indeed satisfied.

The induction hypothesis now yields the conclusion that
$\typdecomp{\mathbb{N}}{\Xi}{G}{B}$
is defined.
A similar observation applies to clauses 1(a), 1(b) and 1(c).
For clause 2, we observe first that substitutions $\theta$ and
$\sigma$ are ``good'' for this case only if they are good for
the recursive calls; hence we may use them without specific reference
to the context.
We then observe that if $\hsubst{\theta}{B}$ is defined, then so must
$\hsubst{\theta}{A_1}$ and $\hsubst{\theta}{A_2}$ be.
Now, if can show that, for the $n$ chosen in the
definition of $\typdecompsans$, it is the case that 
$\subst{\sigma}{\hsubst{\theta}{(G, n: A_1)}}$ is defined and that
$\lfctx{\subst{\sigma}{\hsubst{\theta}{(G,n:A_1)}}}$ has a derivation,
we can invoke the induction hypothesis relative to the second
recursive call to $\typdecompsans$ to conclude that 
$\lftype{\subst{\sigma}{\hsubst{\theta}{(G,n:A_1)}}}{\hsubst{\theta}{A_2}}$
must have a derivation and, hence, that
$\lftype{\subst{\sigma}{\hsubst{\theta}{G}}}{\hsubst{\theta}{\typedpi{x:A_1}{A_2}}}$
has one, as we desired to show.
That $\subst{\sigma}{\hsubst{\theta}{(G, n: A_1)}}$ is defined follows
from the assumption that $\subst{\sigma}{\hsubst{\theta}{G}}$ defined
and the earlier observation that $\hsubst{\theta}{A_1}$ is.
The induction hypothesis invoked relative to the first recursive call
to $\typdecompsans$ allows us to conclude that there is a
derivation for
$\lftype{\subst{\sigma}{\hsubst{\theta}{G}}}{\hsubst{\theta}{A_1}}$.
By assumption, $\lfctx{\subst{\sigma}{\hsubst{\theta}{G}}}$ has a
derivation.
Noting that $\subst{\sigma}{\hsubst{\theta}{(G,n:A_1)}}$ is the same
expression as
$(\subst{\sigma}{\hsubst{\theta}{G}},n:\hsubst{\theta}{A_1})$, we see
that $\lfctx{\subst{\sigma}{\hsubst{\theta}{(G,n:A_1)}}}$
would have a derivation if we could be sure that $n$ is not assigned a
type in $\subst{\sigma}{\hsubst{\theta}{G}}$.
But this must be the case: the way the $n$ is chosen ensures that
$G$ cannot have such a binding and the condition that
$\supportof{\sigma}$ is disjoint from
$\overline{\mathbb{N}}$ ensures that $\sigma$
cannot introduce one.
\end{proof}

\begin{figure}[tbh]

\begin{center}
\begin{tabular}{c}   

\infer[\lfwk]
      {\seq[\mathbb{N}]
           {\Psi}{\Xi}
           {\Omega}
           {\fimp{\fatmann{G}{M:A}{Ann}}{\fatmann{G,n:B}{M:A}{Ann}}}}
      {\begin{array}{c}
           n\mbox{ does not appear in } M,\, A,\, \mbox{or the explicit bindings in } G
         \\[3pt]
         \mbox{ if }\ctxvarty{\Gamma_i}
                             {\mathbb{N}_i}
                             {\ctxty{\mathcal{C}_i}{\mathcal{G}_i}}
                \in\Xi\mbox{ and }\Gamma_i\mbox{ appears in }G,\, \mbox{then }n\in\mathbb{N}_i
         \\[3pt]
         \left\{
         \seq[\mathbb{N}']{\Psi}{\Xi'}{\Omega}{F'}\ |\ 
                (\mathbb{N}',\Xi',F')\in\typdecomp{\mathbb{N}}{\Xi}{G}{B}
         \right\}
       \end{array}}

\\[15pt]

\infer[\lfstr]
      {\seq[\mathbb{N}]
           {\Psi}{\Xi}
           {\Omega}
           {\fimp{\fatmann{G,n:B}{M:A}{Ann}}{\fatmann{G}{M:A}{Ann}}}}
      {n\mbox{ does not appear in }M\mbox{, } A\mbox{, or the explicit bindings in }G}

\\[15pt]
      
\infer[\lfperm]
      {\seq[\mathbb{N}]{\Psi}
           {\Xi}
           {\Omega}
           {\fimp{\fatmann{G}{M:A}{Ann}}
                 {\fatmann{G'}{M:A}{Ann}}}}
      {\begin{array}{c}
          G = G_1,n_1:A_1,n_2:A_2,G_2
          \\[3pt]
          G' = G_1,n_2:A_2,n_1:A_1,G_2
          \\[3pt]
          n_1\mbox{ does not appear in }A_2
       \end{array}}

\\[15pt]

\infer[\lfinst]
      {\seq[\mathbb{N}]
           {\Psi}
           {\Xi}
           {\Omega}
           {\fimp{\fatm{G,G'}{M:A}}
                 {\fimp{\fatm{G}{N:B}}
                       {\fatm{G''}{M':A'}}}}}
      {\begin{array}{c}
           G' = n:B,n_1:A_1,\ldots,n_m:A_m
           \qquad
           \left\{\hsub{\{\langle n, N, \erase{B}\rangle\}}{A_i}{A_i'}
                \ \mid\ 1\leq i\leq m\right\}
           \\[3pt]      
           G''=G,n_1:A_1',\ldots,n_m:A_m'
           \\[3pt]
           \hsub{\{\langle n, N, \erase{B}\rangle\}}{M}{M'}
             \qquad
           \hsub{\{\langle n, N, \erase{B}\rangle\}}{A}{A'}
       \end{array}}
\end{tabular}
\end{center}
\caption{Rules that Encode Meta-Theoretic Properties of LF}\label{fig:rules-meta}
\end{figure}
  
Figure~\ref{fig:rules-meta} presents the proof rules which encode the
content of the LF meta-theorems.
The symbol $Ann$ in the first three rules, which encode weakening,
strengthening, and context permutation, stands for no annotation,
$@^i$ or $*^i$ for some $i$, used in the same manner throughout the
rule instance.
Annotations are permitted in these rules to reflect the fact that
the corresponding meta-theorems guarantee the preservation of the
structure, and thus height, of LF derivations.

We show that these rules are well-defined by demonstrating that they
satisfy the preservation of wellformedness property for sequents and
that they are sound.
\begin{theorem}\label{th:meta-wf}
The following property holds for every instance of the rules in
Figure~\ref{fig:rules-meta}: if the conclusion sequent is well-formed,
the premises expressing typing conditions have derivations and the
conditions expressed by the other, non-sequent premises are satisfied,
then the premise sequents must be well-formed.
\end{theorem}
\begin{proof}
The theorem holds vacuously for the rules \lfstr, \lfperm, and
\lfinst. For \lfwk, Lemma~\ref{lem:typdecomp} yields the conditions
that are required of $\Xi'$ and $F'$ for each of the premise sequents,
and the wellformedness of the conclusion sequent yields the condition
required of the formulas in $\Omega$.
\end{proof}

\begin{theorem}\label{th:meta-sound}
The following property holds for every instance of each of the rules
in Figure~\ref{fig:rules-meta}:
if the premises expressing typing judgements are derivable, the
conditions described in the other non-sequent premises are satisfied
and all the premise sequents are valid, 
then the conclusion sequent must also be valid.
\end{theorem}

\begin{proof}
To establish the theorem, we must show the validity of each closed
instance of the conclusion sequent under the conditions described.
In the case of the \lfwk, \lfstr\ and \lfperm, the notion of validity
must also be relativized to a height assignment.
Since closed instances of atomic formulas are interpreted via the
derivability of LF judgements, we may use the meta-theorems presented
in Section~\ref{ssec:lf-properties} in this argument.
The details are entirely straightforward in the case of the rules
\lfstr, \lfperm and \lfinst, and we therefore avoid their presentation
here.
The argument is more involved for the rule \lfwk: the interpretation
of atomic formulas includes the wellformedness of contexts and types
and we must show that the premises of this rule suffice to ensure this
holds for the extension of the context embodied in the consequent of
the goal formula.
We therefore present the proof for this case in more detail below. 

Let $\theta$ and $\sigma$ be substitutions that identify a closed
instance of the conclusion sequent of this rule.
As noted above, we must show that this instance is valid with respect
to any height assignment $\Upsilon$ if the assumptions of the lemma
are satisfied.  
Let $\overline{\mathbb{N}}$ be the new name set of
$\typdecomp{\mathbb{N}}{\Xi}{G}{B}$.
We show what is required initially under the assumption that
$\supportof{\theta}$ and $\supportof{\sigma}$ are disjoint from
$\overline{\mathbb{N}}$; we will explain later how this condition can
be obviated.

We claim that $\theta$ and $\sigma$ identify closed instances of all
the premise sequents under the above assumption.
Such a sequent must have the form 
$\seq[\mathbb{N}']{\Psi}{\Xi'}{\Omega}{F'}$
where $\mathbb{N}' \subseteq \overline{\mathbb{N}} \cup \mathbb{N}$
and, by Lemma~\ref{lem:typdecomp}, each
$\ctxvarty{\Gamma_i}{\mathbb{N}'_i}{\ctxty{\mathcal{C}_i}{\mathcal{G}_i}}
\in \Xi'$ corresponds to a 
$\ctxvarty{\Gamma_i}{\mathbb{N}_i}{\ctxty{\mathcal{C}_i}{\mathcal{G}_i}}
\in \Xi$
with
$\mathbb{N}'\setminus\mathbb{N} \supseteq \mathbb{N}'_i\setminus\mathbb{N}_i$.
Since $\ctxsanstype{\Xi'} = \ctxsanstype{\Xi}$, we only need to show
that $\theta$ and $\sigma$ are proper with respect to this sequent to
be sure they identify a closed instance of it.
The compatibility of $\langle \theta, \emptyset\rangle$ with the
conclusion sequent implies that $\supportof{\theta} \cap \mathbb{N} =
\emptyset$. 
Since we also have $\supportof{\theta}\cap
\overline{\mathbb{N}}=\emptyset$, it follows that $\supportof{\theta}$ 
is disjoint from the support set of the shown sequent and 
hence $\langle \theta,\emptyset \rangle$ must be substitution
compatible with it. 
Using Theorems~\ref{th:meta-wf} and \ref{th:seq-term-subs-ok}, we see that
the context types in the context variables context that results from
applying $\theta$ to the sequent must be well-formed relative to it. 
Since the application of a term variables substitution does not remove
any entries from a context variables context, there must be an
association in $\hsubst{\theta}{\Xi'}$ for every variable in
$\domain{\sigma}$. 
To determine that $\sigma$ is appropriate for the sequent, it only 
remains to be shown that there is a derivation for
$\ctxtyinst{\supportof{\sigma} \setminus \mathbb{N}'_i}
           {\emptyset}
           {\emptyset}
           {\ctxty{\mathcal{C}_i}{\hsubst{\theta}{\mathcal{G}_i}}}{G_i}$
for each context expression $G_i$ that $\sigma$ substitutes for a
variable $\Gamma_i$. 
The appropriateness of $\sigma$ for the corresponding instance of
the conclusion sequent tells us that
$\ctxtyinst{\supportof{\sigma} \setminus \mathbb{N}_i}
           {\emptyset}
           {\emptyset}
           {\ctxty{\mathcal{C}_i}{\hsubst{\theta}{\mathcal{G}_i}}}{G_i}$
has a derivation.
We would get the desired conclusion from this if
$\supportof{\sigma} \setminus \mathbb{N}'_i = \supportof{\sigma}
\setminus \mathbb{N}_i$ 
But this must be the case: since $\mathbb{N}'\setminus\mathbb{N}
\supseteq\mathbb{N}'_i\setminus\mathbb{N}_i$  it follows that
$\overline{\mathbb{N}} \supseteq \mathbb{N}'_i\setminus\mathbb{N}_i$
and we know that $\supportof{\sigma} \cap \overline{\mathbb{N}} =
\emptyset$.

We turn now to the main task, that of showing that the instance of the
conclusion sequent is valid respect to the height assignment $\Upsilon$.
For this, it suffices to show that if all the formulas in
$\subst{\sigma}{\hsubst{\theta}{\Omega\cup\{\fatmann{G}{M:A}{Ann}\}}}$
are valid with respect to $\Upsilon$, then
$\subst{\sigma}{\hsubst{\theta}{\fatmann{G,n:B}{M:A}{Ann}}}$ is also
so valid.
Using Theorem~\ref{th:weakening}, the interpretation of atomic
formulas, and the fact that $\sigma$ cannot introduce a binding for
$n$ into $G$ if it is appropriate for the instance of the conclusion
sequent resulting from applying $\theta$, we see that this would be
the case if
$\lftype{\subst{\sigma}{\hsubst{\theta}{G}}}{\hsubst{\theta}{B}}$ 
has a derivation.
To show that this is so, we use Lemma~\ref{lem:typdecomp}.
The wellformedness of the conclusion sequent ensures that the
parameters of the call to \typdecompsans\ satisfy the requirements of
the lemma.
The substitutions $\theta$ and $\sigma$ also meet the conditions
required of them by clause 2: they are closed substitutions with
domains identical to $\Psi$ and $\ctxsanstype{\Xi}$, respectively,
$\subst{\sigma}{\hsubst{\theta}{G}}$ and $\hsubst{\theta}{B}$ must
obviously be defined, and
$\lfctx{\subst{\sigma}{\hsubst{\theta}{G}}}$ must have a derivation
because $\subst{\sigma}{\hsubst{\theta}{\fatmann{G}{M:A}{Ann}}}$ is
valid by assumption. 
Thus, the lemma would allow us to conclude that
$\lftype{\subst{\sigma}{\hsubst{\theta}{G}}}{\hsubst{\theta}{B}}$
has a derivation if $\subst{\sigma}{\hsubst{\theta}{F'}}$ is valid
for every $\langle \mathbb{N}',\Xi',F'\rangle \in
\typdecomp{\mathbb{N}}{\Xi}{G}{B}$.
But this must be the case.
These are the goal formulas of closed instances of the premise sequents
that we have assumed to be valid.
Further, the assumption formulas of these sequents are
$\subst{\sigma}{\hsubst{\theta}{\Omega}}$,
which are valid with respect to $\Upsilon$ by assumption.
It follows that each $\subst{\sigma}{\hsubst{\theta}{F'}}$ must be
valid with respect to $\Upsilon$.
Since these formulas do not contain annotations, it is easy to see
that they must be valid without qualification. 

It remains only to dispense with the assumption concerning
$\supportof{\theta}$ and $\supportof{\sigma}$.
If these substitutions do not satisfy the assumption at the outset,
then we may consider their variants $\theta'$ and $\sigma'$ that are
obtained through a permutation that swaps
nominal constants in the set $\overline{\mathbb{N}}$ with ones that do
not appear in
$\overline{\mathbb{N}}\cup\mathbb{N} \cup
    \supportof{\theta} \cup \supportof{\sigma}$.
This permutation will leave the conclusion sequent unaffected and will
ensure that $\supportof{\theta'}$ and 
$\supportof{\sigma'}$ are disjoint from $\overline{\mathbb{N}}$.
We may then use the earlier argument with respect to $\theta'$ and
$\sigma'$ and conclude the proof by invoking
Theorem~\ref{th:perm-subst} and a version of
Theorem~\ref{th:perm-valid} relativized to validity with respect to
height assignments that can be easily shown to be true. 
\end{proof}

\section{An Example Development: Uniqueness of Typing}
\label{sec:examples}

We illustrate the proof system that has been developed in
Section~\ref{sec:proof-system} by showing how a formal proof can be
provided within it for the type uniqueness property of the STLC, the
running example in this paper.
We have discussed how this property can be expressed within \logic\  
in Section~\ref{ssec:logic-examples}.
To recall, it the situation in which the LF signature
parameterizing \logic\ is the one presented in
Figure~\ref{fig:stlc-term-spec} and $c$ denotes the context schema
comprising the single block
\begin{tabbing}
  \qquad\=\kill
  \> $\{t : o\}x:tm,y:\ofty\app x\app t$,
\end{tabbing}
the formula expressing the property of interest is the following:
\begin{tabbing}
\qquad\=\qquad\=\qquad\=\kill
\>$\fctx{\Gamma}{c}{\fall{e:\oty}{\fall{t_1:\oty}{\fall{t_2:\oty}{\fall{d_1:\oty}{\fall{d_2:\oty}{}}}}}}$\\
\>\>$\fimp{\fatm{\Gamma}{e:\tmty}}
          {\fimp{\fatm{\Gamma}{t_1:\tpty}}
                {\fimp{\fatm{\Gamma}{t_2:\tpty}}{}}}$\\
\>\>\> $\fimp{\fatm{\Gamma}{d_1:\ofty\app e\app t_1}}
             {\fimp{\fatm{\Gamma}{d_2:\ofty\app e\app t_2}} 
                   {\fexists{d_3:\oty}
                            {\fatm{.}{d_3:\eqty\app t_1\app t_2}}}}$.
\end{tabbing}
Following the earlier discussion, we will actually construct a proof
for the sequent $\seq[\emptyset]{\emptyset}{\emptyset}{\emptyset}{F}$
where $F$ is the formula  
\begin{tabbing}
\qquad\=\qquad\=\qquad\=\kill
\>$\fctx{\Gamma}{c}{\fall{e:\oty}{\fall{t_1:\oty}{\fall{t_2:\oty}{\fall{d_1:\oty}{\fall{d_2:\oty}{}}}}}}$\\
\>\>$\fimp{\fatm{\Gamma}{e:\tmty}}
          {\fimp{\fatm{\Gamma}{t_1:\tpty}}
                {\fimp{\fatm{\Gamma}{t_2:\tpty}}{}}}$\\
\>\>\> $\fimp{\fatm{\Gamma}{d_1:\ofty\app e\app t_1}}
             {\fimp{\fatm{\Gamma}{d_2:\ofty\app e\app t_2}} 
                   {\fexists{d_3:\oty}
                            {\fatm{\Gamma}{d_3:\eqty\app t_1\app t_2}}}}$.
\end{tabbing}
We can then obtain a proof from the sequent representing the formula
that we really want to prove using a proof of the strengthening lemma
for equality of types identified in Section~\ref{ssec:logic-examples}.
We will not discuss a proof of the strengthening lemma or how exactly
the use of lemmas plays out within our proof system.
Suffice it to say that the former should hold no mysteries after we
have discussed the proof of the sequent
$\seq[\emptyset]{\emptyset}{\emptyset}{\emptyset}{F}$ and the latter
involves the use of the \cut\ rule in the manner that we would
expect. 
Also, our discussion of proofs will be in a ``backwards,''
proof-construction style: given a target or goal sequent, we will
identify a collection of new goal sequents from whose proofs we can
construct a proof of the original one by using rules of the proof
system. 

To get to the details, the informal proof of the formula $F$ was based
on an induction of the height of the typing derivation associated with
the formula $\fatm{\Gamma}{d_1 : \ofty\app e\app t_1}$ that appears as
an antecedent of an implication in $F$.
We can mimic this process in the formal proof through the use of the
\ind\ rule, that would yield as a new goal the task of constructing a
proof for the sequent
$\seq[\emptyset]{\emptyset}{\emptyset}{\{\ltann[1]{F}\}}{\eqann[1]{F}}$;
we use the notation $\ltann[1]{F}$ and $\eqann[1]{F}$ to denote
formulas that are identical to $F$ except that the occurrence of the
formula $\fatm{\Gamma}{d_1 : \ofty\app e\app t_1}$ within them is
annotated with a $*$ and an $@$, respectively.
We may now think of deriving this sequent by using a sequence of
rules for introducing logical symbols into its goal formula, leaving
us with the task of constructing a proof for the following sequent:
\begin{tabbing}
\quad\=\hspace{10cm}\=\kill
\>$\emptyctx;
(e:\oty, t_1:\oty, t_2:\oty, d_1:\oty, d_2:\oty);
\ctxvarty{\Gamma}{\emptyset}{\ctxty{c}{\cdot}};$\\
\>$\{F^*, \fatm{\Gamma}{e : \tmty}, \fatm{\Gamma}{t_1 : \tpty}, \fatm{\Gamma}{t_2 : \tpty}, 
      \eqann{\fatm{\Gamma}{d_1:\ofty\app e\app t_1}}, 
      \fatm{\Gamma}{d_2:\ofty\app e\app t_2}\}$\\ 
\>\>$\longrightarrow
\fexists{d_3:\oty}{\fatm{\Gamma}{d_3:\eqty\app t_1\app t_2}}$
\end{tabbing}

At this point in the informal proof we used case analysis to identify
the different ways in which the formula 
$\fatm{\Gamma}{d_1:\ofty\app e\app t_1}$ may have been derived.
This effect is realized in our proof system by considering the
application of the $\appL$ rule with respect to the assumption formula
$\eqann{\fatm{\Gamma}{d_1:\ofty\app e\app t_1}}$.
Doing so leads to the conclusion that we can complete the proof if we
can construct proofs for the following four sequents:
\begin{enumerate}
\item \label{item1}
\begin{tabbing}
\=\qquad\ \=\hspace{8cm}\=\kill
\>$(n:\oty, n_1:\oty, n_2:\oty);
(t_2:\oty, d_2:\oty, r:\arr{\oty}{\oty}, t:\oty, u:\oty, a:\arr{\oty}{\arr{\oty}{\oty}});
\ctxvarty{\Gamma}{\emptyset}{\ctxty{c}{\cdot}};$\\
\>$\{\ltann[1]{F}, \fatm{\Gamma}{\lamtm\app t\app (\lflam{x}{r\app x}) : \tmty}, \fatm{\Gamma}{\arrtm\app t\app u : \tpty}, \fatm{\Gamma}{t_2 : \tpty}, $\\
\>\> $\fatm{\Gamma}{d_2:\ofty\app (\lamtm\app t\app (\lflam{x}{r\app x}))\app t_2}, 
      \ltann{\fatm{G,n:\tmty}{r\app n : \tmty}}, 
      \ltann{\fatm{\Gamma}{t : \tpty}}, $\\
\>\> $\ltann{\fatm{\Gamma}{u : \tpty}}, \ltann{\fatm{G,n_1:\tmty, n_2:\ofty\app n_1\app t}{a\app n_1\app n_2 : \ofty\app(r\app n_1)\app u}}\}$\\ 
\>\>\>$\longrightarrow
\fexists{d_3:\oty}{\fatm{\Gamma}{d_3:\eqty\app (\arrtm\app t\app u)\app t_2}}$
\end{tabbing}
\item \label{item2}
\begin{tabbing}
\=\qquad\ \=\hspace{8cm}\=\kill
\>$\emptyctx;
(t_1:\oty, t_2:\oty, d_2:\oty, m:\oty, n:\oty, u:\oty, a_1:\oty, a_2:\oty);
\ctxvarty{\Gamma}{\emptyset}{\ctxty{c}{\cdot}};$\\
\>$\{\ltann[1]{F}, \fatm{\Gamma}{\apptm\app m\app n : \tmty}, \fatm{\Gamma}{t_1 : \tpty}, \fatm{\Gamma}{t_2 : \tpty},$\\
\>\> $\fatm{\Gamma}{d_2:\ofty\app (\apptm\app m\app n)\app t_2}, 
      \ltann{\fatm{\Gamma}{m : \tmty}}, \ltann{\fatm{\Gamma}{n : \tmty}}, 
      \ltann{\fatm{\Gamma}{t_1 : \tpty}},$\\
\>\> $\ltann{\fatm{\Gamma}{u : \tpty}},
      \ltann{\fatm{\Gamma}{a_1 : \ofty\app M\app(\arrtm\app u\app t_1)}}, 
      \ltann{\fatm{\Gamma}{a_2 : \ofty\app n\app u}}\}$\\ 
\>\>\>$\longrightarrow
\fexists{d_3:\oty}{\fatm{\Gamma}{d_3:\eqty\app t_1\app t_2}}$
\end{tabbing}
\item \label{item3}
\begin{tabbing}
\=\hspace{9cm}\=\kill
\>$\emptyctx;
(t_2:\oty, d_2:\oty);
\ctxvarty{\Gamma}{\emptyset}{\ctxty{c}{}};$\\
\>$\{\ltann[1]{F}, \fatm{\Gamma}{\emptytm : \tmty}, \fatm{\Gamma}{\unittm : \tpty}, \fatm{\Gamma}{t_2 : \tpty}, 
      \fatm{\Gamma}{d_2:\ofty\app \emptytm\app t_2}\}$\\ 
\>\>$\longrightarrow
\fexists{d_3:\oty}{\fatm{\Gamma}{d_3:\eqty\app \unittm\app t_2}}$
\end{tabbing}
\item \label{item4}
\begin{tabbing}
\=\hspace{8cm}\=\kill
\>$(n:\oty,n_1:\oty);
(t_1:\arr{\oty}{\arr{\oty}{\oty}}, t_2:\arr{\oty}{\arr{\oty}{\oty}}, d_2:\arr{\oty}{\arr{\oty}{\oty}});$\\
\>$\ctxvarty{\Gamma}{\emptyset}{\ctxty{c}{n:\tmty,n_1:\ofty\app n\app (t_1\app n\app n_1)}};$\\
\>$\{\ltann[1]{F}, \fatm{\Gamma}{n : \tmty}, \fatm{\Gamma}{t_1\app n\app n_1 : \tpty}, \fatm{\Gamma}{t_2\app n\app n_1 : \tpty}, 
      \fatm{\Gamma}{d_2\app n\app n_1:\ofty\app n\app (t_2\app n\app n_1)}\}$\\ 
\>\>$\longrightarrow
\fexists{d_3:\oty}{\fatm{\Gamma}{d_3:\eqty\app (t_1\app n\app n_1)\app (t_2\app n\app n_1)}}$
\end{tabbing}
\end{enumerate} 
The first of these sequents covers all of the cases where the head of
$d_1$ is the constant $\oflamtm$, the second where it is 
$\ofapptm$, the third where it is $\ofemptytm$, 
and the fourth where it is a bound variable that appears in the
context that instantiates $\Gamma$. 

In the first of these cases, the informal proof was based on considering
the derivation of the main LF typing judgement associated with 
$\fatm{\Gamma}{d_2:\ofty\app(\lamtm\app t\app (\lflam{x}{r\app
    x}))\app t_2}$.
Given the structure of the type, such a derivation will exist only when
the head of $d_2$ is the constant $\oflamtm$.
In the proof system, this analysis is realized by considering the use
of the \appL\ rule relative to this formula, leading to the conclusion
that we could complete a proof for this subcase if we are able to
provide a proof for the sequent 
\begin{tabbing}
\quad\=\qquad\ \=\hspace{7.5cm}\=\kill
\>$(n:\oty, n_1:\oty, n_2:\oty, n_3:\oty, n_4:\oty, n_5:\oty);$\\
\>$(r:\arr{\oty}{\oty}, t:\oty, u:\oty, a:\arr{\oty}{\arr{\oty}{\oty}}, u_2:\oty, a_2:\arr{\oty}{\arr{\oty}{\oty}});
\ctxvarty{\Gamma}{\emptyset}{\ctxty{c}{\cdot}};$\\
\>$\{\ltann[1]{F}, \fatm{\Gamma}{\lamtm\app t\app (\lflam{x}{r\app x}) : \tmty}, 
         \fatm{\Gamma}{\arrtm\app t\app u : \tpty}, \fatm{\Gamma}{\arrtm\app t\app u_2 : \tpty},
         \ltann{\fatm{\Gamma,n:\tmty}{r\app n : \tmty}},$\\
\>\> $\ltann{\fatm{\Gamma}{t : \tpty}}, \ltann{\fatm{\Gamma}{u : \tpty}}, 
      \ltann{\fatm{\Gamma,n_1:\tmty, n_2:\ofty\app n_1\app t}{a\app n_1\app n_2 : \ofty\app(r\app n_1)\app u}},$\\
\>\> $\fatm{\Gamma,n_3:\tmty}{r\app n_3 : \tmty}, 
      \fatm{\Gamma}{t : \tpty}, \fatm{\Gamma}{u_2 : \tpty},$\\ 
\>\> $\fatm{\Gamma,n_4:\tmty, n_5:\ofty\app n_1\app t}{a_2\app n_4\app n_5 : \ofty\app(r\app n_4)\app u_2}\}$\\
\>\>\>$\longrightarrow
\fexists{d_3:\oty}{\fatm{\Gamma}{d_3:\eqty\app (\arrtm\app t\app u)\app (\arrtm\app t\app u_2)}}$
\end{tabbing}

The informal proof now uses the induction hypothesis based on the
strictly smaller derivation that is assumed to exist for the main
typing judgement associated with the assumption formula
$\fatm{\Gamma,n_1:\tmty, n_2:\ofty\app n_1\app t}
      {a\app n_1\app n_2 : \ofty\app(r\app n_1)\app u}$.
In the formal proof, this effect is realized by using the formula
$\ltann[1]{F}$ together with the other assumption formulas in the
sequent.
Specifically, we would consider a sequence of uses of the
$\ctxL$ and $\allL$, $\impL$,\id\ and \weakening\ rules that would
leave us needing to prove a sequent in which the consequent of a
relevant instance of the formula $F$ has been added to the set of
assumption formulas.
To actually carry out this process, it would be necessary to check
that the context expression
$\Gamma,n_1:\tmty, n_2:\ofty\app n_1\app t$ has the structure needed
to satisfy the context schema represented by $c$, and we would need to
instantiate the quantified variable $e$ in $F$ with $(r\app n_1)$,
$t_1$ with $u_1$, $t_2$ with $u_2$, $d_1$ with $(a\app n_1\app n_2)$,
and $d_2$ with $(a_2\app n_1\app n_2)$.
We leave it to the reader to check that these requirements can be
satisfied and to elaborate the process by which the task can be
reduced to that of constructing a proof for the sequent
\begin{tabbing}
\quad\=\qquad\ \=\hspace{7.5cm}\=\kill
\>$(n:\oty, n_1:\oty, n_2:\oty, n_3:\oty, n_4:\oty, n_5:\oty);$\\
\>$(r:\arr{\oty}{\oty}, t:\oty, u:\oty, a:\arr{\oty}{\arr{\oty}{\oty}}, u_2:\oty, a_2:\arr{\oty}{\arr{\oty}{\oty}}, d_3:\arr{\oty}{\arr{\oty}{\arr{\oty}{\arr{\oty}{\arr{\oty}{\arr{\oty}{\oty}}}}}});$\\
\>$\ctxvarty{\Gamma}{\emptyset}{\ctxty{c}{\cdot}};$\\
\>$\{\ltann[1]{F}, \fatm{\Gamma}{\lamtm\app t\app (\lflam{x}{r\app x}) : \tmty}, 
         \fatm{\Gamma}{\arrtm\app t\app u : \tpty}, \fatm{\Gamma}{\arrtm\app t\app u_2 : \tpty},
         \ltann{\fatm{\Gamma,n:\tmty}{r\app n : \tmty}},$\\
\>\> $\ltann{\fatm{\Gamma}{t : \tpty}}, 
      \ltann{\fatm{\Gamma}{u : \tpty}}, 
      \ltann{\fatm{\Gamma,n_1:\tmty, n_2:\ofty\app n_1\app t}{a\app n_1\app n_2 : \ofty\app(r\app n_1)\app u}},$\\
\>\> $\fatm{\Gamma,n_3:\tmty}{r\app n_3 : \tmty}, 
      \fatm{\Gamma}{t : \tpty}, \fatm{\Gamma}{u_2 : \tpty}, $\\
\>\> $\fatm{\Gamma,n_4:\tmty, n_5:\ofty\app n_1\app t}{a_2\app n_4\app n_5 : \ofty\app(r\app n_4)\app u_2},$\\
\>\> $\fatm{\Gamma,n_1:\tmty,n_2:\ofty\app n_1\app t}
           {d_3\app n\app n_1\app n_2\app n_3\app n_4\app n_5:
               \eqty\app u\app u_2}\}$\\ 
\>\>\>$\longrightarrow
\fexists{d_3:\oty}{\fatm{\Gamma}{d_3:\eqty\app (\arrtm\app t\app u)\app (\arrtm\app t\app u_2)}}$
\end{tabbing}

The key observation now is that the main typing judgement
corresponding to the new assumption formula in this sequent is
derivable only of the head of $d_3$ is an instance of $\refltm$ and
$u$ and $u_2$ are instantiated with identical types. 
This reasoning step is reflected in the formal proof in the
employment, again, of the $\appL$ with respect to this assumption
formula, which results in the conclusion that it suffices to construct
a proof for the sequent
\begin{tabbing}
\quad\=\qquad\ \=\hspace{7.5cm}\=\kill
\>$(n:\oty, n_1:\oty, n_2:\oty, n_3:\oty, n_4:\oty, n_5:\oty);$\\
\>$(r:\arr{\oty}{\oty}, t:\oty, u:\oty, a:\arr{\oty}{\arr{\oty}{\oty}}, a_2:\arr{\oty}{\arr{\oty}{\oty}});
   \ctxvarty{\Gamma}{\emptyset}{\ctxty{c}{\cdot}};$\\
\>$\{\ltann[1]{F}, \fatm{\Gamma}{\lamtm\app t\app (\lflam{x}{r\app x}) : \tmty},
         \fatm{\Gamma}{\arrtm\app t\app u : \tpty}, \fatm{\Gamma}{\arrtm\app t\app u_2 : \tpty},
         \ltann{\fatm{\Gamma,n:\tmty}{r\app n : \tmty}},$\\
\>\> $\ltann{\fatm{\Gamma}{t : \tpty}}, 
      \ltann{\fatm{\Gamma}{u : \tpty}}, 
      \ltann{\fatm{\Gamma,n_1:\tmty, n_2:\ofty\app n_1\app t}{a\app n_1\app n_2 : \ofty\app(r\app n_1)\app u}},$\\
\>\> $\fatm{\Gamma,n_3:\tmty}{r\app n_3 : \tmty},
      \fatm{\Gamma}{t : \tpty}, \fatm{\Gamma}{u : \tpty},$\\
\>\> $\fatm{\Gamma,n_4:\tmty, n_5:\ofty\app n_1\app t}{a_2\app n_4\app n_5 : \ofty\app(r\app n_4)\app u},
      \fatm{\Gamma,n_1:\tmty,n_2:\ofty\app n_1\app t}{u :\tpty}\}$\\ 
\>\>\>$\longrightarrow
\fexists{d_3:\oty}{\fatm{\Gamma}{d_3:\eqty\app (\arrtm\app t\app u)\app (\arrtm\app t\app u)}}$
\end{tabbing}
It is easy to see that this sequent would be derivable if the goal
formula in it is replaced with the one obtained by instantiating the
existentially quantified variable $d_3$ with the term
$(\refltm\app(\arrtm\app t\app u))$. 
However, that suffices to complete the proof of this subcase: we
simply use the $\existsR$ below that derivation to get one for the
sequent that is really of interest. 

The analysis in the three remaining cases all begin with a step
similar to that in the case just considered, that of analyzing the 
assumption formula corresponding to the other posited STLC typing
derivation for the term.
Since the first derivation has fixed the top level structure of the
term, this analysis will identify only one possibility, i.e., that the
final step in that derivation is the same as that in the analysis of
the first derivation. 
In the case corresponding to sequent identified by the
label~\ref{item2}, the argument will again employ the inductive
hypothesis by identifying a suitable instance of the formula
$\ltann[1]{F}$. 
In each of the cases, the conclusion of the proof for that case will
involve the instantiation of the existential goal, and decomposition
of the resulting atomic goal formula.
We leave it to the reader to fill out the details. 

\section{Related Work and Conclusion}\label{sec:conclusion}

We have described a logic called \logic\ in this paper that provides
a means for formalizing properties of LF specifications. 
The atomic formulas in this logic encode typing judgements relative to
an LF signature that parameterizes the logic.
Propositional connectives can be used over these formulas to describe
more complex properties and the logic also supports quantification
over variables denoting LF contexts and terms.
We have complemented the logic by presenting a proof system that can
be used to establish the validity of formulas that express properties
of the specified systems.
In addition to providing rules that interpret the logical symbols, the
proof system builds in the capability to carry out a case-analysis
style reasoning over typing judgements in LF and to reason inductively
based on the heights of LF derivations.
The logic also enables the encoding of generic mechanisms for
reasoning about LF derivability, a fact that we have illustrated by
describing proof rules that embody a collection of commonly used
meta-theorems about LF.
In work that builds on the results of this paper, we have developed a
proof assistant called Adelfa that provides mechanized support for the
proof system and that we have used to demonstrate its effectiveness in
reasoning~\cite{southern21phd,southern21lfmtp}.

There have been other efforts directed at building a capability for
reasoning about LF specifications.
However, the approaches taken within these efforts have differed
from the one that we have explored in this paper.
One prominent approach is that embedded in what might be called the
``Twelf family'' of systems.
The first realization of this approach is the Twelf system
itself~\cite{pfenning99cade}. 
This system uses LF once again to formalize properties of systems
described in it.
More specifically, the properties that are of interest are described
by other LF types that have a functional structure.
The validity of these properties is then demonstrated by constructing
inhabitants of the types that represent them and exhibiting the
totality of these inhabitants as functions via means that are external
to the encoding calculus.
This approach has achieved much success but it is also limited by the
fact that the property that need to be expressed must be transformed
into a form that is encodable by a function type, \ie, by formulas
that have a $\forall\exists$ quantifier structure.
Another drawback of the approach is that is that there isn't an
explicit proof to be extracted at the end of a development for a
property that has ostensibly been demonstrated.
The logic $M_2^+$ has been enunciated towards mitigating this
issue~\cite{schurmann00phd}, and it appears possible to mechanically 
relate the ``reasoning'' embodied in a Twelf development to a proof in
this logic~\cite{wang13lfmtp}.
While some of the proof rules in $M_2^+$ bear a resemblance to the
ones in \logic, there are significant differences in the specific
treatment especially of inductive reasoning.

The Twelf system and the $M_2^+$ logic differ in an another marked way
from the logic that we have described in this paper: they do not
provide an explicit means for quantifying over contexts. It is
possible to parameterize a development by a context description, but
it is one fixed context that then permeates the development.
As an example, it is not possible to express the strengthening lemma 
pertaining to the equality of types in the STLC that we discussed in
Section~\ref{ssec:logic-examples}.
The Beluga system~\cite{pientka10ijcar} alleviates this problem by using
a richer version of type theory that allows for an explicit treatment
of contexts as its basis~\cite{nanevski08tocl}.
Beluga is based on a computational view of reasoning that is similar
in many respects to the philosophy underlying Twelf: one writes
dependently-typed recursive functions and the type system ensures that
the admitted 
functions are ones that are total and hence embody ``proofs'' of the
properties expressed by the types. 
This system is more expressive than Twelf because the type system is
more expressive, but the structure of the formulas that encode
properties is similarly limited. 
A recent development that appears related to this line of work is that
of {\sc Cocon}, a Martin-L\"{o}f style type theory that embeds LF within
a rich dependently typed calculus that supports
recursion~\cite{pientka2019arxiv}. 
It would be interesting to compare the reasoning capabilities that
result from this kind of a combination with what is possible to
achieve with an approach like ours.
We leave a further exploration of this issue to future work.

Another approach that has been explored for reasoning about LF
specifications is based on their translation to a predicate logic
form.
A particular exemplar of this approach is one that translates LF
specifications into specifications in the logic of hereditary Harrop 
formulas~\cite{miller12proghol}, to then be reasoned about using the
Abella system~\cite{southern14fsttcs}.
This approach has the virtue that benefit can be derived from any of
the (generic) reasoning capabilities that have been developed for the
host system. 
In the mentioned example, several such advantages are derived from the
expressiveness of the logic underlying Abella~\cite{gacek11ic}: 
it is possible to define relations between contexts, to treat
binding notions explicitly in the reasoning process through the
$\nabla$-quantifier, and to use inductive (and co-inductive)
definitions in the reasoning logic.
There are, however, some drawbacks to the translation-based
approach. 
Perhaps the most significant problem is that the proof steps that can
be taken under it are determined by the logic of the host system, and
this may allow for more possibilities than are sensible in the LF
context.  
One way to overcome this difficulty is to design macro proof steps
that capture the natural process of reasoning about LF
specifications. 
In this respect, one important outcome of the work that we have
described here, especially for the structure that we have developed
for case analysis, might be an understanding of how one might
build within proof assistants such as Abella a targeted
capability for reasoning about LF specifications. 

\section*{Acknowledgements}

This paper is based upon work supported by the National Science
Foundation under Grant No. CCF-1617771.
Any opinions, findings, and conclusions or recommendations expressed
in this material are those of the authors and do not necessarily
reflect the views of the National Science Foundation.

\bibliographystyle{abbrv}
\bibliography{../../references/master}

\end{document}